\pdfoutput=1 

\documentclass[11pt]{article}
\usepackage{NotesTeX,lipsum}
\usepackage{tikz,xcolor}
\usepackage{karnaugh-map}
\usepackage[utf8]{inputenc}
\usepackage[english]{babel}
\usepackage{amssymb,amsmath,amsfonts,amsthm}
\usepackage{verbatim}
\usepackage{enumerate}
\usepackage{tensor}
\usepackage{bbm}
\usepackage{graphicx}
\newlength{\xxxscale}
\graphicspath{{apics/}}
\usepackage[caption=false]{subfig}

\newtheorem{observation}[theorem]{Observation}
\newcommand{\YX}{\mathcal{Y}\otimes\mathcal{X}}
\newcommand{\XY}{\mathcal{X}\otimes\mathcal{Y}}
\newcommand{\K}{{\mathbb K}}  
\newcommand{\x}{\mathbf{x}}
\newcommand{\y}{\mathbf{y}}
\newcommand{\R}{{\mathbb R}}  
\newcommand{\C}{{\mathbb C}}  
\newcommand{\hilb}[1]{\ensuremath{\mathcal{#1}}} 

\newtheoremstyle{No.}
{}
{}
{}
{}
{\bfseries}
{}
{ }
{\thmname{#1}\thmnumber{ #2}\thmnote{ (#3)}}
{
\theoremstyle{No.}
\newtheorem{illexample}[theorem]{Illustration}}
\newtheorem{corollary}[theorem]{Corollary}
\newtheorem{myexercise}[theorem]{Exercise}
\tcolorboxenvironment{illexample}{colback=pink!20!white, enhanced jigsaw,
colframe=pink!75!black,
before skip=10pt,
after skip=10pt,
breakable}
\tcolorboxenvironment{myexercise}{colback=green!5!white, enhanced jigsaw,
colframe=green!75!black,
before skip=10pt,
after skip=10pt,
breakable,}  
\newenvironment{myfullpage}
    {\smallskip\noindent\begin{minipage}
    {\textwidth+\marginparwidth+\marginparsep}\smallskip\smallskip}
    {\end{minipage}\vspace{.1in}}

\usepackage{tcolorbox}

\DeclareMathOperator{\Span}{span}

\newcommand{\Lx}[1]{{L(\mathcal{#1})}}
\def\1#1{{\bf #1}}
\def\2#1{{\cal #1}}
\def\7#1{{\mathbb #1}}
\newcommand{\I}{\mathbbm{1}} 
\newcommand{\bydef}{\stackrel{\mathrm{def}}{=}}
\newcommand{\dket}[1]{\mbox{$\left|\left.#1\right\rangle\!\right\rangle$}}
\newcommand{\dbra}[1]{\mbox{$\left\langle\!\left\langle #1\right.\right|$}}

\newcommand{\dketdbra}[2]{\mbox{$|#1\rangle\!\rangle\!\!\langle\!\langle #2|$}}

\newtheorem{proposition}[theorem]{Proposition}
\def\3#1{{\sl #1}}
\newcommand{\Tx}[1]{{T(\mathcal{#1})}}
\newcommand{\Cx}[1]{{C(\mathcal{#1})}}
\newcommand{\XX}{\mathcal{X}\otimes\mathcal{X}}
\newcommand{\YY}{\mathcal{Y}\otimes\mathcal{Y}}
\newcommand{\LXY}{L(\mathcal{X}\otimes\mathcal{Y})}
\newcommand{\XZ}{\mathcal{X}\otimes\mathcal{Z}}
\newcommand{\gate}[1]{\ensuremath{\text{\sf #1}}}
\newcommand{\COPY}[0]{\gate{COPY}}
\newcommand{\swap}{\gate{SWAP}}
\newcommand{\SWAP}{\gate{SWAP}}
\newcommand{\W}{\gate{W}}

\newcommand{\CNOT}{\gate{CNOT}}
\newcommand{\NOT}{\gate{NOT}}
\newcommand{\CN}{\gate{CN}}
\newcommand{\XOR}{\gate{XOR}}
\newcommand{\OR}{\gate{OR}}
\newcommand{\AND}{\gate{AND}}
\newcommand{\NAND}{\gate{NAND}}
\newcommand{\XNOR}{\gate{XNOR}}
\newcommand{\NOR}{\gate{NOR}}
\newcommand{\ANF}{\gate{ANF}}
\newcommand{\PPRM}{\gate{PPRM}}

\renewcommand{\H}{\gate{H}}

\newcommand{\GHZ}{\gate{GHZ}}

\newcommand{\inprod}[2]{\ensuremath{\left\langle #1, #2 \right\rangle}}

\newcommand{\isom}{\cong}  
\newcommand{\bv}{e}     
\newcommand{\spidx}[3]{^{(#1)}{}\indices{^{#2}_{#3}}}

\newcommand{\Eqref}[1]{Eq.~\eqref{#1}}

\newcommand{\be}{\begin{equation}}
\newcommand{\ee}{\end{equation}}

\title{\begin{center}{{\Huge{\bf{Lectures on Quantum Tensor Networks}}}}\\
{\itshape a pathway to modern diagrammatic reasoning}
\end{center}}
\author{\Large ~Jacob Biamonte{\footnote{{Version current as of \today.}\\~\\I would be grateful if you email me regarding any typos, errors or omissions you discover.\\~\\ Current version always available at \href{https://www.overleaf.com/read/jkccbhcdqwnh}{https://www.overleaf.com/read/jkccbhcdqwnh}}}
}

\affiliation{ 
Deep Quantum Laboratory\\
Skolkovo Institute of Science and Technology\\
3 Nobel Street\\
Moscow Russia 143026\\

e-mail: \href{j.biamonte@skoltech.ru}{j.biamonte@skoltech.ru}\\
web: \href{http://quantum.skoltech.ru}{http://quantum.skoltech.ru}

\vspace{15 mm}

\begin{center}
  \includegraphics[width=0.5\textwidth]{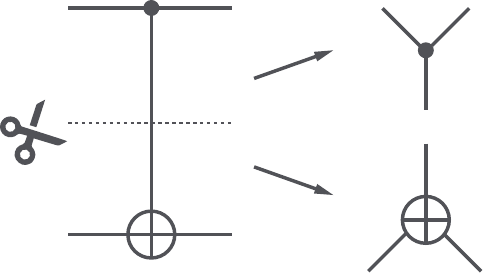}
\end{center}
}

\begin{document}
	\maketitle
	\newpage
	\pagestyle{fancynotes}
	\part*{Foreword}
Tensor network methods represent a collection of techniques to understand and reason about multi-linear maps which have found particular use in applications to quantum information processing.  These methods form the backbone of tensor network contraction algorithms to model physical systems and are used in the abstract graphical languages to represent channels, maps, states and processes appearing across quantum information science.  \\ 
	
In these chapters---which were complied based on years of teaching---we outline the  building blocks needed to understand the salient properties of tensor networks, the associated mathematical techniques and the diagrammatic reasoning language.   \\ 
	
The topic of tensor networks touches on a number of subjects yet the vast majority of writing is much more specific and often limited to be accessible by a narrow community. This book attempts to broadly cover the foundations of tensor network theory as it applies generally across quantum information.  \\ 
	
The aesthetically appealing development of tensor networks as a unifying language across quantum information science has long been close to my own research interests.  I have conducted research on quantum and classical circuits, as it applies to quantum computing.  This research included developing methods to embed logic functions into spin Hamiltonian ground states and the realization of quantum algorithms by quantum circuits.  The settings of both classical and quantum circuits comes with a well known graphical language.  \\

This provided a base to spend time merging ideas from (i) modern tensor networks as they appear in condensed matter physics; (ii) quantum circuits and their graphical language; (iii) aspects of categorical quantum mechanics as well as (iv) the graphical language of digital circuits to create a common notation and to develop and use rewrite rules that intersect these topics.  This book is intended to be self-contained, and accessible to graduate students.  It is hoped that advanced readers will let this book serve as a research reference.   \\

	Sincerely and happy reading, \\ 

~~~\includegraphics[width=0.3\textwidth]{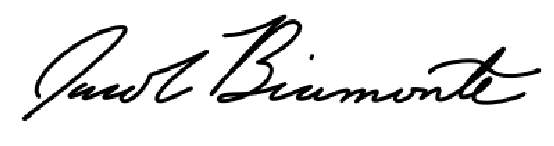} \\  

Jacob Biamonte---Moscow \hfill \today 

\newpage 

\part*{Hall of Fame}
I humbly tip my hat to the following readers.  These brave souls found and reported typos, errors or omissions, improving the book for all future readers. \\ 

\noindent Konstantin Antipin \\
Nick Decroos \\ 
Sergey Filippov \\ 
Aly Nasrallah \\ 
Miles Stoudenmire \\ 
Alireza Yazdi

\pagestyle{plain}

\part{From Tensors to Networks}
\label{sec:tensor}

Tensors are a mathematical concept that encapsulates and generalizes the idea of
multilinear maps, i.e.~functions of multiple parameters that are linear with respect to every parameter.
A tensor network is simply a countable collection of tensors connected by contractions.
{\bf Tensor network methods} is the term sometimes given to the entire collection of associated tools, which are regularly employed in modern quantum information science, condensed matter physics, mathematics and computer science.  

Tensor networks come with an intuitive graphical language that can be used in formal reasoning and in proofs. 
This diagrammatic language found applications in physics at least as early as the 1970s by Roger Penrose~\cite{Penrose}.
Tensor network theory has recently seen many advancements and adaptations to different domains of physics, mathematics and computer science.
An important milestone was David Deutsch's use of the diagrammatic notation in quantum computing,
developing the \emph{quantum circuit} (a.k.a.~quantum computational networks as Deutsch would call them)
model~\cite{Deutsch73}.
Quantum circuits are a special class of tensor networks, in which the arrangement of the tensors and their types are restricted.
A related diagrammatic language slightly before that is due to Richard Feynman~\cite{Fey86}.
The quantum circuit model---now well over two decades old---is widely used
to describe quantum algorithms and their experimental implementations,
to quantify the resources they use (by e.g.~counting the quantum gates required),
to classify the entangling properties and computational power of specific gate families, and more.

There is growing excitement concerning numerical algorithms that 
preform tensor contractions.  These algorithms are important in condensed matter physics and beyond.  There are many reviews and surveys devoted to this important direction---see~\cite{2014AnPhy.349..117O, Vidal2010, MPSreview08, TNSreview09, 2011AnPhy.326...96S, 2010arXiv1006.0675S, Schollw, 2014EPJB...87..280O,2013arXiv1308.3318E,2011JSP...145..891E, 2016arXiv160303039B, MAL-059, Pervishko_2019, MAL-067,2017arXiv170809213R}, as well as, {\it Tensor Networks in a Nutshell}, which I wrote with Ville Bergholm \cite{biamonte2017tensor}. 
Some of the best known applications of tensor networks are 1D Matrix Product States (MPS),  Tensor Trains (TT)~\cite{Oseledets2011}, Tree Tensor Networks (TTN),
the Multi-scale Entanglement Renormalization Ansatz (MERA), Projected Entangled Pair States (PEPS)---which generalize matrix product states to higher dimensions---and various other renormalization methods~\cite{Vidal2010, MPSreview08, TNSreview09, 2011AnPhy.326...96S, 2009PhRvL.102e7202H, 2013arXiv1308.3318E, MAL-059}.  
The excitement is based on the fact that certain
classes of quantum systems can now be simulated more efficiently, studied in greater detail,
and this has opened new avenues for a greater understanding of certain physical systems.  The concept is to factor a quantum state ~$\psi$ into various network structures, as follows.  

\begin{center}
    \includegraphics{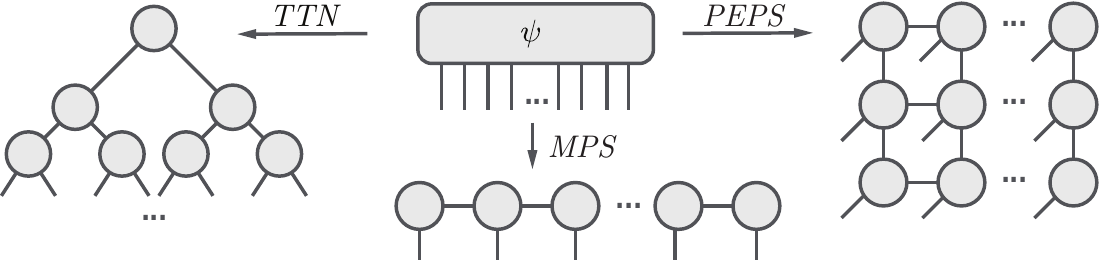} 
\end{center}

These methods approximate a complicated quantum state using a tensor network with a simplistic, regular structure---essentially applying lossy data compression that preserves the most salient properties of the quantum state.  The method is known to efficiently work for certain classes of ground and thermal states. 

We assume that most readers will have a basic understanding of some quantum theory, linear algebra and tensors.

\section{Penrose Graphical Notation for Tensor Networks} 

We will present a variant of the graphical notation used by Penrose \cite{Penrose, Penrose67, road}.  This book presents the modern incarnation of these ideas, building on four ingredients: (i) modern tensor networks as they appear in condensed matter physics; (ii) quantum circuits and their graphical language \cite{NC}; (iii) aspects of categorical quantum mechanics  \cite{Coecke2017} as they describe quantum circuits \cite{CD, redgreen} as well as (iv) the graphical language of digital circuits.  

The output of this merger is an increasingly popular collection of ideas related to the application of tensor networks to quantum information and quantum computation following largely \cite{CTNS, BB11}.  The notation matches quantum circuit notation and the presentation should hopefully be approachable for a wide audience of modern quantum information scientists.  Indeed, the techniques do differ from any of the respective ingredients we have mixed together; so any of the above listed communities---(i), (ii), (iii), (iv)---should go away after reading these notes with new techniques.  

\begin{paragraph}{Tensor Network}
Tensors
\begin{marginfigure}
\centering
\includegraphics[width=5\xxxscale]{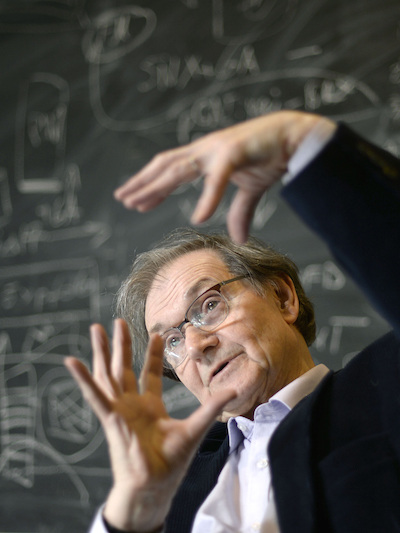}
\caption{``\textit{It now ceases to be important to maintain a distinction 
between upper and lower indices.}''
{-- Roger Penrose, 1971 \cite{Penrose}.}}
\end{marginfigure}
can be thought of as indexed multi-dimensional arrays of complex numbers with respect to a fixed standard basis.

\begin{definition}
 Let $\2 X, \2 Y, \2 Z$ be finite-dimensional complex Hilbert spaces,  ${\mathcal L}(\2X,\2Y)$ is the space of bounded linear operators $A: \2 X\rightarrow \2 Y$ with ${\mathcal L}(\2X)\equiv {\mathcal L}(\2X,\2X)$.
\end{definition}
For example, consider the Hilbert space $\2 X\cong \C^d$, where as is typical in modern quantum theory, we choose our standard basis to be the \emph{computational basis}
$$ 
\{\ket{i}\,:\, i=0,...,d-1\}.
$$ 
Then in Dirac notation a vector $\ket{v}\in \2 X$ is a 1st-order tensor which can be expressed in terms of its tensor components $v_i:=\braket{i}{v}$ with respect to the standard basis as $\ket{v}=\sum_{i=0}^{d-1} v_i\ket{i}$. Similarly one can represent linear operators on this Hilbert space, $A\in {\mathcal L}(\2 X)$,  as 2nd-order tensors with components $A_{ij}:= \bra{i}A\ket{j}$ as $A = \sum_{i,j=0}^{d-1} A_{ij} \ketbra{i}{j}$.

Hence, in Dirac notation the number of indices of a tensors components more define what we called a tensors \emph{order}. Vectors $\ket{v}\in \2 X$ refer to tensors which only have \emph{ket} ``$\ket{i}$'' basis elements, vectors in the dual vector space $\bra{u}\in \2 X^\dagger$ (or more typically denoted $\2X^\star$ or $\2X^*$) refer to those with only \emph{bras} ``$\bra{i}$'', and linear operators on $A\in \2 {\mathcal L}(\2 X)$ refer to tensors with a mixture of kets and bras in their component decomposition.

\begin{remark}
Like Penrose, we use the word \emph{valence} or \emph{order} instead of \emph{rank} when referring to the number of indices on a tensor, since rank is used elsewhere.  A tensor with $n$ indices up and $m$ down is called a valence-$(n,m)$ tensor and sometimes a valence-$k$ tensor for $k=n+m$. 
\end{remark}

\begin{remark}
The concurrent evaluation of all indices returns a complex number.  This is called total or full contraction.
\end{remark}

\begin{remark}
The idea of representing quantum states, operators and maps (etc.) diagrammatically is credited to works by Penrose and is sometimes referred to as \textit{Penrose graphical notation} or \textit{string diagrams}.  Though Penrose unquestionably pioneered many of the applications and uses of the language and deserves credit, Arthur Cayley developed much earlier variants of graphical languages. 
\end{remark}

We mostly adhere to Penrose's notation of representing states (vectors) and effects (dual-vectors) as triangles, linear operators as boxes, and scalars as diamonds, as found in Illustration~\ref{fig:ketbramat}. Here each index corresponds to an open wire on the diagram and so we may define higher order tensors with increasingly more wires. The number of wires is then the order of the tensor, with each wire acting on a separate vector space $\2 X_j$.

\begin{myfullpage}
\begin{illexample}[Graphical depiction of elementary tensors]
\label{sfig}
 Non-zero scalars (d) are also represented as `blank' on the page.  We represent vectors (states) and dual-vectors (effects) as triangles, linear operators as boxes, and scalars as diamonds, with each index of the tensor depicted as an open wire on the diagram. The orientation of the wires determines the type of tensor, in our convention the open end of the wires point to the left for vectors, right for dual-vectors, and both left and right for linear operators.
 \begin{center} 
\begin{tabular}{cccc}
\includegraphics[width=0.11\textwidth]{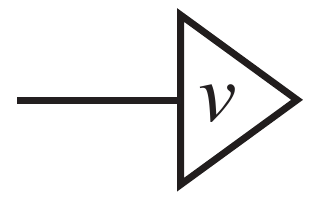} &
\includegraphics[width=0.11\textwidth]{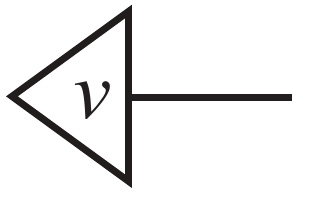} &
\includegraphics[width=0.145\textwidth]{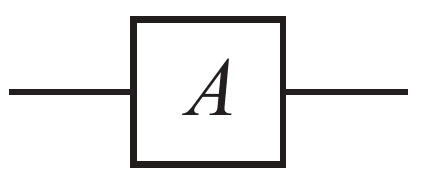} &
\includegraphics[width=0.085\textwidth]{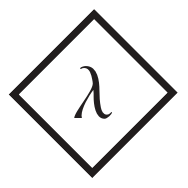}  
\\
\footnotesize{(a) Vector $\ket{v}\in\2 X$} &
\footnotesize{(b) Dual-vector $\bra{v}\in \2 X^\dagger$}&
\footnotesize{(c) Linear Operator $A\in \2 {\mathcal L}(\2 X)$}&
\footnotesize{(d) Scalar $\lambda\in\C$} 
\\~\\
& \includegraphics[width=0.1\textwidth]{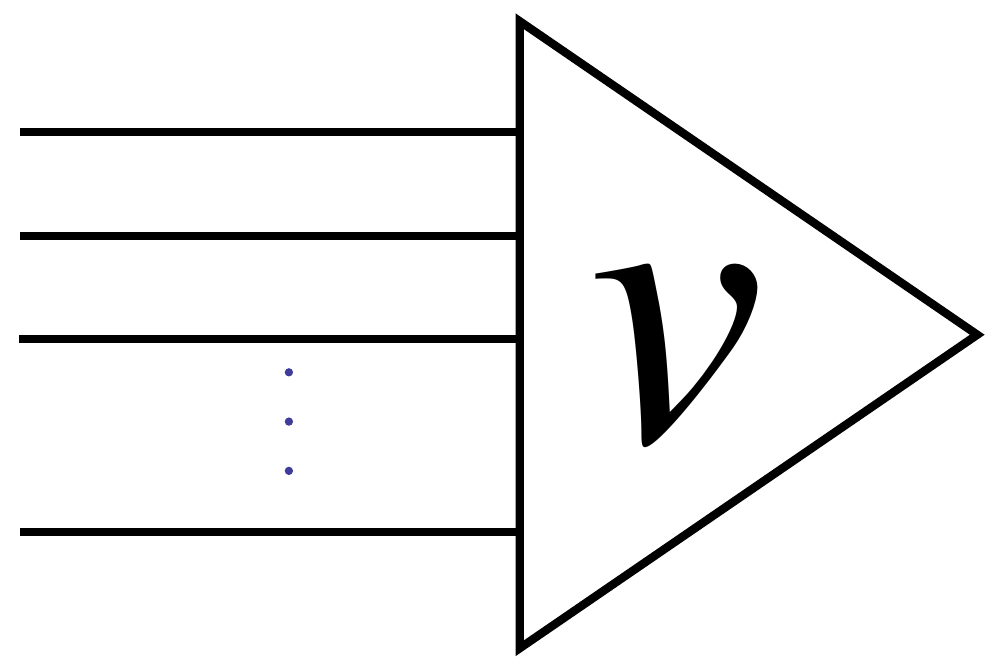} &
\includegraphics[width=0.1\textwidth]{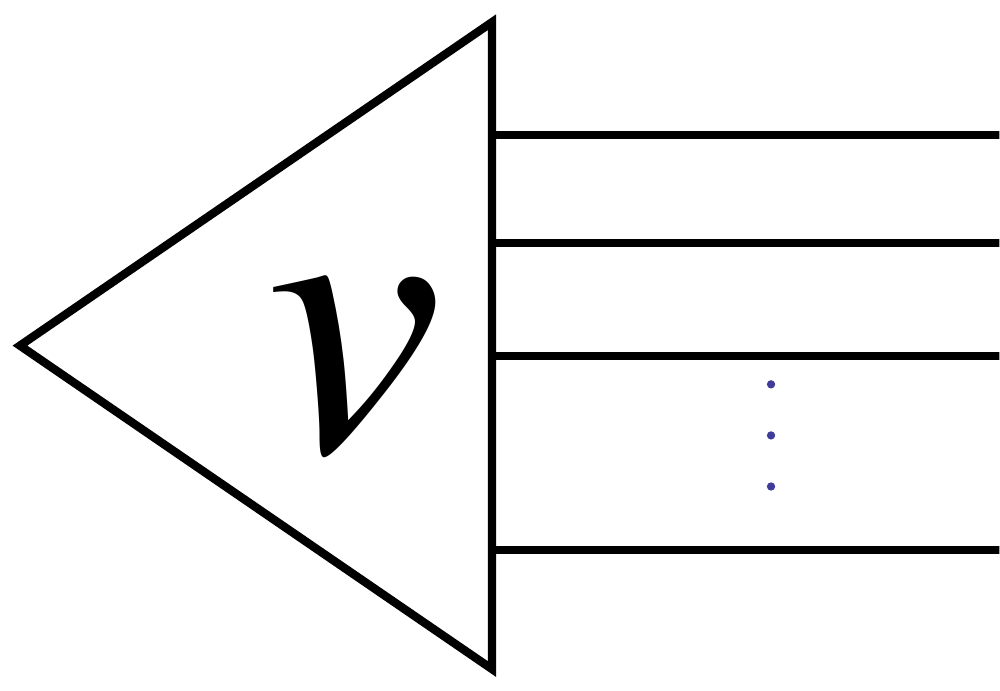}  &
\includegraphics[width=0.145\textwidth]{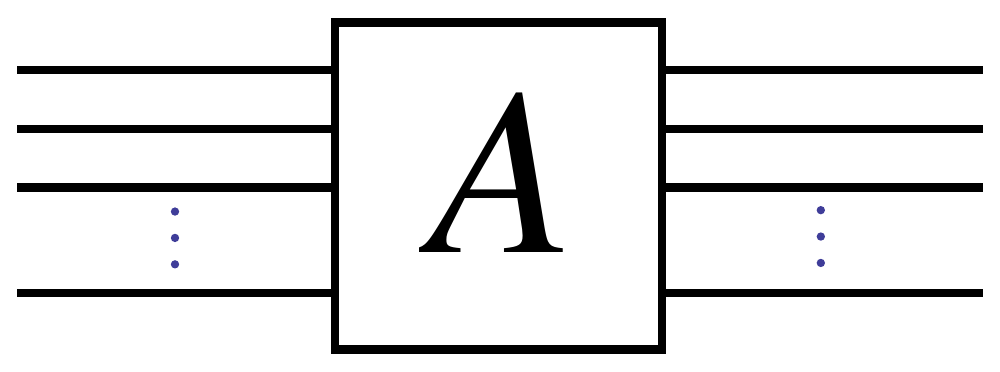}
\\
& \footnotesize{(e) Vector $\ket{v}\in \bigotimes_{i=1}^n \2 X_i$} &
\footnotesize{(f) Dual-vector $\bra{v}\in \bigotimes_{i=1}^n \2 X^\dagger_i$} &
\footnotesize{(g) Linear operator $A:$}
\\
& & &
\footnotesize{$\2 \bigotimes_{i=1}^n \2 X_i \rightarrow \bigotimes_{j=1}^m \2 X_j$}
\end{tabular}
\end{center} 

  \label{fig:ketbramat}
\end{illexample}
\end{myfullpage}
\end{paragraph}

It is typical in quantum physics to think of a tensor as an indexed multi-array of numbers.\mn{Abstract index notation is a mathematical notation for tensors that uses indices as place holders identifying space(s), rather than their components in a particular basis.}  For instance,  
\begin{center}
 \includegraphics[width=6\xxxscale]{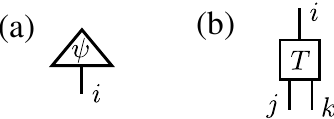}
\end{center}
represent the tensor (a) $\psi_i$ in the space $\2 H_i$ and (b) $T^i_{~jk}$ in the space 
$$
\2H^i\otimes \2H_j \otimes \2 H_k
$$
respectively.  

\begin{remark}[Diagram convention---top to bottom, or right to left]
Open wires pointing towards the top of the page, correspond to upper indices (bras), open wires pointing towards the bottom of a page correspond to a lower indices (kets).  For ease of presentation, we will often rotate this convention $90$ degrees counterclockwise.  
\end{remark}

There are three specific tensors that (essentially) play the role of Kronecker's delta.  These tensors allow for (i) tensor index contraction by diagrammatic connection, (ii) raising and lowering indices, and (iii) they give rise to a duality between maps, states and linear maps in general.  The bijection induced by bending wires, is sometimes called \textit{Penrose Duality}, after its inventor \cite{Penrose}.  As in \cite{Penrose}, these three tensors are given diagrammatically as 
\begin{center}
 \includegraphics[width=12\xxxscale]{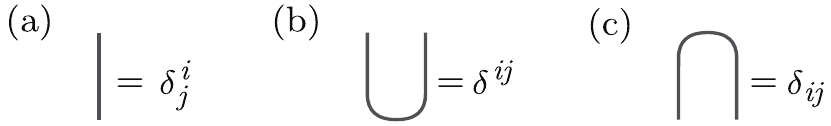}
\end{center}

By thinking of these tensors now in terms of components, e.g.\ $\delta_{ij}=1-(i-j)^2$ for $i,j=0,1$, we note that 
\begin{equation}\label{eqn:delta}
 \I = \sum_{ij}\delta^i_j\ket{i}\bra{j}=\sum_k\ket{k}\bra{k}, 
\end{equation}
\begin{equation}\label{eqn:delta-effect}
 \bra{00}+\bra{11}+\cdots \bra{nn} = \sum_{ij} \delta_{ij}\bra{ij}=\sum_k\bra{kk},
\end{equation}
\begin{equation}\label{eqn:delta-bell}
 \ket{00}+\ket{11}+\cdots \ket{nn} = \sum_{ij} \delta^{ij}\ket{ij}=\sum_k\ket{kk},
\end{equation}
where the identity map (a) corresponds to Equation \eqref{eqn:delta}, the cup (b) to \eqref{eqn:delta-effect} and the cap (c) to \eqref{eqn:delta-bell}.  The relation between these three equations is given by bending wires.  In a basis, bending a wire corresponds to changing a bra to a ket, and vise versa.  

The contraction of two tensor indices diagrammatically amounts to joining those indices with a single wire.  Given tensors 
$T^i_{~jk}$, $A^l_{~n}$ and $B^q_{~m}$ we form a contraction by multiplying by $\delta^j_l\delta^k_q$ resulting in the tensor
\be 
T^i_{~jk}A^j_{~n}B^k_{~m} \bydef \Gamma^i_{~nm},
\ee 
where the tensor $\Gamma^i_{~nm}$ is introduced per definition to simplify notation.  As a linear map, in quantum physics notation, this is typically expressed in equational form as 
\be 
\Gamma = \sum_{inm}\Gamma^i_{~nm}\ket{nm}\bra{i} = \sum_{ijknm} T^i_{~jk}A^j_{~n}B^k_{~m}\ket{nm}\bra{i}. 
\ee 

\begin{remark}
Cups and caps appear on page 231 of {\it Applications of Negative Dimensional Tensors} by Roger Penrose \cite{Penrose}. The composition of cups and caps is the identity; known as the snake equation and also found on page 231 of \cite{Penrose}. 
\end{remark}

\subsubsection*{Connection to quantum computing notation}
As mentioned, tensors are multilinear maps.
They can be expanded in any given basis, and expressed in
terms of their components.
In quantum information science one often introduces a \emph{computational basis}
$\{\ket{k}\}_k$
for each Hilbert space
and expands the tensors in it, using kets ($\ket{~}$) for vectors and bras ($\bra{~}$) for dual vectors:
\be 
T = \sum_{ijk}T\indices{^i_{jk}}\ket{i}\bra{jk}.
\ee 
Here $T\indices{^i_{jk}}$ is understood
not as abstract index notation but as the actual components of the
tensor in the computational basis.
In practice there is little room for confusion.
The Einstein summation convention is rarely used in quantum information science,
hence we write the sum sign explicitly.

So far we have explained how tensors are represented in tensor diagrams,
and what happens when wires are connected.
The ideas are concluded by four
examples; we urge the reader to work through the examples and check the results for themselves.

The first example introduces a familiar structure from linear algebra in tensor form.
The next two examples come from quantum entanglement theory---see connecting
tensor networks with invariants~\cite{BBL11, 2014SIGMA..10..095C}.
The fourth one showcases quantum circuits, a subclass of tensor networks
widely used in the field of quantum information.
The examples are chosen to illustrate properties of tensor networks and should be self-contained.

\begin{remark}
We occasionally will work with equality up to a scalar when manipulating tensor diagrams by hand. This is common and typically amounts to loss of (unit) normalization.  In quantum theory, a global phase is undetectable.  Hence it is common to consider an equivalency class where $\ket{\psi}$ and $e^{\imath \phi}\ket{\psi}$ are equivalent. This is called working in the {\it unit gauge}: in tensor networks we sometimes work in the {\it scalar gauge}, $\mathbb{C}/\{0\}$. This amounts to mapping numbers picked up during calculation as 
$$
\mathbb{C}/\{0\}\rightarrow 1
$$ 
and representing the unit $1$ as a blank on the page. 
\end{remark}

\begin{example}[The $\epsilon$ tensor]\label{ex:epsilon}
A tensor is said to be fully antisymmetric if swapping any pair of
indices will change its sign: $A_{ij} = -A_{ji}$.
The $\epsilon$ tensor is used to represent the fully antisymmetric Levi-Civita
symbol, which in two dimensions can be expressed as
\begin{equation}
 \epsilon_{00}=\epsilon_{11}=0, \qquad
 \epsilon_{01}=-\epsilon_{10}=1.
\end{equation}

The $\epsilon$ tensor can be used to compute the determinant of a matrix. In two dimensions we have
\be
\det(S) = \epsilon_{ij} S\indices{^i_0} S\indices{^j_1}.
\ee
Using this we obtain
\begin{center}
 \includegraphics{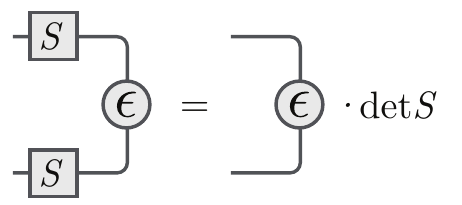}
\end{center}
as can be seen by labeling the wires in the diagram.
In equational form this is
\begin{equation}
\epsilon_{ij} S\indices{^i_m} S\indices{^j_n} = \det(S) \: \epsilon_{mn}.
\end{equation}

In terms of quantum mechanics, $\epsilon$ corresponds to the two-qubit
singlet state:
\be\label{eqn:singlet}
\frac{1}{\sqrt{2}} \ket{\epsilon} = \frac{1}{\sqrt{2}}(\ket{01}-\ket{10}).
\ee
This quantum state is invariant under any transformation of the form $U \otimes U$, where $U$~is a $2\times 2$ unitary,
as it only gains an unphysical global phase factor~$\det(U)$.

\end{example}

\begin{example}[Concurrence and entanglement]\label{ex:concurrence}
Given a two-qubit pure quantum state $\ket{\psi}$,
its \emph{concurrence}
$C(\psi) = |C'(\psi)|$
is the absolute value of the following tensor network expression~\cite{concurrence}:
\begin{center} \label{fig:con}
\includegraphics[width=0.36\textwidth]{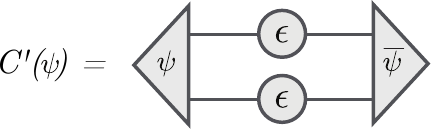}
\end{center}
Here $\overline{\psi}$ is the complex conjugate of~$\psi$ in the computational basis.
The concurrence is an entanglement monotone, a function from states to
nonnegative real numbers that measures how entangled the state is.
$\ket{\psi}$ is entangled if and only if
the concurrence is greater than zero.

Consider now what happens when we act on $\ket{\psi}$ by an arbitrary local unitary operation, i.e.
$$
\ket{\psi'} = (U_1 \otimes U_2) \ket{\psi}.
$$
Using the result of Example~\ref{ex:epsilon} we obtain
\begin{equation}
    C\left((U_1\otimes U_2) \ket{\psi}\right) = C(\psi) |\det(U_1) \det(U_2)|.
\end{equation}
Due to the unitarity 
$$
|\det U_1| = |\det U_2| = 1,
$$
which means that the value of the concurrence
is \emph{invariant} (i.e.~does not change) under local unitary transformations.
This is to be expected, as local unitaries cannot change the amount of entanglement in a quantum state.


More complicated invariants can also be expressed as tensor networks~\cite{BBL11}.
We will leave it to the reader to write the following network as an algebraic expression:
\begin{center}
 \includegraphics[width=0.6\textwidth]{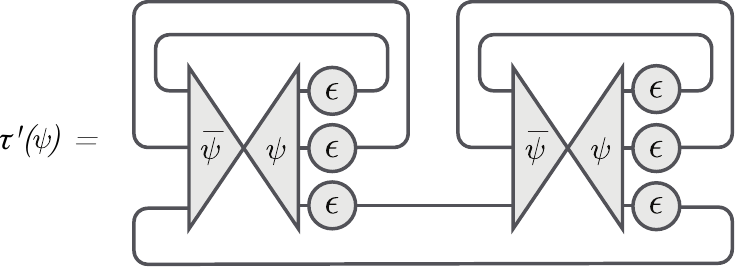}
\end{center}
If $\ket{\psi}$ is a 3-qubit quantum state,
$\tau(\psi) = 2 |\tau'(\psi)|$ represents the entanglement invariant known as the 3-tangle~\cite{2000PhRvA..61e2306C}.
It is possible to form invariants also without using the epsilon tensor.  For example,
the following expression represents the 3-qubit entanglement invariant known as the Kempe invariant~\cite{kempe}:
\begin{equation}\label{eqn:kempe}
K(\psi) = \psi^{ijk}~\overline{\psi}_{ilm}~\psi^{nlo}~\overline{\psi}_{pjo}~\psi^{pqm}~\overline{\psi}_{nqk}.
\end{equation}
The studious reader would draw the equivalent tensor network.
\end{example}

\begin{example}[Quantum circuits] \label{ex:circuits-1}\marginnote{{\it ``I learned very early the difference between knowing the name of something and knowing something.''} --- Richard P.~Feynman, co-discover (with Norman Margolus) of the \CNOT- a.k.a.~Feynman-gate.}

Quantum circuits are a restricted subclass of tensor networks that is widely used in the field of quantum information and is the subject of \S~\ref{sec:btn}. In a quantum circuit diagram each horizontal wire represents the Hilbert space associated with a quantum subsystem,
typically a single qubit.

The tensors attached to the wires represent unitary propagators acting on those subsystems, and are called \emph{quantum gates}.
Additional symbols may be used to denote measurements.
The standard notation is described in~\cite{PhysRevA.52.3457}.  The graphical language of quantum circuits will be explored in detail in \S~\ref{sec:btn}. 

Here we will consider a simple and common quantum circuit, used to  generate entangled Bell states.
It consists of two tensors, a Hadamard gate
(\H{})
and a controlled \NOT{} gate
(\CNOT, denoted by the symbol inside the dashed region):
\begin{center}
\label{eq:bell-circuit}
 \includegraphics{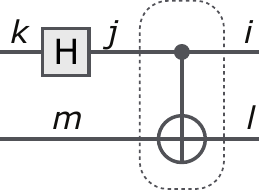}
\end{center}
The \CNOT{} and Hadamard gates are defined as
\begin{align}
\CNOT &= \sum_{ab} \ketbra{a, a \oplus b}{a,b}, \quad \\
\label{eqn:hadamard}
\H &= \frac{1}{\sqrt{2}} \sum_{ab} (-1)^{ab} \ketbra{a}{b},
\end{align}
where the addition in the \CNOT{} is modulo~2.\marginnote{Addition modulo 2: $1\oplus 1 = 0\oplus 0 = 0$, \; $1\oplus 0 = 0\oplus 1 = 1$ (see Appendix \ref{sec:Boolean}).} The reader should verify that acting on the quantum state~$\ket{00}$ the above circuit yields the Bell state
$\frac{1}{\sqrt{2}}(\ket{00}+\ket{11})$, and acting on~$\ket{11}$ it yields the singlet state
$\frac{1}{\sqrt{2}}(\ket{01}-\ket{10})$.
\end{example}

\begin{example} [\COPY{} and \XOR{} tensors: cover art] \label{ex:circuits-2}
One can view the \CNOT{} gate itself as a contraction of two order-three tensors (see \S~\ref{sec:btn} for complete details):
\begin{center}
\includegraphics{cnot-tensors}
\end{center}
The top tensor ($\bullet$ with three legs) is called the \COPY{} tensor.
It equals unity when all the indices are assigned the same value ($0$ or $1$), and vanishes otherwise:

\begin{center}
\includegraphics{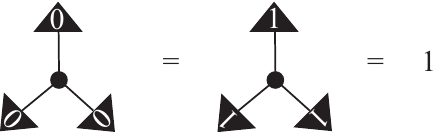}
\end{center}
Hence, \COPY{} acts to copy the binary inputs $0$ and~$1$:
\begin{subequations}
\begin{align}
        \COPY{} \ket{0} &= \ket{0}\ket{0},\\
        \COPY{} \ket{1} &= \ket{1}\ket{1}.
\end{align}
\end{subequations}

The bottom tensor ($\oplus$ with three legs) is called the parity or \XOR{} tensor.
It equals unity when the index assignment contains an even number of $1$s, and vanishes otherwise:

\begin{center}
\includegraphics{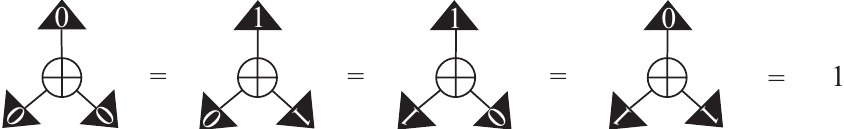}
\end{center}

The \XOR{} and \COPY{} tensors are related via the Hadamard gate \cite{CD, redgreen} as
\be
\label{eq:copy-vs-xor}
\includegraphics{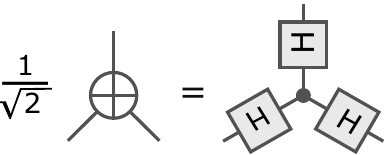}
\ee
where the scalars are often omitted when working in the so called, scalar gauge.  Thus one can think of \XOR{} as being a (scaled) copy operation in another basis:
\begin{subequations}
\begin{align}
        \frac{1}{\sqrt{2}}\XOR{} \ket{+} &= \ket{+}\ket{+},\label{eqn:xcopy}\\
        \frac{1}{\sqrt{2}}\XOR{} \ket{-} &= \ket{-}\ket{-},\label{eqn:xcopy2}
\end{align}
\end{subequations}
where $\ket{+} := \H\ket{0}$ and $\ket{-} := \H\ket{1}$. In terms of components,
\begin{subequations}
\begin{align}
\COPY{}\indices{^{ij}_k}
&= (1-i)(1-j)(1-k)+ijk,\\
\XOR{}\indices{^{qr}_s} &= 1-(q+r+s)+2(qr+qs+sr)-4qrs.
\end{align}
\end{subequations}
 
The CNOT gate is now obtained as the tensor contraction
\begin{equation}\label{eqn:cnot}
\sum_m \COPY{}\indices{^{qm}_{i}} \: \XOR{}\indices{^r_{mj}} = \CNOT{}\indices{^{qr}_{ij}}.
\end{equation}

The \COPY{} and \XOR{} tensors will be explored further in later examples and have many convenient properties~\cite{redgreen, BB11,DBCJ11,B17} which will be explained in \S~\ref{sec:btn}.
\end{example}

\begin{example}[Quantum circuits for cups and epsilon states]\label{ex:state-prep}
The quantum circuit from Example~\ref{ex:circuits-1} is typically used to generate entangled qubit pairs.  
For instance, acting on the state~$\ket{00}$ yields the familiar Bell state---as a tensor network, this is equal to a normalized cup.
Here we also show the mathematical relationship the \XOR{} and \COPY{} tensors have with the cup
($\ket{+}:=\H\ket{0} = \frac{1}{\sqrt{2}}(\ket{0}+\ket{1})$):
\be \label{cnot-bell-states}
\includegraphics{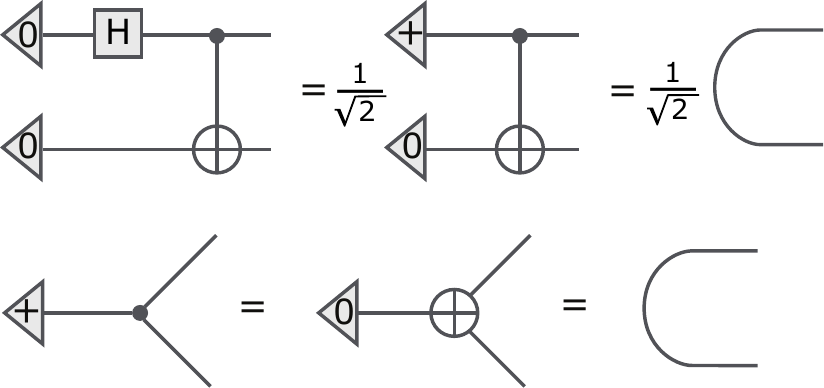}
\ee
Similarly, one can use the circuit~\eqref{eq:bell-circuit} to
generate the epsilon state.
Let us denote the Pauli matrices by $X:=\ket{0}\bra{1}+\ket{1}\bra{0}$, $Y:=-i\ket{0}\bra{1}+i\ket{1}\bra{0}$ and $Z:=\ket{0}\bra{0}-\ket{1}\bra{1}$. 
The $Z$~gate commutes with the \COPY{} tensor, and the $X$ or NOT gate commutes with \XOR{}.
Commuting those tensors to the right hand side,
allows us to apply~\eqref{cnot-bell-states}.  Making use of the Pauli algebra identity $ZX=iY$, one recovers the epsilon state:
\be
\includegraphics{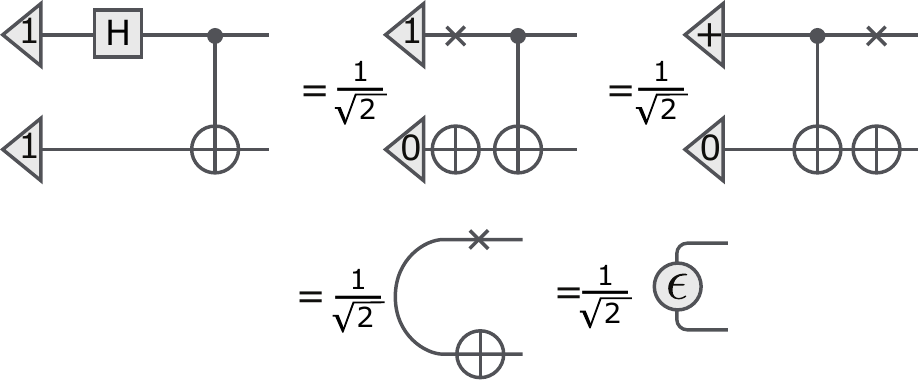}
\ee
\end{example}

\begin{remark}[Graphical tensor calculus \cite{Penrose}]
 While many of the examples we have considered so far are simplistic, in practice tensor networks contain an increasing number of tensors, making it difficult to form expressions using (inherently one-dimensional) equations.  The two-dimensional diagrammatic depiction of tensor networks can simplify such expressions, reduce calculations and often depict internal structure that can lend insight into physical phenomena.  
\end{remark}

Equational identities will also be cast into diagrammatic form.  For example, if $\Gamma$ is totally
symmetric in any arm or leg exchange, then we could adopt the convention to draw it as a circle (b).  The tensor in (c) illustrates the equation $A^{b}_{~dcg}B^{ac}_{~~f}$.  
\begin{center}
 \includegraphics[width=14\xxxscale]{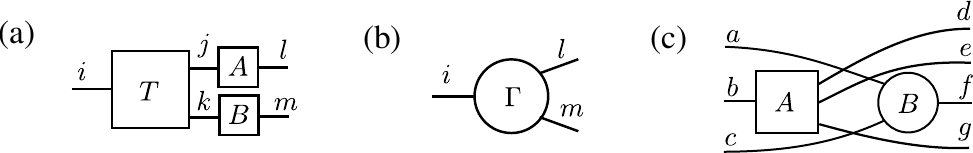}
\end{center}

\begin{remark}[Graphical Trace]
 The trace in the graphical calculus is given by connecting wires to close loops \cite{Penrose}.
Diagram (a) below represents the trace~$A\indices{^i_i}$.
Diagram (b) represents the trace~$B\indices{^{iq}_{iq}}$.
\begin{center}
\includegraphics{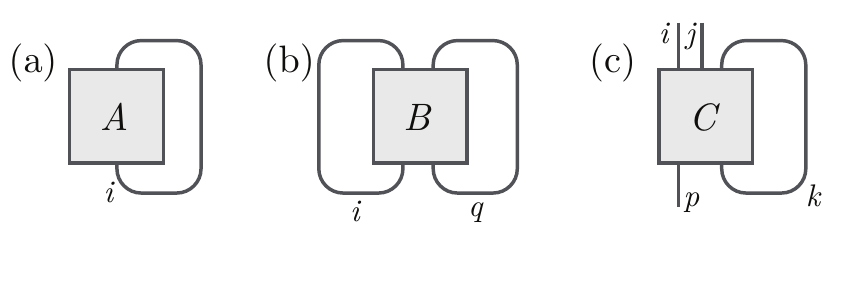}
\end{center}
Partial trace means contracting only some of the outputs with their
corresponding inputs, such as with the tensor $C\indices{^{ijk}_{pk}}$ shown in diagram (c).
\end{remark}


\begin{example}[Partial trace]\label{ex:partial-trace}
The following is an early rewrite representing entangled pairs due to Penrose \cite{Penrose67}.
\begin{center}
\includegraphics{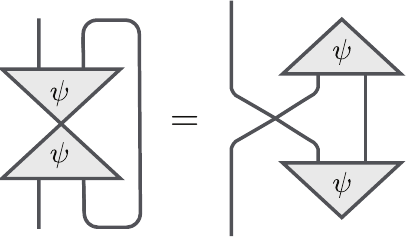}
\end{center}
The diagram on the left represents the partial trace of $\ketbra{\psi}{\psi}$ over the second subsystem.
Readers can prove that this equality follows by interpreting the bent wires as cups and caps, and the crossing wires as \swap{}s.
\end{example}

\begin{example}[Partial trace of Bell states]
 Continuing on from Example \ref{ex:partial-trace}, if we choose $\ket{\psi} = \ket{\cup}$, i.e.~$\ket{\psi}$ is an unnormalized Bell state, we obtain the following. (Compare the following to the diagrams used to model quantum information in the presence of closed timelike curves \cite{Lloyd_2011}.)
\begin{center}
\includegraphics{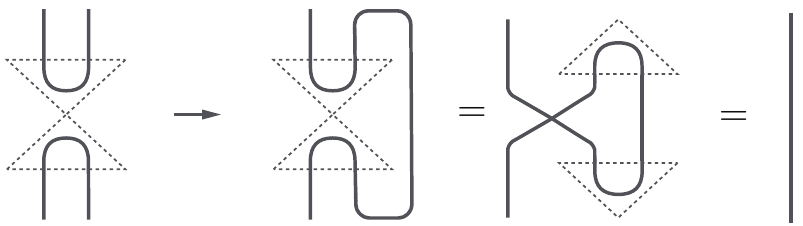}
\end{center}
\end{example}

In general, a quantum state on $n$-spins is a single tensor (such as a box or triangle) with $n$ protruding legs.  Several methods exist to factor such states into elementary building blocks.  For example, it has been shown that the tensor network for each quantum state in the Boolean class of states arises in turn from the classical decomposition of $f$ into fundamental gates.  This is made precise by Theorem \ref{theorem:btns}.  

\begin{remark}[Notation]
We use $\7B$ to denote the Boolean field (this often appears in the literature as $\7F_2$ or as $\7Z_2$), given by an element of the set $\{0,1\}$.  Numbers in $\7B^n$ are then $n$-long Boolean bit strings.  Here and elsewhere, bold font $\textbf{x}$ is shorthand for bit strings $x_1, x_2, \ldots, x_n$.  
\end{remark}

\marginnote{{\it A tensor network representing a Boolean quantum state is determined from the classical network description of the corresponding function}---Theorem \ref{theorem:btns} in \S~\ref{part:count}.}

\begin{definition}[The class of Boolean quantum states --- covered in detail in \S~\ref{sec:btn}]
 Let 
 \be 
 f:\7 B^n \rightarrow \7B ,
 \ee 
 be any switching function.  Then 
 \be 
 \psi_{\7 B} = \sum_{\x}f(\x) \ket{\x},
 \ee 
 is an arbitrary representative in the class of Boolean states.  In this fashion, every Boolean function save the constant zero function, gives rise to a quantum state.  Conversely, every quantum state written in a local basis with amplitude coefficients taking binary values in $\{0,1\}$ gives rise to a Boolean function.  This defines the so called, class of Boolean quantum states (explored in detail in \S~\ref{sec:btn} and \ref{part:count}.
\end{definition}

We will establish the following as Theorem \ref{theorem:btns} in \S~\ref{part:count}: \\

\noindent {\it A tensor network representing a Boolean quantum state is determined from the classical network description of the corresponding function. } \\

\noindent This can be proven by letting each classical gate act on a linear space and from changing the composition of functions, to the contraction of tensors.

\begin{remark}[Quantum Lego blocks: Boolean Tensor Networks]
An example of a Boolean tensor is the \AND-tensor studied in \cite{CTNS}.  This tensor stores the truth table for the local \AND{} function as a superposition.  
\be
\wedge_{ijk} = \sum_{x,y=0,1}\ket{x,y}\ket{x\wedge y} = (\ket{00}+\ket{01}+\ket{10})\ket{0} + \ket{11}\ket{1}.
\ee 
Under Penrose wire-duality, if we bend a wire to raise the index labeled $k$ we arrive at  
 \begin{equation}\label{eqn:AND-map}
 \wedge_{ij}^{~~k}=(\ket{00}+\ket{01}+\ket{10})\bra{0} + \ket{11}\bra{1}.
 \end{equation}
\end{remark}

\begin{remark}[Tensor juxtaposition]
When two or more disconnected tensors appear in the same diagram they are multiplied
together using the tensor product.
In quantum physics notation, they would have a tensor
product sign~$\otimes$ between them.
In the abstract index notation the tensor product sign is omitted.

Tensors can be freely moved past each other.  This is sometimes called planar deformation or rubber sheet topology.
\begin{center}
 \includegraphics{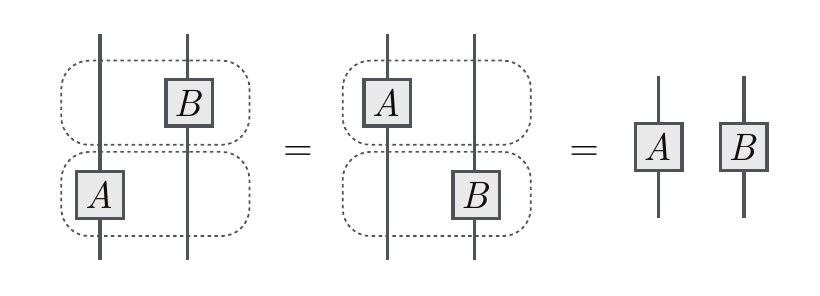}
\end{center}
From the diagram above, using equations we have 
\begin{equation}
 (\I\otimes B)(A\otimes \I) = A\otimes B = (A\otimes \I)(\I\otimes B),
\end{equation}
where we make use of the wire also playing the role of the identity tensor~$\I$---detailed in \S~\ref{sec:Bwires}.
As we shall soon see, wires
are allowed to cross tensor symbols and other wires,
as long as the wire endpoints are not changed.
This is one reason why tensor diagrams are often
simpler to deal with
than their algebraic counterparts.
\begin{center}
   \includegraphics{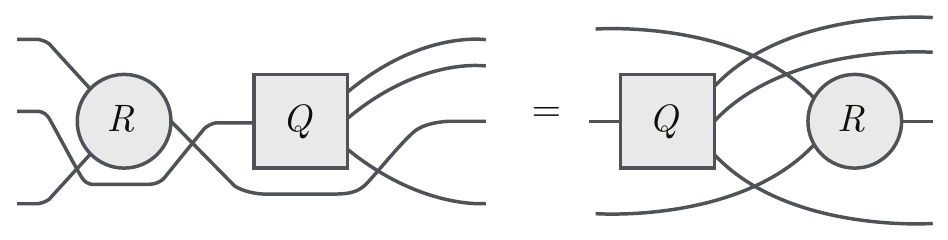}  
\end{center}

In the diagram above we did not label the wires, since
it is an arbitrary assignment. If we did, we could for example denote it as
$Q\indices{^{deg}_{b}}R\indices{^{f}_{ac}}$.
\end{remark}

\section{Penrose Wire Bending Duality}\label{sec:Bwires}
\paragraph{Bending and exchanging wires.}  Let us consider $n$ Hilbert spaces $\2H^{\otimes n}$.  These spaces are essentially equivalent to each other.  We will then consider the \swap{} operator which exchanges the position of two Hilbert spaces in a composite system, moreover it will exchange the $i$th space with the $i+1$th space, for $1\leq i \leq n-1$.  This generates the permutation group with order that divides $n!$, with the generators given diagrammatically as in (a) 
\begin{center}
 \includegraphics[width=8\xxxscale]{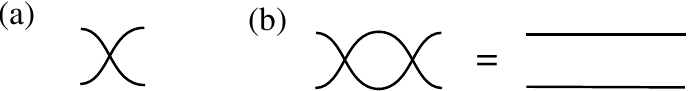}
\end{center}
where (b) shows that (a) is self inverse. The operator is unitary, and (a) may be written as the tensor
$$
\swap^{ij}_{~~kl}=\delta^i_{~l}\delta^j_{~k},
$$ 
or expanded in the computational basis as 
$$
\swap = \sum_{ij} \ketbra{ij}{ji}.
$$
It also has a well-known implementation in terms of three CNOT gates as  
\begin{center}
    \includegraphics{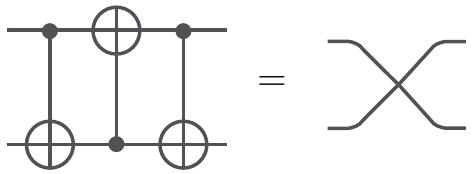}
\end{center}

\begin{example}[\swap{} on two Hilbert spaces]
Let $\2 
X$ and $\2 Y$ be complex Hilbert spaces of dimensions $d_1$ and $d_2$ respectively, then the SWAP operation is the map
\be 
\begin{aligned}
\swap{}: \XY &\rightarrow \YX\\
\swap{}: \ket{x}\otimes\ket{y}&\mapsto \ket{y}\otimes\ket{x},
\end{aligned}
\ee
for all $ \ket{x}\in \2 X,\ket{y}\in\2 Y$.

Given any two orthonormal basis 
$$
\{\ket{x_i}: i=0,\hdots, d_1-1\}
$$
and 
$$
\{\ket{y_j}: j=0,\hdots, d_2-1 \}
$$
for $\2 X$ and $\2 Y$ respectively, we can give an explicit construction for the \swap{} operation as\marginnote{{\it Repeated indices to be summed can share the same color in wire diagrams \cite{2011arXiv1111.6950W}.}}
\begin{equation}
\swap{} =  \sum_{i_1=0}^{d_1-1}\sum_{j_2=0}^{d_2-1}\ketbra{y_j}{x_i}\otimes\ketbra{x_i}{y_j}.
\label{eqn:swap}
\end{equation}

The \swap{} operation is represented graphically by two crossing wires as shown:
\begin{center}
\includegraphics[width=0.35\textwidth]{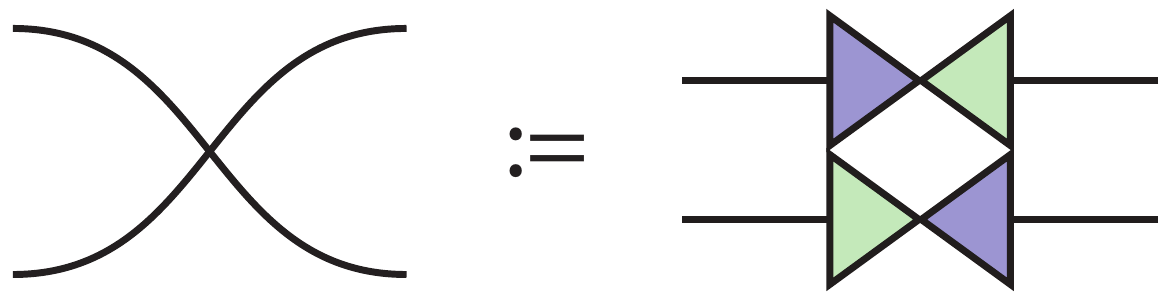} 
\label{fig:swap}
\end{center}
The basis decomposition in~\eqref{eqn:swap} is then an application of the resolution of the identity to each wire.
\end{example}

We will now consider the transformation of raising or lowering an index.  One can raise an index and then lower this index or vice versa, which amounts essentially to the net effect of doing nothing at all.  This is captured diagrammatically by the so called, \textit{snake} or \textit{zig-zag equation}, as
\begin{center}
 \includegraphics[width=.65\textwidth]{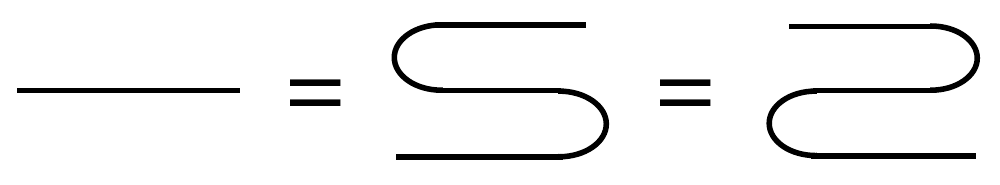}
\end{center}
together with its vertical reflection across the page.  The snake or zig-zag equation in diagrammatic form dates back at least to Penrose \cite{Penrose}.  Given a basis makes a duality between flipping a bra to a ket, that is, raising or lowering an index, precise.  In tensor index notation, it is given simply by $\delta_{ji}\delta^{ik} = \delta_j^{~k}$.

The mathematical rules of tensor network theory assert that the wires of tensors may be manipulated, with each manipulation corresponding to a specific contraction or transformation.

\begin{myfullpage}  
\begin{definition}
 Transposition of 1st-order vectors and dual-vectors, and 2nd-order linear operators is represented by a \emph{bending} of a tensors wires as follows:
\be
\begin{tabular}{c|c|c}
\includegraphics[height=3.5em]{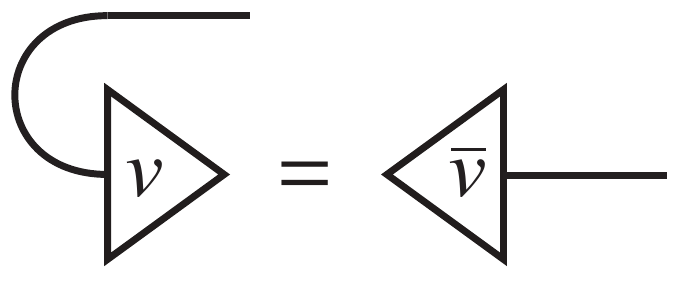} \quad\quad&\quad\quad
\includegraphics[height=3.5em]{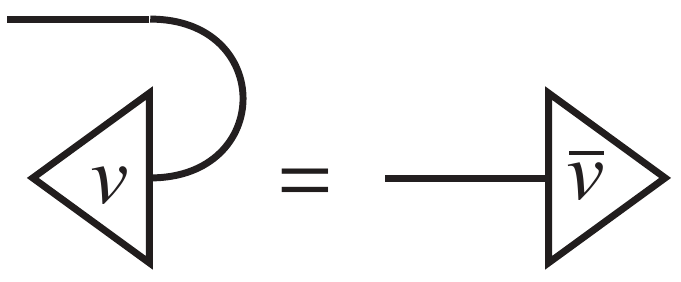} \quad\quad&\quad\quad
\includegraphics[height=3.5em]{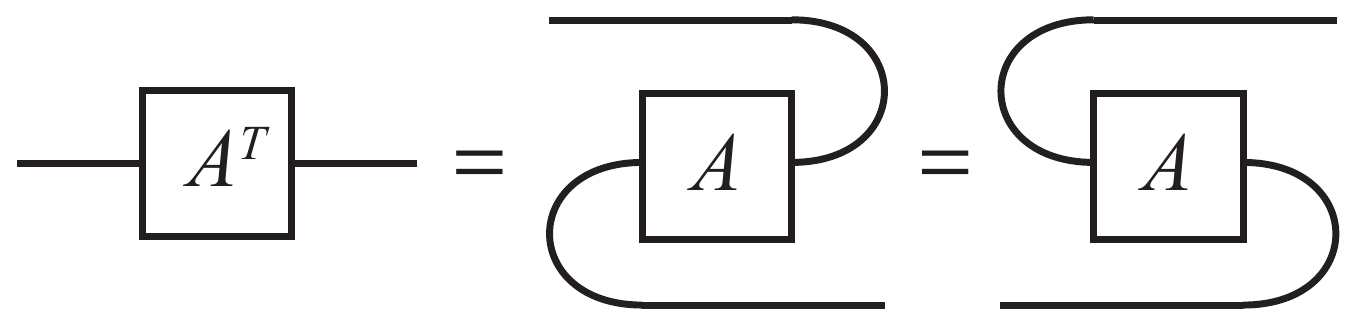} \\~\\
\footnotesize{(a) Vector transposition:} \quad&\quad
\footnotesize{(b) Dual-vector transposition:} \quad&\quad
\footnotesize{(c) Linear operator transposition}\\~\\
\footnotesize{$\ket{v}^T = \bra{\overline{v}}$} \quad&\quad
\footnotesize{$\bra{v}^T=\ket{\overline{v}}$} \quad&\quad
\end{tabular}
\label{fig:transpose}
\ee

Complex conjugation of a tensor's coefficients however is depicted by a bar over the tensor label in the diagram: 
 \be
\begin{tabular}{c|c|c}
\includegraphics[height=3em]{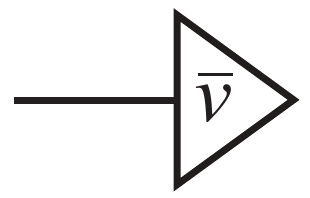} \quad&\quad
\includegraphics[height=3em]{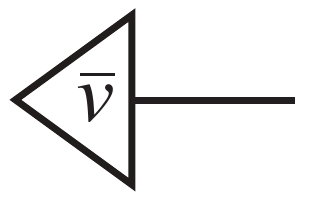} \quad&\quad
\includegraphics[height=3em]{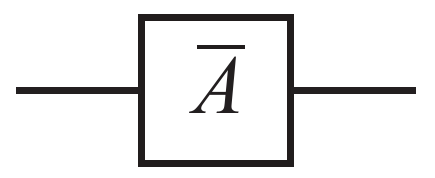} \\
\footnotesize{(a) Complex conjugation of $\ket{v}$} \quad&\quad
\footnotesize{(b) Complex conjugation of $\bra{v}$} \quad&\quad
\footnotesize{(c) Complex conjugation of $A$}
\end{tabular}
\label{fig:conj}
\ee

We stress that under this convention a vector $\ket{v}=\sum_i v_i \ket{i}$ and its hermitian conjugate dual-vector $\bra{v}=\sum_i \overline{v}_i\bra{i}$ are represented as shown in Figure~\ref{sfig}a and~\ref{sfig}b respectively.
\end{definition}
\end{myfullpage}


 \noindent For explicit examples of writing down the equational form of a tensor diagram refer to the proofs in \S~\ref{app:tensor1}.
\paragraph{Equivalence class.}  In further detail, we will consider the class of operations formed from bending tensor wires forwards or backwards using cups and caps, as well as exchanging wires using \swap{}.  We can conceptualize this class of transforms acting on a tensor as, amounting essentially, to matrix reshapes. From the snake equation, action with a cup or cap is invertible and \swap{} is self inverse.  This means that, all possible configurations of a tensors legs using these operations are equivalent, when the equivalence is taken up to Penrose duality.    

\begin{lemma}[Cardinality of index manipulations]
 Given a tensor $T^{ij}$ with fixed labels $i,j$ we can use cups and caps to arrive at 
\be 
 T^{ij}, ~ T^{i}_{~j} , ~ T_{ij}, ~ T_{i}^{~j},
\ee 
the \swap{} operation reorders $i$ and $j$ and then the cups and caps yield 
\be 
 T^{ji}, ~  T^{j}_{~i} , ~  T_{ji}, ~ T_{j}^{~i}.
\ee 
In general, for a tensor with a total of $n$ indices, each index can be up or down, yielding $2^n$ possibilities.  The symmetry group formed by \swap{} is of order $n!$ and acts to arrange the $n$ legs of a tensor, yielding 
\be
n!\cdot 2^n
\ee 
different ways to reorder the indices of a tensor, provided we distinguish $T_{i}^{~j}$ and $T^{j}_{~i}$ etc.  
\end{lemma}

\begin{remark}[Ordering operators by numbers of inputs and outputs]\label{remark:ordering-index}
In the previous remark, we considered $T_{i}^{~j}$ (b) and $T^{j}_{~i}$ (a) etc., as distinct.  For all practical purposes, they are not however.  This is shown as follows.  
 \begin{center}
 \includegraphics[width=7\xxxscale]{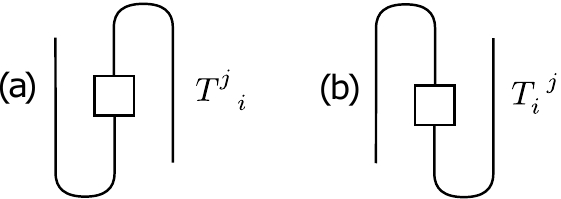}
\end{center}
This shows an awkward property of standard Dirac notation.  Both (a) and (b) represent the same map, but when we write this in a basis, one of them will require us to write $\bra{i}\otimes\ket{j}$.  

With this in mind, we note that the tensor $T_{ij}$ which was considered in the last section actually has 6 unique reshapes, as two of the reshapes are diagrammatically equivalent.  
 \begin{center}
 \includegraphics[width=10\xxxscale]{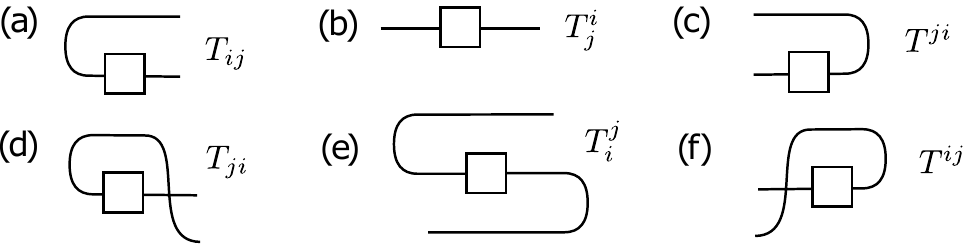}
\end{center}
\end{remark}

The duality is well know, but we have not seen mention of the order.  We call this the \emph{natural tensor symmetry class}.  In Theorem \ref{theorem:symmetry-class} are going to count (i) the number of possible ways a tensor can have its wires bent, either forward or backwards using the cups and caps, in conjunction with (ii) the number of ways a tensor can have its arms and/or legs exchanged.    
\begin{theorem}[Natural tensor symmetry class]
\label{theorem:symmetry-class}
The arms and legs of a tensor $\Gamma^{ij\cdots k}_{qr\cdots s}$ with $m$ input arms $(ij\cdots k)$ and $n$ output legs $(qr\cdots s)$, can be rearranged in  
\be 
(n+m+1)! 
\ee
different ways.
\begin{proof}
Exercise.  
\end{proof}
\end{theorem}

\section{Bending Density Operator Wires}

We will now apply the natural tensor symmetry class counted in Theorem \ref{theorem:symmetry-class} to the analysis of the quantum states arising from bending wires on density operators.  These states are found by bending all the wires of a tensor representing a density operator to the same side, as follows.  
\begin{center}
 \includegraphics[width=10\xxxscale]{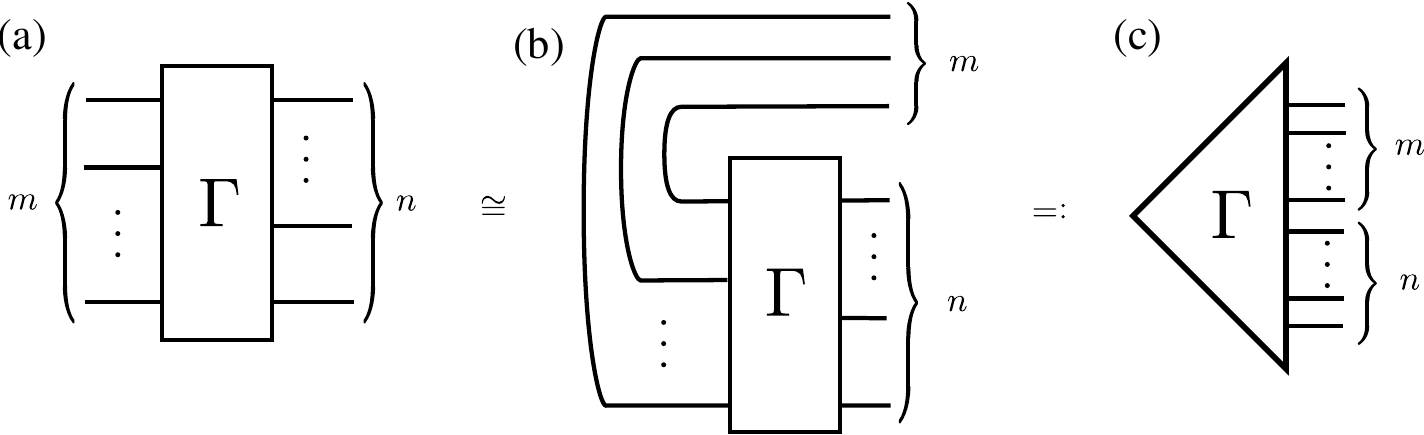}
\end{center} 
For the case of a density map, $m=n$.  

We will first compare the idea of a matrix basis with that of a vector space basis both for an inner product space.  We will use these concepts to study symmetries of 
density operators.  In fact, 
\begin{remark}[Injection from Density Operators to States --- Theorem \ref{theorem:injection}]
 We will soon prove that the existence of an injective map sending each n-party density matrix $\rho$ to a quantum state.  The map is found by bending wires.  The resulting state is naturally equivalent under SWAP to $2n!$ states, by Theorem \ref{theorem:symmetry-class}.    
\end{remark}
The process is invertible.  However, every quantum state does not always give rise to a $\rho$ under wire duality.  The purpose of the present section is to make these statements precise.  

\subsubsection*{Matrix basis}
We expand $d$-dimensional operators using a matrix basis $\{A_i\}$,
which is orthonormal with respect to the \emph{Hilbert-Schmidt} inner product, defined as
\be 
\langle A_i,A_j\rangle :=
\frac{1}{d}\Tr(A_i^\dagger A_j)=\delta_{ij}.
\ee 
The product $A_iA_j$ is given in (a)
\begin{center}
 \includegraphics[width=5\xxxscale]{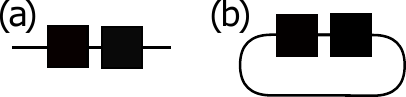}
\end{center} 
and the trace inner product in (b).
Here and elsewhere in this work, all scale factors in the diagrams are omitted graphically, but we note that care must be taken when one is summing over diagrams.

\begin{example}[Pauli matrix basis] 
A nascent example of a matrix basis with the described properties is the Pauli matrices on qubits.  Any operator of type $\C^2\rightarrow\C^2$ can be written in terms of the Pauli matrices as 
\be 
a \I + \vec{p} \cdot \vec{\sigma} = a \I + bX+cY+dZ,
\ee
for $a,b,c,d\in \C$.  Note that this provides a decomposition into a symmetric subspace 
\be
\1S= \text{span}\{\I, X, Z\},
\ee 
and an antisymmetric subspace $\1A=\text{span}\{Y\}$.  This symmetry is exhibited by $\rho = \rho^\top\in\1S$ and antisymmetry as $\rho = -\rho^\top\in\1A$.  Diagrammatically, transposition is done by twisting a map (see (a) or (b) in Remark \ref{remark:ordering-index}) \cite{Kassel}.  
\end{example}

\subsubsection*{Vector basis} For a vector space basis $\{a_i\}$ in $\C^2\otimes \C^2$ we use the typical inner product where
\be
\langle a_i,a_j\rangle = \delta_{ij},
\ee 
which is given diagrammatically as in (b).  (a) Follows from the diagrammatic SVD, and the black square is intended to depict a not necessarily unitary, or for that matter invertible, map.   
\begin{center}
 \includegraphics[width=10\xxxscale]{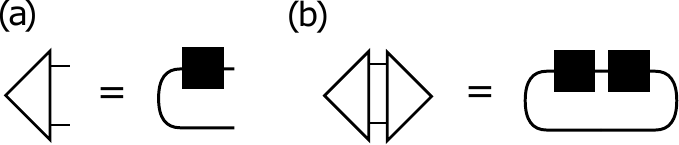}
\end{center} 

\begin{example}[Bell vector basis] 
An example of a vector basis is the Bell basis on qubits.  Any vector
in $\C^2\otimes\C^2$ can be written in terms of this basis as (see
Table~\ref{table:pauli-vs-bell})
\be 
\psi(a,b,c,d) = a \Phi^+ + b \Phi^- + c \Psi^+ + d \Psi^-,
\ee
for $a,b,c,d\in \C$.  This partitions the space into a symmetric subspace 
\be 
\1s=\Span\{ \Phi^+, \Phi^- ,\Psi^+ \},
\ee 
and an antisymmetric subspace $\1a=\Span\{\Psi^-\}$.  This symmetry is illustrated by $\psi(i,j) = \psi(j,i)\in\1s$ and antisymmetry as $\psi(i,j) = -\psi(j,i)\in\1a$.  Diagrammatically, this amounts to letting the swap gate act on both output wires, which serves to exchange them.  
\end{example}

\subsubsection*{Comparison} 
We note that in the diagrammatic language, the Bell vector basis and the Pauli matrix basis become essentially equivalent as they are related by bending wires.  
In particular, we note that they have identical form as the operator basis norm and the inner product of vectors.  These are both, up to a scale factor, identical in the graphical language.  To illustrate our point, we present the following table.  

\begin{table}[h!]
\centering
\begin{tabular}{ccc}
\hline
$\C^2\rightarrow\C^2$ &  & $\C^2\otimes\C^2$\\
\hline 
$\I$ & $\cong$ & $\Phi^+ = \ket{00}+\ket{11}$\\
$\sigma^x$ & $\cong$ & $\Psi^+ = \ket{01}+\ket{10}$\\
$\sigma^y$ & $\cong$ & $\Psi^- = i(\ket{01}-\ket{10})$\\
$\sigma^z$ & $\cong$ & $\Phi^- = \ket{00}-\ket{11}$\\
\hline
\end{tabular}
\caption{Mapping between Pauli matrices and Bell states.
\label{table:pauli-vs-bell}
}
\end{table}
Table~\ref{table:pauli-vs-bell} illustrates the specific mapping between Pauli matrices and Bell states.  For example, the second row represents 
\begin{center}
 \includegraphics[width=10\xxxscale]{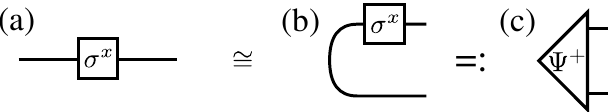}
\end{center} 
where (a) is the $\sigma^x$ matrix.  By bending a wire, we arrive at (b) which is identically equal to the bell state $\Psi^+$ in (c).  

\subsubsection*{Density operators vs pure states.}  We are in a position to carry out an analysis of the class of states found by bending the wires of a density operator all to the same side.  
We first consider the expansion of a density operator into the so called Hilbert-Schmidt basis of Pauli operators.  
We let 
\be 
P_n = \{\I, X, Y, Z\}^{\otimes n},
\ee 
be the set of all $n$ letter words, formed from the alphabet of Pauli
matrices, with $\otimes$ as the concatenation operator.  The span of
$P_n$ forms a Hermitian operator basis as each element is invariant
under the $\dagger$.  We can expand any density operator $\rho$ in
terms of this basis as 
\be 
\rho = \sum a^{ijk\cdots l}\sigma_i \sigma_j\sigma_k\cdots \sigma_l,
\ee 
for an $n$-long index $ijk\cdots l$, where each $i,j,k$ etc indexes an
operator in $P_1$.  We can now study the natural embedding of this
operator $\rho$ on ${\C^2}^{\otimes n}\rightarrow {\C^2}^{\otimes n}$ into
$\psi_\rho \in {\C^2}^{\otimes 2n}$.

\begin{theorem}[Injection from density operators to states]
\label{theorem:injection}
Every density operator $\rho$ on $\2H^{\otimes n} \rightarrow
\2H^{\otimes n}$ gives rise to a state $\psi_\rho$ in $\2H^{\otimes  2n}$.
This state has $2n!$ natural symmetries induced by swap.

\begin{proof}
Each density operator is dual to a state by bending wires.  We will express this starting with the density operator
\be 
\rho = \sum a^{ijk\cdots l}\sigma_i \sigma_j\sigma_k\cdots \sigma_l,
\ee
and writing it as 
\be 
\psi_\rho = \sum a^{ijk\cdots l}\Psi_i\otimes \Psi_j\otimes \Psi_k\otimes \cdots \otimes \Psi_l,
\ee
where $\Psi_m \cong \sigma_m$ per bending a wire on each $\sigma_m$, and arriving at the state $\psi_\rho$.  This results in the quantum state $\psi_\rho$.  Note that 
there are a number of choices of ordering when bending the wires.  In addition, these wires can be arbitrarily ordered after bending 
to still form an essentially equivalent state.  In fact, one can act with the symmetry group on the open wires, to arrive at the natural symmetry class with the same order as the permutation group.  This scenario is depicted below.   
\end{proof}
\end{theorem}

\begin{example}[Injection from density operators to states]
We can illustrate the idea behind Theorem \ref{theorem:injection} in the following figure.  (a) depicts density operator $\rho^i_{~j}$ and (b) $\rho_{ij}$ its state dual. 
 \begin{center}
 \includegraphics[width=8\xxxscale]{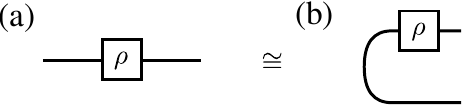}
\end{center}
\end{example}

\begin{remark}[From states to operators by Penrose wire duality]
While every density operator gives rise to a
quantum state, the converse is not necessarily true.
The condition $\rho = \rho^\dagger$ corresponds to $a^{ijk\cdots l} \in \R$
which limits the possible states.
For the case of qubits, a quantum state has
$2(2^n-1)$ real degrees of freedom in general, whereas the states
arising under Theorem \ref{theorem:injection} have $2^n-1$ real degrees of freedom.
\end{remark}

\begin{remark}[In general wire bending is not purification] 
 Let 
\be 
\rho = \frac{1}{2}\I + \sigma.A = \frac{1}{2}\I + a X + b Y + c Z,
\ee 
be the arbitrary state of a qubit.  Then by Theorem \ref{theorem:injection}, 
\be
\psi_\rho = (1+c)\ket{00} + (a+bi)\ket{01} + (a-bi)\ket{10} + (1-c)\ket{11}, 
\ee
is the natural state-dual to $\rho$ with norm $\braket{\psi_\rho}{\psi_\rho} = 2(a^2+b^2+c^2)$. However, we note that 
Theorem \ref{theorem:injection} does not in general provide a purification of $\rho$.  
\end{remark}

\begin{remark}[Symmetric density operators vs symmetric states]
A symmetric single qubit density operator necessarily has $bY=0$ and so gives rise to a symmetric two qubit state with real valued coefficients parametrized by two real degrees of freedom, $\psi_\rho(a, 0, c)$.  Diagrammatically, a symmetric two-party state is invariant under exchange of its legs (a) whereas a symmetric density operator is invariant under exchange of its arms and legs (e.g.\ transpose), (b).  
\begin{center}
 \includegraphics[width=10\xxxscale]{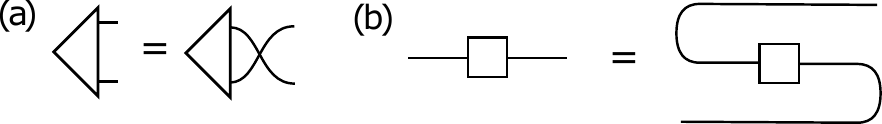}
\end{center}
\end{remark}



\subsection{Bipartite Matrix Operations}
\label{sec:bipartite}

Bipartite matrices are used in several representations of CP-maps, and manipulations of these matrices will be important in the following discussion.
\begin{definition}
Consider two complex Hilbert spaces $\2 X$, and $\2 Y$ with dimensions $d_x$ and $d_y$ respectively. The bipartite matrices we are interested in are then $d^2_x\times d^2_y$ matrices $M\in {\mathcal L}(\2X\otimes\2Y)$ which we can represent as 4th-order tensors with tensor components
\begin{equation}
M_{m\mu,n\nu}:=\bra{m, \mu}M\ket{n,\nu}
\end{equation} 
where $m,n \in \{0,...,d_x-1\}$, $\mu,\nu \in \{0,...,d_y-1\}$ and $\ket{n,\nu}:=\ket{n}\otimes\ket{\nu}\in \2X\otimes\2Y$ is the tensor product of the standard bases for $\2 X$ and $\2 Y$.
\end{definition}

 Graphically this is given by

\begin{center}
  \includegraphics[height=5em]{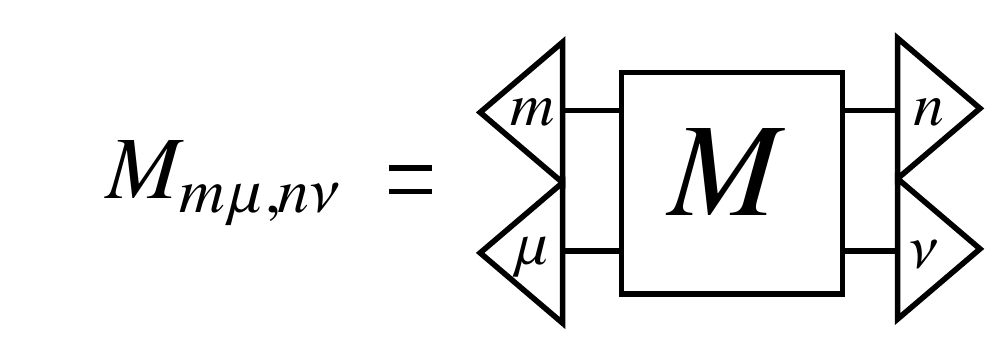}
 \label{fig:bipartite-m-4}  
\end{center}

We can also express the matrix $M$ as a 2nd-order tensor in terms of the standard basis $\{\ket{\alpha} : \alpha = 0,\hdots,D-1\}$ for $\2X\otimes\2Y$ where $D=d_x d_y$. In this case $M$ has tensor components
\begin{equation}
M_{\alpha\beta}=\bra{\alpha}M\ket{\beta}
\end{equation}
This is represented graphically as
\begin{center}
\includegraphics[height=5em]{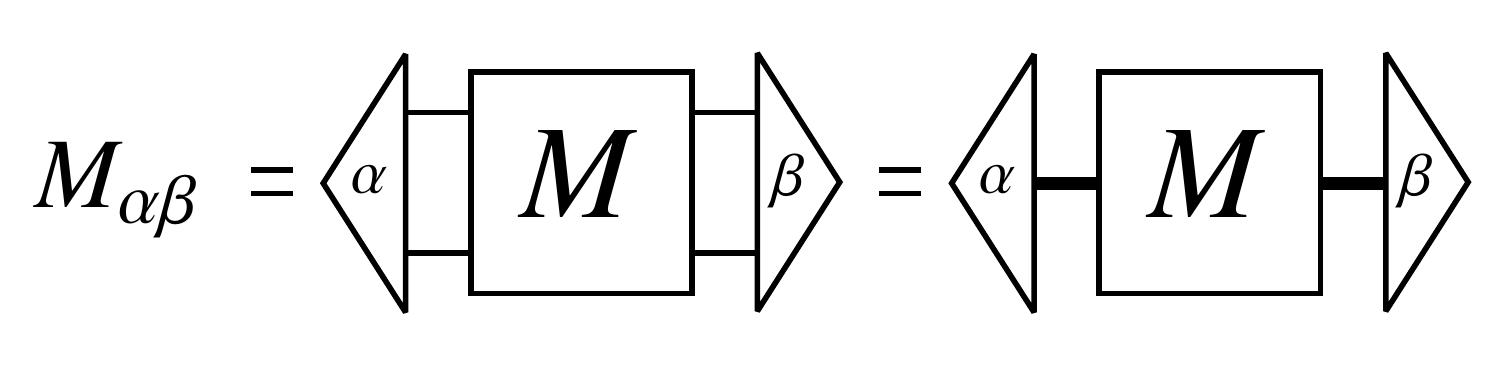}
 \label{fig:bipartite-m-2}
\end{center}

We can specify the equivalence between the tensor components $M_{\alpha\beta}$ and $M_{m\mu,n\nu}$ by making the assignment
\begin{eqnarray}
\alpha &=& d_y m + \mu\\
\beta &=& d_y n +\nu,
\end{eqnarray} 
where $d_y$ is the dimension of the Hilbert space $\2 Y$.

\begin{fullpage}
The bipartite matrix operations which are the most relevant for open quantum systems (see Fig.~\ref{fig:cpreps}) are the \emph{partial trace over $\2 X$}  ($\Tr_{\2 X}$) (and $\Tr_{\2 Y}$ over $\2 Y$), \emph{transposition} ($T$), \emph{bipartite}-\SWAP{} ($S$), \emph{col-reshuffling} ($R_c$), and \emph{row-reshuffling} ($R_r$). The corresponding graphical manipulations are:
\begin{center}
 \begin{tabular}{cccccc}
\includegraphics[width=0.1\textwidth]{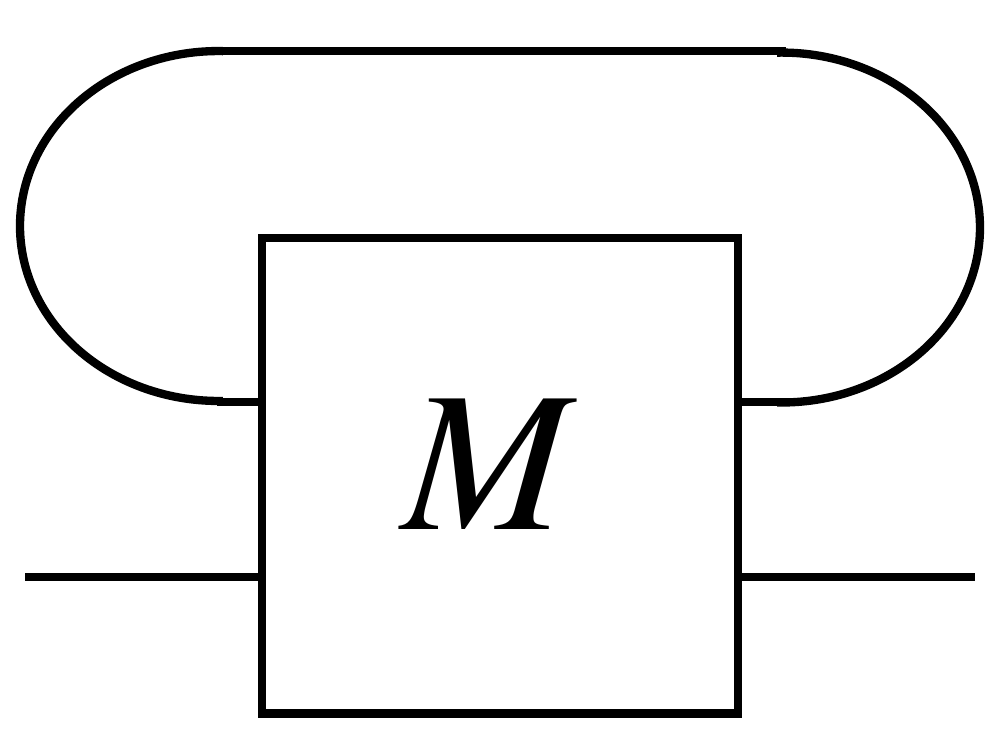}	\quad&\quad
\includegraphics[width=0.1\textwidth]{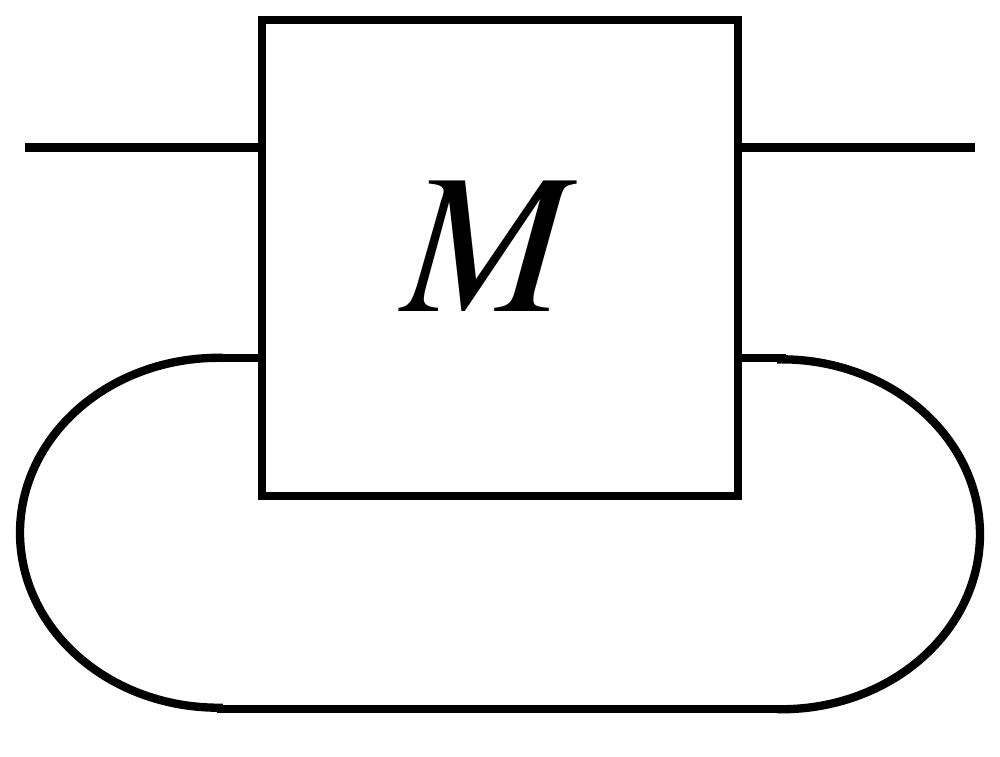}	\quad&\quad
\includegraphics[width=0.1\textwidth]{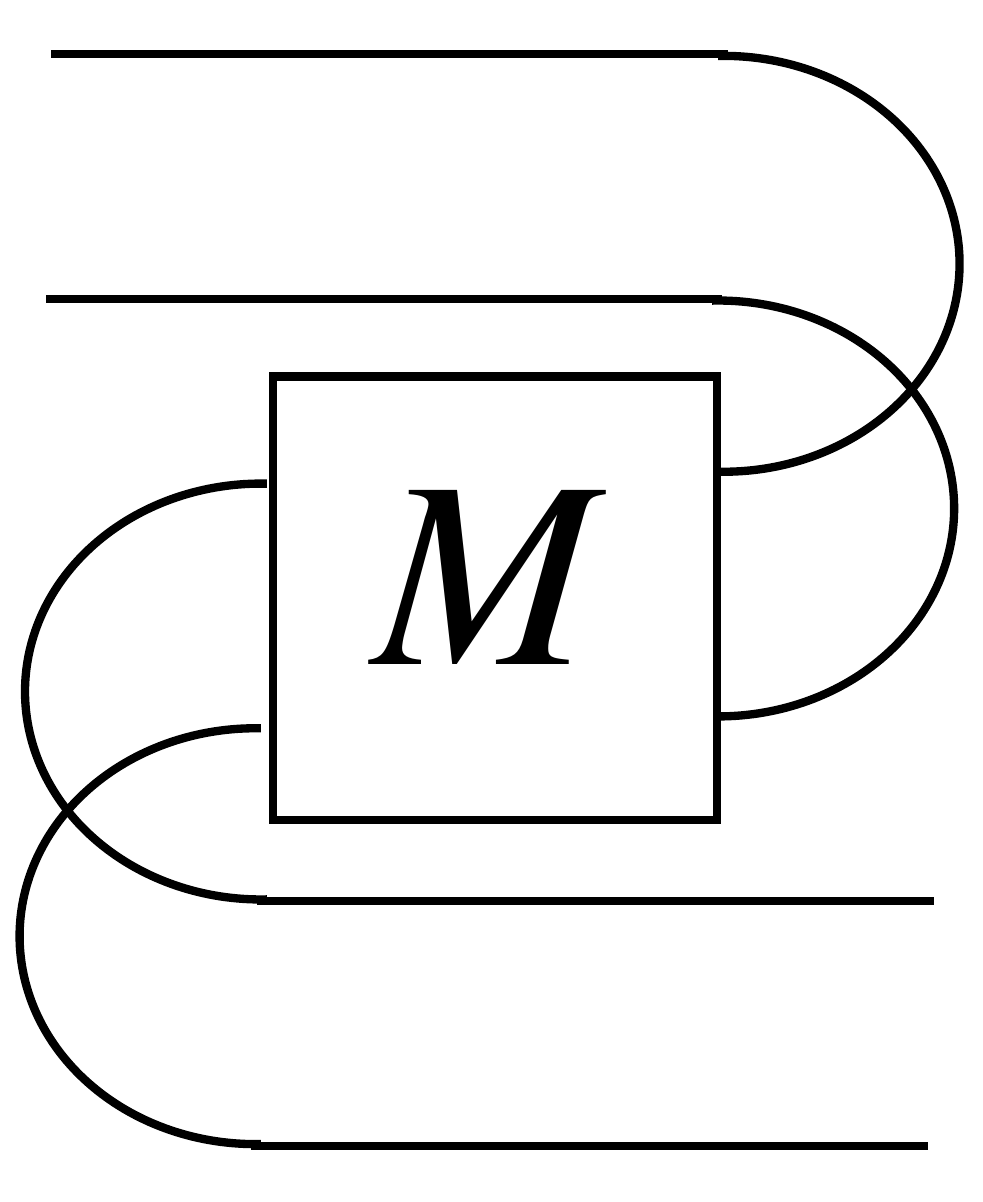}	\quad&\quad
\includegraphics[width=0.13\textwidth,trim=1.5cm 1.5cm 1.5cm 1.5cm,clip]{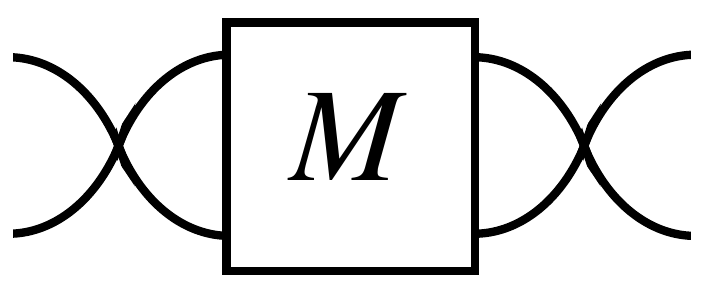}	\quad&\quad
\includegraphics[width=0.1\textwidth,trim=1.5cm 1.5cm 1.5cm 1.5cm,clip]{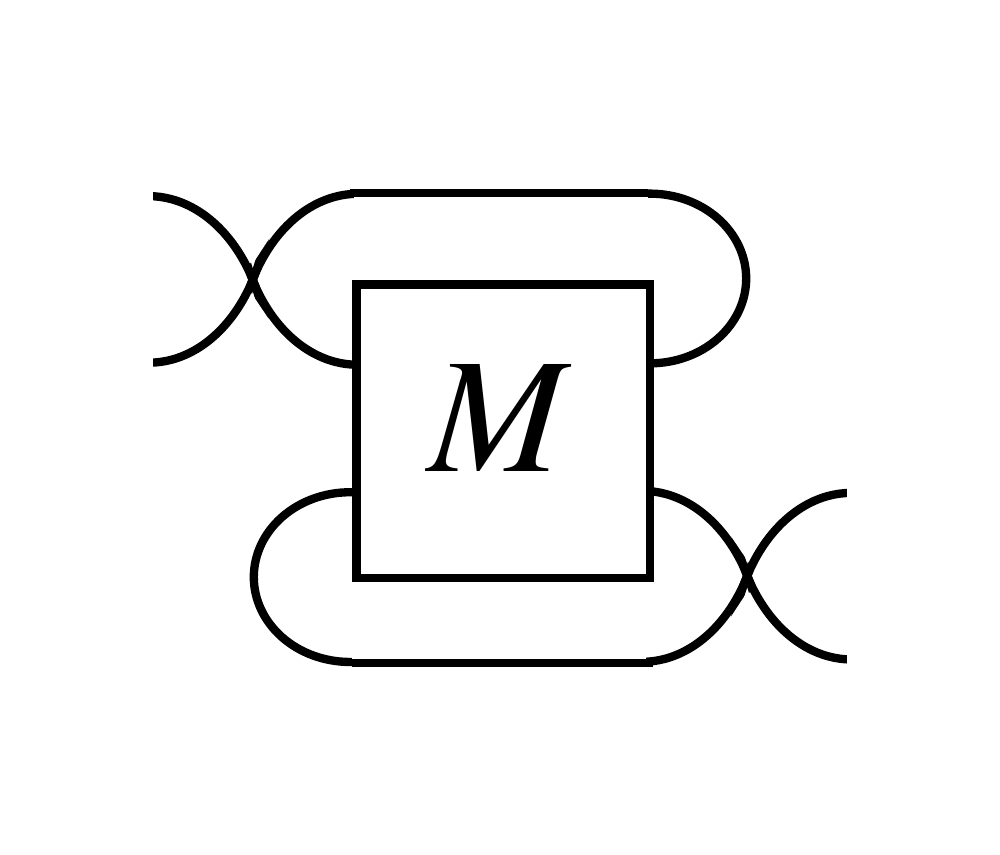}	\quad&\quad
\includegraphics[width=0.1\textwidth,trim=1.5cm 1.5cm 1.5cm 1.5cm,clip]{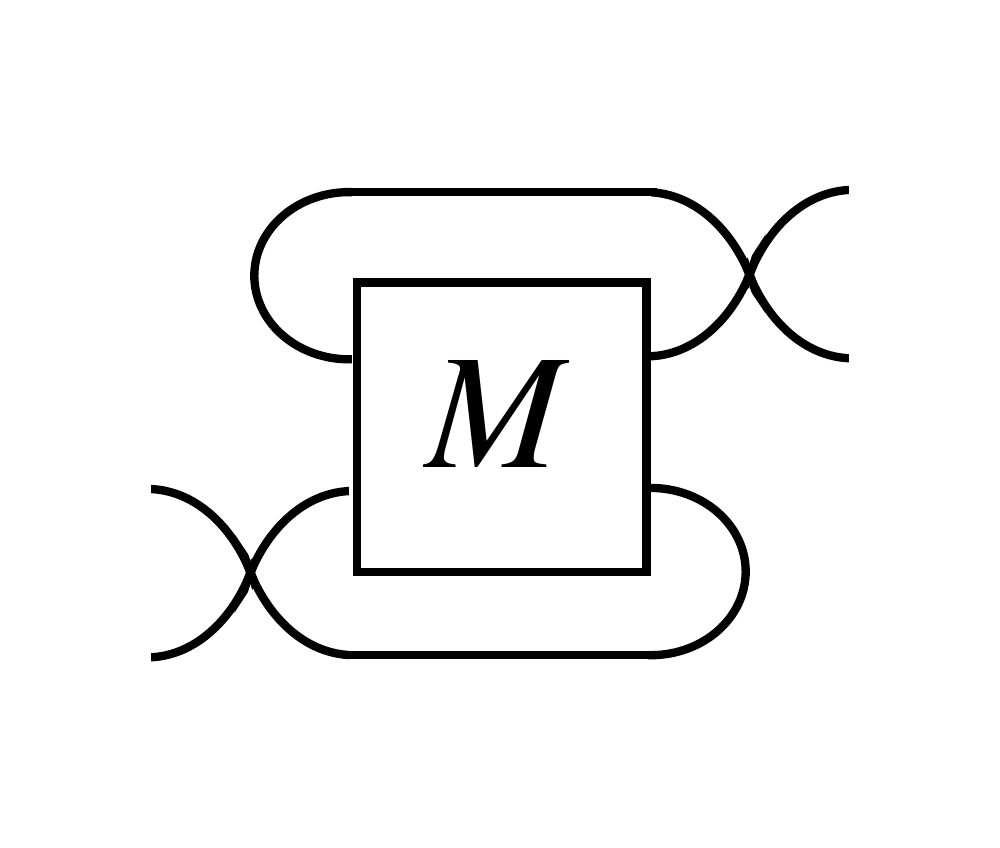}	\\
\footnotesize{(a) Partial Trace} \quad&\quad
\footnotesize{(b) Partial Trace} \quad&\quad
\footnotesize{(c) Transpose} \quad&\quad
\footnotesize{(d) Bipartite-Swap} \quad&\quad
\footnotesize{(e) Row-Reshuffle} \quad&\quad
\footnotesize{(f) Col-Reshuffle} 
\\
\footnotesize{$\Tr_{\2 X}[M]$} \quad&\quad
\footnotesize{$\Tr_{\2 Y}[M]$} \quad&\quad
\footnotesize{$M^T$} \quad&\quad
\footnotesize{$M^S$} \quad&\quad
\footnotesize{$M^{R_r}$} \quad&\quad
\footnotesize{$M^{R_c}$} 
 \end{tabular}
 \label{fig:bipartite}
\end{center}
 In terms of the tensor components of $M$ these operations are respectively given by:
\begin{center}
\begin{tabular}{lll}
Partial trace over $\2X$
\quad&\quad	$\Tr_{\2 X}: {\mathcal L}(\2X\otimes\2Y)\rightarrow {\mathcal L}(\2Y),$ 
\quad&\quad	$M_{m\mu,n\nu} \mapsto \sum_{m} M_{m\mu,m\nu}$
\\
Partial trace over $\2Y$
\quad&\quad	$\Tr_{\2 Y}: {\mathcal L}(\2X\otimes\2Y)\rightarrow {\mathcal L}(\2X)$
\quad&\quad	$M_{m\mu,n\nu} \mapsto \sum_{\mu} M_{m\mu,n\mu}$
\\
Tranpose
\quad&\quad	$T: {\mathcal L}(\2X\otimes\2Y)\rightarrow {\mathcal L}(\2X\otimes\2Y),$ 
\quad&\quad	$M_{m\mu,n\nu} \mapsto M_{n\nu,m\mu}$
\\
Bipartite-\SWAP{}
\quad&\quad	$S: {\mathcal L}(\2X\otimes\2Y)\rightarrow {\mathcal L}(\2Y\otimes\2X),$ 
\quad&\quad	$M_{m\mu,n\nu} \mapsto M_{\mu m, \nu n}$
\\
Row-reshuffling
\quad&\quad	$R_r: {\mathcal L}(\2X\otimes\2Y)\rightarrow {\mathcal L}(\2Y\otimes\2Y,\2X\otimes\2X),$ 
\quad&\quad	$M_{m\mu,n\nu} \mapsto M_{m, n , \mu, \nu} $
\\
Col-reshuffling
\quad&\quad	$R_c: {\mathcal L}(\2X\otimes\2Y)\rightarrow {\mathcal L}(\2X\otimes\2X,\2Y\otimes\2Y),$ 
\quad&\quad	$M_{m\mu,n\nu} \mapsto M_{\nu \mu, n m}$
\end{tabular}
\end{center}
Note that we will generally use reshuffling $R$ to refer to col-reshuffling $R_c$. Similarly we can represent the partial transpose operation by only transposing the wires for $\2X$ (or $\2Y$), and the partial-\SWAP{} operations by only swapping the left (or right) wires of $M$.\\
\end{fullpage}

\subsection{Vectorization of Matrices}
\label{sec:vec}
We now recall the concept of \emph{vectorization} which is a reshaping operation, transforming a $(m\times n)$-matrix into a $(1\times mn)$-vector~\cite{Horn1985}. This is necessary for the description of open quantum systems in the superoperator formalism, which we will consider in \S~\ref{sec:sop}.
\begin{definition}
Vectorization can be done with using one of two standard conventions: \emph{column-stacking} (col-vec) or \emph{row-stacking} (row-vec). Consider two complex Hilbert spaces $\2 X\cong \C^m, \2 Y\cong \C^n$, and linear operators  $A\in {\mathcal L}(\2X,\2Y)$ from $\2 X$ to $\2 Y$. Column and row vectorization are the mappings
\begin{eqnarray}
\mbox{col-vec: } {\mathcal L}(\2X,\2Y)&\rightarrow& \XY: \,\, A\mapsto\dket{A}_c\\
\mbox{row-vec: } {\mathcal L}(\2X,\2Y)&\rightarrow& \YX: \,\,A\mapsto\dket{A}_r
\end{eqnarray} 
respectively, where the operation col(row)-vec when applied to a matrix, outputs a vector with the columns (rows) of the matrix stacked on top of each other. 
\end{definition}

\begin{myfullpage}

\begin{illexample}

Graphical representations for the row-vec and col-vec operations are found from bending a wire to the left either clockwise or counterclockwise respectively:
\begin{center}
\begin{tabular}{c|c}
\includegraphics[height=3.5em]{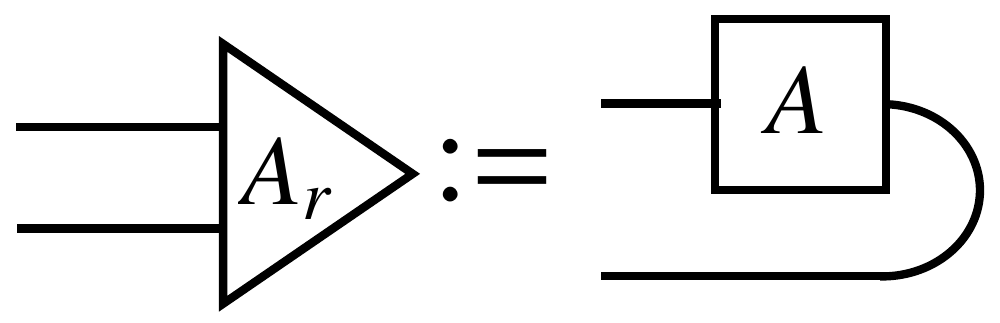} 
\quad\quad&\quad\quad
\includegraphics[height=3.5em]{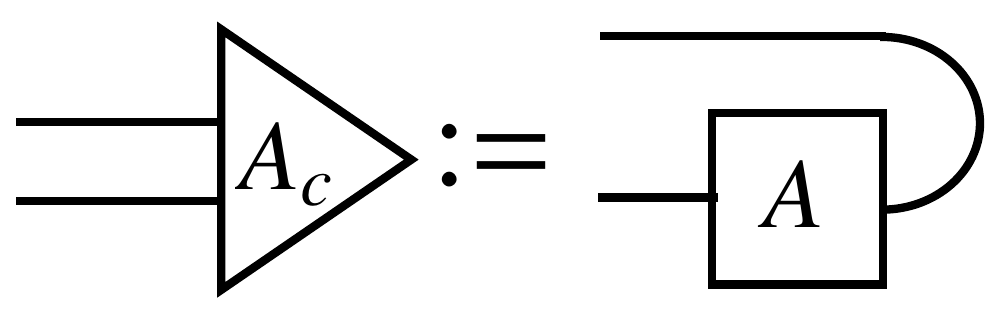} 
\\
\footnotesize{(a) Row-vec} &
\footnotesize{(b) Col-vec}
\end{tabular}
 \label{fig:vectorization}
\end{center}
Vectorized matrices in the col-vec and row-vec conventions are naturally equivalent under wire exchange (the \SWAP{} operation)
\begin{center}
\includegraphics[width=.55\textwidth]{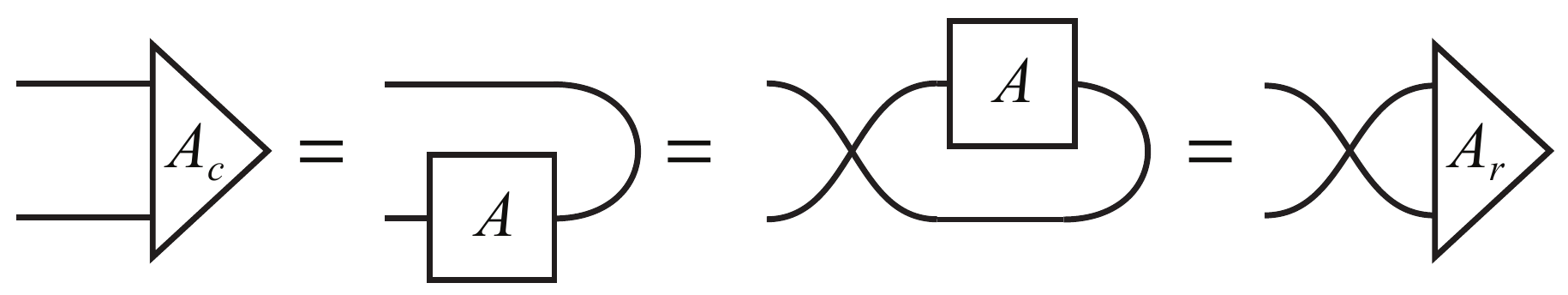} \label{fig:col-to-row}    
\end{center}

In particular we can see that the unnormalized Bell-state $\ket{\Phi^+}\in\2X\otimes\2X$ is in fact the vectorized identity operator $\I\in {\mathcal L}(\2X)$
\begin{equation}
\ket{\Phi^+} = \dket{\I}_r = \dket{\I}_c.
\end{equation}
\end{illexample}
\end{myfullpage}
\begin{definition}
We may also define a vectorization operation with respect to an arbitrary operator basis for ${\mathcal L}(\2X,\2Y)$. Let $\2 X\cong\C^{d_x}, \2 Y\cong \C^{d_y}$, and $\2Z \cong \C^D$ where $D=d_x d_y$. Vectorization with respect to an orthonormal operator basis $\{\sigma_\alpha:\alpha=0,...,D-1\}$ for ${\mathcal L}(\2X,\2Y)$ is given by
\begin{equation}
	\sigma\mbox{-vec: }{\mathcal L}(\2X,\2Y)\rightarrow \2Z :\,\,A\mapsto \dket{A}_\sigma.
\end{equation}
This operation extracts the coefficients of the basis elements returning the vector  
\begin{equation}
	\label{eqn:trace-rep}
	\dket{A}_\sigma := \sum_{\alpha=0}^{D-1} \Tr[\sigma^\dagger_\alpha A] \ket{\alpha}
\end{equation}
where $\{\ket{\alpha}: \alpha=0,...,D-1\}$ is the standard basis for $\2Z\cong\C^D$. This is depicted in our graphical calculus as 
\begin{center}
\includegraphics[width=.65\textwidth]{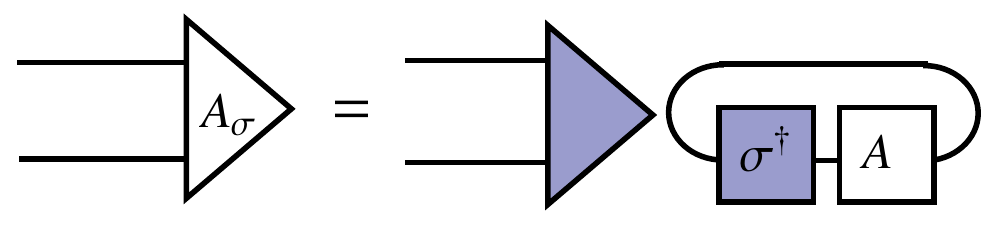} \label{fig:basis-vectorization}  
\end{center}
\begin{remark}
  To distinguish between these different conventions we use the notation $\dket{A}_x$ to denote the vectorization of a matrix $A$, were the subscript $x=c,r,\sigma$ labels which convention we use; either $c$ for col-vec, $r$ for row-vec, or $\sigma$ for an arbitrary operator basis.
\end{remark}
 
For the case $\2 X\cong\2 Y\cong \C^d$, we can define row-vec and col-vec in terms terms of~\eqref{eqn:trace-rep} by taking our basis to be the elementary matrix basis $\{ E_{i,j}=\ketbra{i}{j} : i,j=0,...,d^2-1\}$, and making the assignment $\alpha=di+j$ and $\alpha=i+dj$ respectively. Hence we have
\begin{eqnarray}
\dket{A}_r &:=& \sum_{i,j=0}^{d-1} A_{ij}\, \ket{i}\otimes\ket{j}	\label{eqn:rowvec}\\ 
\dket{A}_c &:=& \sum_{i,j=0}^{d-1} A_{ij}\, \ket{j}\otimes\ket{i}	\label{eqn:colvec}.
\end{eqnarray} 
\end{definition}

\begin{example}
Using the definition of the unnormalized Bell-state $\ket{\Phi^+}$ and summing over $i$ and $j$ one can rewrite~\eqref{eqn:rowvec} and \eqref{eqn:colvec} as
\begin{eqnarray}
\dket{A}_r &=& (A \otimes \I) \ket{\Phi^+}\label{eqn:row-vec}\\
\dket{A}_c &=& (\I \otimes A) \ket{\Phi^+}\label{eqn:col-vec}
\end{eqnarray}
which are the equational versions of our graphical definition of row and col vectorization shown in \eqref{fig:vectorization}.
 \end{example}
 
When working in the superoperator formalism for open quantum systems, it is sometimes convenient to transform between vectorization conventions in different bases. Given two orthonormal operator bases $\{\sigma_\alpha\}$ and $\{\omega_\alpha\}$ for ${\mathcal L}(\2X,\2Y)$, the basis transformation operator 
\begin{equation}
T_{\sigma\rightarrow\omega}: \2Z\rightarrow\2Z: \dket{A}_\sigma\mapsto \dket{A}_\omega
\end{equation} 
transforms vectorized operators in the $\sigma$-vec convention to the $\omega$-vec convention. Graphically this is given by

\begin{center}
\includegraphics[width=0.5\textwidth]{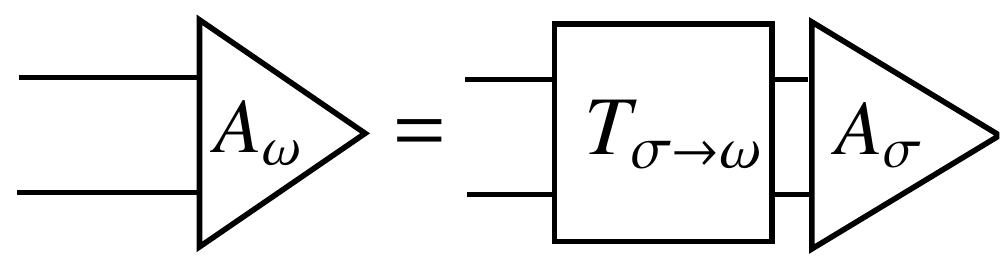} \label{fig:vec-basis-change}
 \end{center}
The basis transformation operator $T_{\sigma\rightarrow\omega}$ is given by the equivalent expressions
\begin{equation}
T_{\sigma\rightarrow\omega} 
	= \sum_\alpha \ket{\alpha}\dbra{\omega_\alpha}_\sigma 	
	= \sum_\alpha \dket{\sigma_\alpha}_\omega\bra{\alpha} 	\label{eqn:vec-change},
\end{equation}
and the corresponding graphical representations are: 
\begin{center}
\includegraphics[width=0.8\textwidth]{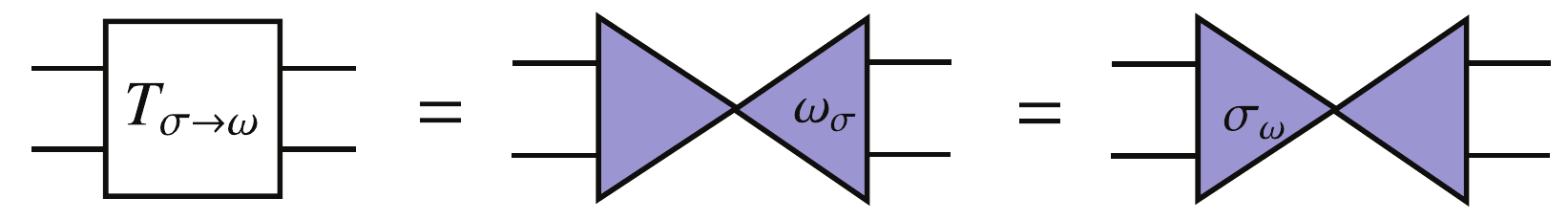}\label{fig:basis-change-op}
\end{center}

As in \cite{2011arXiv1111.6950W}, we tend to use the col-vec convention by default, and drop the vectorization label subscripts unless referring to a general $\sigma$-basis. The main transformation we will be interested in is then from col-vec to another arbitrary orthononormal operator basis $\{\sigma_\alpha\}$. Tensor networks for the change of basis $T_{c\rightarrow\sigma}$ and its inverse $T_{\sigma\rightarrow c}$ are
\begin{center}
 \begin{tabular}{c|c}
\includegraphics[width=0.45\textwidth]{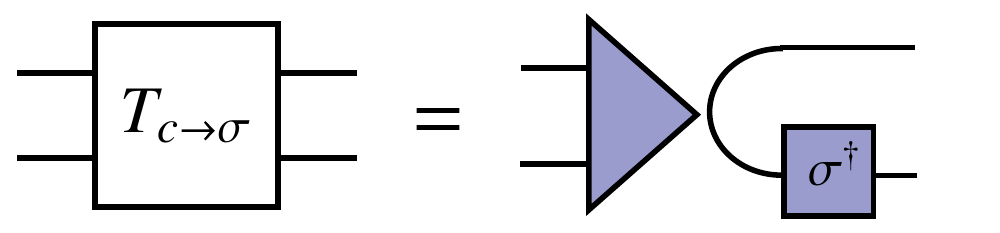}	
\quad\quad&\quad\quad
\includegraphics[width=0.45\textwidth]{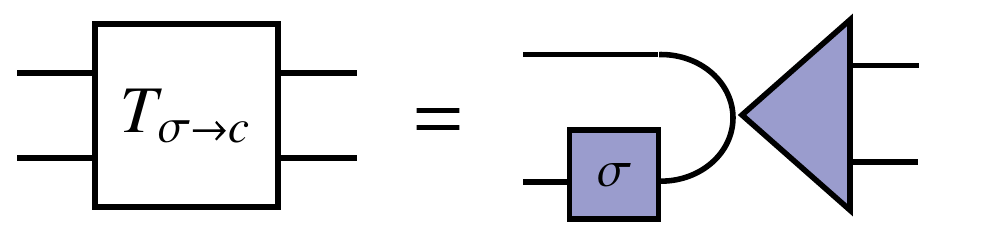}	
\\
\footnotesize{(a) Col-vec to $\sigma$-basis} &
\footnotesize{(b) Row-vec to $\sigma$-basis} 
\end{tabular}
 \label{fig:vec-basis-col}
\end{center}
In the case where one wants to convert to row-vec convention, as previously shown the transformation is given by
\begin{equation}
T_{c\rightarrow r}=T_{r\rightarrow c}= \mbox{SWAP}.
\end{equation}

One final important result that often arises when dealing with vectorized matrices is Roth's Lemma for the vectorization of the matrix product $ABC$~\cite{Horn1985}.  Given matrices $A,B,C\in {\mathcal L}(\2X)$ we have
\begin{eqnarray}
\dket{ABC} &=& (C^T\otimes A) \dket{B}
\label{eqn:roth}
\end{eqnarray}
The graphical tensor network proof of this lemma is as follows:\mn{The theory of tensor networks leverages one to study the
mathematical structure formed by the composition of processes and states on the same footing.} 
\begin{center}
\includegraphics[width=0.45\textwidth]{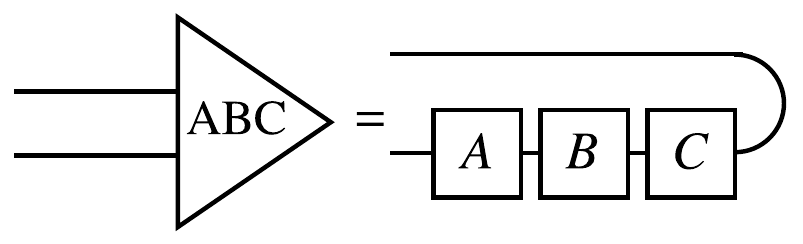}
\includegraphics[width=0.45\textwidth]{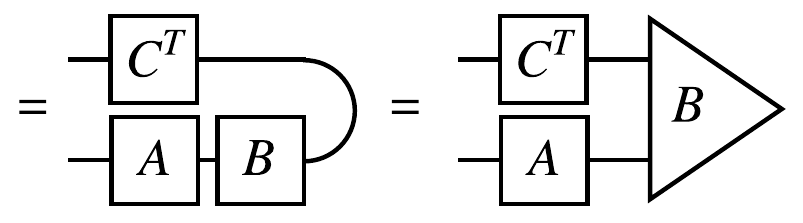}
\label{fig:vecabc}
\end{center}

\section{Worked Examples}
\label{app:tensor1}

We will now prove the consistency of several of the basic tensor networks introduced in Part~\ref{sec:tensor}, and in doing so illustrate how one may use our graphical calculus for diagrammatic reasoning.  

The color summation convention we have presented represents diagrammatic summation over a tensor index by coloring the appropriate tensors in the diagram. In this convention summation over a Kronecker delta, $\sum_{i,j}\delta_{ij}=\sum_{i,j} \braket{i}{j}$, is as shown:
\be
\includegraphics[width=0.3\textwidth]{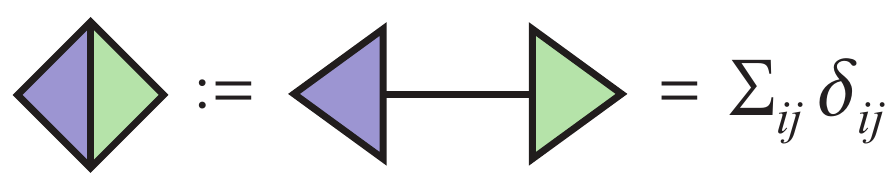}
\label{fig:delta-sum}
\ee
This expression is used in several of the following proofs.

We begin with the proof of the trace of an operator $A$:
\be
\includegraphics[width=0.3\textwidth]{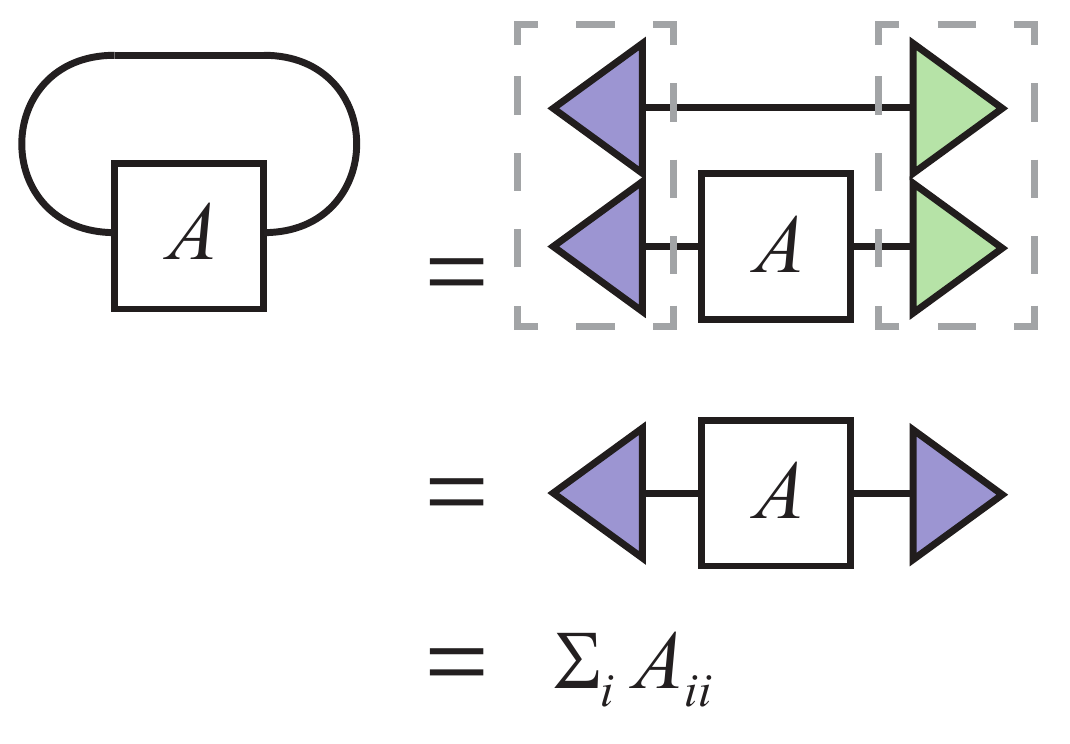}
\label{fig:trace-proof}
\ee
 
For illustrative purposes, to prove this algebraically we note that the tensor networks for trace correspond to the algebraic expressions $\bra{\Phi^+}A\otimes\I\ket{\Phi^+}$ and $\bra{\Phi^+}\I\otimes A\ket{\Phi^+}$, and that
\begin{eqnarray}
\bra{\Phi^+}\I\otimes A\ket{\Phi^+}
&=& \sum_{i,j}\braket{i}{j}\bra{i}A\ket{j}
= \sum_{i,j} \delta_{ij} A_{ij}\nonumber\\
&=& \sum_i A_{ii}\\
&=& \Tr[A].\nonumber
\end{eqnarray}
Similarly we get $\bra{\Phi^+}A\otimes\I\ket{\Phi^+}=\Tr[A]$.

To prove the snake equation we must first make the following equivalence for tensor products of the elements $\ket{i}$ and $\bra{j}$:
\begin{equation}
\bra{j}\otimes\ket{i}\equiv\ket{i}\otimes\bra{j}\equiv \ketbra{i}{j}
\label{eqn:tensor-product-equiv}
\end{equation}
This is illustrated diagrammatically as
\be
\includegraphics[width=0.4\textwidth]{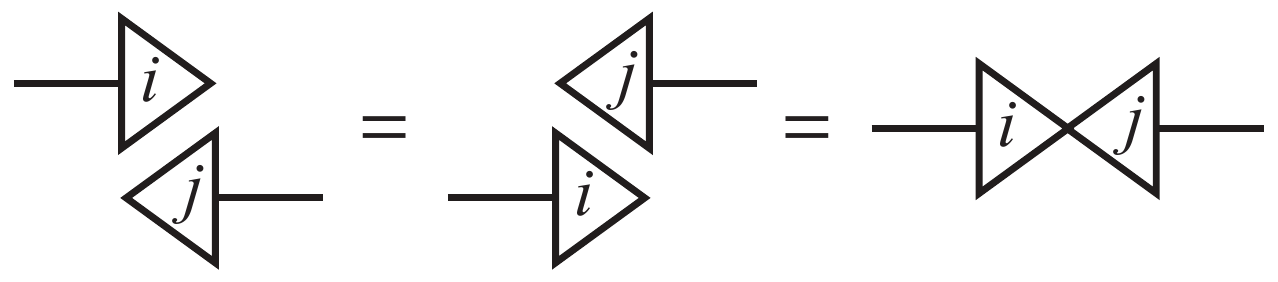}
\label{fig:tensor-product-equiv}
 \ee

With this equivalence made, the proof of the snake-equation for the ``S''  bend is given by
\be
\includegraphics[width=0.3\textwidth]{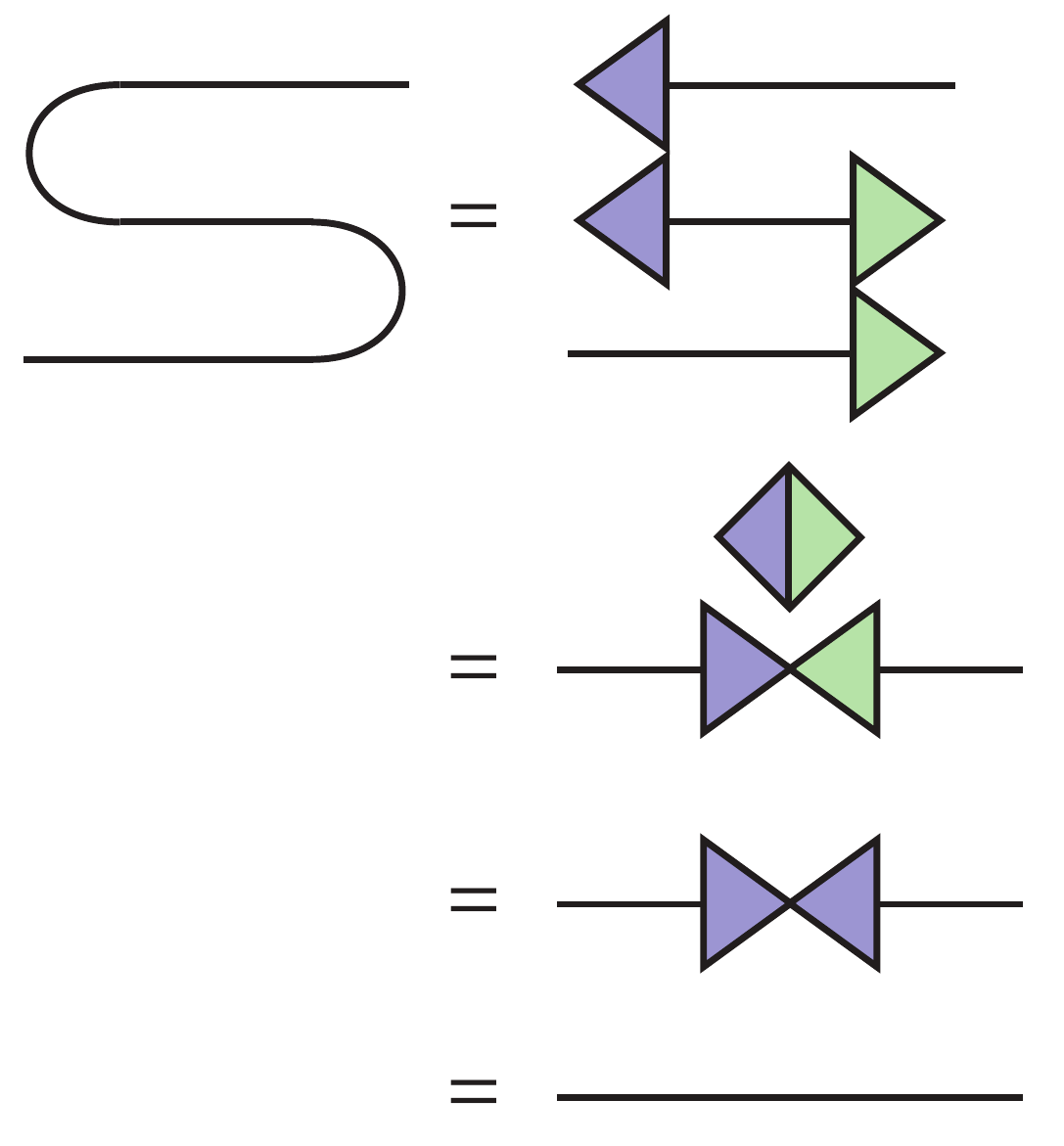}
\label{fig:snake-proof}
\ee
The proof for the reflected ``S'' snake-equation follows naturally from the equivalence defined in \eqref{fig:tensor-product-equiv}.

The proof of our tensor network for the transposition of a linear operator $A$ is as follows:
\be
\includegraphics[width=0.35\textwidth]{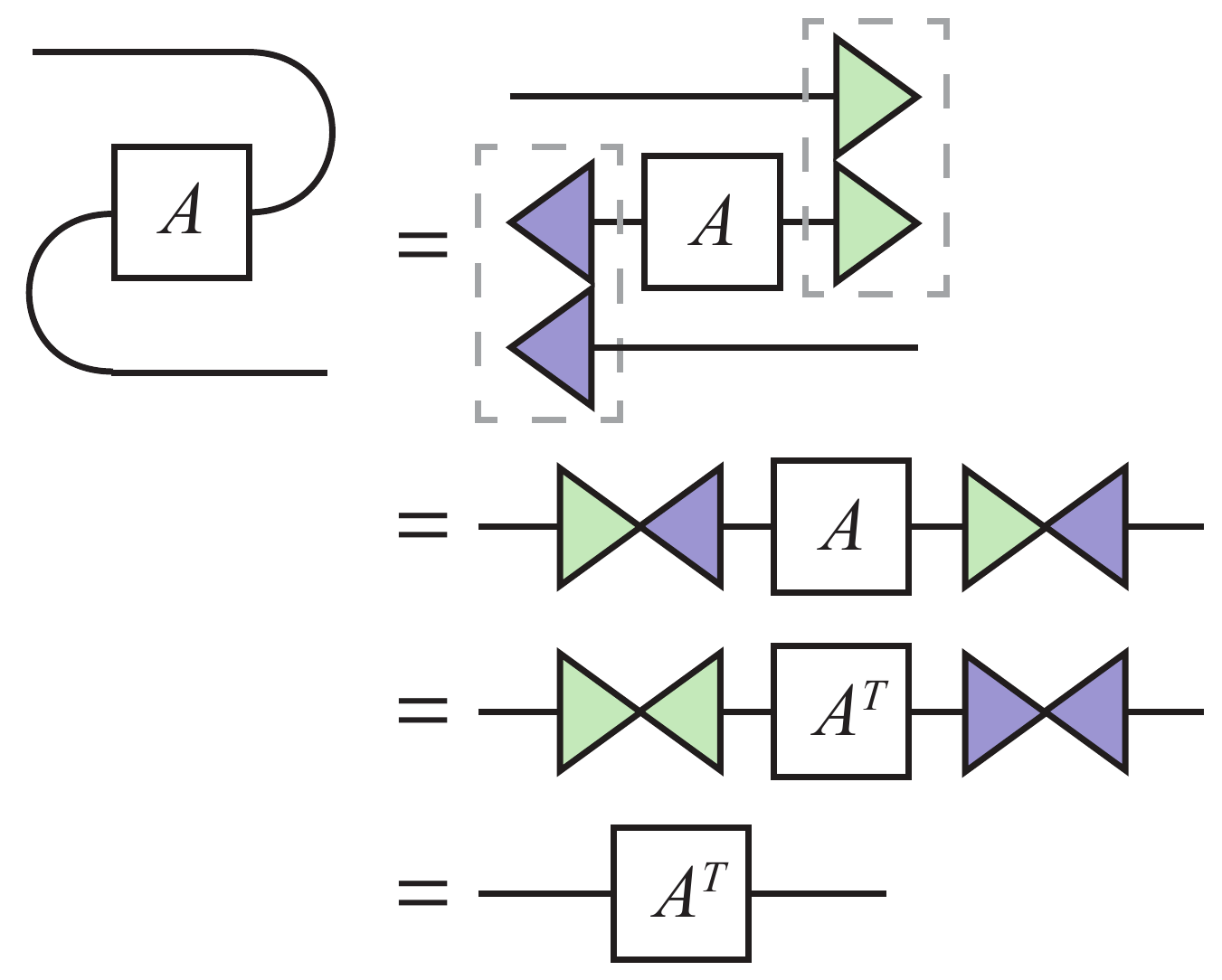}
\label{fig:tensor-transpose-proof-1}
\ee
To prove this algebraically we note that the corresponding algebraic equation for the transposition tensor network is
\begin{eqnarray}
\includegraphics[width=0.09\textwidth]{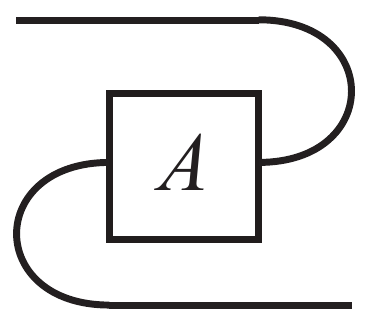}&=&\I\otimes\bra{\Phi^+}(\I\otimes A\otimes\I)\ket{\Phi^+}\otimes\I\\
&=& \sum_{i,j} \bra{j}A\ket{i}\ \ket{i}\otimes\bra{j}\\
&=& \sum_{i,j} \bra{j}A\ket{i} \ketbra{i}{j}\\
&=& \sum_{i,j} \bra{i}A^T\ket{j} \ketbra{i}{j}\\
&=& \sum_{i,j} \ketbra{i}{i}A^T\ketbra{j}{j}\\
&=& A^T.
\end{eqnarray}
The proof for transposition by counter-clockwise wire bending follows from the equivalence relation in \eqref{eqn:tensor-product-equiv} and \eqref{fig:tensor-product-equiv}. 

With the tensor network for transposition of an operator proven, the proof of transposition by contracting through a Bell-state $\ket{\Phi^+}$ is then an application of the snake equation as shown:
\be
\includegraphics[width=0.3\textwidth]{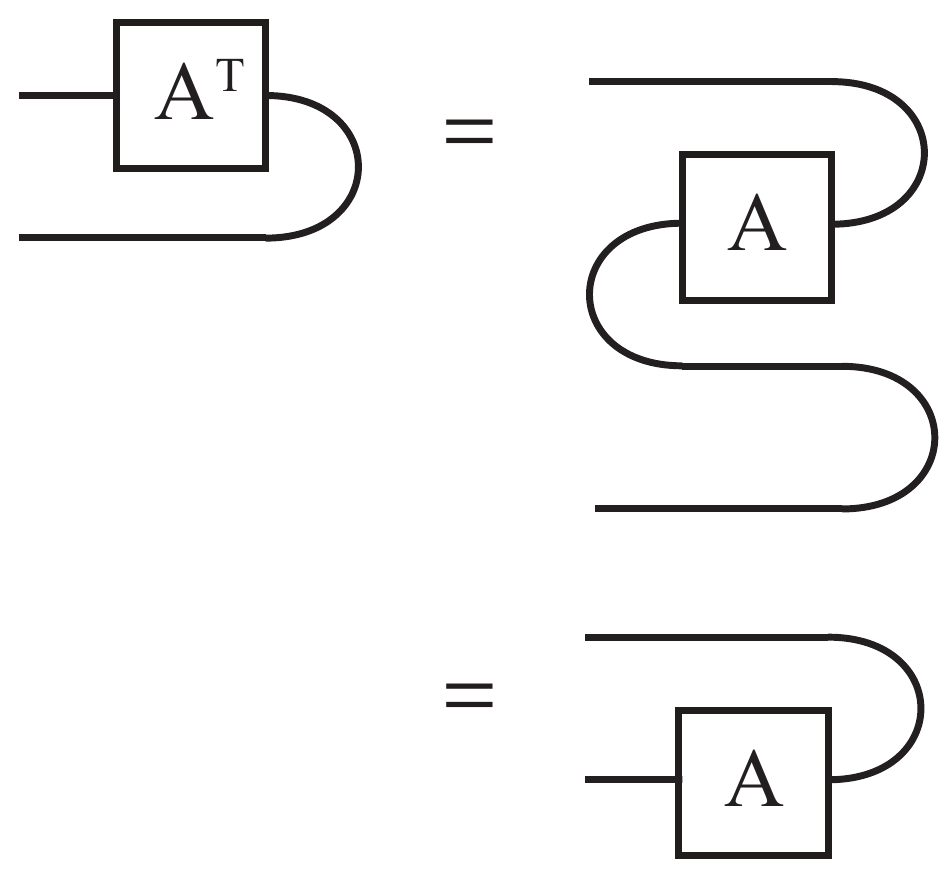}
\label{fig:bell-tr-proof}
\ee  

\section{Problems} 

\begin{myexercise}(Bell basis). 
Using the right hand side equations, 
\begin{eqnarray*}
\sigma_0 = \sum \ketbra{a}{a},\qquad \sigma_1 = \sum \ketbra{1-a}{a},
\\
\sigma_2=i\sum (-1)^{a+1} \ketbra{a}{1-a},\qquad \sigma_3 =\sum (-1)^a\ketbra{a}{a}.
\end{eqnarray*}
for $a\in \{0,1\}$. 
Show that 
\begin{equation}
\frac{1}{\sqrt{2}}\left(\sum_{l\in\{0,1\}} \bra{l} \otimes \bra{l} \right) \sigma_i \otimes \I = \frac{1}{\sqrt{2}}\sum_{l\in\{0,1\}} (\bra{l}\sigma_i) \otimes \bra{l},
\end{equation}
defines the Bell effects (which defines an orthonormal basis in $\mathbb{C}^2\otimes \mathbb{C}^2$). Here $\sigma_i$ indexes the Pauli matrices.\mn{Effects are dual to states.  Also called, costates.}
\end{myexercise}

\begin{remark}
Complete the following problems using standard techniques and then compare this to graphical tensor network approach. 
\end{remark}
\begin{myexercise}
Suppose $\psi \in \mathbb{C}_2^{\otimes 2}$ and $E>0$ with $E\in \mathcal{L}(\mathbb{C}_2^{\otimes 2})$. Show that 
\begin{equation}
    \bra{\psi}E\otimes \I \ket{\psi}, 
\end{equation}
takes the same values when $\psi$ is any of the four Bell states.
\end{myexercise}

\begin{myexercise}[Transpose or Ricochet Trick]
Show that 
\begin{equation}
(M^A\otimes \I^B)\ket{\Phi}^{AB} = (\I^A \otimes (M^\top)^B)\ket{\Phi}^{AB},
\end{equation}
for maximally entangled $\ket{\Phi}^{AB}$ and any matrix $M$. 
\end{myexercise}


\begin{myexercise}
Show that the purity $P(\rho^A)$ is equal to 
\begin{equation}
    P(\rho^A) = \Tr\{ (\rho^A\otimes \rho^{A'}) F^{AA'}\},
\end{equation}
where Hilbert space $A$ is isomorphic to $A'$ and $F^{AA'}$ is the swap operator defined on a basis (indexed by $x,y$) as 
\begin{equation}
    F^{AA'}\ket{x}^A\ket{y}^{A'}=\ket{y}^A\ket{x}^{A'}.
\end{equation}
Hint. Establish that 
\begin{equation}
    \Tr\{f(\rho^A)\}= \Tr\{(f(\rho^A)\otimes \I^{A'} )F^{AA'}\},
\end{equation}
for function $f$ of the operators on $A$. 
\end{myexercise}

\section{Further Reading}

Readers should be aware of the high number of quality tutorials covering various aspects of tensor networks available for free download from the arXiv.org preprint server.  Many but not all of these are also published in journals. For those interested in applications to condensed matter, we particularly recommend. \\

\noindent {\bf Hand-waving and Interpretive Dance: An Introductory Course on Tensor Networks} \cite{2016arXiv160303039B}\\
Jacob Bridgeman and Christopher Chubb \\
J.~Phys.~A: Math.~Theor.~50 223001 (2017)
\href{https://arxiv.org/abs/1603.03039}{arXiv:1603.03039}\\

For those interested in the mathematics of string diagrams (category theory), the most accessible introduction covering the foundations of the building blocks presented in this chapter can be found in. \\

\noindent {\bf A Prehistory of n-Categorical Physics} \cite{Baez}\\
John Baez and Aaron Lauda\\
Deep Beauty 13–128, Cambridge University Press (2011) \href{https://arxiv.org/abs/0908.2469}{arXiv:0908.2469} \\

In addition to this work on the categories of tensor networks \cite{Baez}, readers might also find the survey \cite{catQM} of interest.  

Category theory is worth considering as a mathematical framework to describe wire diagrams (including quantum circuits).  Category theory itself is sophisticated enough to present a theorem that essentially proclaims that the diagrams contain {all} relevant information.  Hence one might say that {\it category theory formally rules out the need for category theory}---provided one knows how to manipulate the tensor diagrams. This is formally stated in the following well-known theorem. 

\begin{theorem}[Coherence for categories \cite{Selinger_2010}] A well-formed equation between two morphism terms in the language of categories follows from the axioms of categories if and only if it holds in the graphical language up to isomorphism of diagrams.
\end{theorem}

Hence, readers should be aware that the formal mathematics of tensor networks finds its roots in {\it tensor category theory} (now often renamed {\it dagger categories}). And that the diagrams themselves are proven to contain the relevant information \cite{Selinger_2010}. The topic of modeling the theory of quantum mechanics using categories was explored intensively in the area the authors from \cite{catQM} named, {\it categorical quantum mechanics}---see the book \cite{Coecke2017} for a survey. 

There are two fundamental types of tensor networks in wide use today.  The most common is similar to quantum circuits (which is the topic of this book).  The second is the braided class of tensor networks, used in topological quantum computing.  In terms of active research, recently a class of tensor networks was discovered by Jaffe, Liu and Wozniakowski---the JLW-model---notably, the wires carry charge excitations \cite{Liu2017, Jaffe2016,Jaffe2018}.  The rules in which network components can be moved, merged and manipulated in a graphical form of reasoning take an elegant form with known applications to quantum protocols~\cite{Jaffe2017}.  For instance the relative locations on wires carries precise meaning and changing the ordering modifies a connected network specifically by a complex number.  The type of isotopy discovered in the topological JLW-model provides an alternative means to reason about quantum information, computation and protocols.  Some open problems related to the JLW-model are given in \cite{JaffeLiu2017}.

\newpage 
\part{Matrix Product States} 

One of the most common uses of tensor networks in quantum information is representing states which belong to small but physically relevant subset of a must larger Hilbert space.  This often includes low-entanglement states. The  backbone of this idea rests on low rank matrix approximations which we will consider in this chapter.  

We will see that if one partitions a network, by cutting it in two, the number of wires that were cut in this process provides an upper bound on the maximum amount of possible entanglement between spins. 
This can be made more precise by considering the (unitarily invariant) entanglement entropy of a bipartite split.  
\begin{equation}\label{eqn:entropy}
E = -\Tr\left(\rho \text{ln} \rho \right) = - \sum_i \lambda_i \text{ln} \lambda_i.
\end{equation}  
Here $\lambda_i$ are the singular values of the reduced density operator of either subsystem.  The quantity is maximized for all $\lambda_i$ equal to the inverse of the dimension of a reduced density matrix. The value of \eqref{eqn:entropy} provides a quantitative measure of correlations.  This will be elaborated as a central concept in what follows.   
\begin{marginfigure}
\centering
\includegraphics[width=6\xxxscale]{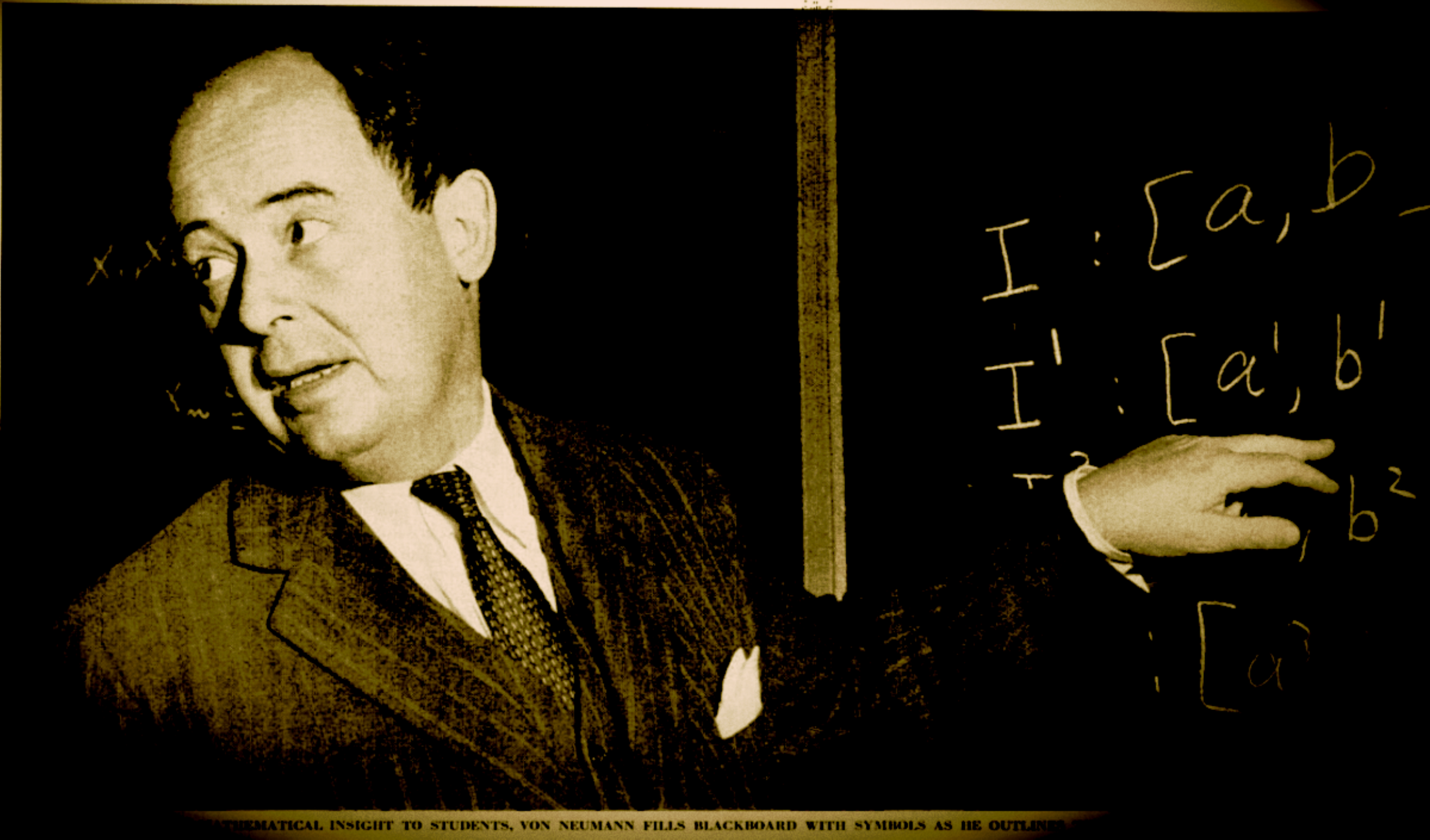}
\caption{``\textit{You should call it entropy, for two reasons. In the first place your uncertainty function has been used in statistical mechanics under that name, so it already has a name. In the second place, and more important, no one really knows what entropy really is, so in a debate you will always have the advantage.}''\\~\\
John von Neumann suggesting to Claude Shannon a name for his new uncertainty function, as quoted in Scientific American {\bf 225}(3) 180, (1971).}
\end{marginfigure}

\section{The Diagrammatic SVD}

In this section, we will introduce a diagrammatic form of the singular
value decomposition (SVD).  It is assumed that the reader has solved Exercise \ref{exercise:svd1} (see also Exercise \ref{exercise:rank1}). 

There are several utilities to our
approach.  The first stems from the fact that the known invariants we
have studied can be simplified by network contraction using the diagrammatic SVD.   
The method factors tensor into well defined building blocks
with simplistic interaction properties: black \COPY-tensors and white
unitary boxes.  We will also consider the iteration of this process, allowing one to arrive at matrix product states (MPS) in terms of our network building blocks. 

An aim of the present work is to consider how graphical depictions of tensors can leverage a better understanding of how certain properties evident in a network are reflected in properties related to a quantum state.  Intuitively one thinks of a connected network as representing a correlated or entangled state.  This chapter will push this idea further.


\COPY{}-tensors have been studied in the setting of the Penrose
tensor calculus, in work dating back at least to Lafont \cite{boolean03} and explored more recently using an alternative notation in the so called ZX-calculus \cite{CD, redgreen} --- see also \cite{BB11, CTNS}.
Here we apply \COPY{}-tensors in the diagrammatic SVD factorization of quantum states.  The key diagrammatic properties are illustrated as follows and explained in detail in \S~\ref{sec:btn}.  

\begin{definition}[Properties of the \COPY-tensor]\label{def:copy-prop}
The copy property is illustrated in (b).  Here a basis state, or copy point $\ket{x}$ is contracted with the \COPY{}-tensor, which then breaks into two copies of $\ket{x}$. These tensors are defined in $d$ dimensions.  If the
tensor was written in the standard basis, it would copy $\ket{0}$ as well as
$\ket{1}$ and for the case of qubits be written as $\ket{0}\bra{00}+\ket{1}\bra{11}$. We use the plus symbol $\ket{+}$ to represent an equal sum over all so called, copy points of a \COPY{}-tensor.  When contracting a \COPY-tensor with such a basis state, the effect is to prune an arm or leg, as shown in (c) and (d).  This is also called a unit.  \COPY-tensors can be composed.  (e) illustrates that appropriate composition of two \COPY-tensors is equal to the identity, as would be expected.
\begin{center}
 \includegraphics[width=15\xxxscale]{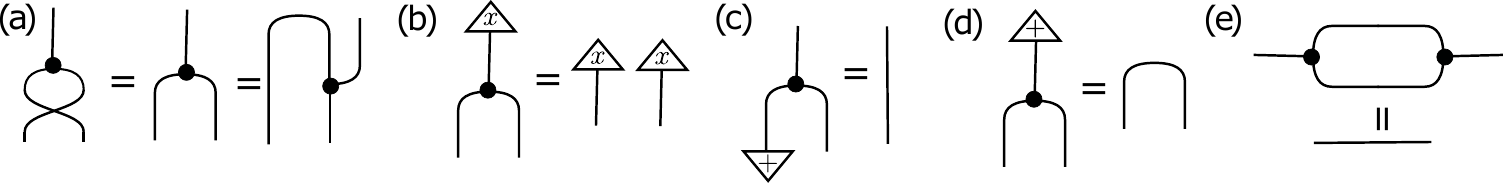}
\end{center}
\end{definition}

\begin{remark}[\COPY-tensors on as generalized delta functions]
 The valence-three \COPY-tensor in terms of components for the case of qubits can be expressed as 
 \be 
 \delta_{ijk} = (1-i)(1-j)(1-k)+ijk, 
 \ee 
 which can be thought of a generalized delta function, where $i,j,k=0,1$.  This can be extended to valence-n \COPY-tensors as 
 \be 
 \delta_{ijk...l} = (1-i)(1-j)(1-k)...(1-l)+ijk...l.
 \ee 
 In figure (e) 
 from Definition \ref{def:copy-prop}, the composition of two \COPY-tensors becomes 
 \be 
 \delta^i_{~jk}\delta^{jk}_{~~m} = \delta^i_{~m}.
 \ee 
We note that the valence-two delta tensor in components is given as 
 \be
 \delta^i_{~j} = 1-(i-j)^2 = (1-i)(1-j)-ij,
 \ee
 where $i,j=0,1$. 
\end{remark}

\begin{myexercise}\label{exercise:rank1}
(Rank-1 projectors). 
Show that a non-trivial operator $P_\star$ is a Schmidt rank-1 projector if and only if it can be written as $\ketbra{\psi}{\psi}$. \\

\end{myexercise}

\begin{theorem}[Diagrammatic SVD]
\label{theorem:diagrammatic-SVD}
Every operator $A: \hilb{H}_1 \to \hilb{H}_2$ can be
factored into a nonnegative order-one tensor $\Sigma$ (unique), an order-three \COPY-tensor
and unitary order-two tensors~$U$ and~$V$:
\begin{center}
 \includegraphics[width=0.7\textwidth]{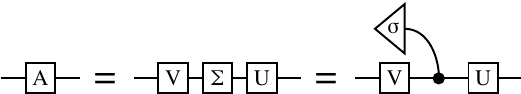}
\end{center}
\begin{proof}
Using the SVD, every valence-two tensor $A^i_l$ can be written as
\be 
A^i_l = U^i_j~\Sigma^j_k~V^k_l,
\ee
where $U$ and $V$ are unitary and $\Sigma$ is diagonal and nonnegative
in the standard bases of $\hilb{H}_1$ and $\hilb{H}_2$.
$\Sigma$ can be written as
\be 
\begin{aligned}
\Sigma =& \sum_{j=0}^{\min(d_1, d_2)-1} \sigma_j \ket{j}_2\bra{j}_1, \\
=& \underbrace{\sum_{i=0}^{\min(d_1, d_2)-1} \ket{i}_2 \bra{i}_1}_{Q_{12}}
\underbrace{\sum_{j} \ket{j}_1\bra{jj}_1}_{\COPY}
\underbrace{\sum_k \sigma_k \ket{k}_1}_{\sigma},
\qquad (\text{where}~~\sigma_k \ge 0).
\end{aligned}
\ee 
We have then expressed the tensor $\Sigma$ as a contraction of an
order-one tensor~$\sigma$ with the \COPY-tensor.
The tensor $Q_{12}$ is only necessary if $\hilb{H}_1$ and $\hilb{H}_2$
have different dimension.
\end{proof}
\end{theorem}

\begin{remark}[Eckart-Young-Mirsky Theorem]
The rank of a matrix~$A$ is the number of non-zero singular values it has.
To determine its optimal rank-$r$ approximation (with $r < \text{rank}(A)$), we can turn to a classic theorem by Eckart and Young which was generalized by Mirsky.

Given the SVD, $A = U\Sigma V^\dagger$, we will discard $\text{rank}(A)-r$ smallest singular values in~$\Sigma$ by setting them to zero, obtaining~$\Sigma'$.
This process is often called trimming.

This gives rise to $A' = U\Sigma' V^\dagger$, an approximation of~$A$. 
\begin{theorem}[Eckart-Young-Mirsky] For $m\times m$ matrices $A$, $A'$
\begin{equation}
\lVert A -A' \rVert =  \min_{\text{rank}(\hat{A}) \le r} \lVert A -\hat{A}\rVert,
\end{equation}
for any unitarily invariant matrix norm $\lVert \cdot \rVert$ with $r < \text{rank}(A)$.
\end{theorem}
\marginnote{\emph{\textit{Interested readers can try to get their hands on copies of
C. Eckart and G. Young, ``The approximation of one matrix by another of lower rank,'' Psychometrika 1, (1936)
and
L. Mirsky, ``Symmetric gauge functions and unitarily invariant norms,'' The Quarterly Journal of Mathematics 11:1, 50--59 (1960).}}}
Here $\hat{A}$ is any approximation to $A$ of the same or lesser rank as~$A'$.
This implies that truncating or trimming $\Sigma$ in this way yields as good of an approximation as one can expect.
In the following section, we will specifically consider the induced error for such an approximation.
\end{remark}

\begin{corollary}[Diagrammatic Schmidt decomposition]
\label{ex:diagrammatic-Schmidt}
Given a bipartite state $\ket{\psi}$, we use the snake equation to
convert it into a linear map (inside of the dashed region).
Now we apply the SVD as in
Theorem~\ref{theorem:diagrammatic-SVD},
resulting in unitary maps $U$ and $V$, a \COPY{}-tensor and the
order-one tensor~$\sigma$ representing the singular values.
Diagram reorganization leads to the diagrammatic Schmidt decomposition
of~$\ket{\psi}$:
\begin{center}
 \includegraphics[width=12\xxxscale]{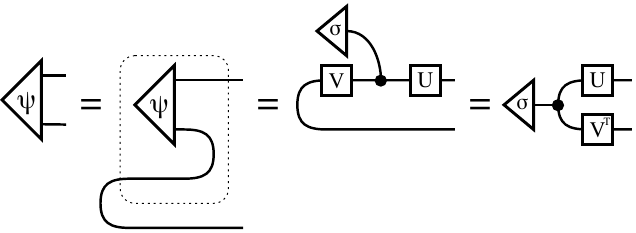}
\end{center}
(Sliding $V$ around the cup takes the transpose but
the resulting map is still unitary).
The singular values $\sigma_0, \ldots, \sigma_{d-1}$ in $\sigma$ correspond to
the Schmidt coefficiets.
\end{corollary}

\begin{example}[Graphical map-state duality]
 In the figure from Corollary \ref{ex:diagrammatic-Schmidt}, we arrive at an example of map-state duality from \cite{Meznaric2014} as 
 \begin{center}
 \includegraphics[width=7\xxxscale]{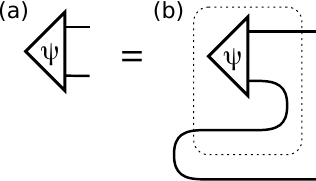}
\end{center}
 In (a) we start with a state $\ket{\psi}$.  We can think of this (vacuously, it would seem) as a state being acted on by the identity operation.  Application of the snake-equation to one of the outgoing wires allows one to transform this into the diagram in (b).  We can think of (b) a Bell state (left) being acted on by a map found from coefficients of the state $\ket{\psi}$.  We have illustrated this map acting on the bell state by a light dashed line around $\ket{\psi}$. 
\end{example}

\begin{definition}[Partition $\chi$]
Given a many-body quantum state, partition the state into two halves and perform the diagrammatic SVD on this system, with respect to this partition.  The term ``$\chi$'' stands for the number of non-zero singular values across a bipartition of a state.  These values are used to compute the entanglement entropy
\eqref{eqn:entropy} of either reduced subsystem.  
\end{definition}

\begin{example}[Entanglement topology]
The most significant topology change occurs when the input state to
the black \COPY-tensor is a copy point --- this causes the diagram to
break into two.  When the input state is a unit for the \COPY-tensor,
the tensor structure is converted to a smooth wire (this is the
maximally entangled case).
\begin{center}
 \includegraphics[width=8\xxxscale]{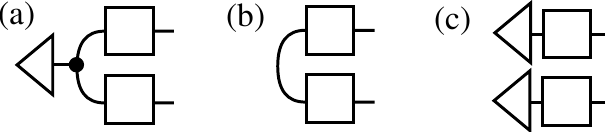}
\end{center}
(a) is the general form of a state $\psi = \sum \alpha^i \ket{\varphi_i}\ket{\phi_i}$ with $\chi>1$ and at least two singular values taking different values, that is  $\exists \lambda_i\neq \lambda_j$ for some $i\neq j$.  The state takes the form in (b) iff the singular values in the triangle from (a) all take the same values, so $\lambda_i=\lambda_j$ $\forall i,j$ with clearly $\chi=d$ .  In this case, the state is LU equivalent to a generalized Bell state, in dimension $d$.  The state takes the form in (c) iff the first singular value equals one, which necessarily implies that the remaining singular values are zero.  In such a case, the state is separable and $\chi=1$.   
\end{example}

\begin{corollary}[Diagrammatic state purification]
The diagrammatic SVD from Theorem \ref{theorem:diagrammatic-SVD} gives rise to a diagrammatic representation of state 
purification.  To begin with consider
\begin{center}
 \includegraphics[width=6\xxxscale]{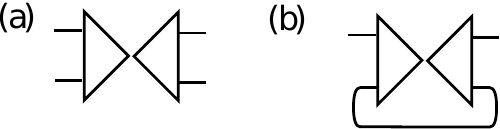}
\end{center}
In (a) we have a two-party pure density state $\ket{\psi}\bra{\psi}$ and in (b) we trace 
out one subsystem $\rho'\bydef\text{Tr}_2(\ket{\psi}\bra{\psi})$.  Now consider 
\begin{center}
 \includegraphics[width=14\xxxscale]{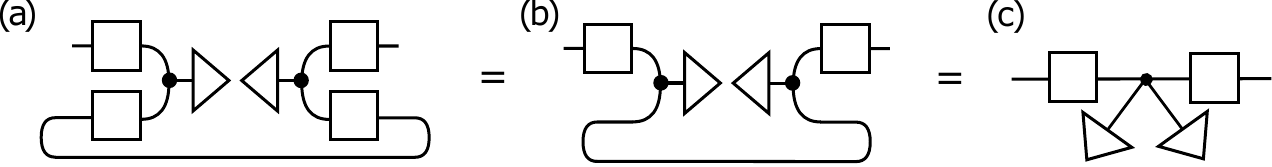}
\end{center}
(a) is found from applying the diagrammatic SVD to the reduced state $\rho'$.  
(b) follows from applications for simple diagrammatic rewrite rules, allowing the bottom unitaries to cancel. 
This follows from pulling both boxes around the bends, which takes the transpose of each map.  We arrive at $U^{\dagger
  \top}U^\top=(UU^{\dagger})^\top=\I$.
In (c) the \COPY-tensors merge, resulting in multiplication of the singular values stored in the valence-one triangular tensors.  

Alternatively, $\ket{\psi}$ can be seen as a purification of $\rho_1$
(the square roots of the singular values multiply (c) resulting again in
$\rho_1$.) These diagrams translate between a purification of a density
operator and the density operator itself.
\end{corollary}

\section{Matrix Product Factorization of States}\label{sec:mps} 

As a key application of the diagrammatic SVD, we will consider Matrix Product States (MPS), an iterative method to factor quantum states into a linear chain of tensors (see \cite{MPSreview08,TNSreview09}).  We will express this factorization in terms of the diagrammatic SVD \ref{theorem:diagrammatic-SVD}, and our focus will be on exposing the degrees of freedom which are invariant under local groups acting on open tensor legs.  A basis to expand any local unitary invariant of an MPS will be explored in \S~\ref{sec:monomial-mps}.

\begin{remark}
The early work in \cite{2007PhRvL..98v0503G, GESP07} cast matrix product states and measurement based quantum computation into the language of tensor networks. 
\end{remark}

\begin{definition}
 Given an $n$-party quantum state $\ket{\psi}$, fully describing this state generally requires an
amount of information (or computer memory) that grows exponentially with~$n$.
If $\ket{\psi}$ represents the state of $n$~qubits,
\be 
\ket{\psi}= \sum_{i j \cdots k}\psi_{i j \cdots k}\ket{i j \cdots k},
\ee 
the number of independent coefficients $\psi_{i j \cdots k}$ in the basis expansion in general would be $2^n$
which quickly grows into a computationally unmanageable number as~$n$ increases.
The goal is to find an alternative representation of~$\ket{\psi}$ which is less data-intensive.
We wish to write $\ket{\psi}$ as 
\begin{equation}\label{eqn:matrix}
\ket{\psi} = \sum_{i j \cdots k} \trace(A_{i}^{[1]} A_{j}^{[2]} \cdots A_{k}^{[n]}) \ket{i j \cdots k},
\end{equation}
where $A_{i}^{[1]}, A_{j}^{[2]}, \ldots, A_{k}^{[n]}$ are indexed sets of matrices and  trace (tr)  closes the boundaries and could be omitted (e.g.~.$A_{i}^{[1]}$ and $A_{k}^{[n]}$ are row and column vectors respectively).
Calculating the components of $\ket{\psi}$ then becomes a matter of calculating the products of matrices,
hence the name \emph{matrix product state}.

If the matrices are bounded in size, the representation becomes efficient
in the sense that the amount of information required to describe them is only linear in~$n$.
The point of the method is to choose these matrices such that they provide a good (and compact) approximation to~$\ket{\psi}$.
For instance, if the matrices are at most $\chi$ by $\chi$, the size of the representation scales as~$n d \chi^2$,
where $d$~is the dimension of each subsystem.
\end{definition}

Without loss of generality, we will apply the MPS method to a four-party state, and explain the procedure in terms of three distinct steps.  Consider a quantum state, expressed as a triangle in the Penrose graphical notation with a label $\textbf{1}$ inside and open legs labeled $i,j,k,m$.  
\begin{center}
 \includegraphics[width=2\xxxscale]{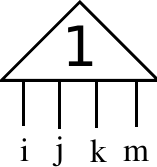}
\end{center}

(Step I). We will now create a partition of the legs of this state, into a first collection containing only leg $i$ and a second collection containing legs $j,k,m$.  We will then apply the diagrammatic SVD across this partition.
The partition is illustrated with the dashed \textit{cut} below in (a).  Figure (b) results from applying the diagrammatic SVD across this partition, factoring the original state with label $\textbf{1}$ in (a) into a valence-two unitary box with label $\textbf{2}$, a valence-one triangle containing the singular values with label $\textbf{3}$, and a 
valence-four triangle with label $\textbf{4}$, all contracted with a \COPY-tensor, as illustrated.  A new internal label (d) for the wire connecting the \COPY-tensor to the valence-four triangle ($\textbf{4}$) was introduced for clarity. (see also Figure \ref{fig:MPS-summary} (a) and (b)). 
\begin{center}
 \includegraphics[width=6\xxxscale]{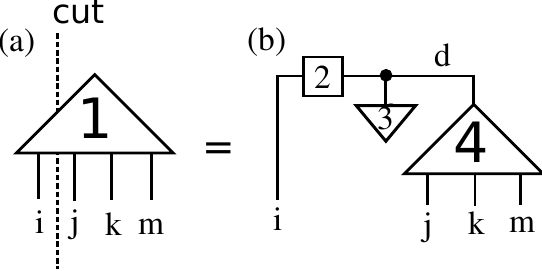}
\end{center}

\begin{remark}[Isometric internal tensors]
 The valence-four triangle tensor in (b) above is actually a unitary map.  The only input leg shown is labeled d.  The other legs are contracted with a fixed basis state $\ket{0}$, from the SVD (a).  We then depict this as the triangle $\textbf{4}$ as in (b).  From the unitarity property, the isometry property follows, as illustrated graphically in (c).  
 \begin{center}
  \includegraphics[width=7\xxxscale]{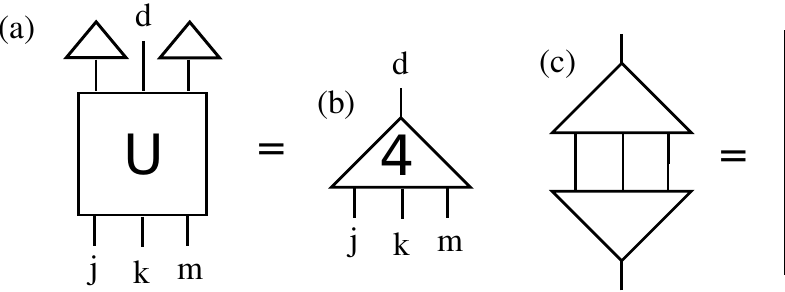}
 \end{center}
\end{remark}

\begin{remark}[Contraction of Unitaries]
In tensor network diagrams, two unitaries compose to form unitaries.  In (a) below, we factor a tensor with three legs into an order-two unitary (white box), a black order three \COPY-tensor contracted with an order-one triangle of singular values and an order-four unitary with label $U$.  In (b) we remove $U$ and act on it with an arbitrary order-two tensor.  These compose to form $\widetilde{U}$ which is still unitary.  In a general MPS, unitaries acting on the open legs do not alter the singular values found in a factorization.
\begin{center}
 \includegraphics[width=12\xxxscale]{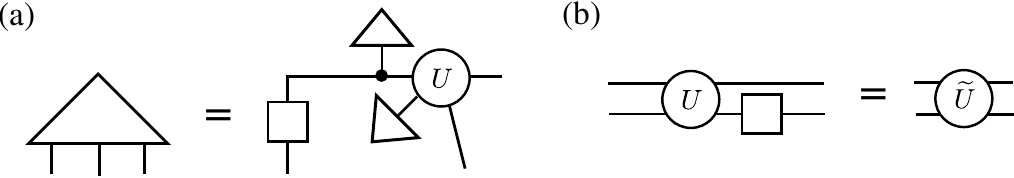}
\end{center}
\end{remark}

(Step II). To illustrate the next step in the factorization, we will remove the tensor labeled $\textbf{4}$ by breaking the wire connecting it to the \COPY-tensor (a).  We will then partition this separate tensor into two halves, one containing wires $d, j$ the other half wires $k,m$.  This partition is illustrated by placing a dashed line (labeled cut) in (a).   We arrive at the the structure in (b), which we have explained in the first step.  (see also Figure \ref{fig:MPS-summary} (b) and (c)).  
\begin{center}
 \includegraphics[width=6\xxxscale]{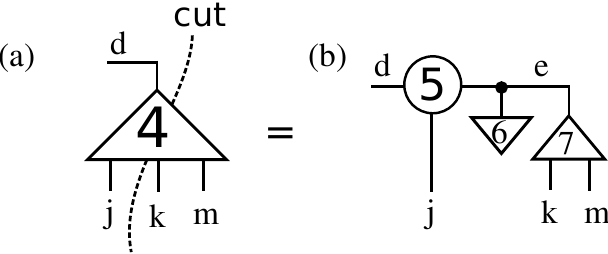}
\end{center}

\begin{remark}[An elementary property of tensor network manipulation]\label{remark:elementary-prop-1}
 It is a fundamental property of tensor network theory that one can remove a portion of a network, alter this removed portion of the network without changing its function, and replace it back into the original network, leaving the function of the original network intact.   
\end{remark}

(Step III).  In the third and final step of the MPS factorization applied to this four-party example, following remark \ref{remark:elementary-prop-1} we first place the tensor we have factored in the second step, back into the original network from the first step, see (a) below.  We then repeat the second step, applied to the triangular isometry tensor, labeled internally with a $\textbf{7}$.  This results in the factorization appearing in (b).  (see also Figure \ref{fig:MPS-summary} (c) and (d)).  
\begin{center}
 \includegraphics[width=10\xxxscale]{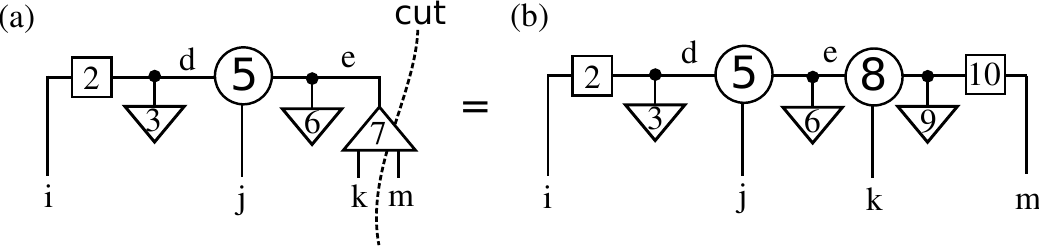}
\end{center}

\begin{remark}[Step $n$]
 The iterative method continues in the same fashion as the first three steps, resulting in a factorization of an $n$-party state.  A summary of the MPS factorisation applied to a four-party state is shown in Figure \ref{fig:MPS-summary}. 
\end{remark}

\begin{figure}[t]
\includegraphics[width=16\xxxscale]{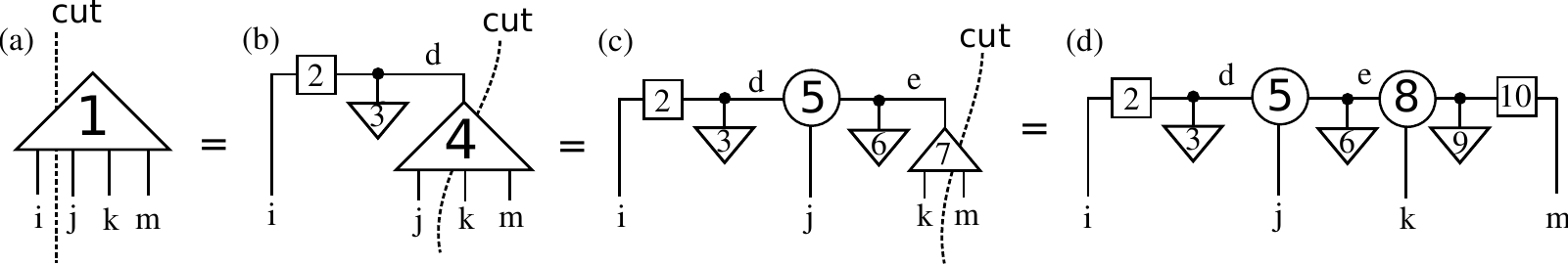}
\caption{(Diagrammatic summary of steps I, II and III). The quantum state (a) is iteratively factored into the 1D Matrix Product State (d). This procedure readily extends to $n$-body states.}
\label{fig:MPS-summary}
\end{figure}

(Summary). We will now consider Figure \ref{fig:MPS-summary}, which summarizes the factorization scheme. In the steps we have outline, we have factored the Figure \ref{fig:MPS-summary} original state (a) into the MPS in Figure \ref{fig:MPS-summary} (d), in terms of the components listed below.  

\begin{itemize}
 \item[(i)] States (labeled $\textbf{3}$, $\textbf{6}$ and $\textbf{9}$; denoted $\phi_3$, $\phi_6$ and $\phi_9$, respectively): $\phi_3=(\lambda_0,\lambda_1)^\top$, $\phi_6=(\lambda_2,\lambda_3,\lambda_4,\lambda_5)^\top$ and $\phi_9=(\lambda_6,\lambda_7)^\top$.  The $\lambda_i$'s are the singular values across each partition.  The number of non-zero singular values ($\chi$) is given by the minimum dimension of the two parties in the cut.  For the case of qubits, the first outside partition has at most two non-zero entries, and the next inside partition has at most $\14$.   One might also consider the singular values as the eigenvalues of either member of the pair of reduced density matrices found from tracing out either half of a partition.  
 \item[(ii)] Unitary gates (labeled $\textbf{2}$ and $\textbf{10}$; denoted $U_2$ and $U_{10}$, respectively).  
 \item[(iii)] Isometries (labeled $\textbf{5}$ and $\textbf{8}$; denoted $I_5$ and $I_8$ respectively).   The isometry condition describes the tensor relation $I_{jq}^d ~\overline{I}^{jq}_r = \delta ^d_{~r}$.  It is a consequence of the fact that tensors $\15$ and $\18$ arise from unitary gates, as explained in Step II.  The isometry condition plays a more relevant role in structures other than 1D tensor chains.   
\end{itemize}

We note that by appropriately combining neighboring tensors as in Figure \ref{fig:MPS-equation} (a), one recovers the familiar matrix product representation of quantum states \ref{fig:MPS-equation} (b).  Matrix product states are written in equational form as  
\begin{equation}\label{eqn:MPS}
 \psi = \sum_{i,j,k,m} A^{[1]}_i A^{[2]}_j A^{[3]}_k A^{[4]}_m\ket{ijkm}.
\end{equation}
Here $A^{[1]}$ becomes a new tensor formed from the contraction of tensors labeled $\12$, $\13$, and $A^{[2]}$ is a contraction of tensors labeled $\15$ and $\16$, etc. The exact grouping has some ambiguity.  

A utility of our approach Figure \ref{fig:MPS-equation} (a) is that the \COPY-tensor is well defined in terms of purely graphical rewrite identities (as seen in Definition \ref{def:copy-prop}).  These graphical relations allow one to gain insights (into e.g.\ polynomial invariants as will be seen), and to contract portions of tensor networks by hand.
The factorization we present however, allows one to perform many diagrammatic manipulations with ease, and exposes more structure inherent in a MPS.  

\begin{remark}[Data compression]
 The compact representation of a MPS is recovered by picking a cutoff value for the singular values across each partition, or a minimum number of allowed singular values.  This allows one to compress data by truncating the Hilbert space and is at the heart of MPS computer algorithms in current use.  
\end{remark}

\begin{figure}[h]
\includegraphics[width=15.5\xxxscale]{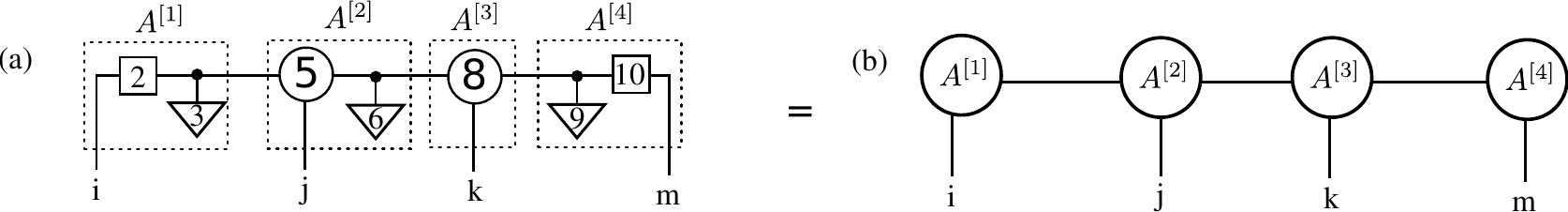}
\caption{Conversion from our notation (a), to conventional MPS notation (b). The factorization methods we have reviewed here allow one to ``zoom in'' and expose internal degree of freedom (a) or ``zoom out'' and expose high-level structure (b). The equational representation of the MPS in (b) is given in \eqref{eqn:MPS}.}
\label{fig:MPS-equation}
\end{figure}

The singular values found from the MPS factorization can be used to form a complete basis to express any quantity related to an MPS that is invariant under local unitary operations.  This includes providing a complete basis to express any entanglement monotone.  

\begin{example}[MPS for the GHZ state]\label{ex:mps-ghz}
The standard MPS representation of the Greenberger-Horne-Zeilinger (GHZ) state is given as \marginnote{Daniel {\bf G}reenberger, Michael {\bf H}orne and Anton {\bf Z}eilinger first studied what is now named the {\sf GHZ}-state in 1989 \cite{GHZ}.}
\be\label{eqn:ghzMPS}
\ket{\text{GHZ}} = \frac{1}{\sqrt{2}}\trace\left( \begin{array}{cc}
\ket{0} & 0 \\
0 & \ket{1} \end{array} \right)^{n}=
\frac{1}{\sqrt{2}}(\ket{00\ldots 0}+\ket{11\ldots 1}).
\ee
Alternatively, we may use a quantum circuit made of \CNOT{} gates to construct
the GHZ state,
and then use the rewrite rules employed in Examples 
\ref{ex:circuits-2} and \ref{ex:state-prep}
to recover the familiar MPS comb-like structure consisting of \COPY{} tensors:
\begin{center}
   \includegraphics{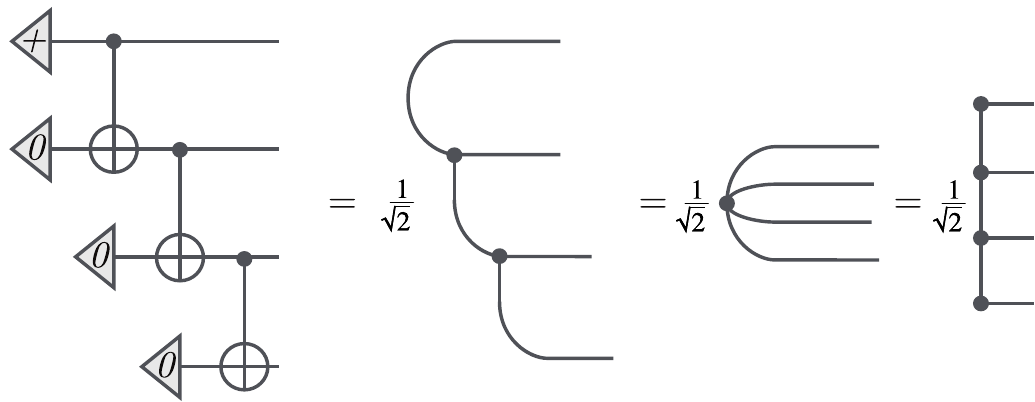}
\end{center}

Diagrammatically, any tensor network formed from connected \COPY{}-tensors reduces to a single dot
with the appropriate number of input and output legs.
Hence one might write the n-party \GHZ-state as 
\be 
\ket{\text{GHZ}} = \frac{1}{\sqrt{2}}\sum_{ijk\ldots l} \COPY^{ijk\ldots l}\ket{ijk\ldots l}.
\ee 
\end{example}

\begin{example}[MPS for the W state]\label{ex:w-state}
Like the GHZ state from Example \ref{ex:mps-ghz}, the $n$-qubit W~state ($n\geq 3$) has the following MPS representation:
\be\label{eqn:wnMPS1}
\begin{aligned}
\ket{W}
=&
\frac{1}{\sqrt{n}}
\begin{pmatrix} \ket{1} & \ket{0} \end{pmatrix}
\left( \begin{array}{cc}
\ket{0} & 0 \\
\ket{1} & \ket{0} \end{array} \right)^{n-2}
\begin{pmatrix} \ket{0}\\ \ket{1} \end{pmatrix},\\
=&
\frac{1}{\sqrt{n}} \left(
\ket{10\ldots 0}+\ket{010\ldots 0}+\ldots +\ket{0\ldots 01}
\right).
\end{aligned}
\ee
\end{example}

\subsection*{Invariant Basis for Matrix Product States}\label{sec:monomial-mps}

We will consider generating a full monomial basis in terms of the singular values found in the factorization of Matrix Product States.  The monomial basis is generated as a matter of convenience and can be used to define a basis for entanglement monotones.   In fact, we will see that the diagrammatic factorization can be used to prove that certain tensor contractions give rise to certain invariants, which allow one to calculate quantities of interest, such as the concurrence or R\'{e}nyi entropy.  These tensor contractions are not dependent on the factorization method used, and are general. The diagrammatic SVD will be used to prove that the contraction of certain tensors, results in an expression that is in terms of the singular values.  

Our objective will be to develop tensor contractions that evaluate to specific quantities of interest.  These quantities of interest will be invariants of the local unitary group.  
Such invariants have expansions in terms of the singular values of reduced density operators.  We will use tensor network contractions to evaluate a full basis that can be used to expand any 
function of the singular values.  This includes quantities of interest such as concurrence and R\'{e}nyi Entropy.  

For the general case, one can, for instance, form a polynomial basis using the elementary symmetric polynomials by combining the $\lambda$ as   
\begin{equation}
 S_1 = \sum_i \lambda_i,
\end{equation}

\begin{equation}
 S_2 = \sum_{i\neq j} \lambda_i\lambda_j,
\end{equation}

\begin{equation}
 S_3 = \sum_{i\neq j\neq k} \lambda_i\lambda_j\lambda_k, 
\end{equation}
and so on.  Such polynomials are used to calculate the d-concurrence, see Definition \ref{def:d-concurrence}.  
Any polynomial in the $S$'s is necessarily a local unitary invariant.  

Another basis of interest, is the basis formed by summed powers of the singular values
\be 
B_n = \sum_i \lambda_i^n,
\ee 
which is of great interest to evaluate R\'{e}nyi's entropy, see Definition \ref{def:R\'{e}nyi}.  
The polynomial $B_2$ is related to the concurrence measure of entanglement in Definition \ref{def:concurrence}.  
For the case of qubits, $B_2=\lambda_0^2 + \lambda_1^2$ which is greater than zero iff the state is entangled.
Increase $n$, pattern continues up to $B_d$.  $B_d$ being greater than
zero implies that $B_n\geq 0$ for all $1\leq n \leq d$.  In addition,
$B_n=0$ implies that $B_k=0$ for all $n\leq k \leq d$.

We will first recall the definition of the concurrence (see \cite{concurrence}) and then the definition of R\'{e}nyi's Entropy (see for instance \cite{Renyi}). 
These quantities are of physical interest.  As will be shown, they can be calculated by contracting specific tensor networks.

\begin{definition}[The concurrence]\label{def:concurrence}
The concurrence of a pure bipartite normalized state $\ket{\psi}$ is defined as 
\be 
C(\ket{\psi}) := \sqrt{\frac{d}{d-1}(1-\text{Tr}\rho^2)},
\ee 
where $\rho$ is obtained by tracing over one subsystem.  The factor $\sqrt{d/(d-1)}$ ensures that $0\leq C(\ket{\psi}) \leq 1$.   
\end{definition}

\begin{remark}[Tensor contractions for the concurrence]
 We will contract tensor networks that evaluate to $B_2=\text{Tr}(\rho^2)$ and these hence can be used to evaluate the concurrence.  
\end{remark}

\begin{definition}[The $d$-concurrence]\label{def:d-concurrence}
 Consider a d x d-dimensional bipartite pure state $\ket{\psi}$ with Schmidt numbers 
 $\lambda:=(\lambda_0, \lambda_1, ..., \lambda_{d-1})$ the d concurrence monotones, $C_k(\ket{\psi})$, $k=1,2,...,d$, of the state $\ket{\psi}$ are defined as follows
 \be 
 C_k(\ket{\psi}) := \left(\frac{S_k(\lambda_0, \lambda_1, ..., \lambda_{d-1})}{S_k(1/d,1/d,...1/d)}\right)^{1/k},
 \ee  
 where $S_k$ is the kth elementary symmetric polynomial.  
\end{definition}

\begin{definition}[The R\'{e}nyi Entropy]\label{def:R\'{e}nyi}
 The R\'{e}nyi entropy of order $\alpha$ is defined to be 
 \be 
S_\alpha := \frac{1}{1-\alpha} \ln\sum \lambda_i^\alpha,  
 \ee  
and in the limit $\alpha \rightarrow 1$ 
\be 
E=\lim_{\alpha \rightarrow 1}S_\alpha = -\sum_i \lambda_i \ln \lambda_i.
\ee 
\end{definition}

\begin{remark}[Lower bounds on Entropy]
Consider a tensor network representing a quantum spin state.  
Partition this state into a block of $k$-spins and a block of $m$-spins. Let the number of wires connecting the two blocks be given as $l$.  Then 
\be 
E\leq \min\{k,m,l\}.
\ee 
\end{remark}

In the remaining sections, we will consider tensor networks that enable the evaluation of the quantities of interest we have mentioned here. 
At the heart of the MPS method, is the factorization of a state into two halves.  Our diagrammatic SVD allows one to prove that certain networks contract to 
quantities that can readily be related to the concurrence or R\'{e}nyi's entropy.  In what follows, we build examples for biaprtite states.  These generally apply to Matrix Product States 
by considering partitions.  In fact, in the diagrammatic language, it is often useful to group $k$-wires into one wire, for the purpose of manipulation \cite{BB11}.

\section{Numerical Tensor Network Algorithms and Packages} 

This book is focused on graphical reasoning and the applications of tensor networks to quantum information science.  A primary driving force behind tensor network applications is numerical algorithms for condensed matter physics application.  This is not our main focus, however. 

However, no book would be complete without touching on these important numerical applications.  Indeed, we have described several tensor contractions which can be utilized to solve counting problems (see \S~\ref{part:count}).  And to simulate quantum systems using matrix product states (Part \ref{sec:mps}). 

Those reading this text that want to explore the numerical application of these ideas should be aware of both open source software packages and papers that describe in detail the most effective tensor contraction algorithms.  That is precisely the objective of this appendix.  

\subsection{Online Resources Describing Tensor Network Software Implementations}
Several webpages are devoted to listing papers and tensor contraction algorithms.  Here we provide a short listing of those papers which go into more detail and are of a primary software centric focus.  We have done our best to be as inclusive as possible, however this list is merely an editors pick and is not designed to be comprehensive. 

\begin{itemize}
  \item[1.] {\sf \bf Resource}. \hfill {\it TensorNetwork.ORG}  \\ 
    {\sf \bf Brief description}. An open-source `living' article containing many tensor network resources, applications, and software.\\ 
{\sf \bf Link}. \url{https://tensornetwork.org} \\
{\sf \bf List of Software Packages}. \url{http://tensornetwork.org/software}

\item[2.] {\sf \bf Resource}. \hfill {\it Tensors.NET} \\ 
{\sf \bf Brief description}. A collection of resources, including links to software, tutorials and code in several languages. \\
{\sf \bf Link}. \url{https://www.tensors.net}
\end{itemize}

\subsection{Open Source Tensor Network Software Packages} 
Software packages and programs to do various tensor contractions and related tasks is increasingly available in a variety of languages.  This includes TEBD programs as well as others.  Here we have done our best to include an active and up to date listing of the main packages.  It was last updated \today.  As software packages can appear and also become inactive at any time, a more updated listing of available packages can be currently be found at \url{http://tensornetwork.org/software}. 

We have attempted to list some of the most common packages below. See also {\bf QUIMB}, which is fully featured and aimed at applications in physics; {\bf TenPy} is known to be a small library but with very good MPS codes; {\bf Cyclops}, {\bf TT-Toolbox}, {\bf Tensor Toolbox} are more focused on the {\it tensor train} formulation of MPS~\cite{Oseledets2011} and on decompositions such as CP ({\it tensor rank decomposition} or {\it canonical polyadic decomposition}). 

\begin{itemize}
\item[1.] {\sf \bf Title of package}.\hfill {\it ITensor} \\
{\sf \bf Language}. C++\\ 
{\sf \bf Brief description}. ITensor or Intelligent Tensor is well featured tensor network library implementing a bit of everything out of the box.\\
{\sf \bf Status}. Active (version 3.1.1) \\
{\sf \bf Link}. \url{http://itensor.org} 

    \item[2.] {\sf \bf Title of package}.\hfill {\it TNT} \\
{\sf \bf Language}. Various\\ 
{\sf \bf Brief description}. The TNT library contains highly optimised routines for manipulating tensors and routines that can be used to build the most common tensor network algorithms \cite{2017JSMTE..09.3102A}.\\
{\sf \bf Status}. Active \\
{\sf \bf Link}. \url{http://www.tensornetworktheory.org} 

\item[3.] {\sf \bf Title of package}.\hfill {\it TensorNetwork} \\
{\sf \bf Language}. Python\\ 
{\sf \bf Brief description}. A Google/X package for tensor networks described in \cite{2019arXiv190501330R} and related to Google's popular tensor flow package (\url{https://www.tensorflow.org}). \\ 
{\sf \bf Status}. Active \\
{\sf \bf Github}. \url{https://github.com/google/tensornetwork} 

\item[4.] {\sf \bf Title of package}.\hfill {\it Quantomatic} \\
{\sf \bf Language}. Python\\ 
{\sf \bf Brief description}. A diagrammatic proof assistant supporting reasoning with diagrammatic languages with applications to tensor networks \cite{Dixon2010, Kissinger2014}. \\ 
{\sf \bf Status}. Active \\
{\sf \bf Github}. \url{https://quantomatic.github.io}

\end{itemize}

\section{Problems} 

\begin{myexercise}[Tensor products, Density operators, Singular values and purification]\label{exercise:svd1}
Let $\psi = \ket{001}+\ket{010}+\ket{100}$ be a state in $\hilb H_1\otimes\hilb H_2$ with $\dim(\hilb H_1)=2$ and $\dim(\hilb H_2)=4$, and let  
\begin{equation}
 V^\dagger = \left(
\begin{array}{cccc}
 0 & 1 & 0 & 0 \\
 \frac{1}{\sqrt{2}} & 0 & 0 & -\frac{1}{\sqrt{2}} \\
 \frac{1}{\sqrt{2}} & 0 & 0 & \frac{1}{\sqrt{2}} \\
 0 & 0 & 1 & 0
\end{array}
\right).
\end{equation}
\begin{itemize}
 \item[(i)] By writing $\phi = \sum_{ij}C_{ij}\ket{i}\ket{j}$, where $\ket{i}\in\hilb H_1$, $\ket{j}\in\hilb H_2$ are both orthonormal basis, state the values of $i\in\{0,1\},j\in\{0,1,2,3\}$ that correspond to non-zero coefficients of $C_{ij}$ and hence express $\psi$ in the basis $\ket{i},\ket{j}$.    
  \item[(ii)]  Write the $2\times 4$ matrix ${\bf C}=(C_{ij})_{ij}$ and show that ${\bf C}{\bf C}^\dagger$ and ${\bf C}^\dagger {\bf C}$ are (non-normalized) density operators, equivalent to $\rho_1$, $\rho_2$ found by tracing over the systems $\hilb H_1$ and $\hilb H_2$ respectively.  (Here the adjoint $\dagger$ means matrix conjugate transpose.)
 \item [(iii)] From the singular value decomposition, one can write ${\bf C}=U\Sigma V^\dagger$.  For $V$ given above, and $U$ the $2\times 2$ identity matrix, find the $2\times 4$ matrix of singular values $\Sigma$.  
 \item[(iv)] Find purifications for $\rho_1$ and $\rho_2$ (other than $\psi$).  
\end{itemize}
\end{myexercise} 

\begin{myexercise}(Lie product formula---a.k.a.~Trotter-Suzuki decomposition). Show that 
\begin{equation}
e^{A + B}=\lim_{k\rightarrow \infty} \left(e^{A/k}e^{B/k} \right)^k
\end{equation}
for Hermitian matrices $A$ and $B$. 
\end{myexercise}

\begin{myexercise}
The Schmidt number of $\ket{\psi}^{AB}$---denoted Sch($\psi$)---is the rank of the reduced density matrix $\rho_A = \Tr_B(\ketbra{\psi}{\psi})$ where rank of a Hermitian operator is defined as the dimension of its support. 
\end{myexercise}

\begin{myexercise}
Let $\ket{\psi} = \alpha \ket{\phi} + \beta \ket{\gamma}$ and prove that 
\begin{equation}
\text{Sch}(\psi)\geq \|\text{Sch}(\phi) - \text{Sch}(\gamma)\|_1
\end{equation} 
\end{myexercise}

\begin{definition}
A {\it multiset} is a modification of a set allowing multiple instances for each element.
\end{definition}

\begin{definition}
 Denote by $\text{Spec}\{A\}$ the multisetset of eigenvalues of square matrix $A$. 
\end{definition}

\begin{myexercise}[Jacobson's Lemma]\label{prob:spec}
Let non-negative $A,B \in \mathcal{L}(\mathbb{C}_n)$.  Show equality of the non-zero elements in the multisets as $\text{Spec}\{AB\}=\text{Spec}\{BA\}$. 
\end{myexercise}

\begin{myexercise}
A matrix $A\in \mathcal{L}(\mathbb{C}_n)$ is non-negative ($A\geq 0$) if
\begin{equation}
    \bra{\psi}A\ket{\psi}\geq 0
\end{equation}
$\forall \psi \in \mathbb{C}_n$. Show that $A\geq 0$ implies the existence of a unique $B$ such that $B^2=A$.
\end{myexercise}

\begin{myexercise}
Let Hermitian $A,B \in \mathcal{L}(\mathbb{C}_n)$. Now suppose that $A^2 = A$ and $B^2=B$ and further that $\{A, B\}=0$.  Let $\|\psi\|_2=1$ and show that 
\begin{equation}
    \bra{\psi}A\ket{\psi}^2 + \bra{\psi}B\ket{\psi}^2\leq 1
\end{equation}
\end{myexercise}

\begin{myexercise}
Consider 
\begin{equation}
\rho = \frac{1}{4}(\I + X\otimes X -Y\otimes Y)
\end{equation} 
and use Lagrange multipliers to find $\|\rho \|_\infty$. 
\end{myexercise}

\begin{myexercise}
Let density operator $\rho^{AB}$ represent the state of a qubit pair (A, B) and define the {\it spin-flipped} density matrix as 
\begin{equation}
    \tilde{\rho}^{AB} = (Y\otimes Y) \bar{\rho}^{AB}(Y\otimes Y)
\end{equation}
where $\bar{\rho}$ is complex conjugate of $\rho$ in the standard basis.  As $\rho^{AB}$, $\tilde{\rho}^{AB}\geq 0$, so is their non-Hermitian product (Problem \eqref{prob:spec}). Consider $\psi$ as the two-qubit Bell state and $\rho = \ketbra{\psi}{\psi}$.  Find the eigenvalues of $\rho \tilde{\rho}$. 
\end{myexercise}

 \part{Boolean Tensor Networks}\label{sec:btn}
\section{Introduction}

Now we will explain part of a tool set and framework largely following \cite{CTNS, BB11}.  \marginnote{{\it ``It is not of the essence of mathematics to be conversant with the ideas of number and quantity.''} --- George Boole} Hence, we will approach tensor networks by focusing on familiar components, namely Boolean logic gates (and
multi-valued logic gates in the case of qudits), applied to the tensor network context. See Appendix \ref{sec:Boolean} for background on Boolean algebra.  The concept of Boolean linearity and non-linearity was used to form a dichotomy between fundamental Boolean tensor building blocks in \cite{CTNS}. 

This subject has an increasingly long history.  The first categorical model of Boolean circuits (as well as some progress on the quantum case) can be found in the seminal work by Lafont~\cite{boolean03}.  In the condensed matter community, Boolean tensors, such as \COPY{}-tensors, are very commonly used.  Closely related to Boolean tensors are stabilizer tensors (\S~\ref{sec:stabtt}), sometimes called Clifford tensor networks.  The so called ZX-calculus \cite{CD, Coecke2017} is also built using many Boolean tensors.

\section{Overview of the Chapter}\label{sec:overview}

\subsection*{Tensor network representations of quantum states}\label{sec:newrep}

A qudit is a $d$-level generalization of a qubit.\mn{A q-dit generalization of the quantum circuits and the ZX-calculus appears in \cite{BB11}.  The ZX-calculus is equivalent to Clifford circuits plus cups and caps to bend wires and compose states. It was first proposed in terms of interacting quantum observables in \cite{redgreen}.}  As has been seen in the last section, a quantum state of $n$-qudits has an exact representation as a order-$n$ tensor with each of the open legs corresponding to a physical degree of freedom, such
as a spin with $(d-1)/2$ energy levels. Such a representation, shown in
Figure~\ref{fig:tensor-networks}(a) is manifestly inefficient since it will have a
number of complex components which grows exponentially with $n$. The purpose of
tensor network states is to decompose this type of structureless order-$n$ tensor into
a network of tensors whose order is bounded. \mn{Lafont appears to be the first to work towards a categorical model of quantum circuits~\cite{boolean03}. Many advancements have subsequently been made \cite{CD, redgreen, CTNS, BB11}.}

There are now a number of ways to
describe strongly-correlated quantum lattice systems as 
tensor-networks. These include
\begin{itemize}
 \item[(i)] Matrix Product States, MPS~\cite{ostromm1995,fannesnachtwern1992,2010NJPh...12b5005C}
 \item[(ii)] Projected Entangled Pair States, PEPS~\cite{verstraete08,TNSreview09}
 \item[(iii)] Multiscale Entanglement Renormalisation Ansatz, MERA~\cite{vidal2007, 2009arXiv0912.1651V}
 \item[(iv)] Tree Tensor Networks, TNN~\cite{PhysRevA.74.022320,2009PhRvB..80w5127T}
 \item[(v)] Boolean Tensor Network States, BTNS \cite{CTNS, BB11}. 
\end{itemize}

 The central problem faced by all types of tensor networks is that the
resulting tensor network for the quantity $\bra{\psi}(\2 O\ket{\psi}$, where $\2O$
is some product operator, needs to be efficiently contractible (efficient is taken to mean polynomial in the problem size) if any physically meaningful calculations, e.g., expectation values, correlations or probabilities, are to be
computed. For MPS and TTN efficient contractibility follows from the 1D chain or tree-like geometry, while for MERA it follows from its interesting causal cone structure~\cite{2009arXiv0912.1651V}. For PEPS and BTNS, however, exact contraction is not proven to be efficient in general, but can often be rendered efficient if approximations are made~\cite{verstraete08,TNSreview09}.

For MPS and PEPS, shown in Figures~\ref{fig:tensor-networks}(b) and (c), the resulting network of tensors follows the geometry of the underlying physical system, e.g., a 1D chain and 2D grid, respectively. 

Alternatively a Tensor Tree Network (TTN) can be
employed which has a hierarchical structure where only the bottom layer has open physical legs, as shown in Figure~\ref{fig:tensor-networks}(d) 
for a 1D system and Figure~\ref{fig:tensor-networks}(e) for a 2D one.

For MERA the network is similar
to a TTN, as seen in Figure~\ref{fig:tensor-networks}(f) for 1D, but is instead comprised of alternating layers of order-four unitary and order-three isometric
tensors. 

A Boolean tensor network state (BTNS) contains some algebraically constrained 
tensors obeying some clearly defined diagrammatic laws, along with possible generic
tensors. Indeed, when recast, certain widely used classes of tensor network states can be readily exposed
as examples of BTNS \cite{CTNS}.
Specifically, variants of PEPS have been proposed called
string-bond states~\cite{PhysRevLett.100.040501}.  Although these string-bond states,
like PEPS in general, are not
efficiently contractible, they are efficient to sample.

\begin{myfullpage}
\begin{illexample}
\label{fig:tensor-networks}

\begin{center}
     \includegraphics[width=.95\textwidth]{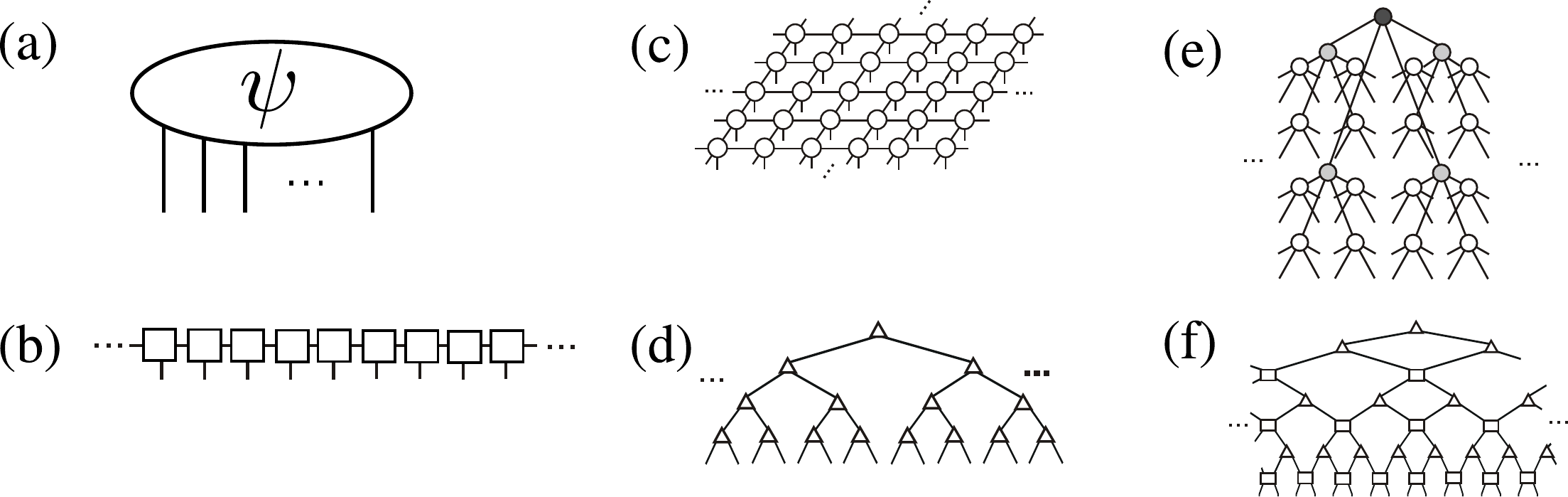}
\end{center}
(a) A generic quantum state $\ket{\psi}$ for $n$ degrees of freedom represented as a tensor with $n$ open legs. (b) A comb-like MPS tensor network for a 1D chain system~\cite{ostromm1995,fannesnachtwern1992}. (c) A grid-like PEPS tensor network for a 2D lattice system~\cite{verstraete08,TNSreview09}. (d) A TTN for a 1D
chain system where only the bottom layer of tensors possess open physical legs~\cite{PhysRevA.74.022320,2009PhRvB..80w5127T}. (e) A TTN for a 2D lattice system. (f) A hierarchically structured MERA network for a 1D chain system possessing unitaries (order-4 tensors) and isometries (order-3 tensors)~\cite{vidal2007, 2009arXiv0912.1651V}. This tensor network can also be generalized to a 2D lattice (not shown).
\end{illexample} 
\end{myfullpage}

\begin{remark}[Generalisation of String Bond States]
 When recast, certain widely used classes of tensor network states can be readily exposed
as examples of BTNS. Specifically, variants of PEPS have been proposed called
string-bond states~\cite{PhysRevLett.100.040501}.  Although these string-bond states,
like PEPS in general, are not efficiently contractible, they are efficient to sample. By this we mean that for
these special cases of PEPS, any given amplitude of the resulting state (for a fixed computational basis state) can be extracted exactly and efficiently,
in contrast to generic PEPS. This permits variational quantum Monte-Carlo calculations to be performed on string-bond states 
where the energy of the state is stochastically
minimized~\cite{PhysRevLett.100.040501}. This remarkable property
follows directly from the use of a tensor, 
called the {\sf COPY}-tensor, which forms one of several tensors in the fixed
toolbox considered in great detail later in this lecture. 

As its name suggests, the {\sf COPY}-tensor
duplicates inputs states in the computational basis, and thus with these inputs
breaks up into disconnected components, as depicted in
Figure~\ref{fig:stringbonds}(a). By using the
{\sf COPY}-tensor as the ``glue" for connecting up a TNS, the ability to sample the
state efficiently is guaranteed so long as the individual parts connected are themselves
contractible. The generality and applicability of this trick can be seen by examining
the structure of string-bond states, as well as other types of 
similar states like entangled-plaquette-states~\cite{2009NJPh...11h3026M} and
correlator-product states~\cite{2009PhRvB..80x5116C}, shown in Figure~\ref{fig:stringbonds}(c)-(e).  
A long-term aim of this work is that by presenting our toolbox of
tensors, entirely new classes of BTNS with similarly desirable contractibility
properties can be devised. 
\end{remark}
\begin{myfullpage}
\begin{illexample}
\begin{center}
\includegraphics[width=17\xxxscale]{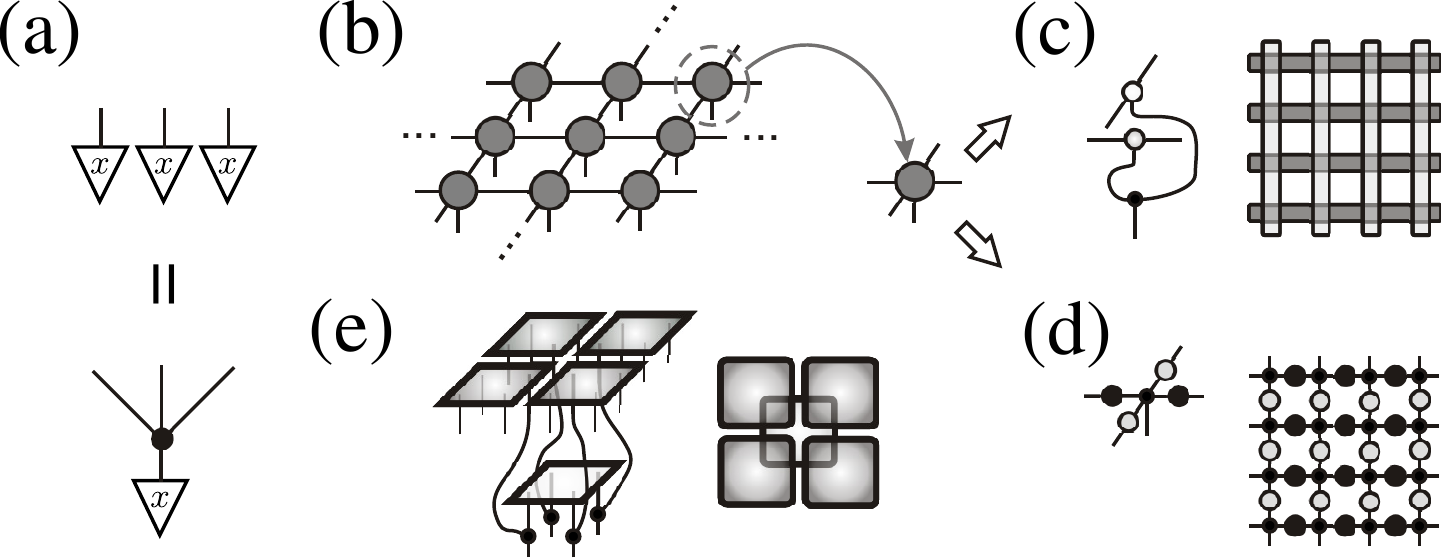}
\end{center}
(a) One of the simplest tensors, the {\sf COPY}-gate or the {\sf COPY}-dot in classical boolean circuits, copies computational basis
states
$\ket{x}$ where $x=0,1$ for qubits and $x=0,1,...,d-1$
for qudits.  The tensor subsequently breaks up into disconnected states.
(b) A generic PEPS in which we expose a single generic order-5 tensor. This
tensor network can neither be contracted nor sampled exactly and efficiently.
However,
if the tensor has internal structure exploiting the {\sf COPY}-tensor, then efficient
sampling
becomes possible. (c) The tensor breaks up into a vertical and a horizontal order-3
tensor joined by the {\sf COPY}-tensor. Upon sampling computational basis states the
resulting contraction reduces to many isolated MPS, each of which are exactly
contractible, for
each row and column of the lattice. This type of state is known as a string-bond
state and can be readily generalized~\cite{PhysRevLett.100.040501}. (d) An even
simpler case is to break the tensor
up into four order-2 tensors joined by a {\sf COPY}-tensor forming a co-called
correlator-product state~\cite{2009PhRvB..80x5116C}. (e) Finally, outside the PEPS
class, there are entangled
plaquette states~\cite{2009NJPh...11h3026M} which join up overlapping tensors (in
this case order-4 ones
describing a $2\times 2$ plaquette) for each plaquette. Efficient sampling is again
possible due to the {\sf COPY}-tensor.\label{fig:stringbonds}
\end{illexample}
\end{myfullpage}

\begin{remark}[From qubits, to qtrits, ..., qdits]
 By invoking known theorems asserting the universality of multi-valued logic~\cite{MVL1} (also called $d$-state switching), our methods can be readily applied to tensors of any finite dimension. This was considered explicitly in \cite{BB11}, where a higher dimensional graphical calculus was developed. 
\end{remark}

\subsection*{Tensor network components defined by diagrammatic laws}\label{sec:comdef}

Here we will review the collection of tensors that form a universal
tensor tool box.
In mathematical logic and computer science, formal semantics is an important field of study.  Throughout this lecture series, we largely adopt the semantics 
developed in \cite{CTNS, BB11} which offer a natural extension of the graphical language of quantum circuits in wide spread use in quantum physics.  This was done by merging the overlapping concepts in various fields into a common language that deviates as little as possible from the standard language of quantum circuits.

To get an idea of how the tensor calculus will work, consider Figure~\ref{fig:F2-presentation},
which forms a presentation of the linear fragment of the Boolean
calculus~\cite{boolean03}): that is, the calculus of Boolean algebra we represent on quantum states, restricted
to the building blocks that can be used to generate linear Boolean functions---as described in \cite{CTNS}.  This is the fragment exactly considered in what is called the ZX-calculus \cite{CD, redgreen}. 

\begin{remark}[Quantum linear states are non-trivial]
 In the setting of tensors, the linear fragment of the tensor calculus is already non-trivial.  In fact, 
 these are the building blocks that appear in my important quantum information protocols and are the backbone of the widely studied class of stabilizer states.  For instance, in \cite{DBCJ11} the authors construct exactly contractible 2D networks representing topological quantum states.
\end{remark}

To recover the full Boolean-calculus, we must append a non-linear Boolean gate as done in \cite{CTNS}: we
use the \AND{}-gate.  Figure~\ref{fig:F2-presentation} together with
Figure~\ref{fig:extraF2} form a full presentation of the
calculus~\cite{boolean03}. The origin and consequences of these relations will be
considered in full detail in \S~\ref{sec:components}. The presentations in Figure~\ref{fig:F2-presentation} together with
Figure~\ref{fig:extraF2} represent a
complete set of defining equations, see Lafont~\cite{boolean03}.

\begin{myfullpage}
\begin{illexample}
\label{fig:F2-presentation}
A summary of the linear fragment of the Boolean
calculus on tensors (reproduced from Lafont \cite{boolean03} and written to match the common quantum circuit notation as in \cite{CTNS}). The plus ($\oplus$) tensors are \XOR{} and the black ($\bullet$) tensors represent
{\sf COPY}.  The details of (a)-(g) will be given in
Sections~\ref{sec:components}.  For
instance, (d) represents the bialgebra law and (g) the Hopf-law (in
the case of qubits $x\oplus x =0$, in higher dimensions
the units $\bra{+}$ becomes
$\bra{0}+\bra{1}+\cdots+\bra{d-1}$). (Read top to bottom.)
\begin{center}
    \includegraphics{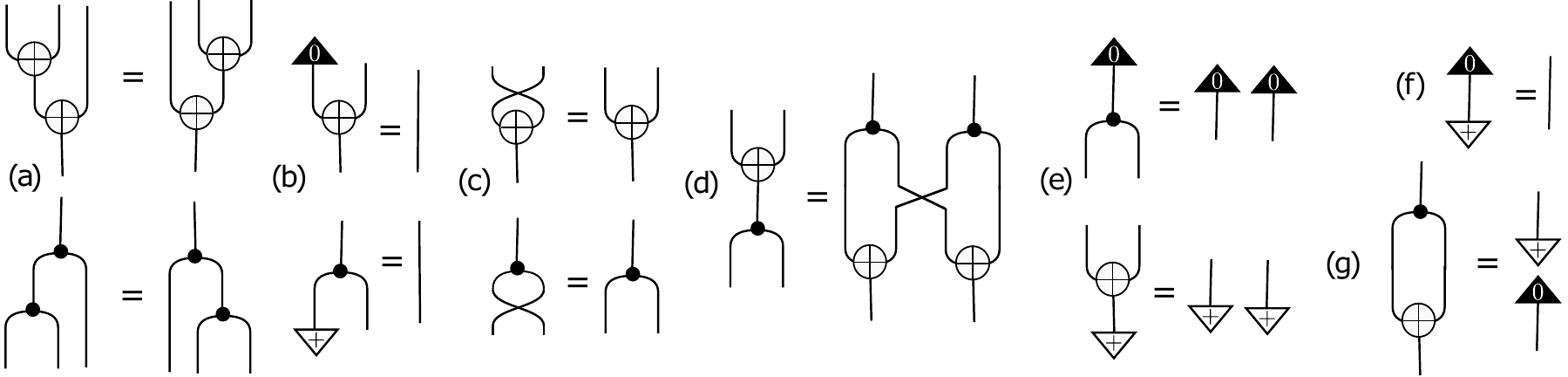}
\end{center}
\end{illexample}
\end{myfullpage}

\begin{remark}[Alternative approaches to the linear fragment]\label{remark:related-work}
The structures in Figure \ref{fig:F2-presentation} which are found by casting classical circuits into tensor networks 
are used to form the building blocks needed to represent \CNOT{}-gates and are related to other approaches \cite{Kissinger09, DP10, Euler09} 
which have been used as a graphical language for measurement based quantum computation and for graph states \cite{DP10,Euler09}.
Our method of arriving at this collection of tensors (Figure \ref{fig:F2-presentation}) affords more general options and our presentation of the linear fragment 
here offers (i) improved semantics and (ii) a better theoretical understanding by pinpointing precisely that these networks correspond to the so called linear fragment 
of the \XOR{} or mod sum algebra carries with it new proof techniques.  These results were found by casting the theory of classical networks into a theory of tensors, 
which carried with it all of the known and desirable graphical rewrite properties from classical networks, and from this and some other methods, in \cite{CTNS} we 
assert that we have subsumed the existing graphical languages present in quantum information science by considering the 
the graphical system appearing in Figure \ref{fig:extraF2} together with the linear fragment from Figure \ref{fig:F2-presentation}.   
\end{remark}

\begin{myfullpage}
\begin{illexample}
\label{fig:extraF2}

\begin{center}
\includegraphics[width=\textwidth]{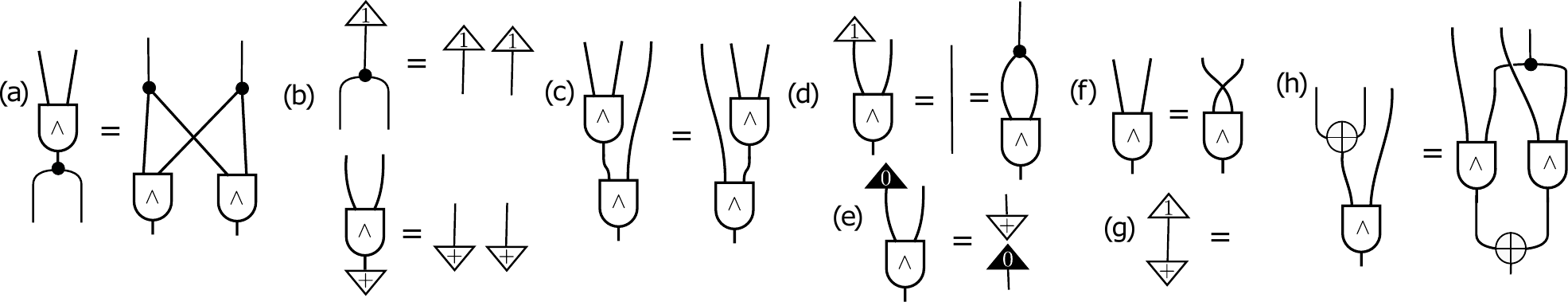}
\end{center}
A summary of the quantum \AND-tensor calculus we apply to quantum information processing and tensor networks (reproduced from Lafont \cite{boolean03} and written to match the common quantum circuit notation as in \cite{CTNS}).  This figure with 
Figure~\ref{fig:F2-presentation} is a summary of the Boolean-calculus. 
The details of (a)-(g) will be given in
Sections~\ref{sec:components}.  For
instance, (h) represents distributivity of \AND{}($\wedge$) over \XOR{} ($\oplus$),
and (d) shows that $x\wedge x = x$. (Diagrams read top to bottom.)
\end{illexample}
\end{myfullpage}

\paragraph{Bending wires.}  Proceeding axiomatically we need to add additional tensors
to represent operators and quantum states. Our network model of quantum states
requires that we are able to bend wires.  As is well known in modern algebra, we 
can hence define transposition graphically (see Figure~\ref{fig:adjoints} (d)). Cups and caps (wire bending) was also used in the categorical model of teleportation \cite{catQM}---see also the early work on graphical representations of atemporal circuits \cite{Atemp06} and the diagramatic model of teleportation therein.  

The way forward is to add what mathematicians refer to as \textit{compact
structures} (see \S~\ref{sec:bends} for further 
details).  These compact structures are given diagrammatically as
\begin{center}
\includegraphics[width=5\xxxscale]{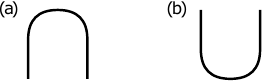}
\end{center}
and as will be explored in \S~\ref{sec:bends} these two structures allow us to
formally bend wires and to define the transpose of a linear
map/state, and provide a formal way to reshape a matrix.  We understand (a) above as
a cup, given as the generalized
Bell-state 
\be 
\sum_{i=0}^{d-1}\ket{ii}= \ket{00}+\ket{11}+...
\ee 
and (b) above as the so-called cap, Bell-costate
\be 
\sum_{i=0}^{d-1}\bra{ii}=\bra{00}+\bra{11}+...
\ee 
or \textit{effect}.

\begin{remark}[Normalization factors omitted]
As we have mentioned before, we will often omit global scale factors (contracted tensor
networks with no open wires are sent to blank on the page).  This is done for ease of presentation.  We note that for Hilbert space $\2 H$ there is a natural isomorphism
\begin{equation*}
\7 C\otimes \2 H \cong \2 H \cong \2 H\otimes \7 C,
\end{equation*}
which allows one to define equality up to a scale factor (called the scalar gauge).  Care must be taken when summing over diagrams where a relative scale factor could exist. 
\end{remark}

As readers will recall from \S~\ref{sec:tensor}, compact structures provide a formal way to bend wires --- indeed, we can now connect a diagram represented with an operator with spectral decomposition 
$$
\sum_i
\beta_i\ket{i}\bra{i},
$$ 
bend all the open wires (or legs) towards the same direction
and it then can be thought of as representing a state 
$$
\sum_i \beta_i\ket{i}\ket{\overline{i}},
$$ 
where overbar is complex conjugation), bend them
the other way and it then can be thought of as
representing a measurement outcome 
$$
\sum_i
\beta_i\bra{\overline{i}}\bra{i}, 
$$
that is an effect.  One
can also connect inputs to outputs, contracting indices and creating larger and
larger networks. With these
ingredients in place, let us now consider the class of Boolean quantum states.

\begin{remark}[Overbar notation]
The isomorphism 
\begin{equation}
    \sum_i \beta_i\bra{\overline{i}}\bra{i}\cong \sum_i
\beta_i\ket{i}\bra{i} \cong \sum_i
\beta_i\ket{i}\ket{\overline{i}},
\end{equation}
for a real valued basis becomes
\begin{equation}
    \sum_i
\beta_i\bra{i}\bra{i}\cong \sum_i
\beta_i\ket{i}\bra{i} \cong \sum_i
\beta_i\ket{i}\ket{i},
\end{equation}
which amounts to flipping a bra to a ket and vise versa.  Here we will always assume a real valued 
basis so will always omit the overbar on kets.  
\end{remark}

\subsection*{Defining the class of Boolean tensor network states}\label{sec:c}
Figure \ref{fig:andtensor} which depicts a simple but
key network building block: the use of the so-called ``quantum \AND-tensor'' which we consider in detail in \S~\ref{sec:AND}.  This is a representation of
the familiar Boolean operation in the bit pattern of a three-qubit quantum state as
\begin{equation*}
\ket{\psi_\AND} \bydef
\sum_{x_1,x_2\in\{0,1\}}\ket{x_1}\otimes\ket{x_2}\otimes\ket{x_1\wedge
x_2}=\ket{000}+\ket{010} +\ket{100} +\ket{111},
\end{equation*}
and hence the truth table of a function is encoded in the bit pattern of the
superposition state. This utilizes a representation
of Boolean gates on quantum states.

\begin{figure}[t]
\centering
\includegraphics[width=12\xxxscale]{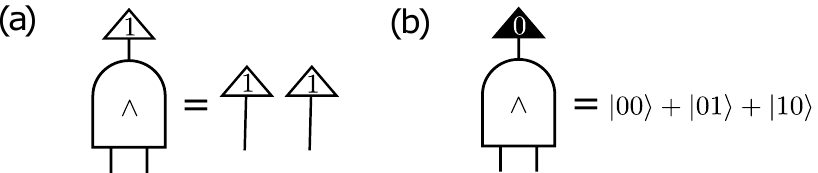}
\caption{Example of the Boolean quantum \AND-state or tensor from \cite{CTNS}.  In (a) the
tensors output is contracted with $\bra{1}$ resulting in the tensor splitting to 
the product state $\ket{11}$.  In (b) the tensors output is contracted with $\bra{0}$ resulting in the entangled state
$\ket{00}+\ket{01}+\ket{10}$.}\label{fig:andtensor}
\end{figure}

We desire to construct tensor networks with components that take binary values $0$ or $1$.  This is done by contracting the output of a switching tensor network with $\ket{1}$.  The function realized in the tensor network is constructed in such a way that 
any time the input qubits states represent a
desired term in a quantum state
(e.g.\ create a function that outputs logical-one on designated inputs
$\ket{00}$, $\ket{01}$ and $\ket{10}$ and zero otherwise as shown in Figure
\ref{fig:andtensor}). We then insert a $\ket{1}$
at the network output.  This procedure
recovers the desired Boolean state as illustrated in Figure \ref{fig:bs}(a) with the
resulting state appearing in~\eqref{eqn:Booleaninput}.  
\be \label{eqn:Booleaninput}
\begin{aligned}
&\sum_{x_1,x_2,...,x_n\in\{0,1\}}\braket{1}{f(x_1
, x_2 ,...,x_n)}\ket{x_1,x_2,...,x_n} =\\ &\sum_{x_1,x_2,...,x_n\in\{0,1\}}f(x_1
, x_2 ,...,x_n)\ket{x_1,x_2,...,x_n}.
\end{aligned}
\ee 
The network representing the circuit is read backwards from output to
input.  Alternatively the full class of Boolean states is defined as:

\begin{definition}[The Class of Boolean Quantum States]
We define the class of Boolean states as those states which can be
expressed up to a global scalar factor in the form \eqref{eqn:Booleanstates}
\begin{equation}\label{eqn:Booleanstates}
\sum_{x_1,x_2,...,x_n\in\{0,1,...,d-1\}}\ket{x_1,x_2,...,x_n}\ket{f(x_1
, x_2 ,...,x_n)},
\end{equation} 
where 
$$
f:\{0,1\}^n\rightarrow \{0,1\},
$$
is a switching function and the sum is
taken over all variables $x_j$ taking 0 and 1 for qubits (see Figure \ref{fig:bs} (a)).  
\end{definition}

\begin{remark}[Better notation]
 In practice it is often simpler to express equations such as 
 \be 
 \sum_{x_1,x_2,...,x_n\in\{0,1\}}f(x_1, x_2,...,x_n)\ket{x_1,x_2,...,x_n},
\ee 
as 
\be 
\sum_{\textbf{x}} f(\textbf{x})\ket{\textbf{x}},
\ee 
where the sum is over all assignments of $\textbf{x}:= x_1,x_2,...,x_n$.  
\end{remark}

\begin{figure}[ht]
\centering
\includegraphics[width=7\xxxscale]{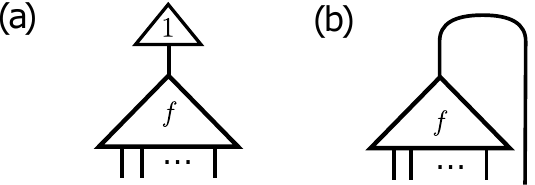}
\caption{A general Boolean quantum state arising from function $f$ can
either be formed as (a) by network contraction with a logical-one at the output of the circuit as
described by~\eqref{eqn:Booleaninput} or (b) by bending the output of the
tensor network around, as in~\eqref{eqn:Booleanstates}.}\label{fig:bs}
\end{figure}

\begin{fullpage}

\begin{example}[{\sf GHZ}-states and \W-states]
 Examples of Boolean states include the familiar {\sf GHZ}-state
$\ket{00\cdots0}+\ket{11\cdots1}$ which on qudits in dimension $d$ becomes 

\begin{equation}
 \ket{{\sf GHZ}_d}= \sum_{i=0}^{d-1} \ket{i}\ket{i}\ket{i} =
\ket{0}\ket{0}\ket{0}+\ket{1}\ket{1}\ket{1}+\cdots+\ket{d-1}\ket{d-1}\ket{d-1}, 
\end{equation}
as well as the \W-state
$\ket{00\cdots1}+\ket{01\cdots0}+\cdots+\ket{10\cdots0}$ which again on qudits
becomes
\be
\begin{aligned}\label{eqn:w-stated}
&\ket{\W_d} := \sum_{i=1}^{d-1} \sum_{j=1}^{3} (X_j)^i\ket{0}\ket{0}\ket{0} =
\ket{0}\ket{0}\ket{1}+\ket{0}\ket{1}\ket{0}+\ket{1}\ket{0}\ket{0}+\ket{0}\ket{0}\ket{
2}+\\
&\ket{0}\ket{2}\ket{0}+
\ket{2}\ket{0}\ket{0}+\cdots \cdots+\ket{0}\ket{0}\ket{d-1}+\ket{0} \ket{d-1}\ket{0} +\ket{d-1}\ket{0}\ket{0}. 
\end{aligned}
\ee

In~\eqref{eqn:w-stated} the operator
$X\ket{m}=\ket{m+1(\text{mod}~d)}$ is
one way to define negation in higher dimensions \cite{BB11}.  The subscript labels the ket
(labeled 1,2 or 3 from left to right) the operator acts $i$ times.
\end{example}
\end{fullpage}
\begin{remark}[Extensions to arbitrary quantum states]
What is clear from this definition is that Boolean states are always composed of equal
superpositions of sets of computational basis states, as the allowed scalars take binary values, 0,1. 
Despite this apparent limitation,
tensor networks composed only of Boolean components can nonetheless describe
any quantum state. To do this we require a minor extension to include superposition
input/output states, e.g. order-1 tensors of the form
$\ket{0}+\beta_1\ket{1}+\cdots+\beta_{d-1}\ket{d-1}$. This
gives a universal class of generalized Boolean tensor networks
which subsumes the important subclass of Boolean states.
This class is then shown to form a nascent example
of the exhaustiveness of BTNS and to give rise to a wide class of
quantum states that we show are exactly and efficiently sampled \cite{CTNS}.   
\end{remark}

 \begin{remark}[Comparison to other approaches]
  The theory of tensor network states has received recent interest fueled by developments that have been made related to 
  the important problem of quantum simulation, using tensor contraction algorithms \cite{MPSreview08,TNSreview09}. There is also an established language of quantum circuits, appearing in most text books on quantum computing and quantum information.  These circuits 
  are effectively tensor networks and efforts have been made to form an extension and unite the two \cite{BB11}.  
 \end{remark}

\section{Quantum Legos: a tensor tool box}\label{sec:components}
A key point to this is that the introduction of Boolean logic gate tensors
into the tensor network context allows the seminal logic gate universality results from classical network
theory to be applied in the setting of tensor network states. 

\begin{remark}[Dual spaces]
Any vector space $\2V$ has a dual $\2V^*$: this is the space of linear functions $f$
from $\2V$ to the ground field $\7C$, that is $f: \2V \rightarrow \7C$. This defines
the dual uniquely.  We must however fix a basis to identify the vector
space $\2V$ with its dual.  Given a basis, any basis vector $\ket{i}$ in $\2V$ gives
rise to a basis vector $\bra{j}$ in $\2V^*$ defined by $\braket{j}{i} = \delta^j_i$ (Kronecker's delta). 
This defines an isomorphism $\2V\rightarrow \2V^*$ sending $\ket{i}$ to $\bra{i}$ and
allowing us to identify $\2V$ with $\2V^*$.  In what follows, we will fix a particular
arbitrarily chosen basis (called the computational basis in quantum information
science).  We will now concentrate on Boolean building blocks that are used
in our construction.  
\end{remark}


\subsection*{Review of Boolean algebra}

Here we have reviewed Boolean tensor building blocks.  These building blocks appear in many applications of tensor network algorithms, including \cite{tommy}. Here we encourage the readers to review 
Boolean algebra to better understand the presented structures.  
We advise readers to quickly review Appendix \ref{sec:XOR} on \XOR-algebra as well as 
Appendix \ref{appendix:kmap} on the method of Karnaugh map equation reduction.  The following sections
will assume these methods are known to the reader.

\subsection{{\sf COPY}-tensors: the ``diagonal''}\label{sec:COPY}

The copy operation arises in digital circuits~\cite{Davio78, Weg87}
and more generally, in the context of category theory and algebra, where it is
called a diagonal \cite{CD, redgreen}. The
operation is readily defined in any finite dimension as
\begin{equation}
\bigtriangleup \bydef \sum_{i=0}^{d-1} \ket{ii}\bra{i}.
\end{equation}
As $\ket{0}$ and
$\ket{1}$ are
eigenstates of $\sigma^z$, we might give $\bigtriangleup$ the alternative name of
{\sf Z}-copy.  In the
case of qubits {\sf COPY} is
succinctly presented by considering the map $\bigtriangleup$ that copies
$\sigma^z$-eigenstates:
$$
\bigtriangleup:\7C^2\rightarrow\7C^2\otimes\7C^2:
\begin{cases}
\ket{0} \mapsto \ket{00}\\
\ket{1} \mapsto \ket{11}
\end{cases}
$$
This map can be written in operator form as $\bigtriangleup:\ket{00}\bra{0} +
\ket{11} \bra{1}$ and
under cup/cap induced duality (on the right bra) this state becomes a {\sf GHZ}-state
as
$\ket{\psi_{\sf GHZ}} = \ket{000} + \ket{111}\cong\ket{00}\bra{0} +
\ket{11} \bra{1}$. The standard properties of {\sf COPY} are
given diagrammatically in Figure~\ref{fig:copygate} and a list of its relevant
mathematical properties are found in Table~\ref{fig:copygatesum}.

\begin{illexample}
\label{fig:copygate}
\begin{center}
\includegraphics[width=.95\textwidth]{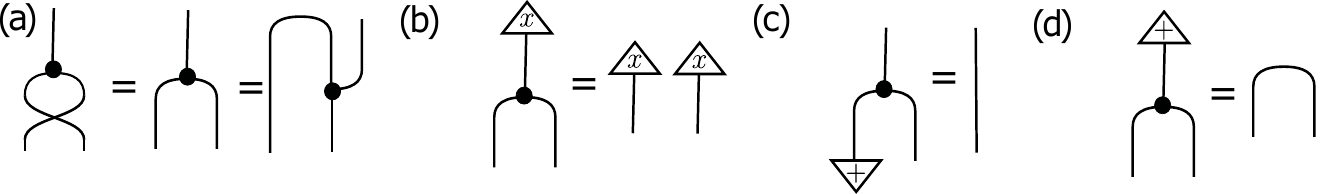}
\end{center}
Some diagrammatic properties of the {\sf COPY}-tensor. (a) Full-symmetry.
(b) Copy points, e.g.\ $\ket{x}\mapsto\ket{xx}$ for $x=0,1$ for qubits.  (c) The unit --- in this case
the unit corresponds to deletion, or a map to the
terminal object which is given as $\bra{+}\bydef\bra{0}+\bra{1}$ for qubits and
$\bra{+}\bydef\bra{0}+\bra{1}+\cdots+\bra{d-1}$ for $d$
dimensional qudits.  (d) Co-interaction with the unit creates a Bell state.
\end{illexample}

\begin{remark}[The {\sf COPY}-gate from \CNOT]
The \CNOT-gate is defined as $\ket{0}\bra{0}_1\otimes \I_2 +
\ket{1}\bra{1}_1\otimes \sigma^x_2$.  We will set the input that the target acts on
to $\ket{0}$ then calculate $\CNOT (\I_1\otimes
\ket{0}_2)=\ket{0}\bra{0}_1\otimes \ket{0}_2 + \ket{1}\bra{1}_1\otimes
\ket{1}_2$.  We have hence defined the desired {\sf COPY} map copying states from the
Hilbert space with label $1$ (subscript) to the joint Hilbert space labeled $1$ and
$2$.   
\end{remark}

\begin{remark}[The types of possible states built from \COPY]
An alternative definition of the \COPY-tensor would be to define the operation by raising or lowering indices
 on $\delta^{ij}_{~~k}$, a Kronecker delta function.   
 In that regard, one might write the n-party \GHZ-state as 
 \be 
 \psi_\GHZ = \sum \delta^{ijk...l}\ket{ijk...l}.
 \ee 
Tensor products of state of this form are precisely the only types of states constructible with the \COPY-tensor alone.  
\end{remark}

\subsection{\XOR-tensors: the ``addition''}\label{sec:XOR}
The classical \XOR-gate implements exclusive disjunction or addition (mod 2 for
qubits) and is denoted by the
symbol $\oplus$~\cite{Cohn62, xor70}. We note that for multi-valued
logic a modulo subtraction gate can also be defined as in~\cite{BB11}.

\begin{remark}[relation to \COPY~\cite{CD, redgreen}]
 As is a well known fact in algebra, the \XOR-gate is simply a Hadamard
transform of the {\sf COPY}-gate, appropriately applied to all of the tensors
legs.  This can be captured diagrammatically in the slightly different form: 
\begin{center}
\includegraphics[width=0.250\textwidth]{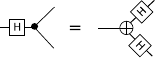}
\end{center}
\end{remark}

\begin{remark}[Symmetry]
 The \XOR-tensor is symmetric under leg exchange.  In components, if we write 
 $X_{ijk}$ then $X_{ijk}=1$ for $i=j=k=0$ (mod 2) and 0 otherwise.  
\end{remark}

To define the gate on the computational basis, we
consider $f(x_1,x_2)=x_1\oplus x_2$ then $f=0$ corresponds to
$(x_1,x_2)\in\{(0,0),(1,1)\}$ and $f=1$ corresponds to $(x_1,x_2)\in\{(1,0),(0,1)\}$,
where the
truth table for $\XOR$ follows. 
\begin{center}
\begin{tabular}{c|c|c}
$~x_1~$ & $~x_2~$ & $f(x_1,x_2)=x_1\oplus x_2$ \\ \hline
0 & 0 & 0 \\
0 & 1 & 1 \\
1 & 0 & 1 \\
1 & 1 & 0
\end{tabular}
\end{center}
Under cap/cap induced duality, the state defined by $\XOR$ is given as
\begin{equation}
\ket{\psi_\oplus}\bydef\sum_{x_1,x_2\in\{0,1\}}\ket{x_1}\ket{x_2}\ket{f(x_1,x_2)}
=\ket { 000 }
+\ket { 110 }
+\ket{011}+\ket{101}, 
\end{equation}
which is in the {\sf GHZ}-class by LOCC equivalence viz.
$\ket{\psi_\oplus}=\H\otimes\H\otimes\H(\ket{000}+\ket{111})$.  The operation of
\XOR~is
summarized in Table~\ref{fig:xorgatesum}.  Since the
\XOR-gate is related to the {\sf COPY}-gate by a change of basis, its diagrammatic
laws
have the same structure as those illustrated in Figure~\ref{fig:copygate}.  The gate
acting backwards (co-\XOR) is
defined on a basis as follows:
$$
\oplus:\7C^2\rightarrow\7C^2\otimes\7C^2:
\begin{cases}
\ket{0} \mapsto \ket{00}+\ket{11}\\
\ket{1} \mapsto \ket{10}+\ket{01}
\end{cases}~~~~\text{or equivalently}~~~~~~~\begin{cases}
\ket{+} \mapsto \ket{++}\\
\ket{-} \mapsto \ket{--}
\end{cases}
$$

\subsection*{Generating the affine class of networks}
Thus far we have presented the {\sf XOR}- and {\sf COPY}- gates.  This system allows
us to create the linear class of Boolean functions.  As explained in the present
subsection, this class can be extended to to the affine class by introducing either a
gate that acts like an inverter, or by appending a constant $\ket{1}$ into our
system.  This constant will allow us to use the {\sf XOR}-gate to create an
inverter.   

\begin{definition}[Complemented vs uncomplemented Boolean variables]
A \textit{complemented Boolean variable} is a
Boolean variable that appears in negated form, that is $\neg x$ or
written equivalently as $\overline{x}$. Negation of a Boolean variable $x$ can be
expressed as the \XOR{} of
the variable with constant $1$ as $\overline{x}=1\oplus x$. Whereas
\textit{Uncomplemented Boolean
variables} are Boolean variables that do not appear in negated form (e.g.\ negation
is not allowed).  Linear Boolean functions contain terms with Uncomplemented Boolean
variables that
appear individually (e.g.\ variable
products are not allowed such as $x_1x_2$ and higher orders etc., see
\S~\ref{sec:Boolean}).   
\end{definition}

\begin{definition}[Linear Boolean functions]
Linear Boolean functions take the general form
\be
f(x_1,x_2,...,x_n)=c_1x_1\oplus c_2x_2\oplus ...\oplus c_nx_n,
\ee 
where the vector $(c_1,c_2,...,c_n)$ uniquely determines the function.   
\end{definition}

\begin{definition}[Affine boolean functions]
The affine
Boolean functions take the same general form as linear functions.  However, functions
in the affine class allows variables to appear in both complemented
and uncomplemented form. Affine Boolean functions take the general form
\begin{equation}\label{eqn:affine1}
f(x_1,x_2,...,x_n)=c_0\oplus c_1x_1\oplus c_2x_2\oplus ...\oplus c_nx_n,
\end{equation}
where $c_0=1$ gives functions outside the linear class.  From the identities,
$1\oplus 1=0$ and $0\oplus x=x$ we require the introduction of only one constant
($c_0$), see Appendix \ref{sec:Boolean}. 
\end{definition}

Together, \XOR{} and {\sf COPY} are not universal for classical circuits. When used
together, \XOR- and {\sf COPY}-gates compose to create networks representing the class of linear
circuits. The affine circuits are generated by considering the constant $\ket{1}$.
The state $\ket{1}$ is
indeed copied by the black tensor. However, our axiomatization
(Figure~\ref{fig:F2-presentation}) proceeds through 
considering the \XOR- and {\sf COPY}-gates together with $\ket{+}$, the unit for
{\sf COPY} and $\ket{0}$ the unit for \XOR.  It is by appending the constant
$\ket{1}$ into the
formal system (Figure~\ref{fig:F2-presentation}) that the affine class of circuits
can be realized.

\begin{remark}[Affine functions correspond to a basis]
Each affine function is labeled by a corresponding bit pattern.  This can be
thought of as labeling the computational basis, as
states of the form $\ket{\{0,1\}^n}$ are in correspondence with polynomials in
algebraic normal form (see Appendix \ref{sec:Boolean}).
\end{remark}

\subsection{Quantum \AND-state tensors: Boolean universality}\label{sec:AND}
The proceeding sections have introduced enough machinery to generate the linear and
affine classes of classical circuits.  These classes are not universal.  To recover
a universal system one will introduce the \AND{} gate as a tensor \cite{CTNS}.  The multiplicative unit for this gate is $\bra{1}$ and
so can be used to elevate the linear fragment to the affine class.  

The \AND{} gate (that is, $\wedge$) implements logical conjunction~\cite{Davio78,
Weg87}.  The \AND-gate relates to the \OR-gate via De
Morgan's law.
This can be captured diagrammatically as
\begin{center}
\includegraphics[width=0.35\textwidth]{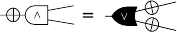}
\end{center}
To define the gate on the computational basis, we
consider $f(x_1,x_2)=x_1\wedge x_2$ which we write in short hand as $x_1 x_2$.  Here
$f=0$
corresponds to $(x_1,x_2)\in\{(0,0),(0,1),(1,0)\}$ and $f=1$ corresponds to
$(x_1,x_2)=(1,1)$.

Under cap/cap induced duality, the state defined by $\AND$ is given as
\begin{equation}\label{eqn:quantumAND}
\ket{\psi_\wedge}\bydef\sum_{x_1,x_2\in\{0,1\}}\ket{x_1}\ket{x_2}\ket{f(x_1,x_2)}
=\ket { 000 }
+\ket { 010 }
+\ket{010}+\ket{111}.
\end{equation}
The key diagrammatic properties of {\sf AND}
are presented  in Figure~\ref{fig:ANDgate} and the gate is 
summarized in Table~\ref{fig:andgate}. 

The gate acting backwards (co-\AND) is
defined on a basis as follows: 

$$
\wedge:\7C^2\rightarrow\7C^2\otimes\7C^2:
\begin{cases}
\ket{0} \mapsto \ket{00}+\ket{01}+\ket{10}\\
\ket{1} \mapsto \ket{11}
\end{cases}~\text{or}~~\begin{cases}
\ket{+} \mapsto \ket{++}\\
\ket{-} \mapsto \ket{00}+\ket{01}+\ket{10} - \ket{11}
\end{cases}
$$

\begin{illexample}
 \label{fig:ANDgate}
Salient diagrammatic properties of the \AND-tensor.
\begin{center}
\includegraphics[width=0.95\textwidth]{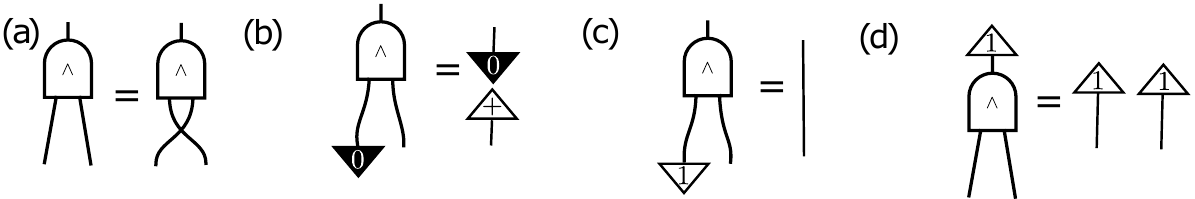}
\end{center}
 (a) Input-symmetry. (b)
Existence of a zero or fixed-point.
(c) The unit $\ket{1}$.  (d) Co-interaction with the unit
creates a product-state.  Note that the gate forms a valid
quantum operation when run backwards as in (d).
 \end{illexample}

\begin{example}[\AND-states from Toffoli-gates]\label{ex:ANDfromtoff}
The \AND-state is readily constructed from the Toffoli gate \cite{CTNS} as illustrated in
Figure~\ref{fig:ANDfromtoff}.  This allows some
interesting states to be created 
experimentally, for instance, post-selection of the output to $\ket{0}$
would yield the state $\ket{00}+\ket{01}+\ket{10}$. 
\end{example}

 \begin{remark}[Quantum universality and universal
states]\label{re:quantumuniversal}
 The problem of determining universal quantum gate families has received significant
 research interest resulting in the simplistic universal gate sets appearing
 in~\cite{A03,Shi:02,RG02} and elsewhere.  It is even known that Toffoli and Hadamard are
universal
 for quantum computation~\cite{A03}. Toffoli can be generated by combining one
 \AND-state and two {\sf COPY}-states (see also Figure \ref{fig:hadamard}).  \mn{In \cite{CTNS} the ZX calculus (Clifford gates plus cups and caps to bend wires) by considering the addition of \AND{}-states which can represent Toffoli gates was shown to be quantum computationally universal \cite{CTNS}.  With the addition of scalars, the \AND+ZX calculus presented in \cite{CTNS} was proven to be approximately universal  for linear maps between qubits.} 
 \end{remark}

\begin{theorem}[Toffoli contracts to give the \AND{}-state \cite{CTNS}]\label{fig:ANDfromtoff}
The following rewrites exhibit the use units to prepare the \AND-state \cite{CTNS}.  Using this state together with single qubit \NOT-gates, one can construct tensor networks which any Boolean qubit state as well as any of the states appearing in Table~\ref{fig:BooleanStates}.  We note that the box around the Toffoli gate (left) is meant to illustrate a
difference between our notation and that of quantum circuits.  In our notation, those dots inside the box would merge which of course is not a valid unitary gate.
\begin{center}
\includegraphics[width=0.95\textwidth]{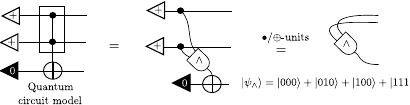}
\end{center}
\end{theorem}

\begin{theorem}
 Hadamard follows from contracting the \AND-state 
together with $\ket{-}\bydef\frac{1}{\sqrt{2}}(\ket{0}-\ket{1})$ \cite{CTNS}.  
\begin{proof}
\begin{center}\label{fig:hadamard}
    \includegraphics[width=9\xxxscale]{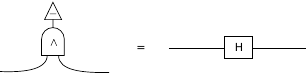}
\end{center}
Follows by direct calculation. 
\end{proof}
\end{theorem}

\begin{definition}
 Hadamard
states are defined as 
$$
\ket{\psi_\H} = \ket{00}+\ket{01}+\ket{10}-\ket{11}. 
$$
\end{definition}

\begin{lemma}
Tensor contractions formed from Hadamard states, {\sf COPY}-and \AND-states are universal for quantum computation \cite{CTNS}. 
\begin{proof}
The proof \cite{CTNS} follows from wire bending duality and the proof that Hadamard and Toffoli are universal
for quantum circuits~\cite{A03}.
\end{proof}
\end{lemma}

\subsubsection{Summary of the {\sf XOR}-algebra on tensors}

We will now present the three previously referenced Tables (\ref{fig:copygatesum},
\ref{fig:xorgatesum} and \ref{fig:andgate}) which summarize the quantum logic tensors
we introduced in the previous subsections (\ref{sec:COPY}, \ref{sec:XOR} and
\ref{sec:AND}).  The tables contain entries listing properties that describe how
the introduced network components interact \cite{boolean03, CD, redgreen, CTNS, BB11}.  These interactions
are defined diagrammatically and explained in \S~\ref{sec:components}.  

\begin{center}
\begin{table}[ht]
\begin{tabular}{|c|c|c|c|}\hline
Gate Type & Co-copy point(s) & Unit & Co-unit Interaction\\\hline
{\sf COPY} & $\ket{0}$,$\ket{1}$  & $\ket{+}$  & Bell state: $\ket{00}+\ket{11}$
 \\ \hline\hline
Symmetry & Associative & Commutative & Frobenius Algebra\\\hline
Full  & Yes         & Yes & Yes (Fusion Law) \\ \hline
\end{tabular}
\caption{Summary of the {\sf COPY}-gate from \S~\ref{sec:COPY}.}\label{fig:copygatesum}
\end{table}
\end{center}

\begin{center}
\begin{table}[ht]
\begin{tabular}{|c|c|c|c|}\hline
Gate Type & Co-copy point(s) & Unit & Co-unit Interaction\\\hline
\XOR & $\ket{+}$,$\ket{-}$  & $\ket{0}$  & Bell state: $\ket{00}+\ket{11}$
 \\ \hline\hline
Symmetry & Associative & Commutative & Frobenius Algebra\\\hline
Full  & Yes         & Yes & Yes (Fusion Law) \\ \hline
\end{tabular}
\caption{Summary of the \XOR-gate from \S~\ref{sec:XOR}.}\label{fig:xorgatesum}
\end{table}
\end{center}

\begin{center}
\begin{table}[ht]
\begin{tabular}{|c|c|c|c|}\hline
Gate Type & Co-copy point(s) & Unit & Co-unit Interaction\\\hline
\AND & $\ket{1}$  & $\ket{1}$  & Product state: $\ket{11}$
 \\ \hline\hline
Symmetry   & Associative & Commutative & Bialgebra Law\\\hline
Inputs  & Yes         & Yes         & Yes (with {\sf GHZ}) \\ \hline
\end{tabular}
\caption{Summary of the \AND-gate from \S~\ref{sec:AND}.}\label{fig:andgate}
\end{table}
\end{center}

\subsection*{co-{\sf COPY}: the co-diagonal}\label{sec:co-diagonal}
What is evident from our subsequent discussions on logic gates is that in the context
of tensors, the bending of wires implies that gates can be used both forwards in backwards. We can therefore form tensor networks from Boolean gates
in a very
different way from classical circuits. Indeed, it becomes possible to flip a {\sf
COPY} operation upside down,
that is, instead of having a single leg split into two legs, have two legs merge into
one.  In
terms of tensor networks, co-{\sf COPY} is simply
thought of as being a dual (transpose) to the familiar {\sf COPY} operation. This is
common in
algebra: to consider the dual
notation to algebra, that is co-algebra.  In general, while a product is a joining or pairing (e.g.\ taking two
vectors and producing a third) a co-product is a co-pairing taking a single vector
in the space $\2A$ and producing a vector in the space $\2A\otimes \2 A$.

\begin{remark} [co-algebras~\cite{FA}]
co-algebras are structures that are dual (in the sense of
reversing arrows) to unital associative algebras such as {\sf COPY} and
\AND{} the axioms of which we formulated in terms of picture calculi
(Sections \ref{sec:COPY} and \ref{sec:AND}).  Every co-algebra, by
(vector space) duality, gives rise to an algebra, and in finite dimensions, this duality goes in both directions.  
\end{remark}

Co-{\sf COPY} can be thought of as applying a delta function in the transition from
input to output.  That is, given a copy point $x=0,1,...,d-1$ for qudits on dim $d$.
 Depicting {\sf COPY} as the map
$\bigtriangleup$
\begin{equation}
\bigtriangleup(\ket{x})=\ket{x}\otimes\ket{x},
\end{equation}
we define co-{\sf COPY} by the map $\bigtriangledown$ such that 
\begin{equation}
\bigtriangledown(\ket{i},\ket{j})=\delta_{ij} \ket{i},
\end{equation}
that is, the diagram is mapped to zero (or empty) if the inputs $\ket{i}$, $\ket{j}$
do not agree.  This is succinctly expressed in terms of a delta-function dependent on
inputs $\ket{i}$, $\ket{j}$ where $i,j=0,1,...,d-1$ for qudits of dim $d$.

\begin{example}[Simple co-pairing]
Measurement effects on tripartite quantum systems can be thought of as
co-products.  This is given as a map from one system (measuring the first) into two
systems (the effect this has on the other two). {\sf GHZ}-states are prototypical
examples of co-pairings.  In this case, the measurement outcome of $\ket{0}$
($\ket{1}$) on a single subsystem sends the other qubits to $\ket{00}$ ($\ket{11}$)
and by linearity this sends $\ket{+}$ to $\ket{00}+\ket{11}$.  
\end{example}

\subsection*{The remaining Boolean tensors: \NAND-states etc.}
We have represented a logical system on tensors --- this enables us
to represent any Boolean function as a connected network of tensors and hence any Boolean state.
We chose as our generators,
constant $\ket{1}$, {\sf COPY}, \XOR, \AND{}.  Other generators could have also been
chosen such as \NAND-tensors. Our choice however, was made as a matter of convenience. If we
had considered other generators, we could have ended up considering the
following cases: weak-units~(Definition \ref{def:weakunits}) and fixed point
pairs~(Definition \ref{def:fixedpoints}). 


\begin{definition}[Weak units]\label{def:weakunits}
An algebra (or product see Appendix \ref{sec:newalgebra}) on a tripartite state
$\ket{\psi}$ has a unit (equivalently, one has that the state is unital) if there
exists an effect $\bra{\phi}$ which
the product acts on to produce an invertible map $B$, where $B=\I$ (see
Example~\ref{ex:weakunits}). If no such $\bra{\phi}$ exists to make $B=\I$, and $B$
has an inverse, we call $\bra{\phi}$ a weak unit, and say the state $\ket{\psi}$ is
weak unital and if $B\neq \I$ and $B^2=1$ we call
the algebra on $\ket{\psi}$ unital-involutive.  This scenario is given
diagrammatically
as:
\begin{center}
\includegraphics[width=6.5\xxxscale]{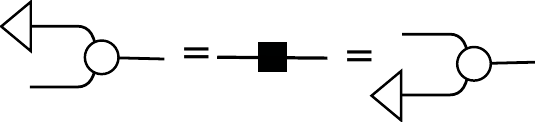}
\end{center}
\end{definition}
\begin{example}[\NAND{} and \NOR{}]\label{ex:weakunits}
\NAND{} and \NOR{} have weak units, respectively given by $\ket{1}$ and $\ket{0}$.
These weak units are unital-involutive.
\begin{equation}
 \ket{\psi_\NAND} = \ket{001}+\ket{011}+\ket{101}+\ket{110},
\end{equation}
\begin{equation}
 \ket{\psi_\NOR} = \ket{001}+\ket{010}+\ket{100}+\ket{110}.
\end{equation}
For $\ket{\psi_\NAND}$ to have a unit, there must exist a $\ket{\phi}$ such that
\begin{equation}
 \braket{\overline{\phi}}{0}\ket{01}+\braket{\overline{\phi}}{0}\ket{11}+\braket{
\overline{\phi}}{1}\ket{01} +\braket{\overline{\phi}}{1}\ket{10}=\ket{00}+\ket{11},
\end{equation}
and hence no choice of $\ket{\phi}$ makes this possible, thereby confirming the
claim.
\end{example}

\begin{definition}[Fixed Point Pair]\label{def:fixedpoints}
An algebra (see Appendix \ref{sec:newalgebra}) on a tripartite state
$\ket{\psi}$ has a \textit{fixed point} if
there exists an effect $\bra{\phi}$ (the fixed point) which
the product acts on to produce a constant output, independent of the other input value. 
For instance, in Figure \ref{fig:fixedpoints}(c) on the left hand side the effect
$\bra{1}$ induces a map (read bottom to top) that sends $\ket{+}\mapsto \ket{1}$. 
Up to a scalar, this map expands linearly sending both basis effects $\bra{0}$,
$\bra{1}$ to to the constant state $\ket{1}$. If the resulting output is the same as
the fixed point, we say $\bra{\phi}$ has a zero ($\ket{1}$ is the zero for the
\OR-gate in Figure \ref{fig:fixedpoints}(c)). A
fixed point pair consists of two algebras with fixed points, such that the fixed
point of one algebra is the unit of the other, and vise versa (see
Figure~\ref{fig:fixedpoints}). Diagrammatically this
is given in Figure \ref{fig:fixedpoints2}.  
\end{definition}

\begin{figure}[ht]
\centering
\includegraphics[width=11\xxxscale]{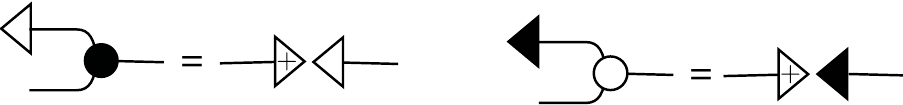}
\caption{Diagrammatic equations satisfied by a fixed point
pair (see Definition \ref{def:fixedpoints}). }\label{fig:fixedpoints2}
\end{figure}

\begin{figure}[ht]
\centering
\includegraphics[width=11\xxxscale]{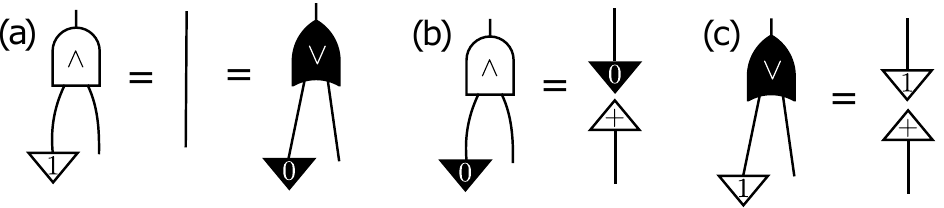}
\caption{\AND{} and \OR{} tensors form a fixed point pair.  The unit for \AND{}
($\ket{1}$
see a) is the zero for \OR{} (c) and vise versa: the unit of \OR{} ($\ket{0}$
see a) is the zero for \AND{} (b).}\label{fig:fixedpoints}
\end{figure}

\subsection*{Summarizing: network composition of quantum logic tensors}
We have considered sets of universal classical structures in our tensor
network model.  In classical computer science, a universal set of gates is able to
express any $n$-bit Boolean function
\begin{equation}
  f:\7B^n\rightarrow\7B:(x_1,...,x_n)\mapsto f(x_1,...,x_n),
\end{equation}
where we note that $\7Z_2\cong\7B$ allowing us to use the alternative notation for
$f$ as $f:\7Z_d^n\rightarrow\7Z_d$ with $d=2$ for the binary case. Universal sets
include
\begin{itemize}
    \item[1.] \{{\sf COPY}, \NAND\},
    \item[2.] \{{\sf COPY}, \AND,~\NOT\},
    \item[3.] \{{\sf COPY}, \AND,~\XOR, $\ket{1}$\},
    \item[4.] \{\OR,~\XNOR, $\ket{1}$\} and others.  
\end{itemize}
One can also consider the states $\ket{\psi}$ formed by the bit patterns of these functions $f(x_1,x_2)$ as
\be
\ket{\psi_f} = \sum_{x_1,x_2\in\{0,1\}}\ket{x_1}\ket{x_2}\ket{f(x_1,x_2)}.
\ee
This allows a wide class of states to be constructed effectively.  In the
following
Table (\ref{fig:BooleanStates}) we illustrate the quantum states representing the classical
function of two-inputs.

The bit pattern of the following quantum states represents a Boolean function (given by the subscript) such that the right most bit is the Boolean functions output, and the two left bits are the functions inputs, and the non-linear Boolean functions are
on the left side of the table and the linear functions on the right.  Consider the state $\ket{\psi_\AND}$, and Boolean variables $x_1$ and $x_2$, then the superposition $\ket{\psi_\AND}$ encodes the function $\ket{x_1,x_2,x_1\wedge x_2}$ in each term in the superposition, and
$$ 
\ket{\psi_\AND}=\sum_{x_1,x_2\in\{0,1\}}\ket{x_1,x_2,x_1\wedge
x_2}.
$$ 

As outlined in the text, cup/cap induced-duality allows us (for instance) to express this state as the operator
$$ 
\ket{0}\bra{00}+\ket{0}\bra{01}+\ket{0}\bra{01}+\ket{1}\bra{11}:\ket{x_1,x_2} \mapsto\ket{x_1\wedge x_2}
$$
which projects qubit states to the $\AND$ of their bit value.

\begin{center}
\begin{table}[ht]
\begin{tabular}{|c|c|}\hline
non-linear & linear (Frobenius Algebras) \\\hline
$\ket{\psi_\AND}=\ket{000}+\ket{010}+\ket{100}+\ket{111}$ &  \\
~$\ket{\psi_\OR}=\ket{001}+\ket{011}+\ket{101}+\ket{111}$ &
~$\ket{\psi_\XOR}=\ket{000}+\ket{011}+\ket{101}+\ket{110}$\\
$\ket{\psi_\NAND}=\ket{001}+\ket{011}+\ket{101}+\ket{110}$ &
$\ket{\psi_\XNOR}=\ket{001}+\ket{010}+\ket{100}+\ket{111}$\\
~$\ket{\psi_\NOR}=\ket{001}+\ket{010}+\ket{100}+\ket{110}$ & \\\hline
\end{tabular}
\end{table}\label{fig:BooleanStates}
\end{center}

\subsubsection{Merging {\sf COPY}-tensors by node equivalence}
{\sf COPY}-tensors are readily generalized to an arbitrary number of input and
output legs. As one would rightly suspect,
a {\sf COPY}-tensor with $n$~inputs and $m$~outputs corresponds to an
$n+m$-partite {\sf GHZ}-state.  Neighboring tensors of the same type can be merged into
a single tensor: this is called node equivalence in digital circuits. {\sf
COPY}-tensors represent Frobenius algebras\mn{In the work \cite{COECKE_2012}, the \COPY{}-tensor, through its properties as a Frobenius algebra \cite{DP13}, was shown to be equivalently characterised by an orthogonal basis for a finite-dimensional Hilbert space.}~\cite{Carboni-Walters,FA}.

\begin{theorem}[Node equivalence or fusion law]\label{theorem:spider} Given a
connected graph with $m$~inputs and $n$~outputs comprised solely of \COPY-tensors of
equal dimension, this map
can be equivalently expressed as a single $m$-to-$n$ tensor, as shown as 
\begin{center}
    \includegraphics[width=0.95\textwidth]{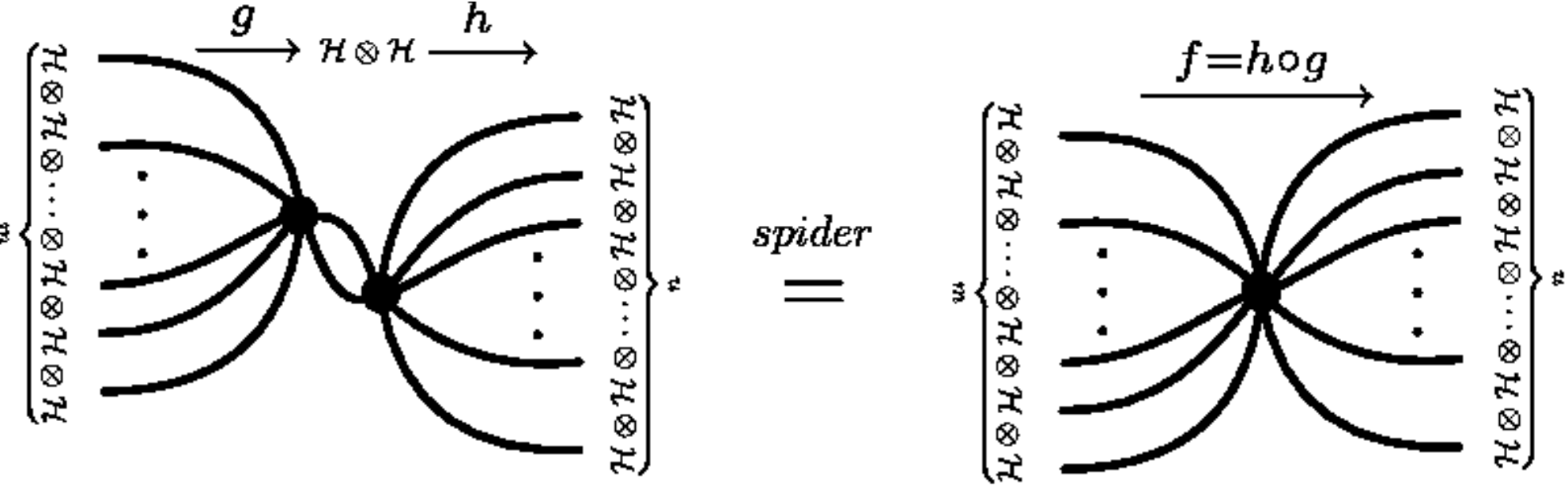}
\end{center}
Node equivalence or fusion law.  Connected black-tensors ($\bullet$) as well as connected plus-tensors
($\oplus$) can be merged and also split apart at will.  The intuition for digital or qudit circuits follows by connecting a state $\ket{\phi}$ to one of the legs and iterating over a complete basis $\ket{0}$, $\ket{1}$,...,$\ket{d-1}$.\\ 

This rule goes by many different names, depending on the community.  For example, in classical circuits this is called node equivalence whereas researchers in categorical quantum mechanics, credit node equivalence as their own {\it spider law} \cite{redgreen}. 
\label{fig:Spider}
\end{theorem}

\subsection*{Associativity, distributivity and commutativity}

The products we have considered are all associative and commutative.  As algebras,
\AND{}, \XOR{} and {\sf COPY}{} are associative, unital commutative algebras.  This
was
already expressed diagrammatically in Figures~\ref{fig:F2-presentation}(a) and
Figure \ref{fig:extraF2}(c).
The diagrammatic laws relevant for this subsection represent the following
Equations
\begin{equation}
(x_1\wedge x_2)\wedge x_3= x_1\wedge (x_2\wedge x_3),
\end{equation}
\begin{equation}
(x_1\oplus x_2)\oplus x_3= x_1\oplus (x_2\oplus x_3).
\end{equation}
Distributivity of \AND{} over \XOR{} then becomes (see (h) in
Figure~\ref{fig:extraF2})
\begin{equation}
(x_1\oplus x_2)\wedge x_3= (x_1\wedge x_2)\oplus (x_1\wedge x_2).
\end{equation}
We have commutativity for any product symmetric in its inputs: this is
the case for \AND{} and \XOR.

\subsection*{Bialgebras on tensors}\label{sec:bialgebra}
There is a powerful type of algebra that arises in our
setting: a bialgebra defined graphically on tensors in Figure
\ref{fig:BialgebraAxioms} (see Kassel, Chapter III~\cite{Kassel}, \cite{FA} and \cite{redgreen}).  

 Such an algebra is simultaneously a unital associative algebra and co-algebra (for
the associativity condition see (b) in Figure~\ref{fig:BialgebraAxioms}).
Specifically, we consider the following two ingredients:
\begin{description}\addtolength{\itemsep}{-0.5\baselineskip}
\item[(i)] A product (black tensor) with a unit (black triangle) see the right hand
side of Figure~\ref{fig:BialgebraAxioms}(a).
\item[(ii)] A co-product (white tensor) with a co-unit
(white triangle) see the left hand side of Figure~\ref{fig:BialgebraAxioms}(a). 
\end{description}
To form a bialgebra, these two ingredients above must be characterized by
the following four compatibility conditions:
\begin{description}\addtolength{\itemsep}{-0.5\baselineskip}
\item[(i)] The unit of the black tensor is a copy-point of the white tensor as in (e)
from Figure~\ref{fig:BialgebraAxioms}.
\item[(ii)] The (co)unit of the white tensor is a copy-point of the black tensor as in (d)
from Figure~\ref{fig:BialgebraAxioms}.
\item[(iii)] The bialgebra-law is satisfied given in (c) from
Figure~\ref{fig:BialgebraAxioms}.
\item[(iv)] The inner product of the unit (black triangle) and the co-unit (white
triangle) is non-zero (not shown in Figure~\ref{fig:BialgebraAxioms}).
\end{description}

\begin{figure}[ht]
\centering
\includegraphics[width=0.95\textwidth]{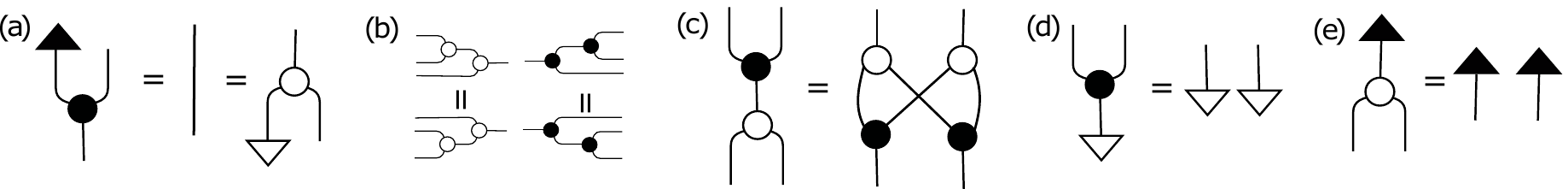}
\caption{Bialgebra axioms~\cite{FA} (scalars are omitted).  (a) unit laws (these are
of course
left and right units); (b) associativity; (c) bialgebra; (d,e) co-{\sf COPY}
points.}\label{fig:BialgebraAxioms}
\end{figure}

\begin{example}[{\sf GHZ}, \AND{} form a bialgebra \cite{CTNS}]
We are in a position to study the interaction of {\sf GHZ}-\AND.   This interaction
satisfies the equations in Figure \ref{fig:BialgebraAxioms}: (a) the bialgebra law;
(b) the co-copy point of \AND{} is $\ket{1}$; and (c) the co-interaction with the
unit for {\sf GHZ} creates a compact structure.  In addition, (a) and (b) show the
copy points for the black {\sf GHZ}-tensor; in (c) we have the unit and fixed point
laws.
\end{example}

Even if a given product and co-product do not satisfy all of the compatibility
conditions (given in (a), (b), (c), (d), (e) in Figure \ref{fig:BialgebraAxioms}),
and hence do not form bialgebras, they can still satisfy the bialgebra law which is
given in Figure \ref{fig:BialgebraAxioms}(c).  Examples of states
that satisfy the bialgebra law in Figure \ref{fig:BialgebraAxioms}(c), but are not
bialgebras are given in Definition~\ref{def:bialgebra}. Notice that bialgebra
provides a highly constraining characterization of the tensors involved and is
tantamount to defining a commutation relation between them.

\begin{definition}[Bialgebra Law~\cite{FA}]\label{def:bialgebra}
A pair of quantum states (black, white tensors) satisfy the bialgebra law if (c)
in Figure \ref{fig:BialgebraAxioms} holds.  The Boolean states, \AND, \OR, \XOR,
\XNOR, \NAND, \NOR{} all satisfy the bialgebra law with {\sf COPY}.  
\end{definition}

\subsubsection{Algebras on valence-3 tensors}\label{sec:hopf}
A particularly important class of bialgebras are known as Hopf-algebras~\cite{FA}.
This is characterized by the way in which algebras and co-algebras can interact.  This
is captured by the Hopf-law, where the linear map $A$ is known as the antipode.

\begin{definition}[Hopf-Law~\cite{FA}]
A pair of quantum states satisfy the Hopf-Law if an $A$ can be found such that the
following equations hold:
\begin{center}
\includegraphics[width=12\xxxscale]{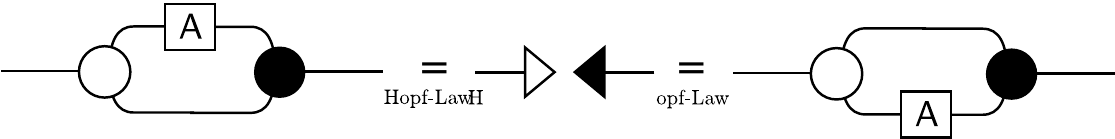}
\end{center}
\end{definition}

\begin{example}[\XOR{} and {\sf COPY} are Hopf-algebras on Boolean States
\cite{boolean03}]
It is well known (see e.g.\ \cite{boolean03}) that the Boolean state \XOR, satisfies
the Hopf-algebra law with trivial antipode ($A=\I$) with {\sf COPY}.  Recall Figure
\ref{fig:F2-presentation}(g). 
\end{example}

\subsection*{Bending wires: compact structures}\label{sec:bends}
As mentioned in the preliminary section (\ref{sec:overview}), we make use of
what's called a \textit{compact structure} in category theory which amounts to
introducing cups and caps, to provide a formal way to bend wires and define
transposition. See Figures
\ref{fig:cupsetc} and \ref{fig:adjoints}.  

A \textit{compact structure} on an object $\2H$ consists of another object $\2H^*$
together with a pair of morphisms (note that we use the equation $\2H^*=\2H$ in
Hilbert space making objects self dual which simplifies what follows).
\begin{equation*}
  \eta_{\2H} : \7C \longrightarrow {\2H} \otimes {\2H}, ~~~~~~~~~~~~~~\epsilon_{\2H}
: {\2H} \otimes {\2H} \longrightarrow \7C,
\end{equation*}
where the standard representation in Hilbert space with dimension $d$ and basis
$\{\ket{i}\}$ is given by
\begin{equation*}
\eta_{\2H} = \sum_{i=0}^{d-1}\ket{i}\otimes\ket{i}, ~~~~~~~~~~~~~~ \epsilon_{\2H} =
\sum_{i=0}^{d-1}\bra{i}\otimes\bra{i},
\end{equation*}
and in string diagrams (read from the top to the bottom of the page) as
\begin{center}
  \includegraphics[width=0.25\textwidth]{cup_and_cap}
\end{center}
These cups and caps give rise to cup/cap-induced
duality: this amounts to being able to create a linear map that ``flips'' a bra to a
ket (and vise versa) and at the same time taking an (anti-linear) complex conjugate.
 In other words, the cap $\sum_{i=0}^1\bra{ii}$ sends  quantum state
$\ket{\psi}=\alpha\ket{0} + \beta\ket{1}$ to $\alpha\bra{0} + \beta\bra{1}$ which is
equal to the complex conjugate of $\ket{\psi}^\dagger =
\bra{\psi}=\overline{\alpha}\bra{0} + \overline{\beta}\bra{1}$.  Diagrammatically,
the dagger is given by mirroring operators across the page, whereas transposition is
given by bending wire(s). Clearly, $\bra{\overline{\psi}}=\alpha\bra{0} +
\beta\bra{1}$.  

In the case of relating the Bell-states and effects to the identity operator, under
cup/cap-induced duality, we flip the second ket on $\eta_{\2H}$ and the first bra on
$\epsilon_{\2H}$.  This relates these maps and the identity $\I_{\2H}$ of the
Hilbert space: that is, we can fix a basis and construct
invertible maps sending $\eta_{\2H}~\cong~\I_{\2H}~\cong~\epsilon_{\2H}$.  More
generally, the maps $\eta_{\2H}$ and $\epsilon_{\2H}$ satisfy the
equations given in Figure \ref{fig:cupsetc} and their duals
under the dagger. 

A second way to introduce cups and caps is to consider a \textit{Frobenius
form}~\cite{FA} on
either of the structures in the linear fragment from Figure \ref{fig:F2-presentation}
({\sf COPY} and \XOR).  This is simply a functional that turns a
product/co-product into a cup/cap.  This allows one to recover the above compact
structures (that is, the cups and caps given above) as
\begin{center}
\includegraphics[width=7\xxxscale]{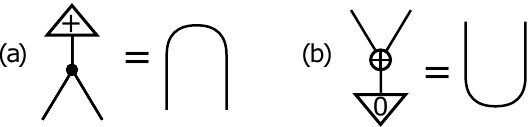}
\end{center}
Again, we will use these cups and caps as a formal way to bend wires in tensor
networks: this can be thought of simply as a reshape of a matrix.

\begin{figure}[ht]
\centering
  \includegraphics[width=0.7\textwidth]{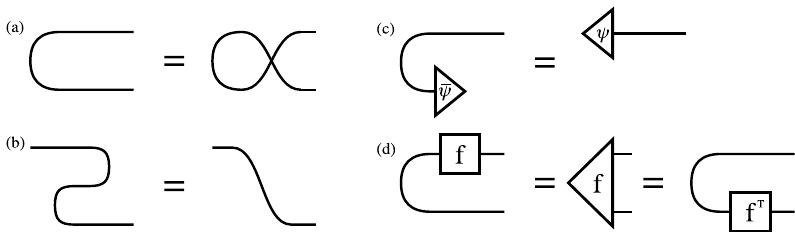}
\caption{Cup identities. (a) Symmetry. (b) Conjugate state.
(c) the \textit{snake equation}.
(d) Sliding an operator around a cup transposes it.}\label{fig:cupsetc}
\end{figure}

\begin{figure}[ht]
\centering
  \includegraphics[width=0.7\textwidth]{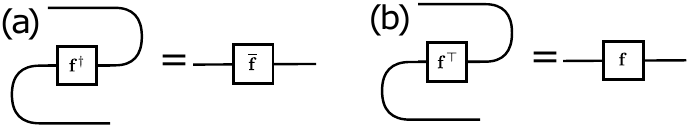}
\caption{Diagrammatic adjoints.  Cups and caps allow
us to take the transpose of a
linear map.  Note that care must be taken, as flipping
a ket $\ket{\psi}$ to a bra $\bra{\psi}$ is conjugate
transpose, and bending a wire is simply transposition, so the conjugate must be
taken: e.g.\ acting on $\ket{\psi}$ with a cap given as $\sum_i \bra{ii}$ results in
$\bra{\overline \psi}$.  }\label{fig:adjoints}
\end{figure}

\section{Examples of Boolean Tensor Network States}\label{sec:CTN}

\subsection*{Constructing Boolean states}\label{sec:W}
Since the fixed building blocks of our tensor networks are the logic tensors \AND,
\OR,
\XOR\, and {\sf COPY}, along
with ancilla bits, we can immediately apply the universality of these elements for classical circuit construction
to guarantee that any Boolean state has a tensor network decomposition.
However our construction goes beyond this because as we have seen,
Boolean tensor networks can be deformed and rewired in ways which are not
ordinarily permitted in the standard acyclic-temporal definition of classical
circuits. The \W-state will be shown to provide a non-trivial example of this.

\begin{example}[Functions on \W- and {\sf GHZ}-states]\label{ex:WandGHZ}
We consider the function $f_\W$ which outputs logical-one given input bit string
$001$, $010$ and $100$ and logical-zero otherwise.  Likewise the function
$f_{{\sf GHZ}}$ is defined to output logical-one on input bit strings $000$ and $111$
and
logical-zero otherwise.  See Examples \ref{ex:WandGHZBoolean} and
\ref{ex:mlWandGHZ} which
consider representation of these functions as polynomials.  We will continue to work
with a linear representation of functions on quantum states; here
bit string
$000\mapsto \ket{000}$ (etc.).
\end{example}

\begin{example}[MPS form for \W-state]
Like the {\sf GHZ} state, the \W-state has a
simple MPS representation
\be\label{eqn:wnMPS}
\ket{\W_n} = \bra{0}\left( \begin{array}{cc}
\ket{0} & 0 \\
\ket{1} & \ket{0} \end{array} \right)^{n}\ket{1}=
\ket{10...0}+\ket{01...0}+...+\ket{00...1}.
\ee
This description \eqref{eqn:wnMPS} is succinct.   All MPS-states have essentially this 
same topological or network structure. In contrast, our
categorical construction
described below breaks this network up further.
\end{example}

\begin{remark}[Exact-value functions]
The function $f_\W$ takes value logical-one on input vectors with $k$ ones for a
fixed integer $k$.
Such functions are known in the literature as Exact-value symmetric Boolean
functions.  When cast into our framework, exact-value functions give
rise to tensor networks which represent what are known as Dicke states
\cite{2010NJPh...12g3025A}. 
\end{remark}

\begin{example}[Function realization of $f_\W$ and
$f_{{\sf GHZ}}$: the Boolean case]\label{ex:WandGHZBoolean}
One can express (using $\overline x$ to mean Boolean variable negation)
\begin{equation}
f_\W(x_1,x_2,x_3)=\overline x_1\overline x_2x_3\oplus
x_1\overline x_2\overline x_3\oplus \overline x_1x_2\overline x_3
\end{equation}
by noting that each term in the disjunctive normal form of $f_\W$ are disjoint, and
hence \OR~maps to \XOR~as $\vee\mapsto \oplus$.  The algebraic normal form (see
Appendix~\ref{sec:Boolean}) becomes
\begin{equation}\label{eqn:fWBoolean}
f_\W(x_1,x_2,x_3)=x_1 \oplus x_2\oplus x_3\oplus x_1 x_2 x_3
\end{equation}
\begin{equation}
f_{{\sf GHZ}}(x_1,x_2,x_3) =1\oplus x_1\oplus x_2\oplus x_3\oplus x_1 x_2\oplus
x_1x_3\oplus x_2x_3
\end{equation}
\end{example}

\begin{example}[Function realization of $f_\W$ and $f_{{\sf GHZ}}$: the set
function case] \label{ex:mlWandGHZ}
Set functions are mappings from the family of subsets of a finite ground set (e.g.\
Booleans) to the real or complex numbers.  In the circuit theory literature,
functions from the Booleans to the reals are known as pseudo-Boolean functions and
more commonly as
multi-linear polynomials or forms (see~\cite{JDB08} where these functions are used to
embed a co-algebraic theory of logic gates in the ground state energy configuration of
spin models). There exists an algebraic normal form and hence a unique multi-linear
polynomial representation for each pseudo-Boolean
function (see Appendix~\ref{sec:Boolean}).  This is found by mapping the negated
Boolean variable as
$\overline{x}\mapsto
(1-x)$.  For the {\sf GHZ}- and \W-functions defined in Example~\ref{ex:WandGHZ} we
arrive
at the unique polynomials \eqref{eqn:mlGHZ} and \eqref{eqn:mlW}.
\begin{equation}\label{eqn:mlGHZ}
f_{{\sf GHZ}}(x_1,x_2,x_3)=1 - x_1 - x_2 + x_1 x_2 - x_3 + x_1 x_3 + x_2 x_3
\end{equation}
\begin{equation}\label{eqn:mlW}
f_{\W}(x_1,x_2,x_3) = x_1 + x_2 + x_3 - 2 x_1 x_2 - 2 x_1 x_3 - 2 x_2 x_3 + 3 x_1
x_2 x_3
\end{equation}
These polynomials \eqref{eqn:mlGHZ} and \eqref{eqn:mlW} are readily translated into
Boolean tensor networks.
\end{example}

\begin{example}[Network realisation of \W- and {\sf GHZ}-states]
A network realization of \W- and {\sf GHZ}-states in our framework then follows by
post-selecting the relevant network to $\ket{1}$ on the output bit --- leaving the
input qubits to
represent a \W- or {\sf GHZ}-state respectively. An example of this is shown in
Figure~\ref{fig:Wfunction}.
\end{example}

\begin{figure}[ht]
\centering
\includegraphics[width=14\xxxscale]{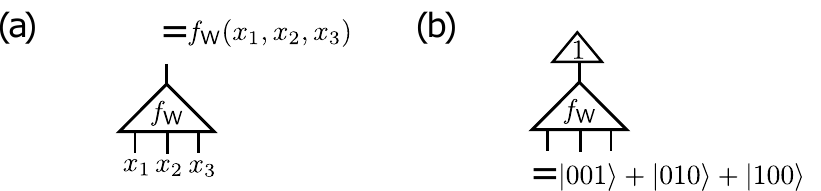}
\caption{Left (a) the circuit realization (internal to the triangle) of the function
$f_\W$ of e.g.\ \eqref{eqn:fWBoolean} which outputs logical-one given input $\ket{x_1
x_2 x_3}=$ $\ket{001}$, $\ket{010}$ and $\ket{100}$ and logical-zero
otherwise. Right (b) reversing time and setting the output to $\ket{1}$ (e.g.\
post-selection) gives a network representing the \W-state.  The na\"{i}ve
realization of $f_\W$ is given in Figure~\ref{fig:CategoricalWStates}
with an optimized co-algebraic construction shown in Figure
\ref{fig:CategoricalWStates}.}\label{fig:Wfunction}
\end{figure}

\begin{figure}[ht]
\centering
\includegraphics[width=7\xxxscale]{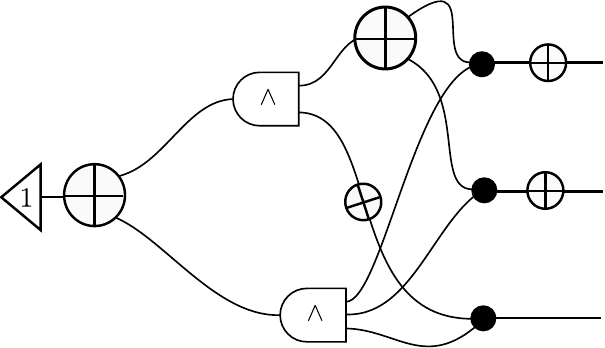}
\caption{Na\"{i}ve BTNS realization of the familiar \W-state
$\ket{001}+\ket{010}+\ket{100}$.  A standard (temporal) acyclic classical circuit
decomposition in terms of the \XOR-algebra realizes the function $f_\W$ of three
bits.  This function is given a representation on tensors.  As illustrated, the
networks input
is post selected to $\ket{1}$ to realize the desired
\W-state.}\label{fig:temporalWStates}
\end{figure}

\begin{figure}[ht]
\includegraphics[width=15\xxxscale]{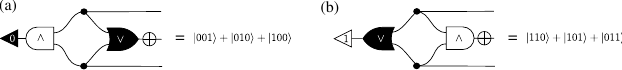}
\caption{\W-class states in the Boolean tensor network state formalism.  (a)
is the standard \W-state.  (b) is found from applying De Morgan's law (see
\S~\ref{sec:AND}) to (a) and rearranging after inserting inverters on the output
legs. Notice the atemporal nature of the circuits, as
one gate is used forwards, and the other backwards.}\label{fig:CategoricalWStates}
\end{figure}

Two different Boolean tensor constructions for the building blocks of the \W-state are shown
in Figure~\ref{fig:temporalWStates} and  Figure~\ref{fig:CategoricalWStates}. Notice that in Figure~\ref{fig:CategoricalWStates} the
resulting tensor network forms an atemporal classical circuit and is much more efficient than
the na\"{i}ve construction in Figure~\ref{fig:temporalWStates}. Moreover by
appropriately daisy-chaining the networks in
Figure~\ref{fig:CategoricalWStates} we construct a Boolean tensor network for
an $n$-party \W-state
as shown in Figure~\ref{fig:nPartyCategoricalWStates}. Contrast this with other factorizations appearing in the literature~\cite{Kissinger09}. The resulting form of this tensor network is entirely equivalent
(up to regauging) to the MPS description given earlier, but now reveals internal structure of the state in terms of BTNS building blocks.

\begin{figure}[ht]
\begin{center}
  \includegraphics[width=14\xxxscale]{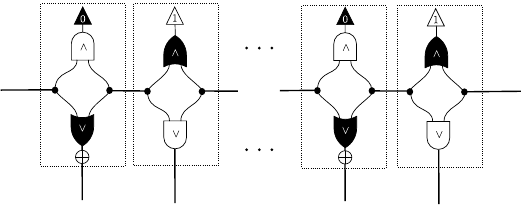}  
\end{center}
\caption{\W-state ($n$-party) in the Boolean tensor network state formalism. The comb-shaped feature of efficient network contraction remains, with the internal structure of the network components exposed in terms of well understood algebraic structures.}\label{fig:nPartyCategoricalWStates}
\end{figure}

\section{Discussion} 

We have introduced a class of quantum states, known as Boolean quantum states.  This class is of interest, since it allows one to study quantum states using the well understood Boolean algebra.  In addition, states in this class have an evident tensor network.  

\begin{theorem}[The Class of Boolean Quantum States \cite{CTNS}]
 Every switching function $f(\x)$ gives rise to a quantum state with binary coefficients in $\{0,1\}$.  Moreover, a tensor network 
 representing this state is determined from the classical network description of $f(\x)$.  
\end{theorem}

\begin{remark}[From composition to contraction]
The quantum tensor network is found by letting each classical gate act on a linear space and from changing the composition of functions, to the contraction of tensors.
\end{remark}

\begin{remark}[Shannon and Davio decomposition]
 There are several elegant methods that allow one to factor a classical function into networks of different structures.  We should only have time to briefly mention these.  Those interested can consult wikipedia for Shannon and Davio decompositions.  
\end{remark}

We have examined in some detail a Boolean tensor tool box.  This is perhaps the most accessible part of tensor network states, due to its strong relation to Boolean circuits and widespread awareness of techniques to manipulate standard Boolean networks.  This tool box forms {\it the glue} or building blocks behind the more complicated applications.  

\section{Problems}

\begin{myexercise}[LOCC equivalence]
 Show that the \AND-tensor is locally bit-flip equivalent to \NAND-, \NOR-, and \OR-tensors.  
\end{myexercise}

\begin{example}[Two-site reduced density operator of $n$-party {\sf GHZ}-states]\label{fig:Spider-comb}
{\sf GHZ}-states on $n$-parties have a well known matrix product expression given as
\be
\begin{aligned}
\ket{{\sf GHZ}_n} &= \text{Tr}\left[ \left( \begin{array}{ccccc}
\ket{0} & 0 & \cdots & 0 \\
0 & \ket{1} & \cdots & 0 \\
\vdots & \cdots        & \ddots & \vdots \\
0 &  \cdots & 0 & \ket{d-1} 
\end{array} \right)^{n}\right]\\
&=
\ket{0}\ket{0}\ket{0}+\ket{1}\ket{1}\ket{1}+\cdots+\ket{d-1}\ket{d-1}\ket{d-1}.
\end{aligned}
\ee
Such MPS networks are
known to be efficiently contactable.  We note that the networks in
Figure~\ref{fig:Spider} do not appear \textit{a priori} to be contractible due to the
number of open legs. What makes them contractible (in their present from) is that
the tensors obey the fusion 
law allowing them to be deformed into a contractible MPS network.  The reduced density matrix of an n-party
{\sf GHZ}-state then becomes (a) in
Figure~\ref{fig:ReducedGHZ} and the expectation value of an observable is shown in
(b) where we included the normalisation constant.
\begin{center}
    \includegraphics[width=0.75\textwidth]{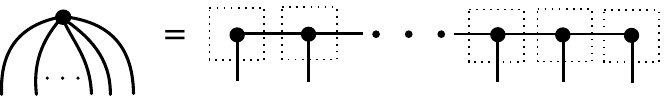}
\end{center}
The {\sf GHZ}-state tensor is simply a order-$n$ {\sf COPY}-tensor.  Node
equivalence
implies that this tensor can be deformed into any network
geometry including a MPS comb-like structure (right).
\begin{center}
    \includegraphics[width=\textwidth]{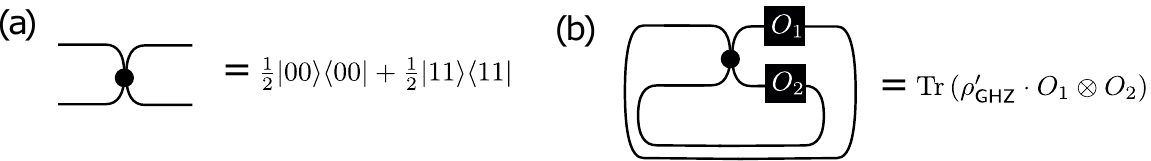}
    \label{fig:ReducedGHZ}
\end{center}
Reduced density operator.  Left (a) reduced density operator $\rho_{\sf GHZ}'$ found from applying the fusion law to a n-qubit {\sf GHZ}-state.  Right (b) the expectation value of observable $O_1\otimes O_2$ found from connecting the
observable and connecting the open legs (i.e.\ taking the trace).
\end{example}

\begin{myexercise}
Write down the matrix product state for the n-party GHZ  and n-party W-states. Let us define the two-point correlation as \begin{equation}
    C_{i,j}(\psi, A_i, A_j) =  \bra{\psi}A_i A_j \ket{\psi} - \bra{\psi}A_i  \ket{\psi}\bra{\psi}A_j \ket{\psi}
\end{equation}
Find $C(j\leq n):=C_{1,j}(\psi, X_1, X_j)$ for both GHZ- and W- where $X$ is the familiar Pauli matrix. Using Mathematica, make a publication quality plot for some fixed $n$ (e.g.~label everything and create a caption explaining the plot).  
\end{myexercise}

\begin{myexercise}[Basic properties of Boolean quantum states]
This exercise considers some elementary properties of Boolean quantum states.  
\begin{itemize}
 \item[(i)] Count the number of Boolean quantum states on n-qubits. Write truth tables for all two-input boolean functions and label appropriately those columns corresponding to \NOT, \AND, \XOR, \XNOR, \NAND, \NOR, and \OR.   
 \item[(ii)] Let $ \psi^2_{\7B} = c_0 \ket{00} + c_1\ket{01} + c_2 \ket{10} + c_3 \ket{11}$ represent a boolean quantum state, and hence $\forall i, c_i=0,1$.  Also let the output of all possible two-bit functions be 
 given by a truth vector $\textbf{c} = (c_0, c_1, c_2, c_3)$ and classify entangled vs non-entangled boolean states on two qubits.  
 \item[(iii)] Hence, using the result in (ii) or otherwise, show that the possible values of entanglement (using the quantity $K_1$ or $J_2$) is course grained and give the possible values.
\end{itemize}
\end{myexercise}

\begin{example}[valence-(0,3) tensor
factorization]\label{theorem:factorisation}
Every tensor in $\2A\otimes\2A\otimes\2A$ has a factorization as: (i) two valence-(1,0) tensors in $\2A$ (e.g.~the triangles); (ii) two valence-(2,1) \COPY-tensors; (iii) two unitary
morphisms of type $\2A\rightarrow\2A$ (white boxes).  This factorization is
given
diagrammatically as follows. 
\begin{center}
\includegraphics[width=8\xxxscale]{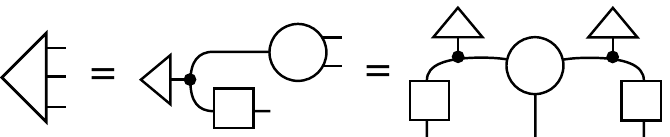}
\end{center} 
\end{example}

\begin{myexercise}[MPS factorization of the \AND-state]
 Recall the MPS factorization of quantum states covered in lecture I.
Present an MPS factorization of the \AND-state
\be 
\psi_\wedge = \ket{000}+\ket{010}+\ket{100}+\ket{111}
\ee 
into five elementary tensors, as illustrated in Example
\ref{theorem:factorisation}.   Write this as a matrix product.  (Note: it
is acceptable to use a computer to calculate the SVD. Another method follows
from considering the eigenvalues of reduced density states.)
\end{myexercise}

\begin{myexercise}[The class of linear quantum states]
 We define the linear class of quantum states as quantum states of the form 
 \be 
 \psi_{\oplus L} = \sum c_1x_1\oplus c_2x_2\oplus ...\oplus c_nx_n \ket{x_1, x_2, ..., x_n}.
 \ee
 Here we will consider a 1D system.  What is the maximum possible $\chi$ (as defined in lecture I) in a partition in one of these states?  
\end{myexercise}

\begin{myexercise}[Correlations in polarity states (Optional)]
Here let us consider a state defined as 
\be 
\psi:=\sum (-1)^{f(\textbf x)}\ket{\textbf x},
\ee 
where the sum is over all $\textbf x$.  These states are often considered in quantum algorithm theory.  
Prove that $\psi$ is separable iff the function implements a linear function.  
\end{myexercise}

\begin{definition}[Boolean Density Operators]
 In general, a Boolean density matrix takes the following form.  
 \begin{equation} 
\begin{split}
\rho_{\7 B} & = \sum f(\x)f(\y)\ket{\x}\bra{\y} =\\
 & = \sum f(x_1, ..., x_k, ..., x_n)f(y_1, ..., y_k, ...., y_n) \ket{x_1, ..., x_k, ..., x_n}\bra{y_1, ..., y_k, ..., y_n}
\end{split}
\end{equation}
\end{definition}
\begin{myexercise}[The Boolean trace theorem]
Show that performing the partial trace over the kth subsystem of a Boolean density state results in 
 \begin{equation} 
\begin{split}
\text{Tr}_k \rho_{\7 B}  & = \delta_{x_k, y_k} \rho_{\7 B}= \\
 & =  \sum f(x_1, ...,x_k, ..., x_n)f(y_1, ..., x_k, ..., y_n) \ket{x_1, ..., x_n}\bra{y_1, ..., y_n}
\end{split}
\end{equation}
\end{myexercise}

\begin{myexercise}[Two-qubit entanglement]
To prepare for the next problem, here we will continue our study of two-qubit entanglement, by considering again the results in Lecture I.  
 \begin{itemize}
 \item[(i)] By considering a general quantum state
 \be 
 \alpha \ket{00} + \beta\ket{01}+\gamma\ket{10}+\delta\ket{11},
 \ee  
and by using wire bending duality, construct the induced matrix $M$. 
  \item[(ii)]  Show that $\text{Det}(M)$ vanishes identically for $\alpha = ac$, $\beta = ad$, $\gamma = bc$ and $\delta = bd$ with $a,b,c,d\in \7C$ and factor states of this type into a product of local states $\phi_1(a,b)\phi_2(c,d)$.  
  \item[(ii)] Using this result from (ii) or otherwise, show that $\text{Det}(M)$ vanishes iff the the state giving rise to $M$ under wire duality takes the form $\phi_1(a,b)\phi_2(c,d)$.
 \end{itemize}
\end{myexercise}

\begin{myexercise}[Quantum \AND-tensors] 
With the contractions in Figure \ref{fig:andtensor} (a) and (b) in mind, consider instead the contraction with the state 
\be 
\alpha\ket{0} + \beta\ket{1},
\ee 
and hence, form an order-two tensor $T$ with components take values in $\alpha$ and $\beta$.  
\begin{itemize}
 \item[(i)] Write down the quantum state resulting from this contraction.  By bending a wire, write down the resulting two by two matrix, and label this $M(\alpha, \beta)$.  
 \item[(ii)] By considering the matrix determinant or otherwise, find the values 
 that $\alpha$ and $\beta$ must take for the state to be separable. 
 \item[(iii)] Consider now products of $M(\alpha, \beta)$.  Give a one sentence proof or disproof of the following statements:  The products of $M$ for (a) a magma (or groupoid); (b) a semigroup; (c) a monoid.  
 \item[(iv)] Are there values of $\alpha$ and $\beta$ that make a state local unitary equivalent to a bell-state? Are there values of $\alpha$ and $\beta$ that can break leg exchange symmetry in the state? 
\end{itemize}
\end{myexercise}

\begin{myexercise}
~
\begin{itemize}
 \item[(i)] Write down the functions for a two-party bell state, a GHZ and a W state.  Use the trace theorem to provide analytical closed formula for the reduced density states found 
 from tracing over the last bit.  Expand this as a matrix and compare to the result obtained using standard methods.  
 \item[(ii)] Using standard properties of Boolean algebra, show that $\rho_{\7 B}^2 = \rho_{\7 B}$ and hence that $\rho_{\7 B}$ is a projector.
 \item[(iii)] Show that $U=\I -2\rho_{\7 B}$ is self-adjoint and unitary.  Give values for the matrix trace and determinant of $U$. 
 \item[(iv)] Show that $(1-\rho_{\7 B})$ and $(\rho_{\7 B})$ project onto eigenspaces of $U$ and relate the dimension of the subspaces to properties of the boolean function.   
\end{itemize}
\end{myexercise}

\begin{myexercise}
Define the Werner states acting on $\mathbb{C}_2\otimes \mathbb{C}_2$ as 
\begin{equation}
    \rho_r = r\ketbra{\phi^+}{\phi^+}+ \frac{1-r}{4}\I
\end{equation}
where $\sqrt{2}\ket{\phi^+}=\ket{00}+\ket{11}$ is the standard Bell state and $r$ takes only values in the real interval $[0,1]$. Find (i) the matrix trace of $\rho_r$ and (ii) the eigenvalues of $\rho_r$.  The concurrence is 
\begin{equation}
    C(\rho) = \text{max}\{\lambda_1 - \lambda_2 - \lambda_3 - \lambda_4, 0\}
\end{equation}
where $\lambda_1 \geq \lambda_2 \geq \lambda_3 \geq \lambda_4$
Using part (ii) or otherwise, find $C(\rho_r)$. 
\end{myexercise}

\begin{myexercise}[The Shannon effect --- Research]
For many fundamental functions one can design efficient circuits.
Is it possible to compute each function by an efficient circuit? In 
1949 Shannon proved that almost all functions are hard functions, optimal
circuits for almost all functions have exponential size and linear depth.
This was proved quite easily using a counting argument \cite{Weg87}. The number of circuits with small circuit size or
small depth grows much slower than the number of different Boolean
functions implying that almost all functions are hard. This means that
a random Boolean function is hard with very large probability.  
\begin{itemize}
 \item[(i)] The circuits we consider here don't have the same temporal structure enforced by classical gates.  In that regard, we have seen that the function for the \W-state can be realized using two \COPY-gates, one \AND, and one \OR.  
 \item[(ii)] With this observation in mind, what can we say about the Shannon effect in the present case?  
\end{itemize}
\end{myexercise}
 \part{Symmetries and Stabilizer Tensor Theory}

\section{Introduction}
In quantum theory, one predicts experiments by performing calculations involving the processes that transform the mathematical representative of a quantum system, such as operators on a state space or an observable.  This mathematical structure has proven very rich, and consequently well studied.  It is often the case that mathematical structures, such as symmetries or other algebraic properties, allow us to understand much more about a situation at hand.  

The theory of Penrose tensor networks leverages one to study the mathematical structure formed by the composition of processes themselves. Tensor networks are not just a new conceptual tool, but in fact are becoming applicable across a range of disciplines.  Primarily they have practical applications in the context of condensed matter and many-body physics.  There are still fundamental questions related to these tools, left almost entirely unexplored.  

In this lecture, we will examine some basic aspects of symmetry in tensor network states.\mn{When describing known results in quantum information, tensor network practitioners should hope to find simple and elegant rewrites between structurally intuitive and aesthetically pleasing representations.}

\section{Symmetries in Tensor Network States}

Let $\psi^{ijk}$ be a tensor with valence-(3,0).  In the graphical tensor notation of Penrose \cite{Penrose} (who also often used triangles to depict quantum states) we represent this as  
\begin{center}
 \includegraphics[width=0.07\textwidth]{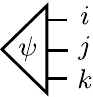}
\end{center}
When we refer to tensor symmetry, we typically mean under index exchange.  Formally, 

\begin{definition}[Exchange or index symmetric]
\label{def:ex-symmetry}
 A valence-($n$,$m$) tensor is {\it wire exchange symmetric} if one can exchange any of the arms amongst themselves, or any of the legs amongst themselves, with the effect leaving the tensor unchanged.  
\end{definition}

\begin{definition}[Full symmetry]
\label{def:full-symmetry}
A valence-($n$,$m$) tensor is {\it fully symmetric} if in addition to being exchange symmetric per Definition \ref{def:ex-symmetry}, the tensor is left invariant under arm and leg exchange. 
\end{definition}

\begin{remark}[The trivial representation]
 An exchange symmetric tensor of valence-($n$,$m$) carries the trivial representation of the symmetric group of order $n!$ (arm exchange symmetric) and order $m!$ leg exchange symmetric.  
\end{remark}

The diagrammatic generator of the symmetry group is expressed in (a).  
\begin{center}
 \includegraphics[width=0.5\textwidth]{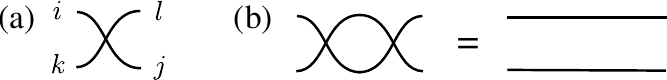}
\end{center}
Here (a) represents the tensor $\delta^i_j\delta^k_l$ and (b) represents the contraction 
\begin{equation}
  (\delta^i_j\delta^k_l)(\delta^i_j\delta^k_l)=\delta^i_i\delta^k_k.   
\end{equation}
In diagrammatic form, this 
exchange operator is used to generate the group of operators $S_1$ in (a), $S_2$ in (b) and $S_3$ of order $3!$ in (c).  
\begin{center}
 \includegraphics[width=.8\textwidth]{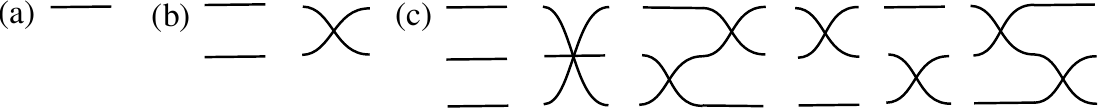}
\end{center}

Returning now to $\psi^{ijk}$.  The state $\psi^{ijk}$ is index symmetric provided it is left invariant when acted on with any operator from (c) above.  

\begin{remark}[Symmetry breaking]
 Contraction of symmetric tensors breaks symmetry.  Given two tensors $\Gamma_1$ and $\Gamma_2$ satisfying Definition \ref{def:full-symmetry} the contraction of $\Gamma_1$ and $\Gamma_2$
 over one or more indices is not necessarily symmetric.  Examples abound.  The sum of tensors (defined for tensors of the same valence) is symmetric.
\end{remark}

\paragraph{Symmetrizers.}
Let us consider methods to perform the symmetrization of tensors.  In the definition we avoid clutter by letting the reader 
figure out how the group should act on the arms and legs of the tensor.  

\begin{definition}[Symmetrizer]
 Let us define the operator, acting on appropriate types 
 \be 
 R_G:\Gamma \mapsto \frac{1}{|G|}\sum_{g\in G} g\{\Gamma\}
 \ee 
 then $R_G\{\Gamma\}$ is necessarily symmetric under the group $G$.  
\end{definition}

\begin{example}[Example of symmetrization]
Consider $\phi^{ijk}$ which is not necessarily symmetric.  We will symmetrize the first two indices labeled 
$i$ and $j$ under $S_2$.  In this case we find 
\be 
R_G\{\phi^{ijk}\}= \frac{1}{2}\phi^{ijk} + \frac{1}{2}\phi^{jik}. 
\ee
The shorthand notation for this procedure and expansion is to write round brackets over the indices we are symmetrizing as $\phi^{(ij)k}$.  
\end{example}

\begin{remark}[Symmetrization can vanish]
 It is possible to have a tensor vanish under $R_G$.  
 Examples include the antisymmetric state $\psi_\epsilon = \ket{01}-\ket{10}$.  
 Here 
 \be
 R_G\{\psi_\epsilon^{~ij}\} = \frac{1}{2} \psi_\epsilon^{~ij} +\frac{1}{2} \psi_\epsilon^{~ji} = \frac{1}{2} \psi_\epsilon^{~ij} - \frac{1}{2} \psi_\epsilon^{~ij}. 
 \ee 
\end{remark}

\begin{example}[Symmetric over other groups]
 We have defined $R_G$ in a way where we can pick groups $G$ other than the symmetry group acting on the tensor indices.  Consider say 
 $G=\{H\otimes H\otimes H, \I\}$ where $H^2=\I$ is the standard Hadamard transformation.  We consider the valence-(1,2) \COPY-tensor $\delta^i_{~jk}$ and let $R_G$ act as 
 \be 
   R_G\{\delta^i_{~ij}\} = \frac{1}{2}\delta^i_{~jk} + \frac{1}{2}\oplus^i_{~jk} 
 \ee 
here $\oplus^i_{~jk}$ and from now on, $\oplus^i_{~jk}$ denotes the \XOR-tensor.  The resulting tensor still has three indices.  This operation is given diagrammatically as 
\begin{center}
 \includegraphics[width=.75\textwidth]{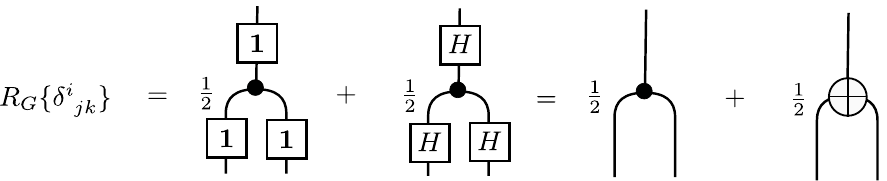}
\end{center}
\end{example}

\paragraph{Antisymmetrizers.}  One can also consider antisymmetrizers. These are used regularly in the study of fermionic particles for instance.  

\begin{definition}[Antisymmetrizer]
 We define the operator of antisymmetrization as 
 \be 
 \2 A = \frac{1}{N!}\sum_{g\in S_n}(-1)^\pi g 
 \ee 
 where $\pi = 0$ for $g$ even and $\pi = 1$ for $g$ odd.  
\end{definition}

\begin{example}[Antisymmetrization]
 Consider the tensor $\Gamma^{ij}$ of valence-(2,0).  Then 
 \be 
 \2A\{\Gamma^{ij}\} =\frac{1}{2}\Gamma^{ij} - \frac{1}{2}\Gamma^{ji}
 \ee 
 and provided the tensor $\2A\{\Gamma^{ij}\}$ is non-zero, it is said to carry the sign representation of $S_2$.  It is standard to write 
 antisymmetrization over indicies, by including them in square brackets.  The example here would then be $\Gamma^{[ij]} = \2A\{\Gamma^{ij}\}$. 
\end{example}

\begin{remark}[Antisymmetrization has a non-trivial kernal]
 Examples abound of tensors which map to zero under $\2A$.  For example, there is no fully antisymmetric tensor on $(\7C^{2})^{\otimes 3}$.  
\end{remark}

\begin{example}[Considering other groups]
 We will let 
 \be 
 T^i_{~jk} = \delta^i_{~jk} - \oplus^i_{~jk}
 \ee 
 then if we consider action under
 \begin{equation}
     G=\{H\otimes H\otimes H, \I\}, 
 \end{equation}
we have that the action
 \begin{equation}
     H\otimes H\otimes H\{T^i_{~jk}\}=-T^i_{~jk}
 \end{equation}
and 
\begin{equation}
    \I\{T^i_{~jk}\} = T^i_{~jk}. 
\end{equation}
\end{example}

Let us now consider some of the symmetries present in the tensors we have defined in the first two lectures.  We will start with the \COPY-tensor.  
\begin{center}
 \includegraphics[width=.4\textwidth]{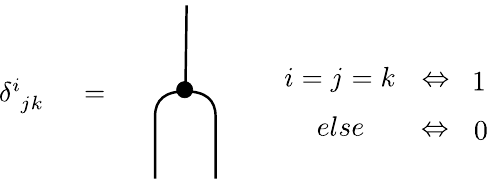}
\end{center}
This is written using standard quantum theory notation as 
\be 
\sum_{ijk}\delta^i_{~jk}\ket{jk}\bra{i}. 
\ee 
As we recall, copy is fully symmetric. In addition, 
\begin{remark}[The types of possible states built from \COPY]
A concise definition of the \COPY-tensor is given as any combination of raised or lowered indices on 
 $\delta^{ijk}$, a Kronecker delta function on three indices.   
 In step, one might write the n-party \GHZ-state as 
 \be 
 \psi_\GHZ = \sum \delta^{ijk...l}\ket{ijk...l}. 
 \ee 
Tensor products of states of this form are precisely the only types of states constructible with the \COPY-tensor alone.  
\end{remark}
\begin{figure}[h]
\includegraphics[width=15\xxxscale]{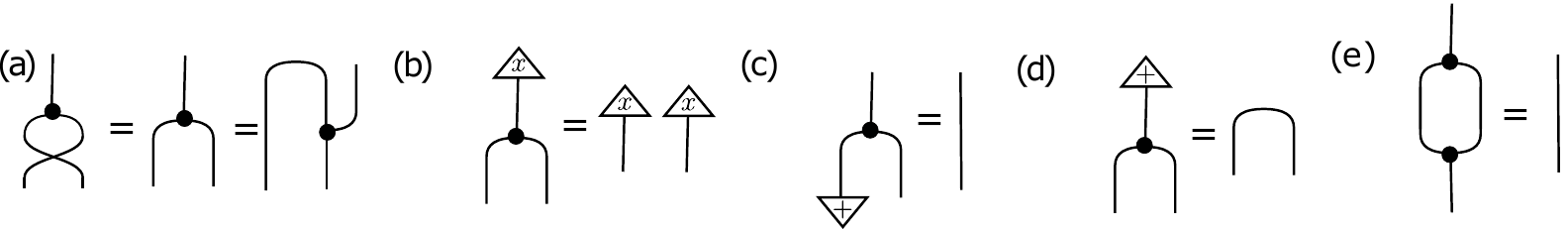}
\caption{Diagrammatic properties of the \COPY-tensor \cite{boolean03}. (a) Full-symmetry.
(b) Copy points, e.g.\ $\ket{x}\mapsto\ket{xx}$ for $x=0,1$ for qubits.  (c) The unit 
given as $\bra{+}\bydef\bra{0}+\bra{1}$.  (d) Applying the unit to the top leg creates a cap.  (e) Copy then delete is the same as identity.}\label{fig:copygate1}
\end{figure}

We are going to study these diagrammatic equations from Figure \ref{fig:copygate1} under local unitary change of basis.  
We will first define the rotation operator 
\begin{definition}[Tensor rotation operator]
 Let us define a rotation operator which acts on a valence-(n,m) tensor as follows.  On the n input arms, $U$ is applied, 
 on the m output legs, $U^\dagger$ is applied.  We will define this map as $U_G$ which is given diagrammatically as follows 
\begin{center}
 \includegraphics[width=.4\textwidth]{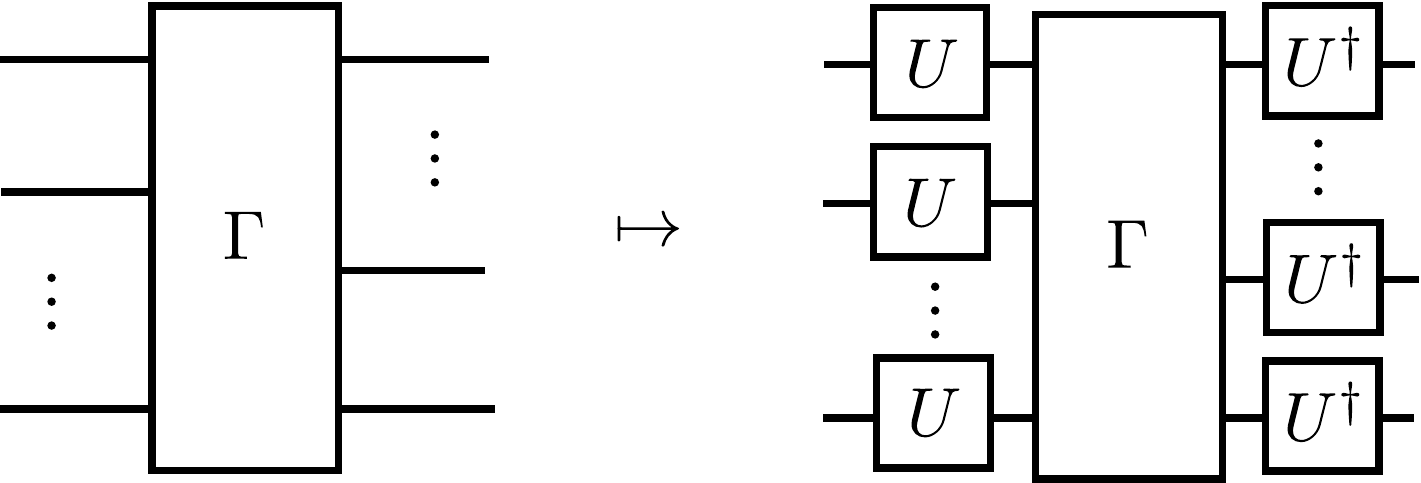}
\end{center}
\end{definition}

\paragraph{Rotations of \COPY-tensors.}\label{sec:generalcopy}
It is of course desirable to extend the definitions of the valence-(1,2) \COPY-tensors to act on arbitrary bases.  We would then change to a basis 
$B$ as 
 \begin{center}
\includegraphics[width=0.60\textwidth]{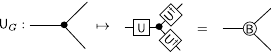}
\end{center}

\begin{example}[Engineering a map to copy a particular basis]
Say you have an orthonormal basis $\{\ket{\psi},\ket{\psi^\bot}\}$, which we will call $\sf B$ and you wish to define a \COPY-tensor on this basis as $\ket{\psi}\mapsto \ket{\psi,\psi}$ and $\ket{\psi^\bot}\mapsto \ket{\psi^\bot,\psi^\bot}$.  To represent such a tensor diagrammatically we draw 
 \begin{center}
\includegraphics[width=0.40\textwidth]{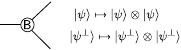}
\end{center}
To transform the \COPY-tensor to copy this new basis, define the unitary 
\be
U = \sum_i\ket{i}\bra{\phi_i}
\ee
where $\{\phi_i\}_i$ and $\{i\}_i$ are bases for the same space such that $\braket{\phi_i}{\phi_j}=\delta_{ij}$ and $\braket{i}{j}=\delta_{ij}$.  The map $U$ is unitary as 
\be
U^\dagger U = (\sum_i\ket{\phi_i}\bra{i})(\sum_j\ket{j}\bra{\phi_j})=\sum_j\ket{\phi_i}\delta_{ij}\bra{\phi_j}=\I. 
\ee
\end{example}

\begin{remark}
There are important cases were we only consider transforming \COPY-tensors using self adjoint unitary maps.  Hadamard is an example with
\begin{equation}
    H=\frac{1}{\sqrt{2}}(X+Z)=H^\dagger. 
\end{equation}
Every projector gives rise to such a map as $U= \I -2P$ is unitary when $P^2=P$.  The map is invertible, and hence every self adjoint unitary gives rise to a projector, illustrating the 
bijection between self adjoint unitary maps and projects.  
\end{remark}

\begin{example}[Transforming \COPY-tensors into \XOR-tensors]
 It is well known in the theory of algebra that Hadamard transforms (the Fourier transform over $\7Z_2$) relate \COPY-tensors and \XOR-tensors.  This is also known in quantum circuits.  For the higher dimensional case see \cite{BB11}.  
 The transformation of a \COPY-tensor into an \XOR-tensor is given diagrammatically as follows.  
  \begin{center}
\includegraphics[width=0.20\textwidth]{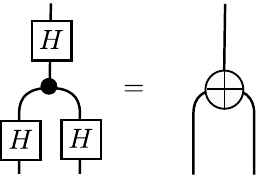}
\end{center}
\end{example}

\begin{remark}[Arbitrary rotations]
 We will now mention exactly how the general form of the diagrammatic laws can transform properly under rotation.  
The general form is straightforward.  States rotate under $U^\dagger$ and tensors rotate as we have already defined.  
\begin{center}
\includegraphics[width=0.60\textwidth]{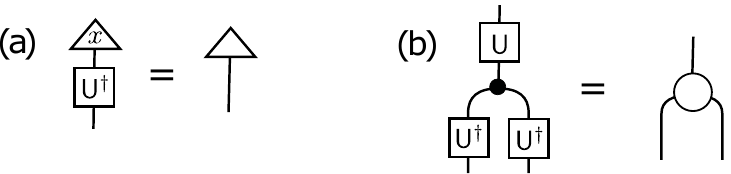}
\end{center}
\end{remark}

\begin{myexercise}[General rotation] 
 Prove that if the \COPY-tensor is rotated as by $U$ above, then the diagrammatic identities hold.  
\end{myexercise}

\section{The Interaction of Networks Comprised of \texorpdfstring{$\delta^{ijk}$}{} and \texorpdfstring{$\oplus^{ijk}$}{}}

We will now consider networks comprised of \COPY-tensors and \XOR-tensors.  Let us first recalls DeMorgan's law.  

\begin{definition}[DeMorgan's Law]
DeMorgan's law is a relationship between logical \AND, ($\wedge$) and logical \OR, ($\vee$).  The relationship 
is found by considering what in quantum circuit theory would be called bit flips ($X$), 
induced by logical negation ($\neg$).  The equation reads 
\be 
\neg (a \wedge b) = (\neg a)\vee (\neg b) 
\ee 
In diagrammatic form, this equation becomes 
\begin{center}
\includegraphics[width=0.45\textwidth]{demorgan}
\end{center}
we note that $X^2 = \I$. 
\end{definition}

\begin{remark}[Relating \COPY{} and \XOR-tensors]
A similar structure to the DeMorgan's law relating \AND{} and \OR-tensors also holds 
between \COPY- and \XOR-tensors.  In this case, instead of using negation, we use the Hadamard gate $H$.  
In diagrammatic form, this becomes 
\begin{center}
\includegraphics[width=0.40\textwidth]{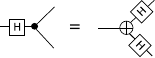}
\end{center}
and we note that $H^2=\I$.  
\end{remark}

To continue, let us recall the definitions of the tensor at hand.  In constraint equation form, together with graphical representations, these read as 
\begin{center}
\includegraphics[width=0.910\textwidth]{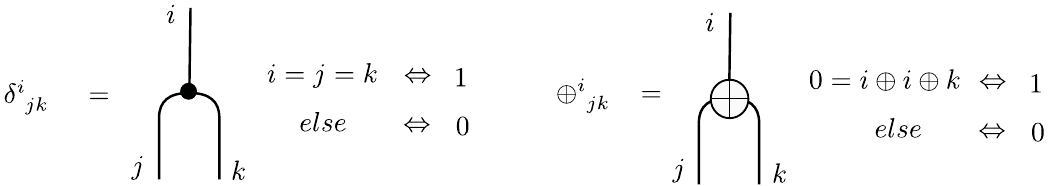}
\end{center}
from the constraint equations, we can readily construct, a sort of ``truth table'', which gives the value of the tensor contraction, provided we contract the wires (labeled $i$, $j$ and $k$) with states $\ket{0}$ or $\ket{1}$.  We illustrated these possibilities in the contraction table below.  

\begin{center}
\begin{tabular}{ccc||c||c}
i         & j         & k         & $\oplus_{ijk}$ & $\delta_{ijk}$\\\hline
$\ket{0}$ & $\ket{0}$ & $\ket{0}$ & 1 & 1\\
$\ket{0}$ & $\ket{0}$ & $\ket{1}$ & 0 & 0\\
$\ket{0}$ & $\ket{1}$ & $\ket{0}$ & 0 & 0\\
$\ket{0}$ & $\ket{1}$ & $\ket{1}$ & 1 & 0\\
$\ket{1}$ & $\ket{0}$ & $\ket{0}$ & 0 & 0\\
$\ket{1}$ & $\ket{0}$ & $\ket{1}$ & 1 & 0\\
$\ket{1}$ & $\ket{1}$ & $\ket{0}$ & 1 & 0\\
$\ket{1}$ & $\ket{1}$ & $\ket{1}$ & 0 & 1
\end{tabular}
\end{center}

\begin{myexercise}[Linearity of tensor contraction]
 Tensor contraction is linear. Given the table above, use linearity to determine the contraction found from other states, such as say $\ket{+}$, $\ket{-}$, etc. 
\end{myexercise}

To consider further properties relating the interaction of \COPY-tensors and \XOR-tensors we will recall Lafont's presentation of the Boolean calculus \cite{boolean03, CTNS}, given in Figure \ref{fig:F2-presentation1}.

\begin{myfullpage}
\begin{illexample}
\label{fig:F2-presentation1}
\begin{center}
\includegraphics[width=20\xxxscale]{F2-presentation1}
\end{center}
Lafont's presentation of the linear fragment of the Boolean calculus \cite{boolean03}.  
(a) associativity; (b) unit laws; (c) input symmetry (as mentioned, these tensors are in fact fully symmetric 
by Definition \ref{def:full-symmetry}); (d) bialgebra law; (e) illustrates that the unit for \XOR{} is a copy point for the \COPY-tensor and vice versa; (f) the inner product $\braket{0}{+}\in\7C$.  Global scale factors are represented as blank space on the page; (g) is the Hopf law.
\end{illexample}
\end{myfullpage}

These rules, and also the algebraic properties of \XOR-algebra result in the following class of functions (the only type possible to construct using the operations at hand.  

\begin{definition}[Linear and affine boolean functions]
Linear Boolean functions take the general form
\be
f(x_1,x_2,...,x_n)=c_1x_1\oplus c_2x_2\oplus ...\oplus c_nx_n
\ee 
where the vector $(c_1,c_2,...,c_n)$ uniquely determines the function.  The affine
Boolean functions take the same general form as linear functions.  However, for functions
in the affine class, variables can appear in both complemented
and uncomplemented form. Affine Boolean functions take the general form
\begin{equation}\label{eqn:affine}
f(x_1,x_2,...,x_n)=c_0\oplus c_1x_1\oplus c_2x_2\oplus ...\oplus c_nx_n
\end{equation}
where $c_0=1$ gives functions outside the linear class.  From the identities,
$1\oplus 1=0$ and $0\oplus x=x$ we require the introduction of only one constant
($c_0$). 
\end{definition}

When contracted, \XOR- and {\sf COPY}-tensors compose to create tensor networks representing the class of linear
quantum networks. This gives rise to the following class of quantum states.  

\begin{definition}[The class of linear quantum states]
 We define the linear class of quantum states as quantum states of the form 
 \be 
 \psi_{\oplus L} = \sum c_1x_1\oplus c_2x_2\oplus ...\oplus c_nx_n \ket{x_1, x_2, ..., x_n}
 \ee
\end{definition}

\begin{myexercise}
  Here we will consider a 1D system.  What is the maximum possible $\chi$ (as defined in lecture I) in a bi-partition of a linear quantum state?   
\end{myexercise}

\begin{myexercise}[Correlations in polarity states (Optional)]
Here let us consider a state defined as 
\be 
\psi:=\sum (-1)^{f(\textbf x)}\ket{\textbf x}
\ee 
where the sum is over all $\textbf x$.  These states are often considered in quantum algorithm theory.  
Prove that $\psi$ is separable iff the function implements a linear function.  
\end{myexercise}

\begin{example}[Examples of linear quantum states]

\be 
\Phi^+ =\sum a\oplus b \ket{a,b} = \ket{01}+\ket{10} 
\ee 

\be 
\Psi^+ = \sum a\oplus \neg b \ket{a,b} = \ket{00} + \ket{11} 
\ee 

\be 
H\otimes H\otimes H \ket{\GHZ} = \sum a\oplus b \oplus c \ket{a,b,c}= \ket{010}+\ket{100}+\ket{001}+\ket{111} 
\ee 
\end{example}

\paragraph{The connection to quantum circuits.}
Lafont's 2003 paper \cite{boolean03} considered an algebraic theory of classical switching networks, but also considered certain aspects of quantum networks.  We also mention the results in \cite{BB11}. These ideas are readily applied to quantum circuits, and combined with the known gate identities appearing in text books on quantum information science. The development was in part, influenced by quantum circuits \cite{NC}, developments in the use of tensor networks in condensed matter, development of the ZX-calculus \cite{redgreen} as well as Lafont's influential work \cite{boolean03}. 

The structures in Figure \ref{fig:F2-presentation1} are also related to other approaches \cite{Kissinger09, DP10, Euler09, Edwards10}.
which have been used as a graphical language for measurement based quantum computation and for graph states \cite{Duncan2010, DP10, Euler09}.
Our method of arriving at this collection of tensors (Figure \ref{fig:F2-presentation1}) affords more general options and our presentation of the linear fragment 
here offers (i) improved semantics and (ii) a better theoretical understanding by pinpointing precisely that these networks correspond to the so called linear fragment 
of the \XOR-algebra \cite{CTNS}.

\begin{remark}[From \CNOT{} to \COPY-tensors or \XOR-tensors]
The following figure illustrates the contractions needed to transform a \CNOT{}-gate into either a \COPY-tensor (top) or an \XOR-tensor (bottom).  
 \begin{center}
\includegraphics[width=0.5\textwidth]{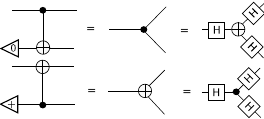}
\end{center}
These types of identities are common when considering interacting quantum observables \cite{CD, redgreen} as well as quantum circuits \cite{CD, redgreen,BB11}. 
\end{remark}

\begin{example}[\GHZ-class circuits]
We can realize the \GHZ-sate by the following circuit: 
 \begin{center}
\includegraphics[width=.95\textwidth]{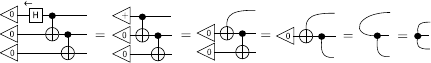}
\end{center}
The simplification from left to right is a sequence of contractions \cite{BB11}, which recover 
the familiar form of the \COPY-tensor.  On the other hand, one could also realize \GHZ{} by bending a wire as follows. 
 \begin{center}
\includegraphics[width=0.80\textwidth]{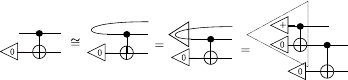}
\end{center}
These circuits scale to create $n$-qubit \GHZ-states in the evident way. 
\end{example}

\begin{example}[Transformations of the controlled Z-gate \cite{CD}]
The following illustrates (graphically) the well known identity that there is a symmetry between the control and target on a controlled Z-gate.  This sequence of rewrites is a law relating \COPY- and \XOR-tensors.  
  \begin{center}
\includegraphics[width=0.75\textwidth]{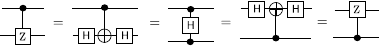}
\end{center}
Such a transformation is common in applications of graphical rewrite systems to measurement based quantum computation \cite{DP10} and appeared in the early work \cite{CD}. 
\end{example}

Now we will recall Lafont's diagrammatic presentation of the bialgebra law, between contractions of \COPY- and \XOR-tensors.  This relationship also appeared in early work on the ZX-calculus in \cite{CD, redgreen}. 

\begin{definition}[\COPY{} and \XOR-tensors form a (scaled) bialgebra \cite{boolean03}]
As mentioned, we often consider {\it equality} as being up to a scalar.  One can think of the bialgebra law as a type of commutation relationship amongst tensor pairs \cite{boolean03}---see early applications to quantum computing in \cite{CD,redgreen} which defined the {\it scaled} bialgebra law.  Turning it side ways we have (ignoring relative scale factors)
\begin{center}
\includegraphics[width=0.4\textwidth]{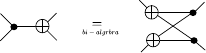}
\end{center}
\end{definition}

\begin{myexercise}[\COPY{} and \XOR-tensors form a bialgebra \cite{boolean03}]
Verify that \COPY{} and \XOR-tensors form a bialgebra \cite{boolean03}, up to a global scale factor, which should be determined.   
\end{myexercise}

We will now relate this abstract definition to well known quantum circuit identities.  Let us first recall the well known factorization of the \swap{} gate into a triple product of \CNOT-gates.  

\begin{example}[Factorization of the \swap-gate]
As can be shown using the \XOR-algebra (Appendix \ref{sec:XOR}), the \CN-gate together with its horizontally mirrored pair allows one to construct the \swap-gate as follows.  
\begin{center}
 \includegraphics[width=0.40\textwidth]{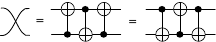}
\end{center}
A fully graphical proof of this relationship first appeared in the 2008 work \cite{CD} which introduced what is now called the ZX-calculus--see also \cite{redgreen}. 
\end{example}

\begin{theorem}[Relating \swap{} and the scaled bialgebra law] 
In the past, we have provided a diagrammatic proof that the square of the controlled not gate is equal to the identity.  This in turn allows one to relate the \swap-gate (using its factorization above) and the bialgebra law.  
\begin{center}
\includegraphics[width=0.36\textwidth]{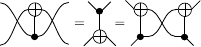}
\end{center}
So the factorization of \swap{} together with \CN$^2=\I$ is enough to have the bialgebra law on tensors.  
\end{theorem}

We will now introduce the gate-copy rewrite rule. Gate-copy allows one to
\textit{pull} controls and targets through each other.  When this happens, they are
copied, along with the attaching wires, leaving the attaching tensor intact.  

\begin{myfullpage}
\begin{theorem}[Gate-copy] 
\label{theorem:gatecopy}
The following graphical rewrites in (a) and (b) hold.    
 \begin{center}
 \includegraphics[width=1\textwidth]{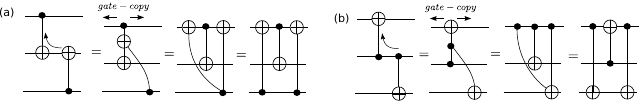} 
\end{center}
\begin{proof}
The proof of gate-copy follows from application of the bialgebra and fusion laws, as follows.     
 \begin{center}
 \includegraphics[width=0.65\textwidth]{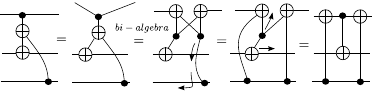} 
\end{center}
\end{proof}
\end{theorem}
\end{myfullpage}

 \begin{fullpage}
 \begin{example}[Circuit simplification using gate-copy]\label{fig:ZZZsim}
Here we apply gate-copy to simplify the circuit from~\cite{NC} designed to
simulate time evolution under the $\sigma^z\sigma^z\sigma^z$ Hamiltonian.
\end{example}
 \begin{center}
   \includegraphics[width=.95\textwidth]{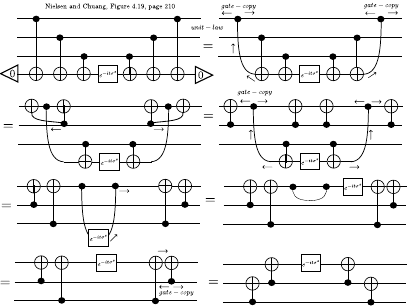} %
 \end{center}
Starting from the circuit from Figure 4.19 on page 210 of~\cite{NC}, we
apply a sequence of transformations including the Gate-copy reduction rule introduced in
Theorem~\ref{theorem:gatecopy}.  The network resulting from the simplification appears in the bottom
right. \\ 
\end{fullpage}

\paragraph{Hopf law.}  We recall Lafont's diagrammatic form of the Hopf-law \cite{boolean03}.  

\begin{theorem}[\COPY- and \XOR-tensors satisfy the Hopf-law \cite{boolean03}] 
The following diagrammatic equations, reproduced from \cite{boolean03}, depict the Hopf-law, which is satisfied by \COPY- and \XOR.  
\begin{center}
 \includegraphics[width=0.12\textwidth]{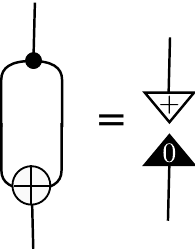} 
\end{center}
Exploration of the consequences of this identity in terms of quantum circuit manipulation can be found in the ZX-calculus \cite{redgreen} as well as other works on categorical models of quantum circuits \cite{BB11}. 
\end{theorem}

Using the Hopf law, we can justify a well known gate identity.  This identity was generalized to arbitrary finite dimensions in \cite{BB11}.  
\begin{center}
 \includegraphics[width=0.5\textwidth]{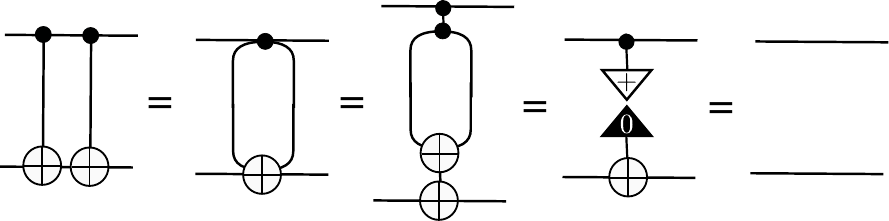} 
\end{center}

We have considered the key properties and defining equations of the \XOR- and \COPY-tensors.  We will use these results as 
building blocks for the sections that follow on from here. 

\section{Stabilizer Tensor Theory} \label{sec:stabtt}
Let us recall the notation of a stabilizer.  \marginnote{{\it ``We may always depend on it that algebra, which cannot be translated into good English and sound common sense, is bad algebra.'' --- William Kingdon Clifford}}

\begin{definition}[General stabilizer]
A unitary map $U$ stabilizes a quantum state $\psi$ iff 
\begin{equation}
U\psi = (+1)\psi. 
\end{equation} 
Stabilizers of $\psi$ from a group represented trivially by $\psi$.  
\end{definition}

We are concerned here with a subclass of the above definition.  We are concerned with what are commonly known as ``stabilizer states''.  

\begin{definition}[Stabilizer state]
 An $n$-qubit stabilizer state $\psi$ is defined by $n$ commuting and different operators $\1 S$ formed from the Pauli algebra with presentation
 \begin{equation}
     \{\I, X, Y, Z, \pm i, \cdot , \otimes\}. 
 \end{equation}
These $n$ operators generate the stabilizer group of order $2^n$ for $\psi$.   Clearly we have that $s\psi = \psi, \forall s\in \1 S$.    
\end{definition}

\begin{example}[Single qubit stabilizer states]
Here are the single qubit states stabilized by the Pauli-group
\begin{itemize}\addtolength{\itemsep}{-0.25\baselineskip}
 \item[(i)] \textbf{Pauli-X}: $\sigma^x$ stabilizes $\ket{+}=\ket{0}+\ket{1}$
and $-\sigma^x$ stabilizes $\ket{-}=\ket{0}-\ket{1}$
 \item[(ii)] \textbf{Pauli-Y}: $\sigma^y$ stabilizes
$\ket{y_+}=\ket{0}+i\ket{1}$ and $-\sigma^y$ stabilizes
$\ket{y_-}=\ket{0}-i\ket{1}$
 \item[(iii)] \textbf{Pauli-Z}: $\sigma^z$ stabilizes $\ket{0}$ and $-\sigma^z$
stabilizes $\ket{1}$
\end{itemize}
\end{example}

The following is simply a reminder of the essential properties of the Pauli operators.  These are useful for general knowledge and for working through the details of later calculations, but could be skimmed on a first read.

\begin{definition}
 Recall from angular momentum theory, the familiar Pauli Matrices.  We let $\sigma_1\equiv\sigma_x$, $\sigma_2\equiv\sigma^y$ and $\sigma_3\equiv\sigma^z$ which satisfy 
 \begin{itemize}
  \item[(i)] $[\sigma_i,\sigma_j] = 2i\epsilon^{ijk}\sigma_k$
  \item[(ii)] complex conjugation is generated by $\sigma_j$ $\forall w\in \{i,j,k\}$ as $\sigma_j\sigma_{w}^*\sigma_j=-\sigma_w$ for $w\neq j$
  \item[(iii)] $\text{Tr}(\sigma_i\sigma_j)=2\delta_{ij}$
 \end{itemize}
 Note that it is common to change subscripts to superscripts $\sigma_x\equiv\sigma^x$ to distinguish powers of operators $\sigma_x^2$ from operators acting on specific indices $\sigma^x_3$ --- that is, $\sigma_x$ acting on the third qubit and not $\sigma_x$ cubed.  This should be evident from context.
\end{definition}

These familiar operators from quantum mechanics form what is called a Geometric Algebra (a.k.a. Clifford Algebra). Consider $\{\sigma^l_i\}$ as a basis for a real left and right distributive vector space, such that
\begin{equation}\label{eqn:geometric1}
\sigma_x^2=\sigma_y^2=\sigma_z^2=\I=-i\sigma_x\sigma_y\sigma_z.
\end{equation}
Then consider the product
\begin{equation}\label{eqn:geoprod}
 \sigma_i\sigma_j = \delta_{ij}\I + i\epsilon^{ijk}\sigma_k~~~\text{(geometric product of vectors)}.
\end{equation}
It is clear that for $i\neq j$
\be \label{eqn:vanish}
\{\sigma_i^l,\sigma_j^l\}=0
\ee
meaning that the vectors anti-commute viz.,
\be \label{eqn:anti}
\sigma_i^l\sigma_j^l=-\sigma_j^l\sigma_i^l.
\ee
We also note that 
\be \label{eqn:eyepauli}
\{\sigma^i,\sigma^j\}=2\delta_{ij}\I.
\ee

\begin{myexercise}
 Verify Equations~\eqref{eqn:geoprod}, \eqref{eqn:vanish}, \eqref{eqn:anti} and \eqref{eqn:eyepauli}. 
\end{myexercise}

Single qubit density operators can be expanded in terms of a dot product of a polarisation vector $\underline{P} :=(p_1,p_2,p_3)$, $\forall i, p_i \in \7 R$, and a sigma vector $\underline{\sigma}:= (\sigma_1, \sigma_2,\sigma_3)$ as
\begin{equation}
\rho = \frac{1}{2}\I + \underline{\sigma}.\underline{P} =\frac{1}{2}\I +  p_i\sigma^i = p_1\sigma^1 + p_2\sigma^2 + p_3\sigma^3. 
\end{equation}  

Clearly the vectors $\underline{P}$ are in the vector space $\7 R^3$.  We can elevate this vector space to a Hilbert space by defining an inner product between vectors --- that is a map $(-,-)$ taking two elements from the vector space (in this case Hamiltonians on $\7C^2$) and producing a scalar in the underlining field $\7C$.
\begin{equation*}
 (-,-):(\7C^2\rightarrow\7C^2)\times(\7C^2\rightarrow\7C^2)\rightarrow \7C
\end{equation*}
\be 
(\underline{\sigma}.\underline{A}, \underline{\sigma}.\underline{B})\mapsto Tr[(\underline{\sigma}.\underline{A}).(\underline{\sigma}.\underline{B})] = \underline{A}.\underline{B}
\ee

\begin{myexercise}[Product of vectors] 
 Verify that straight forward calculation also yields
\be
(\underline{\sigma}.\underline{A},\underline{\sigma}.\underline{B})=A_iB_i + i\epsilon_{ijk}\sigma^k = \underline{A}.\underline{B} + i\underline{\sigma}(\underline{A}\wedge \underline{B}).
\ee
\end{myexercise}

\begin{myexercise}[Eigenvalues] 
 Verify, by using the fact that for $\underline{A}=\underline{B}$ the cross product vanishes or otherwise, that the eigenvalues of $(\underline{A}.\underline{\sigma})$ are $\pm|\underline{A}|$, where both roots necessarily appear as $\underline{A}.\underline{\sigma}$ is traceless. 
\end{myexercise}

We can relate symmetry in density operators and symmetry in states as follows.  

Before continuing on to define the Clifford group of quantum circuits, we will consider an example of a stabilizer state.  

\begin{example}[The Bell state is a stabilizer state]
 The bell state 
 \be 
 \Phi^+ = \sum a \oplus \neg b \ket{a, b} = \ket{00}+\ket{11}
 \ee 
 is a stabilizer state with stabilizer group 
 \be 
 \1S= \{\I, XX, -YY, ZZ\}
 \ee 
Note that this matches the following claims in the definition (i) the stabilizer group is of order $2^n$, here $n=2$, with generators given by e.g. $XX$, $ZZ$, etc. and (ii) the group is abelian.  
\end{example}

\begin{myexercise}
 Verify that $XX$, $-YY$ and $ZZ$ are in fact stabilizers of the Bell state $\Phi^+$ graphically.  
 Using the rules of the Pauli algebra, verify that these operators commute.  
\end{myexercise}

\begin{myexercise}
 Find the stabilizers for the following states.  
 \be 
 \Phi^- = \ket{00}-\ket{11}
 \ee 
 \be 
 \Psi^- = \ket{01}-\ket{10}
 \ee 
 Hint.  One method would be to use the Pauli algebra and the stabilizes for $\Psi^+$.  First cancel the local operator separating the given states from 
 $\Psi^-$, apply the known stabilizer and then reapply the operator to return to the starting state.  Other methods exist. 
\end{myexercise}

\section{The Clifford Group}

\marginnote{{\it 
The following are equivalent. 
\begin{itemize}
    \item[1.] The ZX-calculus \cite{Coecke2017}.
    \item[2.] Clifford gates plus cups and caps. 
    \item[3.] Ability to bend wires and compose stabilizer states.
\end{itemize}
}}

There is another useful way of generating stabilizer states.  We will consider a quantum circuit acting on the state $\ket{0}^{\otimes n}$.  The circuit is comprised solely of Clifford gates.

\begin{definition}[Clifford gates]
The collection of Clifford gates is as follows. (a) \CNOT{}; (b) the Hadamard gate $H = \frac{1}{\sqrt{2}}(X+Z)$; (c) The phase gate $P = \ket{0}\bra{0} + i \ket{1}\bra{1}$, and (d) the Pauli gates generated by the Pauli algebra on a single qubits, with 16 elements (the group generated by $\{X, Y, Z\}$).    
 \begin{center}
 \includegraphics[width=0.65\textwidth]{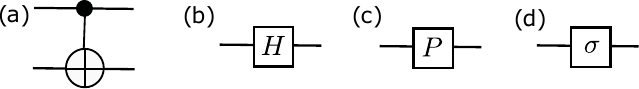} 
\end{center}
\end{definition}

\begin{remark}[Single qubit Clifford group]
The standard properties of single qubit gates follow. 
\begin{itemize}
 \item[(i)] $HXH=Z$; and $HZH = X$
 \item[(ii)] $PXP^\dagger = Y$; and $PYP^\dagger = Z = P^2$
\end{itemize}
These gates above generate the single qubit Clifford group.   
\end{remark}

With the definitions of the the Clifford group in place, one can state the alternative definition of stabilizer states.  

\begin{definition}[Stabilizer states]
If $\psi$ can be produced from the all-$\ket{0}$ state by Clifford gates, then $\psi$ is stabilized by $2^n$ tensor products of Pauli matrices or their sign opposites (where $n$ is the number of qubits).  This means that the stabilizer group is generated by $\log(2^n)=n$ such tensor products.  The state $\psi$ is then the \textit{stabilizer state} uniquely determined by these generators.
\end{definition}

\begin{remark}[Properties of Clifford circuits]
We now list elementary properties of Clifford circuits.  
 \begin{itemize}
  \item[(i)] Clifford circuits generate the Clifford group.
  \item[(ii)] Let $P_n$ be the collection of $4^n$ $n$-letter words with $\otimes$ as concatenation generated from the alphabet
  \begin{equation}
      \{\pm\I, \pm X,\pm Y,\pm Z, \pm i\I, \pm iX,\pm iY,\pm iZ\}. 
  \end{equation}
All operators $g$ in the Clifford group acts as an involution when $P_n$ is conjugated by $g$, that is $g P_n g^\dagger = P_n$.    
 \item[(iii)] Note, alternative notation to write $P_1$ could be 
 \be 
 P_1 = \{\pm 1, \pm i\}\{\I, X, Y, Z\}. 
 \ee
 \item[(iv)] Defining properties of the Pauli matrices:
 \begin{equation}
     X^2=Y^2=Z^2=\I = -iXYZ.
 \end{equation}
  \end{itemize}
\end{remark}

\begin{myexercise}
 Show that conjugation by $H$ lifts to an involution on $P_1$ by considering (iv) with $H= \frac{1}{\sqrt{2}}(Z+X)$.
\end{myexercise}

\section{Stabiliser Tensor Theory}

Let us cast the stabilizer theory into a theory of tensors, in the Penrose graphical calculus.  We will first define the notation of an abstract stabilizer for a tensor.  

\begin{definition}[Abstract stabilizer]
 Let $\Gamma$ be a valence-($n$,$m$) tensor.  A stabilizer for $\Gamma$ is given by $m+n$ local invertible operators, satisfying 
  \begin{center}
 \includegraphics[width=0.5\textwidth]{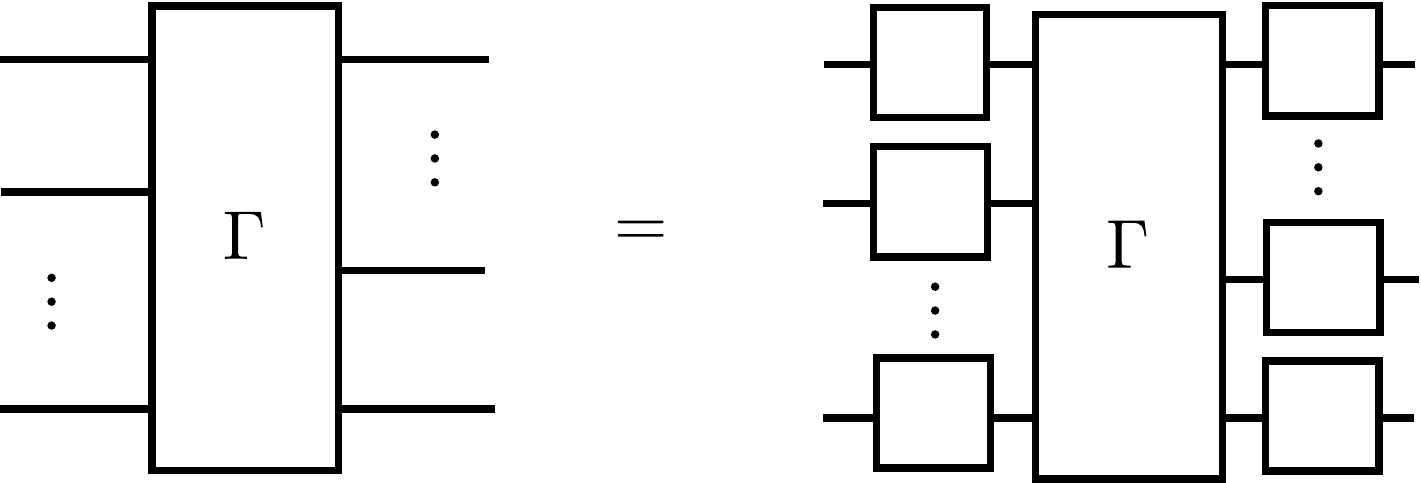} 
\end{center}
\end{definition}

We then turn to a definition of tensors, which as we will soon show, generate the Clifford circuits as a subclass.  

\begin{theorem}[Stabilizer tensors]
Contraction of tensors taken from the following are sufficient to generate the stabilizer group.  
\begin{center}
 \includegraphics[width=0.5\textwidth]{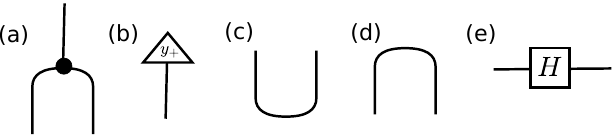} 
\end{center}
Where (a) is the \COPY-tensor, (b) is the y-plus-state $\ket{y_+}\bydef\ket{0}+i\ket{1}$, the cup (c) and cap (d) are to bend wires and hence reshape maps and take transposes, (e) is the Hadamard gate.  
\begin{proof}[Stabilizer tensors]
 Let us first consider generating the states $\ket{+}$, $\ket{-}$, $\ket{y_-}$. These are found from the following tensor contractions.
\begin{center}
 \includegraphics[width=0.5\textwidth]{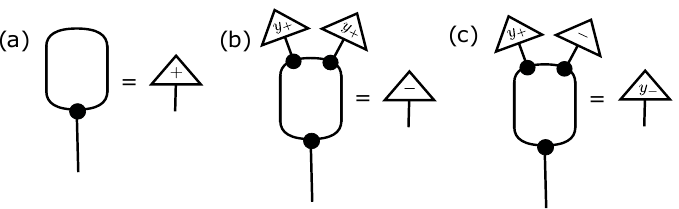} 
\end{center}
From these states, the following contractions generate our elementary gates.
\begin{center}
 \includegraphics[width=0.5\textwidth]{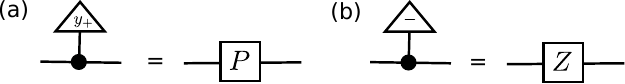} 
\end{center}
We then note that 
\be 
HZH = X
\ee 
\be 
P^\dagger = P^2P 
\ee 
and that 
\be 
PXP^\dagger = PXP^3 = Y
\ee 
We recover the \COPY-tensor by contracting with Hadamard gates and \CNOT{} is found from reshapes of the following contraction
\begin{center}
 \includegraphics[width=0.16\textwidth]{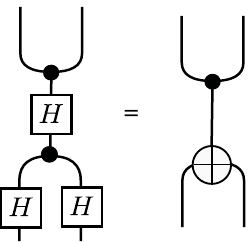} 
\end{center}
We then recover the generators of the Clifford group.   
\end{proof}
\end{theorem}

Now we will consider stabilizers of the \COPY-tensor.  

\begin{example}[Stabilizers of \COPY] The \COPY-tensor has stabilizer
generators $\sigma^x_1\otimes\sigma^x_2\otimes\sigma^x_3$ and
$\sigma^z_i\otimes\sigma^z_j$ which uniquely determine
$\psi_{GHZ}=\ket{000}+\ket{111}$ and result in the following stabilizer group of order $2^3=8$. 
\be
\{\sigma^x\sigma^x\sigma^x,-\sigma^x\sigma^y\sigma^y,
-\sigma^y\sigma^x\sigma^y,-\sigma^y\sigma^y\sigma^x,
\I\otimes \sigma^z\otimes\sigma^z,\sigma^z\otimes\I\otimes \sigma^z,
\sigma^z\otimes\sigma^z\otimes\I, \I\}
\ee 
Diagrammatically these relations are given in Figure~\ref{fig:stabGHZ}.   

\end{example}

\begin{figure}[h]
\centering
\includegraphics[width=0.7\textwidth]{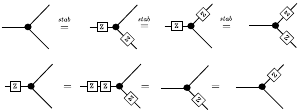}\\
\includegraphics[width=0.750\textwidth]{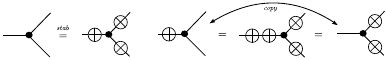}\\
\includegraphics[width=0.60\textwidth]{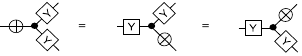}
\caption{(Top) Diagrammatic depiction of the stabilizer equation
$\sigma^z_i\sigma^z_j(\ket{000}+\ket{111})=\ket{000}+\ket{111}$. (Bottom)
Uses the stabilizer identity together with $\sigma_z^2=\I$ to show that
the $\sigma^z$ commutes with the \COPY-tensor. (Middle) Diagrammatic depiction of the stabilizer
equation
$\sigma^x_1\otimes\sigma^x_2\otimes\sigma^x_3(\ket{000}+\ket{111})=\ket{000
}+\ket{111}$.  (Bottom) Diagrammatic depiction (up to a sign) of the
stabilizer equation
$-\sigma^x_i\otimes\sigma^y_j\otimes\sigma^y_k(\ket{000}+\ket{111})=\ket{
000}+\ket{111}$.} 
\label{fig:stabGHZ}
\end{figure}

\begin{myexercise}[Gottesman-Knill Theorem] 
A graphical rewrite proof (by bounding the number of rewrites) of the
Gottesman-Knill theorem follows by considering the action of the \COPY- and \XOR-tensors on $\sigma^z$ and $\sigma^x$.  Derive the gate identities in paper \cite{GN98} and 
prove the main theorem using the methods from this lecture.  
\end{myexercise}

\begin{myexercise}[Stabilizers for $\delta_{ijkl}$]
Write down the $2^4$ stabilizers of the state 
\be 
 \sum_{ijkl\in\{0,1\}^4} \delta^{ijkl}\ket{ijkl}
\ee 
Now consider stabilizers in (a).  We can expand this arriving at (b).  What are the conditions on $C$ and $D$ such that $A$, $B$, $E$, $F$ are stabilizers?  Compare this with the 
stabilizers given for \COPY, prior to contraction.  
\begin{center}
 \includegraphics[width=0.65\textwidth]{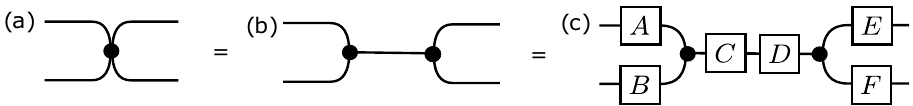} 
\end{center}
\end{myexercise}

We can now consider how stabilizers of a tensor transform, when the tensor undergoes a local change of basis.  

\begin{theorem}[Transformation properties of stabilizers]
Let $\Gamma$ be a tensor with stabilizer $A\otimes B\otimes \cdots \otimes C$.   Then if we rotate $\Gamma$ as 
$\Gamma' = U_G(\Gamma)$, then 
\be
U_G(A\otimes B\otimes \cdots \otimes C) = U_G(A)\otimes U_G(B)\otimes \cdots \otimes U_G(C)
\ee
is a stabilizer for $\Gamma'$.  
\end{theorem}

\section{Boolean Stabilizer States}

Now one can consider the intersection found from the class of Boolean states which are in addition to being Boolean states, also stabilizer states.  We will construct them explicitly for the case of a single qubit, and leave the two qubit case 
as homework.  

\section*{Single qubit Boolean Stabilizer States}

There is one stabilizer state for each single qubit boolean state.  This correspondence is a special case 
of the dim $=2$ state space and does not hold in any higher dim.  \marginnote{{\it ``To know all about anything is to know how to deal with it under all circumstances.'' --- 
William Kingdon Clifford} }

\subsection*{Single bit functions} 
Consider 
\be 
f:\7 B^2 \rightarrow \7 B 
\ee 
then a general function $f$ is expanded as a sum-of-products (LHS) and as a pseudo Boolean form (RHS)  
\be 
c_0 \overline{x}\vee c_1 x = c_0 + (c_1 - c_0)x 
\ee 
for $c_0, c_1, x \in \7 B$.  There are exactly $2^{2^n}$ boolean functions.  Evaluated at $n=1$ gives four possible single qubit Boolean states.  One of these however, corresponds to 
$\psi =0$ and so is trivial.  The others are $\ket{0}$, $\ket{1}$ and $\ket{0}+\ket{1}$.  

\subsection*{Single qubit stabilizer operators} 

We will consider $\pm Z$ and $X$.  These are the only possible stabilizer operators that stabilize single qubit boolean states.  
\be 
X(c_0\ket{0} + c_1\ket{1}) = c_1 \ket{1} + c_0 \ket{0} 
\ee 
and hence, $X\psi = +1 \psi$ iff $c_1 = c_0$.  Over $\7B$ this has non-trivial solutions (that is, $\|\psi\| > 0$) for 
$c_0=c_1=1$ and hence, we recover the only boolean state stabilized by $X$ as $\ket{0}+\ket{1}$.  We will then consider 
\be 
\pm Z \psi = \pm c_0 \ket{0} \mp c_1 \ket{1} 
\ee 
and $\pm Z \psi = +1 \psi$ iff $\pm c_0 = \mp c_1$.  It follows that 
\be 
+ Z \Rightarrow c_0 = 1, c_1 = 0 
\ee 
\be 
-Z \Rightarrow c_0 = 0, c_1 = 1
\ee 

\subsection*{Correspondence between stabilizer states and boolean states} 

In this section we consider the correspondence between single qubit stabilizer states and single 
qubit boolean states.  Let us introduce two boolean variables, $b_0$ and $b_1$.  We then write 
\be 
(-1)^{b_1}(1-b_0)Z + b_0 X
\ee 
We find that $b_0$ $\Rightarrow$ $c_0=c_1=1$.  The case that $b_0=0$ implies that $b_1$ decides $c_0, c_1$ as 
\be 
b_1=1 \Rightarrow c_1=1, c_0=0
\ee 
\be
b_1=0 \Rightarrow c_1=0, c_0=1
\ee 
We will then parameterize the single qubit boolean state in terms of $b_0$ and $b_1$ as 
\be 
\psi_{\7 B} = (b_0 \vee\overline b_0 \wedge b_1)\ket{0} + (b_0\vee \overline b_0 \wedge \overline b_1)\ket{1} = (b_0 + b_1-b_0b_1)\ket{0} + (1+b_0 - b_1 + b_0b_1)\ket{1}
\ee 
and hence the choice $c_0 = b_0 + b_1-b_0b_1$ and $c_1 = 1+b_0 - b_1 + b_0b_1$ provides the correspondence.  

\begin{remark}
 In or around circa.~2010, I discussed development of classification scheme for Boolean stabilizer states with Oscar Dahlsten.  I believed Oscar solved that problem, specifying exactly the class of Boolean stabilizer states. 
\end{remark}

\begin{myexercise}[Two qubit Boolean stabilizer states]
 For two qubits, there are $2^{2^2}-1 = 15$ boolean states.  Let $q$ and $r$ take boolean values.  Then we arrive at the following truth table.  

\begin{center}
\begin{tabular}{l||l|l||l|l|l|l|l|l|l|l|l|l|l|l|l|l}
$q$ $r$ & $f_0$  & $f_1$  &   &   &   &   &   &   &   &   &   &   &   &   &   & $f_\wedge$ \\\hline 
$0$ $0$ & 0 & 1 & 1 & 1 & 1 & 0 & 1 & 1 & 0 & 0 & 1 & 0 & 1 & 0 & 0 & 0 \\
$0$ $1$ & 0 & 1 & 1 & 1 & 0 & 1 & 1 & 0 & 1 & 0 & 0 & 1 & 0 & 1 & 0 & 0 \\
$1$ $0$ & 0 & 1 & 1 & 0 & 1 & 1 & 0 & 1 & 1 & 1 & 0 & 0 & 0 & 0 & 1 & 0 \\
$1$ $1$ & 0 & 1 & 0 & 1 & 1 & 1 & 0 & 0 & 0 & 1 & 1 & 1 & 0 & 0 & 0 & 1
\end{tabular}
\end{center}

\begin{itemize}
 \item[(i)] Label \OR, \NOR, \NAND, \XOR{} in the table above.  
 \item[(ii)] Using the definition of $K_1$, determine which states in this table are separable, and classify them based on the resulting values of $K_1$. 
 \item[(iii)] Which states in the table are not stabilizer states?  
\end{itemize}
\end{myexercise}

\begin{myexercise}[Invariants of symmetric Boolean states]
~
\begin{itemize}
 \item[(i)] Write the general form of a symmetric three qubit boolean state.  How many possible symmetric boolean states are there? 
 \item[(ii)] Consider the discriminant of the cubic.
\be 
\Delta = a_0^2 a_3^2 - 6 a_0 a_1 a_2 a_3 + 4 a_0 a_2^3 - 3a_1^2 a_2^2 + 4 a_1^3 a_3
\ee 
and find maximum and minimum values when $\forall i, a_i\in \{0,1\}$.  
\end{itemize}
\end{myexercise}

\section{Problems} 

\begin{theorem}[Sufficient expression stabilizer states]
Let 
\be 
f, g, k:\7B^n\rightarrow \7B
\ee 
then the quantum state 
\be 
\psi_{\7B} = \sum (-1)^{f(\1 x)}(i)^{g(\1 x)}k(\1 x)\ket{\1 x}
\ee 
is sufficient to express any stabilizer state.  
\begin{proof}[Normal forms on stabilizer states]
 Each stabilizer state is an equally weighted superposition with coefficients $\pm1, \pm i, 0$.  The functions $f, g, k$ determine these for each basis vector 
 $\ket{\1 x}$.  
\end{proof}
\end{theorem}

\begin{remark}[Normal forms]
 Using a PPRM from lecture II, we can expand $f, g, k$ to a normal form.  The functions then become uniquely determined by a coefficient vector.  
\end{remark}

\begin{myexercise}[Stabilizer states as pseudo Boolean forms]
 Find $f$, $g$, and $k$ to express the following states (note, the second is not a stabilizer state). 
 \begin{itemize}
  \item[(i)] $\psi_1 = \ket{000} + i \ket{111}$
  \item[(ii)] $\psi_2 = \ket{001} + i \ket{010} - i\ket{100}$  
 \end{itemize}
\end{myexercise}

\begin{myexercise}[Stabilizer generators]
\marginnote{{\it ``An expert is someone who knows some of the worst mistakes that can be made in their subject and how to avoid them.''}--- Werner Heisenberg}
 Consider $U$ as an arbitrary Clifford circuit.  Then, 
 \be 
 \psi = U\ket{0}^{\otimes n}
 \ee 
 is an arbitrary stabilizer state.  Show that evolution of $Z_i$ under $U$ in the Heisenberg picture
 \be 
 UZ_iU^\dagger
 \ee 
 is necessarily a stabilizer for $\psi$.  
\end{myexercise}

\begin{myexercise}[Further Exercises on Pauli Matricies] 
~
\begin{itemize}
 \item[(i)] Let $\underline{P}.\underline{\sigma} := p_0 X + p_1 Y + p_2 Z$ where $|P|=1$. Show that $\exp(-i\frac{\theta}{2}\underline{P}.\underline{\sigma})=\I \cos(\theta/2)-i(\underline{P}.\underline{\sigma})\sin(\theta/2)$ and find the values of $\theta$, $\underline{P}$ to recover the Hadamard gate, up to a phase factor.  

\item[(ii)] Find the time of the evolution of the Hamiltonian $\ket{11}\bra{11}$ to create a {\sf CZ}-gate, then write down a quantum circuit in terms of {\sf H}  and {\sf CZ} to create a {\sf CNOT}-gate.  What are the input states needed to use the {\sf CNOT}-gate to prepare the singlet state $\ket{\Psi^-}=\ket{01}-\ket{10}$?  

\item[(iii)] Show that the \swap{} operator $\frac{1}{2}(\I + \underline{\sigma}_A\cdot\underline{\sigma}_B)$ permutes the values of bits $A$ and $B$ as 
  \begin{equation}
   \text{\swap} \ket{i_A}\ket{i_B} = \ket{i_B}\ket{i_A}
  \end{equation}
  where the notation $ \underline{\sigma}_A\cdot\underline{\sigma}_B$ is typically said to stand for the scalar and tensor product:  
  $ \underline{\sigma}_A\cdot\underline{\sigma}_B = \sum_{i = 1}^3 \sigma_i^A \otimes \sigma_i^B.$ 

\item[(iv)] In the computational basis, express the general form of a two-qubit symmetric eigenstate of the \swap{} operator and count the real degrees of freedom.  Repeat this for anti-symmetric eigenstates (e.g. $\swap\ket{\psi}=-\ket{\psi}$).  
\item[(v)] Using the notation from (iii) above, find a value for $q$ to show that the two-site quantum Heisenberg model $J\underline{\sigma}_1\cdot \underline{\sigma}_2$ can be written as $\frac{J}{2}\left((\underline{\sigma}_1+\underline{\sigma}_2)^2-q\I\right)$ and show that $\ket{\Psi^\pm}=\ket{01}\pm\ket{10}$ are energy eigenstates.  
\end{itemize}
\end{myexercise}

 \part{Tensor Networks and Entanglement Invariants} 
\label{sec:qubit-invariants}

It is typical to consider a symmetry of a density operator $\rho$ or state $\psi$ as an operator satisfying 
\be 
\rho = V\rho V^\dagger
\ee 
\be 
\psi = M \psi
\ee
and so one could say that $\rho$ is invariant under $V$ and $\psi$ is
invariant under $M$ --- see Example \ref{ex:group-sym}.
If the relations holds for all elements of a matrix group, then we say
that $\rho$ or $\psi$ are invariant under the left multiplicative action of the group.
There are more subtle symmetries however.  
These are formed by considering polynomials in the coefficients of $\rho$ or $\psi$ that are left invariant under the action of a group.  The study of such polynomials is known as 
Invariant Theory \cite{hilbert}.  David Hilbert made notable progress on the topic of invariant theory, which he perused throughout his life.  There has been past work on considering these invariants in the context of quantum information science.  Some of our personal favorites include \cite{entinv2, entinv3, entinv1} as well as some work more closely related to this chapter \cite{BBL11}. Algebraic geometry has also been applied to tensor networks in \cite{2014SIGMA..10..095C}. 

\begin{example}[Group Symmetry of $\rho$]\label{ex:group-sym}
 We will consider a general density operator $\rho$ and look for $V\in U(d)$ such that 
 \be 
 \rho = V\rho V^\dagger
 \ee 
 this implies that $[\rho, V]=0$ and we arrive at a basis for $G$ by noting the unitary operators that commute with $\rho$.  
 That is, $\{\ket{\lambda_i}\}_i$ such that $\rho = \sum_i p_i\ket{\lambda_i}\bra{\lambda_i}$.  It then follows from $VV^\dagger = \I$ that 
 every $V\in G$ can be written as 
 \be 
 V= \sum_i e^{i\theta_i}\ket{\lambda_i}\bra{\lambda_i}
 \ee 
\end{example}


\subsection*{Introduction to polynomial invariants}
\label{sec:comdef1}

One can form polynomials out of the coefficients of a state or an
operator.  These algebraic invariants, are called polynomial.  For example, given a state with coefficients
$\alpha^{ij}$,
\be 
\psi = \sum_{ij} \alpha_{ij}\ket{ij},
\ee 
we could form a real valued polynomial function out of these variables $\alpha^{ij}$ and their complex conjugates $\overline \alpha_{ij}$ 
\be 
f(\alpha_{00},\alpha_{01},\alpha_{10},\alpha_{11}, \overline \alpha_{00},\overline \alpha_{01},\overline \alpha_{10},\overline \alpha_{11})
\ee 
Acting on the state $\psi$ with some linear transformation induces in turn an action 
of this linear transformation on the polynomial $f$.  We will be concerned with the case that the linear transformation can be any element of a group~$G$.  If the polynomial $f$ in the coefficients of the state remains unchanged 
under the induced action of all $g\in G$, then the polynomial is said to be a polynomial invariant under~$G$.

For example, the polynomial $J_1$ \eqref{eqn:J1} corresponds to the norm of the state, and is invariant under unitary transformations of $\psi$. 

\begin{equation}\label{eqn:J1}  
J_1 :=  \sum_{ij} \alpha^{ij}\overline{\alpha}_{ij} 
\end{equation} 

\begin{remark}[Basis independence]
To form a polynomial out of the coefficients of a state, one first chooses a basis to express the state in.  The coefficients of the state will change depending on the basis chosen.  A polynomial invariant that is invariant under any group that contains the local unitary group as a subgroup is however inherently (local) basis independent.  The basis chosen to express the polynomial depends on the local unitary group, however the polynomial is invariant under the local unitary group, by construction. 
\end{remark}

\begin{remark}[Polynomial invariants map states to scalars]\label{re:invariance}  
Polynomial invariants, are state independent.  In other words, an invariant is a function of a quantum state and (if proven to be an invariant) is invariant still for any quantum state.  A polynomial invariant is invariant for all states acted on by some group, but the numerical value can differ from state to state.  The numerical value of the invariant illuminates important properties about the specific state in question.  
\end{remark}

\begin{example}[Example of $\7Z_2$ invariance]
\label{ex:Z2-invariance}
Here we will illustrate properties of forming polynomial invariants out of the coefficients of a state by considering a toy example.  Let 
\be 
\psi = \sum_{\textbf{x}}c_{\textbf{x}}\ket{\textbf{x}}
\ee 
be a quantum state of two qubits.  For the purpose of this example, we will explore what a polynomial invariant is by considering a state-specific example (in contract to remark \ref{re:invariance}).  That is, we will pick specific values of the $c_{\textbf{x}}$'s to illustrate our point.  

We must pick a matrix group $G$ acting on states in $\C^2\otimes \C^2$.  Each $g\in G$ in turn induces an action on polynomials in $c_{\textbf{x}}$ as follows 
\begin{equation}\label{eqn:induced-action}
c_{\textbf{x}}\mapsto \langle\textbf{x},g\psi\rangle, ~g\in G
\end{equation} 
As an illustrative example, we pick $G=\{\sigma_x\sigma_x, \I\}$ as the group $G$ and  
\be 
\psi_2=c_{00}\ket{00}+c_{11}\ket{11}
\ee 
for the state.  We then consider polynomials 
\be
f(c_{00}, c_{11}, \overline c_{00}, \overline c_{11}) 
\ee 
that are invariant under $G$.  

As stated, action of the group $G$ on the state, induces an action on any polynomial in the coefficients of the state via \eqref{eqn:induced-action}.  
For $g\in G$ it is standard to write this action using the notation of putting $g\in G$ in the superscript \eqref{eqn:group-notation}.  
\begin{equation}\label{eqn:group-notation} 
f^g(c_{00}, c_{11}, \overline c_{00}, \overline c_{11}):=f((c_{00}, c_{11}, \overline c_{00}, \overline c_{11})g^\top)
\end{equation}  
then $f^{\I}=f$ and 
\be
f^{\sigma_x\sigma_x}(c_{00}, c_{11}, \overline c_{00}, \overline c_{11})=f(c_{11}, c_{00}, \overline c_{11}, \overline c_{00})
\ee
It becomes clear that one can write certain polynomials $f$ that are invariant, e.g. 
\be 
f= (c_{00} + \overline c_{00})(c_{11} + \overline c_{11})
\ee 
under this group and others that are not, e.g. 
\be 
f = (c_{00} + \overline c_{11})(c_{11} + \overline c_{00})
\ee
This example served to illustrate several key aspects of polynomial invariants, but as this toy example is state specific and basis dependent,  it does not 
capture the philosophy and key aspects present in the more general setting.  
\end{example}

\subsection*{A Graphical Language for Polynomial Entanglement Invariants} 
The invariants we have studied have all taken a remarkably simplistic form when cast into our framework.  The key to the simplification we have found was though our introduction of the diagrammatic SVD in Theorem \ref{theorem:diagrammatic-SVD}.  In fact, the study of invariants was a motivating factor which lead us to introduce this factorization, which we soon found to have other applications.  Through this factorization and through other methods we are currently exploring, we have found that Penrose graphical tensor notation and the invariant theory of Hilbert et al. connect very well together.  We are aware of a more cumbersome graphical language appearing in the classic literature and reviewed in the book by Oliver \cite{oliver}, which he credits to Cayley.  We feel the approach taken here, facilitated by the diagrammatic SVD and developed particularly to study entanglement invariants is significantly more natural.

\subsection*{Invariants of pure qubit states}

\subsubsection*{Tensor contractions for LOCC
\sn{LOCC: Local Operations and Classical Communication.}
invariants}
We have seen in Example~\ref{ex:Z2-invariance} that when we pick a
finite group and a basis for our state, one can form polynomials that
are invariant under the group action.
The standard case, however, is to consider invariance under the local
action of a continuous group, and to consider this action on arbitrary states, of appropriate dimension.

\begin{definition}[LOCC Equivalence]
Two states are LOCC equivalent iff they are equivalent
under local unitary transformations, that is, action of the group 
\be 
G_{\text{LOCC}} := U(1) \times SU(d_1)\times SU(d_2)\times \ldots  \times SU(d_n) 
\ee 
Equivalence under LOCC yields a partitioning of 
states.  Two states are in the same equivalence class iff they are related by a local unitary transformation. 
LOCC equivalence gives rise to a partitioning of the set of entanglement values. 
An entanglement measure should remain constant on the equivalence classes.
\end{definition} 

\begin{remark}[Single qubit invariants]
 When considering the local unitary group acting on a single pure qubit, 
the only polynomial invariant is the norm of the state 
\be 
J_1 = \sum \psi^i\overline{\psi}_i
\ee 
invariant under $U(1)\times SU(2)$. This is because the local orbit moves through the angles on the Bloch sphere, and so the norm is the only 
possible ambiguity.  By fixing the norm, we fix the only invariant.  
\end{remark}

We have made adaptations to the diagrammatic language the utility of which we will first illustrate by considering 
the case of two-qubits.  These adaptations allow for the confluent graphical contraction to evaluate 
invariants to equate them in terms of singular values.  This is illustrated by the following Theorem. 

\begin{lemma}[Graphical contraction of the norm $J_1$]
The graphical language enables a sequence of tensor contractions to relate the polynomial invariant $J_1$ to the singular values of the state.    

\begin{proof}
 For pure states of two qubits,
starting from Equation \eqref{eqn:J1} 
\begin{equation}  
J_1 :=  \sum_{ij} \alpha^{ij}\overline{\alpha}_{ij} 
\end{equation} 
we arrive at the contracted
network in (a) below.  To show that this contracts to a value determined by the singular values of the state, 
we apply the diagrammatic SVD which exposes the
internal network building blocks shown in (b).  
\begin{center}
 \includegraphics[width=14\xxxscale]{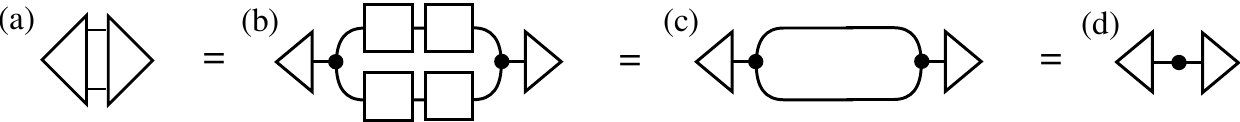}
\end{center}
The white boxes (valence-two tensors) represent unitary maps.  They therefore cancel 
resulting in the network illustrated in (c).  The \COPY-tensor
contracted as in (c) results in the identity wire (recall Figure (e) in Definition \ref{def:copy-prop}).  We then find the
inner product of the triangular tensors containing singular values,
\be 
\Rightarrow J_1 = \sum_i \lambda_i^2 = \lambda_0^2 +\lambda_1^2~ (= 1~\text{for a normalized state}) 
\ee 
Which is the desired result.  
\end{proof}
\end{lemma}

In addition to the norm $J_1$, for pure two-qubit states 
one finds another independent invariant given as\sn{A way to understand why we can expect only two invariants for two pure qubits is by considering the reduced density operator of a single qubit.  This can be done because there is only one possible partition for a two-qubit state.  This density operator has two singular values and we should be able to find (algebraically independent) invariants that separate the orbit, based on these two singular values. } 
\begin{equation}\label{eqn:J2} 
J_2 =  \sum \alpha^{i j}\alpha^{k l}\overline{\alpha}_{i l}\overline{\alpha}_{kj}
\end{equation} 
The invariant is at quartic order and its tensor diagram is given by contracting two copies of the state, and two copies of its dual under the dagger.
\begin{center}
 \includegraphics[width=2\xxxscale]{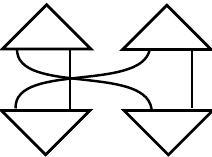}
\end{center}

\begin{lemma}[Graphical contraction of $J_2$]
The graphical language enables a sequence of tensor contractions to relate the polynomial invariant $J_1$ to the singular values of the state.

\begin{proof}
From the graphical expression of $J_2$ (a) we arrive at (b) which is found from applying the diagrammatic SVD to (a).  
The two pairs of unitary order-two tensors cancel (c) and after contraction we are left with
a product of four valence-one triangles (d).  The center portion, attached to the four triangles (c) contracts to a single valence-four \COPY-tensor as\mn{This elementary property of contraction, can also be seen as the \textit{special} case of a more general normal form, see page 64 of \cite{FA}.  It can also be seen as a tensor version of node equivalence in digital circuits \cite{CTNS}.} 
\be 
\delta^i_{~jk} \delta^q_{~rs}\delta^{js}_{~~l}\delta^{rk}_{~~~m} = \delta^{iq}_{~~lm}
\ee 
$J_2$ thus becomes a product of singular values and hence (b) provides a
graphical variant of \eqref{eqn:J2} given as 
\be
J_2 = \sum_i \lambda_i^4 =  \lambda_0^4 + \lambda_1^4 = 1 - 2\lambda_1^2 + 2 \lambda_1^4
\ee 
where the last step follows from the constraint on the states norm.  
\begin{center}
 \includegraphics[width=14\xxxscale]{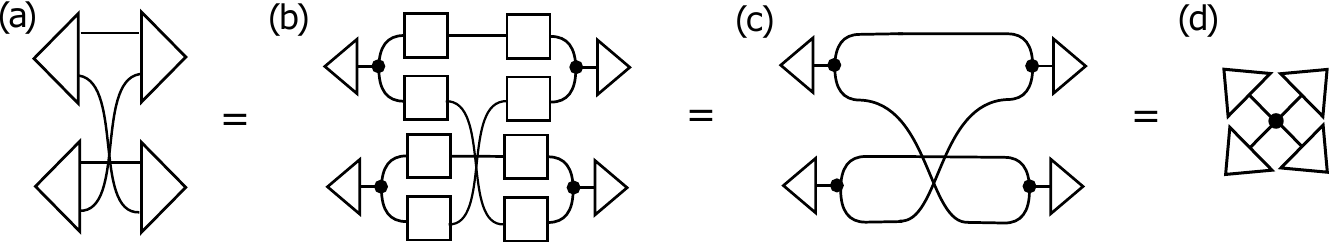}
\end{center} 
\end{proof}
\end{lemma}

These graphical contractions inspired the following derivations.  

\begin{lemma}[General formulas equating state coefficients to singular values]
 Given a general two-qubit state
\begin{equation}
|\Psi\rangle = \alpha \ket{00} + \beta \ket{01} + \gamma \ket{10} + \delta \ket{11}
\end{equation} 
it can be shown that the Schmidt coefficients are given as 
\begin{equation}
\lambda_{1,2}^2=\frac{1}{2}\left( 1 \pm \sqrt{1-4|\beta\gamma - \alpha\delta|^2} \right)
\end{equation} 
where $\alpha$, $\beta$, $\gamma$, and $\delta$ are complex numbers.
\end{lemma}

\begin{corollary}
 The invariance of $J_2$ implies the invariance of the Schmidt coefficients and vise versa.
\end{corollary}
\begin{proof}
It can be shown that the invariant $J_2$ can be expressed as 
\begin{equation} 
J_2=1-2 |\beta\gamma - \alpha\delta|^2  
\end{equation} 
which implies that 
\be 
|\beta\gamma - \alpha\delta|^2=\frac{1-J_2}{2}
\ee 
and therefore
\begin{equation}
\lambda_{1,2}^2=\frac{1}{2}\left( 1 \pm \sqrt{2J_2-1} \right)
\end{equation}  
\end{proof}

\begin{remark}[Relating $J_2$ to entanglement]
 Parameterizing $\lambda_0:= \cos \theta$ and $\lambda_1:= \sin
\theta$, with $0\leq \theta \leq \pi/4$ (which gives $\lambda_0 \geq
\lambda_1$) the invariant $J_2$ becomes
\be 
J_2(\theta) = \cos^4 \theta + \sin^4 \theta= \frac{1}{4} (3 + \cos(4 \theta))
\ee 
where the last line again follows from the norm.  The angle $\theta$
now becomes a characterization of the entanglement.  $\theta = 0$ iff
the state is separable and $\theta = \pi /4$ iff the state is locally
equivalent to a maximally entangled Bell state.  We can see this from
the following plot of $J_2(\theta)$.  For small angles, the value of the plot is $\approx1$ as the angle increases, 
the entanglement increases and the value of $j_2$ goes to its maximum value of one half (at $\theta/\pi = 1/4$). 
\begin{center}
 \includegraphics[width=0.65\textwidth]{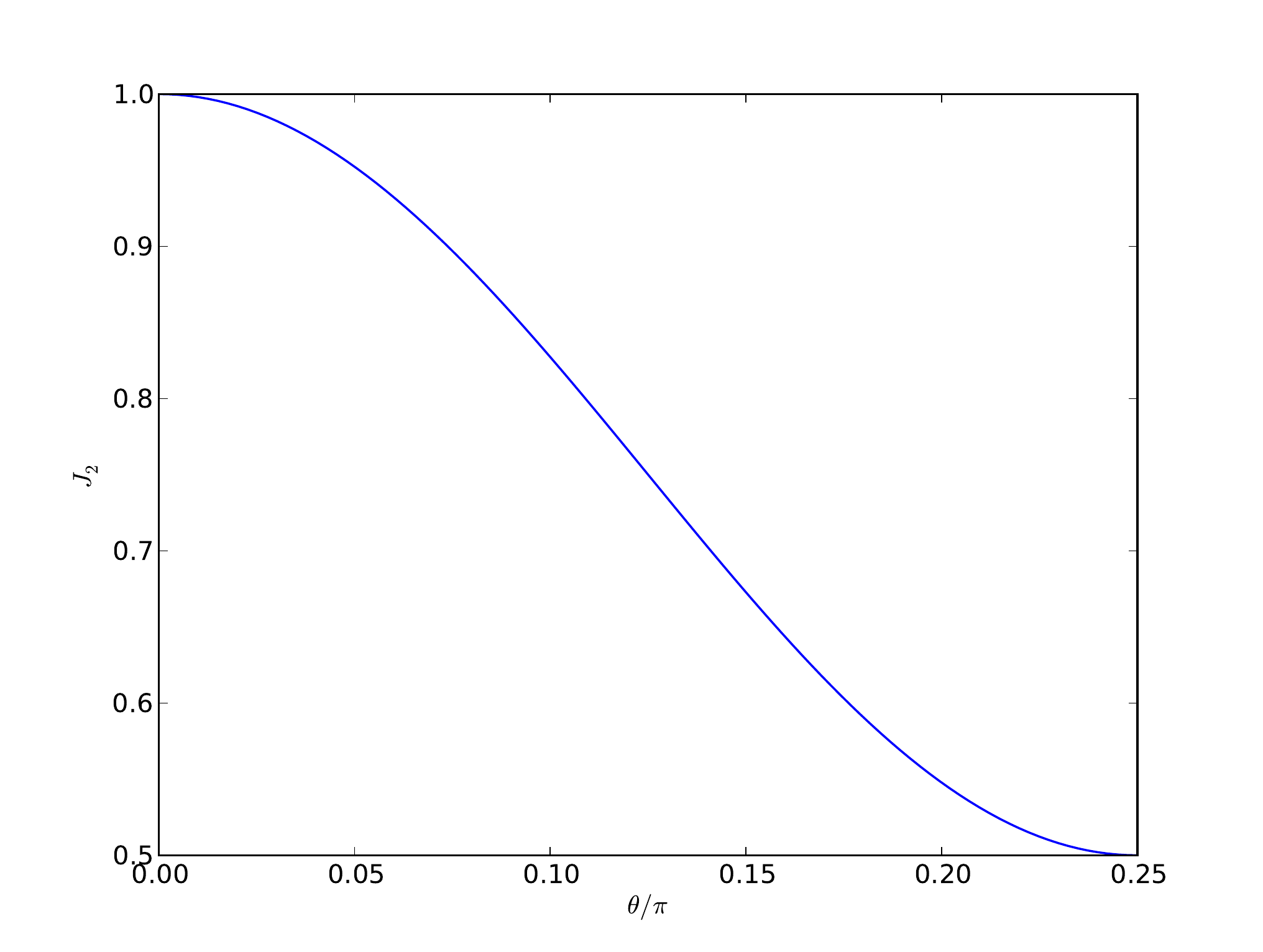}
\end{center}
\end{remark}

\begin{remark}[Algebraic independence of $J_1$ and $J_2$: the fundamental ring]
$J_1$ and $J_2$ are the only algebraically independent invariants under local action of the unitary group.
Any polynomial function of such invariants is also a polynomial
invariant.  In this fashion, it is a remarkable feature that functions
of $J_1$ and $J_2$ are all that is needed to express any local unitary
invariant of two-qubit pure states.  This elementary result follows from a
much more powerful and general result in classical invariant theory.
That is a proof by Hilbert that the ring of polynomial invariants is
finitely generated \cite{hilbert}.  This corresponds to freely
generated linear sums and products of $J_1$, $J_2$ (the ring \begin{equation}
    \{J_1, J_2, (\R, +, \cdot)\}. 
\end{equation}
Any minimal complete set of invariants that
can freely generate the full ring are called \emph{fundamental invariants}.
\end{remark}


\paragraph{$n$-qubit LOCC invariants.}

The method we have described in detail for two qubits is readily extended to $n$-qubits.  
Consider the general expression for an $n$-qubit pure state 
\be 
\psi = \sum \psi_{ijk...}\ket{ijk...}
\ee 
A general polynomial of the state coefficients together with their complex conjugates is expressed as
\begin{equation}\label{eqn:n-qubit-invariant}
\sum c^{ijk...}_{qrs...lmn}\overline{\psi}^{qrs...}\overline{\psi}^{lmn...}\cdots\psi_{ijk...}\cdots
\end{equation} 

If the polynomial \eqref{eqn:n-qubit-invariant} has an equal numbers of $\psi$’s 
and $\overline{\psi}$’s and all the indices of the $\psi$ are contracted using the
invariant tensor $\delta$ with those of the $\overline{\psi}$, each index being
contracted with an index corresponding to the same party then the polynomial is 
manifestly invariant under LOCC transformations.

\begin{remark}[Connection to the diagrammatic language] 
Invariant polynomials \eqref{eqn:n-qubit-invariant} can be written in terms of permutations
on the indices, and given diagrammatically as contraction with a permutation operator.  
\end{remark}

\begin{remark}[Fundamental ring of invariants]
Although we can generate a full basis of invariants \eqref{eqn:n-qubit-invariant}, except in rare cases, a minimal set of invariants is not known. 
\end{remark} 

\subsubsection*{Tensor contraction for SLOCC\sn{SLOCC: Stochastic Local Operations and Classical Communication.} invariants}   
SLOCC equivalence amounts to equivalence under local invertible 
transformations. This results in a coarser partitioning of the set of states 
as compared to LOCC equivalence, since SLOCC shows which states are
accessible to different parties with non-zero probability.  The following result was proven in \cite{3qubits}.  

\begin{theorem}[SLOCC]\label{theorem:SLOCC}
Two states are SLOCC-equivalent iff there exists
an invertible local operator relating them given by the action of the local general linear group \cite{3qubits} 
\be 
GL(d_1, \C)\times GL(d_2, \C)\times \ldots  \times GL(d_n, \C). 
\ee 
\end{theorem}

\begin{remark}[SLOCC under the special linear group]\label{remark:SL}
Up to a scale factor, and without loss of generality, we will instead consider action of the special linear group, $SL(2,\C)$ on qubits.  
\be
SL(2, \C) \times SL(2, \C) \times \ldots  \times SL(2, \C) 
\ee 
The precise relation between $GL$ and $SL$ invariants is mentioned in Remark \ref{remark:GL-vs-SL}.  We then consider transforming an $n$-party state $\psi$ to $S\ket{\psi}$ with $S:= \bigotimes_{i=0}^n S_i$ where each $S_i$ has unit determinant.  Action of either group is known to correspond precisely to quantum operations that preserve true entanglement under a general quantum operation.  
 
\end{remark}

\paragraph*{Graphical properties of the Levi-Civita symbol.} In order to study 
SLOCC invariants, and their tensor networks, we will need to recall 
the Levi-Civita symbol and consider a tensor network representation.  

\begin{definition}[Levi-Civita symbol]
The order-$n$ Levi-Civita symbol $\varepsilon^{ij \cdots k}$
is the fully antisymmetric tensor with coefficients in $\{-1,0,1\}$.
All its legs have the same dimension $d = n$.
\be
\varepsilon^{ij \cdots k} = \left\{
\begin{array}{rl}
1  & \text{when} \: (i,j,\ldots, k) \: \text{is an even permutation of the index values,}\\
-1 & \text{when} \: (i,j,\ldots, k) \: \text{is an odd permutation, and}\\
0  & \text{otherwise}.
\end{array}
\right.
\ee




\end{definition}

Before proceeding to an analysis of the SLOCC invariant tensor network structure, we must explore some properties of the epsilon tensor.  

\begin{remark}[Matrix determinant]
The determinant for an $n$-by-$n$ matrix
$A$ can be expressed in terms of the Levi-Civita symbol as follows:
\be
\label{eq:epsilon-det}
\det (A) = \varepsilon^{i_0 \cdots i_{n-1}} A\indices{^{0}_{i_0}} \cdots A\indices{^{n-1}_{i_{n-1}}}.
\ee
Diagrammatically,
the determinant takes the form
\begin{center}
 \includegraphics[width=2.5\xxxscale]{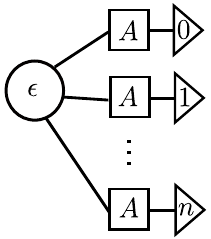}
\end{center}
or equivalently, the determinant is directly proportional to the expected value of the operator 
\begin{equation}
   \bra{\epsilon} \bigotimes_{i=1}^n A_i \ket{\epsilon} 
\end{equation}
with respect to the epsilon state in that dimension $n$. 
\end{remark}

\begin{lemma}[Invariance of epsilon under $SL(n, \C)$]
\label{lemma:epsilon-invariance}
We will consider the representation $L$ of the group $SL(n, \C)$ on
${\C^n}^{\otimes n}$, defined as
\be
L(S) \ket{\psi} := S^{\otimes n} \ket{\psi},
\ee
where $S \in SL(n, \C)$.
The order-$n$ epsilon state is invariant under this representation:
\begin{align*}
L(S) \ket{\varepsilon}
&= S\indices{^{i'}_i} S\indices{^{j'}_j} \cdots S\indices{^{k'}_k}
\varepsilon^{ij\cdots k} \ket{i' j' \cdots k'}
= \varepsilon^{i'j' \cdots k'} \ket{i' j' \cdots k'} = \ket{\varepsilon}.
\end{align*}
\end{lemma}

\begin{proof}
If an index value is repeated in the set~$\{i', j', \ldots, k'\}$,
the antisymmetry of epsilon dictates that the expression must
vanish. Hence the only nonvanishing index combinations are
permutations of the $n$ allowed index values.
Using~\eqref{eq:epsilon-det}, we can see that the permutation
$$
(i', j', \ldots, k') = (0, 1, \ldots, n-1)
$$
corresponds to
$\det(S)$, which equals~$1$ since $S \in SL(n, \C)$.
Invoking the antisymmetry of epsilon again we conclude that all even
permutations must also yield~$1$, and all odd permutations~$-1$, thus
recovering the definition of epsilon.

\begin{center}
 \includegraphics[width=4\xxxscale]{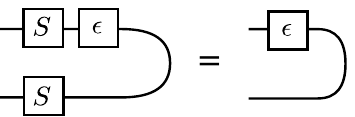}
\end{center}

However, we will then note that under wire duality we arrive at the relation 
\begin{center}
 \includegraphics[width=8\xxxscale]{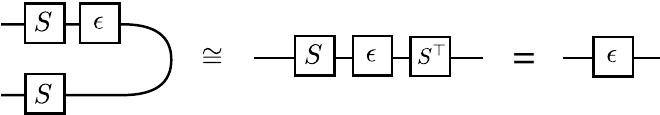}
\end{center} 
which says that $S\varepsilon S^\top = \varepsilon$.   So under wire duality, we find that $(S\otimes S) \psi_\varepsilon \cong S\varepsilon S^\top$.  
\end{proof}

%
%
%
%
%

\begin{remark}[Defining equations relating $\delta$ and $\epsilon$]
 It is a classical result of invariant theory that every identical relation satisfied by contractions of $\delta$'s and $\epsilon$'s can be built up from the following relations (see e.g.\ \cite{Penrose67}).  
 \be 
 \epsilon_{AB}=-\epsilon_{BA}; ~~~~~~~~~~ \epsilon^{AB} = -\epsilon^{BA}
 \ee
 \be 
 \delta^B_A\delta^C_B=\delta^C_A; ~~~~~~~~~ \delta^B_A\epsilon^{AC}; ~~~~~~~~ \delta^B_A\epsilon^{AC}=\epsilon^{BC}; ~~~~~~~~ \epsilon_{AB}\epsilon^{AC} = \delta^C_B
 \ee
 \be 
 \delta^A_A = 2 = \epsilon_{AB}\epsilon^{AB}
 \ee 
 Also, the identity, 
 \be 
 \epsilon_{AB}\epsilon_{CD} + \epsilon_{AD}\epsilon_{BC} + \epsilon_{AC}\epsilon_{DB} = 0
 \ee 
 holds, together with the equivalent identities obtained by raising indices of the above, with Levi-Civita symbols, e.g. 
 \be 
 \delta^B_A\delta^D_C - \delta^D_A\delta^B_C= \epsilon_{AC}\epsilon^{BD}
 \ee  
 We also note that when we assign values to $A, B=0,1$ we arrive at the following relations, between the indices and the tensor components as 
 \be 
 \epsilon_{AB}= A-B
 \ee 
 \be 
 \delta^A_B = 1 - (A-B)^2 
 \ee 
\end{remark}

\begin{remark}[Matrix determinant in dim $=2$]
The case of the determinant in two dimensions becomes
\be 
\sum \varepsilon_{ij}A_{i0}A_{j1} = \varepsilon_{01}A_{00}A_{11} +  \varepsilon_{10}A_{10}A_{01}
\ee 
where $\varepsilon_{10} = -1$ and $\varepsilon_{01} = 1$.  
In two dimensions, when all $i,j,m,n$ are in $\{0,1\}$, the following identities hold.  
\be 
\varepsilon_{ij} \varepsilon^{mn} = \delta^m_i \delta^n_j - \delta^n_i \delta^m_j
\ee 
\be 
\varepsilon_{ij} \varepsilon^{in} = \delta^n_j
\ee 
\be 
\varepsilon_{ij} \varepsilon^{ij} = 2
\ee   
\end{remark}

\begin{remark}[Epsilon symbol in two dimensions]
In two dimensions the Levi-Civita symbol is an element of the Pauli group:
$\varepsilon = i Y = -XZ = ZX$, that is,\footnote{Note that the diagram is not written using the convention 
from quantum circuits, where the order of composition of operations in an equation is reversed in the circuit.} 
\begin{center}
 \includegraphics[width=8\xxxscale]{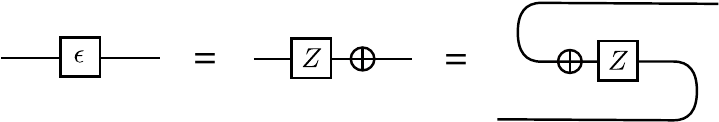}
\end{center}
where the last step of bending output to inputs and inputs to outputs
takes the transpose of a linear map.
\end{remark}

\paragraph*{SLOCC invariants.}
For pure two-qubit states, the single $SL$ invariant is the determinant
of the coefficient matrix $\alpha$ as 
\begin{equation}\label{eqn:K1}
K_1 = \sum \varepsilon_{i j}\varepsilon_{k l} \alpha^{i k} \alpha^{j l}  = 2\det (\psi^i_{~j}) = 2(\alpha^{00}\alpha^{11} - \alpha^{01}\alpha^{10})
\end{equation}
here $\varepsilon$ is the fully antisymmetric tensor on two indices defined as $\varepsilon_{00} = \varepsilon_{11} = 0$ and $\varepsilon_{01} = -\varepsilon_{10}$.  

\begin{remark}
 The norm is not a SLOCC invariant.  
\end{remark}

\begin{remark}[Algebraic independence of $K_1$: the fundamental ring]
 Any other polynomial invariant of the induced action of the special linear group is necessarily a polynomial in $K_1$. Since there is only one invariant in this special case, each member of this class could be considered a fundamental invariant.  Finding the fundamental invariants becomes complicated for higher dimensional systems.  
\end{remark}

With the epsilon identities in place, we are now able to consider the tensor structure of~$K_1$:
\begin{equation}\label{eqn:K_1} 
K_1 = \sum \varepsilon_{i j}\varepsilon_{k l} \alpha^{i k} \alpha^{j l} 
\end{equation}

\begin{lemma}[Diagrammatic contraction of $K_1$]
 The diagrammatic language contracts the tensor network for $K_1$ \eqref{eqn:K_1} to be twice the determinate of the matrix of state coefficients $(\bra{\Phi^+}\otimes \I)\ket{\psi}$ found from bending a wire on $\psi$. 
 \end{lemma}
 
 \begin{fullpage}
 \begin{proof}[Diagrammatic contraction of $K_1$]
In the following figure, (a) represents $K_1$ which simplifies to (d) using our previous results.  
\begin{center}
 \includegraphics[width=15\xxxscale]{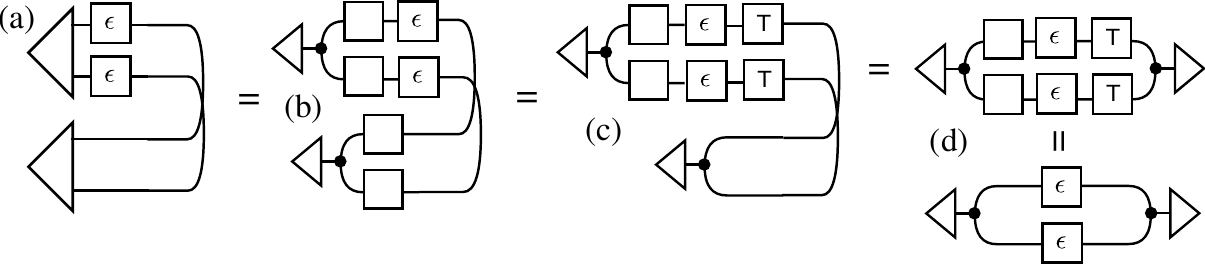}
\end{center} 
In (c) we have added $\top$ symbols in the unitary boxes, to denote transpose as found from sliding a box around a bent wire.  Figure (d) shows the $SU(2,\C)$ invariant of $\varepsilon$ as in Lemma \ref{lemma:epsilon-invariance}.   

To continue the analysis of this network, we will need to introduce a few more identities that we have proven in detail elsewhere \cite{CTNS, BB11}.  We need to then consider the stabilizer group of the \COPY-tensor.  This is an eight element group 
\be 
\{\I, X\otimes X\otimes X, -X\otimes Y\otimes Y, -Y\otimes X\otimes Y, -Y\otimes Y\otimes X, Z\otimes Z\otimes \I, Z\otimes \I\otimes  Z, \I\otimes  Z\otimes Z\}
\ee 
and then we consider the diagrammatic laws on a choice of generators 
\begin{center}
 \includegraphics[width=17\xxxscale]{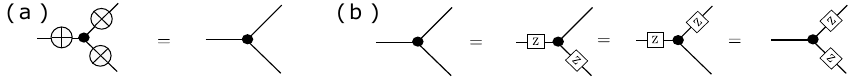}
\end{center} 
From these identities we then return to our equation for the determinant.  
\begin{center}
 \includegraphics[width=16\xxxscale]{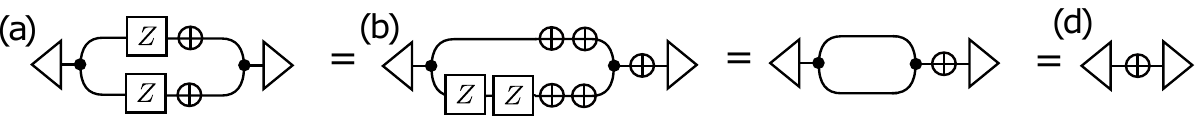}
\end{center} 
In (a) we express epsilon in using the Pauli algebra.  In (b) we apply the fact that the \COPY-tensor is stabilized
by the tensor product of three bit flip operators $X\otimes X\otimes
X$, as well as pair products of $Z$'s.  We combine this with the
fact that $X^2=\I = Z^2$ to show that the internal order-two tensors
all cancel out in (b).  This results in the network (d).  What remains
is one bit flip operator.  It can be contracted with an valence-one
triangular tensor, performing the map $\lambda_i \mapsto \lambda_j$
for $i\neq j$.  This leads to the inner product
\be 
(\lambda_0, \lambda_1).(\lambda_1, \lambda_0) = 2\lambda_1\lambda_0 = 2\det (\psi) 
\ee 
which matches precisely what we expect.  
\end{proof}
\end{fullpage}

 In Theorem \ref{theorem:SLOCC} we recalled the definition of SLOCC invariance, with respect to $GL$, the general linear group.  We then stated that there is a connection to the special linear group in Remark \ref{remark:SL}.  We make this connection precise in the following remark.  
 
\begin{remark}[Relating GL and SL in SLOCC]\label{remark:GL-vs-SL}
Let $G\in GL$.  We have that $\det(G)= k$, for some constant $k$ and
that $S\in SL$ has $\det(S)=1$.  Let $S = w G$, then letting 
\be 
w = \frac{1}{\det(G)^{\frac{1}{n}}}, ~ \Rightarrow ~ \det(S)=\det{wG}=1
\ee 
From this it follows that 
\be 
G \ket{\psi} = \det(G)^{\frac{1}{n}}S\ket{\psi}
\ee 
and so the transformations are related by constant factors (global scale factor).  Each polynomial SLOCC invariant is a homogeneous function with respect to this global scale factor.  The invariant changes only by a factor that does not depend on the state, but only on the transformation $G,S$.  
\end{remark}

\begin{fullpage}
\begin{example}[Calculated Values of the Invariants for Example States]
Here we calculate the values of each invariant for several common states of interest. The reduced \AND-state is found from recalling Equation \eqref{eqn:AND-map}.  
\begin{equation}
 \wedge_{ij}^{~~k}=(\ket{00}+\ket{01}+\ket{10})\bra{0} + \ket{11}\bra{1}
\end{equation}
We evaluate this Boolean tensor at $k=0$ as 
\begin{equation}
 \wedge_{ij}^{~~0}=\left(\ket{00}+\ket{01}+\ket{10})\bra{0} + \ket{11}\bra{1}\right)\ket{0} = \ket{00}+\ket{01}+\ket{10}
\end{equation}

\begin{center}
\begin{tabular}{r||c|c|c|c}
Type & State & $J_2$ & $K_1$ & $(\lambda_0, \lambda_1)$ \\\hline \hline
Product  & $\ket{0}\ket{0}$ & 1  & 0 & $(1,0)$\\\hline
Bell  & $\frac{1}{\sqrt 2}(\ket{00}+\ket{11})$ & $\frac{1}{4}$ & $\frac{1}{2}$ & $\frac{1}{\sqrt 2}(1,1)$ \\\hline
 \AND-type & $\frac{1}{\sqrt 3}(\ket{01}+\ket{10}+\ket{00})$  & $\frac{7}{9}$ & $\frac{1}{3}$ & $\lambda_\pm=\sqrt{\frac{1}{6} \left(3\pm\sqrt{5}\right)}$\\\hline 
General & $\sqrt{\lambda}\ket{00}+\sqrt{1-\lambda}\ket{11}$ & $2\lambda(\lambda-1)+1$ & $2\sqrt{\lambda(1-\lambda)} $ & $(\sqrt{\lambda}, \sqrt{1-\lambda})$ 
\end{tabular}
\end{center}
\end{example}
\end{fullpage}

\subsection*{Invariant composition law} 

Consider two initially non-interacting quantum wavefunctions $\ket{\psi}$ and $\ket{\phi}$.  If these states interact by a coupling unitary propagator $U$, the polynomial invariants of each state individually are not necessarily polynomial invariants, of the join system $U\ket{\phi}\otimes\ket{\psi}$.  

It is desirable to devise methods to compose invariants.  By this we mean that if we know the polynomial invariants of $\ket{\psi}$, $\ket{\phi}$ and $U$, together with a tensor network defining their composition, from this information alone, how can we determine the resulting invariants of the joint system?  In the present section, we solve this problem for the case of bipartite systems of equal dimension.  

\paragraph{Composition of the invariant $K_1$}
We will consider a quantum state $\psi\in \7H\otimes \7H$.  We will act on this state locally on each leg, with invertible linear maps.  A possible scenario is depicted as follows.  On the left, we have a state $\in \7H\otimes \7H$ factored using the diagrammatic SVD.  
We then act on this state with invertible maps $S_1$ and $S_2$.  
\begin{center}
 \includegraphics[width=2.5\xxxscale]{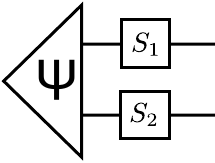}
\end{center}

The invariant $K_1$ is equal to twice the determinate of the coefficient matrix, found from bending either wire back, on the tensor representing the state.  The invariant $K_1$ is equal to zero iff $\det(\psi)=0$.  This is true iff $\psi$ is separable, thereby partitioning the SLOCC class into a disjoint union: entangled vs. not. We can calculate $K_1$ for not only the state, but also for the maps $S_1$ and $S_2$.  Given the quantities 
\be 
K_1(\psi),~K_1(S_1), ~ K_1(S_2)
\ee 
we can deduce from the standard properties of the determinant, the value of the invariant found from acting on $\psi$ with the map $S_1\otimes S_2$ as 
\be 
K_1(S_1\otimes S_2 \psi)= K_1(\psi)K_1(S_1)K_1(S_2)
\ee
Such a scenario extends to the case of qudits, and the value of the invariants evaluates to the product of singular values of all parties in the contracted network.

\section{Generating Invariants for General Density Operators}
A method to systematically generate polynomial invariants of density
operators acted on by the local unitary group exists and was shown to
me by Markus Grassl \cite{entinv2}.  Here we will cast this method into the Penrose tensor calculus.  The method generates a complete basis of monomials that are necessarily
invariants of the local unitary group acting on a density matrix.  A polynomial invariant 
could be expressed in this complete basis.  Although we can generate a complete basis of invariants, except in rare cases, finding the minimal collection of polynomial invariants is computationally difficult.  A utility of generating this basis, comes from the fact that the tensor networks we consider can be used to calculate any properties that are invariant under the action of the local unitary group.  This includes R\'{e}nyi entropies and the entanglement spectrum.  

\paragraph{The invariant basis.}  Consider a density 
operator $\rho\in \7H\rightarrow \7H$.  We can write 
this in a basis and consider acting on $\rho$ by a given a group $G$.  
A polynomial in 
\be 
f(x_0, x_1, x_2, x_3, \overline{x}_0, \overline{x}_1, \overline{x}_2, \overline{x}_3)=:f(\1x) 
\ee 
is invariant under $G$ iff $\forall g\in G$ 
\be 
f^g(\textbf x) = f(\textbf x) 
\ee 
where the notation $f^g(\textbf x)$ means 
\be 
f\small{\left(\left[g(\textbf x)^\top \right]^\top\right)} = f\left[(\textbf x)g^\top \right]
\ee 
so we let $g$ act on $\textbf x$ and then let $f$ act --- $f$ is constant under $G$ iff $f$ is invariant.  

The question we are considering is how to generate 
a complete monomial basis where each basis element is invariant under $G$.  This would then imply that any holomorphic function, can be written in this basis.  In fact, such a function would (i) necessarily be invariant under $G$ and (ii) be a function of the density operator $\rho$.    


First consider the network 
\begin{center}
 \includegraphics[width=2.5\xxxscale]{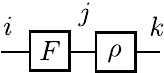}
\end{center}
representing the equation $F^i_{~j}\rho^j_{~k}$ then Tr$(F\rho)=\delta_{ik}F^i_{~j}\rho^j_{~k}$ that is 
\begin{center}
 \includegraphics[width=4\xxxscale]{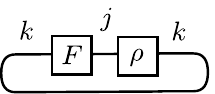}
\end{center}
$= F^i_{~j}\rho^j_{~i}$ If we were considering 
\be 
\rho = x_0\ket{0}\bra{0} + x_1\ket{0}\bra{1} + x_2 \ket{1}\bra{0} + x_3 \ket{1}\bra{1}
\ee 
then this trace would evaluate to 
\be 
F_{00}x_0 +F_{01}x_1 + F_{10}x_2 + F_{11}x_3
\ee 
and we can pick an valence-four $F$ such to expand the second order monomials as Tr$(F\rho^{\otimes 2}) = F^{il}_{jk}\rho^j_i\rho^l_k$ which translates to the tensor network as 
\begin{center}
 \includegraphics[width=4\xxxscale]{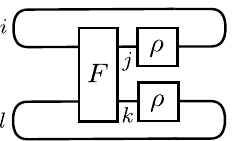}
\end{center}
This procedure carries on in this fashion, arriving at a monomial of order $n$ as 
\be 
\Tr(F\rho^{\otimes n}) = F^{il\cdots q}_{jk\cdots r}\rho^j_i\rho^l_k\cdots \rho^q_r
\ee 
The question then changes.  We have generated a complete basis in the coefficients of the density operator, with coefficients in $F^{il\cdots q}_{jk\cdots r}$.  We then will act on 
the density operators by elements $g\in G$ 
\be 
\rho' = g \rho g^{-1}
\ee 
and search for $F^{il\cdots q}_{jk\cdots r}$ that satisfy 
\begin{equation}\label{eqn:density-invariant} 
\Tr(F g^{\otimes n} \rho^{\otimes n}g^{\otimes n -1}) = \Tr(F \rho^{\otimes n})
\end{equation}  
Equation \eqref{eqn:density-invariant} is satisfied by monomial invariants.  From this it follows that 
\begin{equation} 
\Tr(F g^{\otimes n} \rho^{\otimes n}g^{-1\otimes n }) = \Tr(\rho^{-1\otimes n }F g^{\otimes n} g^{\otimes n})
\end{equation}  
$\forall \rho$ and so 
\be 
[F, g^{\otimes n}]=0, ~\forall g\in G
\ee 
We are then faced with finding matrices $F$ that commute with $g^{\otimes n}$ for each $g\in G$. This problem was solved in a different setting around 1937, and the solution is roughly stated in the following theorem.  

\begin{theorem}[R. Brauer, 1937]
 The algebra of matrices that commute with each $U^{\otimes n}$ for $U \in  U(n)$ is
generated by a certain representation of the permutation group.
\end{theorem}

The permutation group has a well known and evident diagrammatic form. In (a) we show the elements of the permutation group on one system.  In (b) we show the elements on two systems.  (c) Illustrates the permutation group on three elements, $S_3$ of order $6$.\sn{Of possible related interest, is the diagrammatic presentation of the Temperly-Lieb algebra \cite{TL07}.} 
\begin{center}
 \includegraphics[width=14\xxxscale]{symmetric-group}
\end{center}
We then carry on to evaluate 
\be 
\Tr(S_k g^{\otimes n} \rho^{\otimes n}\rho^{-1\otimes n})
\ee 
to form invariant monomials.  We consider the case where we wish to generate an order one invariant 
in (a) below.  Order two invariants are found below in (b).  We note that the first invariant is simply the square of the invariant 
found in (a).  The second however, is independent.  An example of an invariant of order three is given in (c). The situation continues.  The basis is finitely generated form the Cayley-Hamilton Theorem.  We have that every density operator satisfies its own characteristic polynomial, giving an $n$th order polynomial equation in $\rho$ that vanishes.  
\begin{center}
 \includegraphics[width=10\xxxscale]{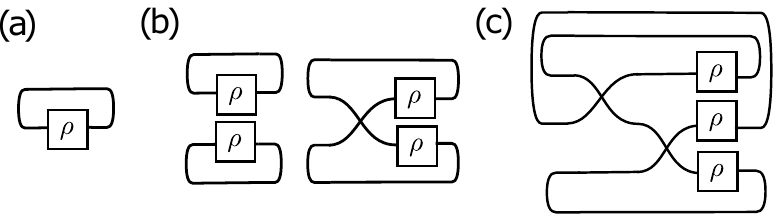}
\end{center}

\paragraph{The general form.}  We consider the monomial invariant generated by tracing over the contraction with the SWAP-operator in (a) below.  Here we have acted on $\rho$ with some unitary operation $U$, and give a diagrammatic proof that the network contracts to a quantity that is invariant, under unitary transformations of $\rho$.  

To see this, we slide $U, U^\dagger$ around the bends, resulting in (b).  The diagram reduces to (c), showing that the invariant evaluates to $\Tr(\rho^2)$.  
\begin{center}
 \includegraphics[width=14\xxxscale]{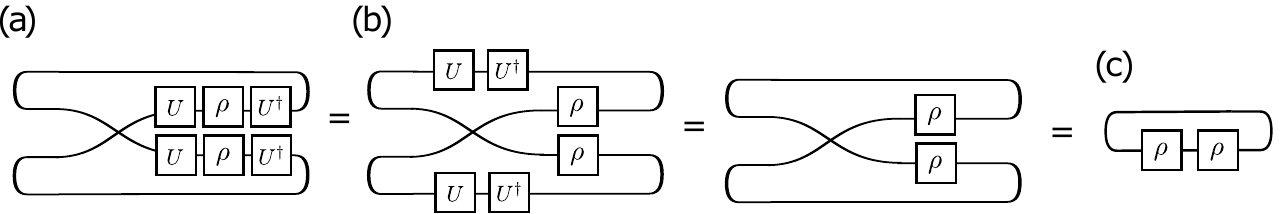}
\end{center}

\begin{remark}[Generating a complete basis of invariants by tensor contraction]
To generate a complete basis, we have to consider operations like SWAP that permute all elements to different elements. These are called complete permutations.   
\end{remark}

\begin{lemma}[Diagrammatic proof of the basis]
 Every complete permutation on $n$ wires contracted with $n$ tensor copies of $\rho$ and traced over, contracts to the elementary diagram giving the trace of $\rho^n$.  
\end{lemma}

The unitary invariants of any contraction with a complete permutation operator is readily shown diagrammatically.  Any complete permutation is can be shown to contract to (a) below.  It then possible to understand a complete basis of invariants by evaluating the quantity  as shown in (c) below.  
\begin{center}
 \includegraphics[width=14\xxxscale]{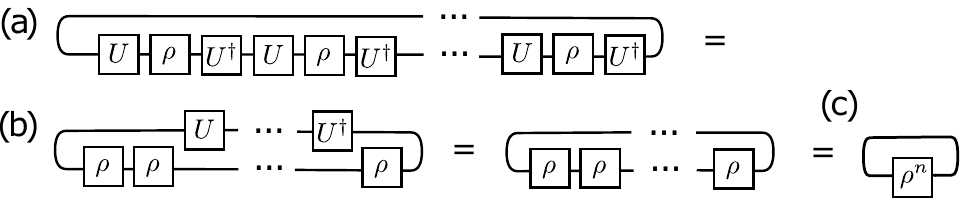}
\end{center}
Note that in (b), we have pulled both of the unitaries around bends.  We should have written $U^\top$ and $U^{\dagger\top}$ but we have omitted the transpose symbol $\top$ for ease of producing the figures.  These unitary maps still contract to identity as is quickly verified.

\begin{example}[Trace invariants for qubit density operators]
We will consider generating the fundamental invariants for a qubit density operator.  
The first invariant follows from tracing over the identity operator (the only permutation group element on a single system) 
\be
I_1 =\Tr(\rho)
\ee 
is the constant norm.  The second is 
\be
I_2 = \Tr(\rho^2)
\ee
which turns out to be precisely what is known in other areas of the literature as the purity.  In terms of the singular values we have 
\be
I_1=\lambda^+ + \lambda^- = 1
\ee
which is the norm and 
\be
I_2=\lambda_+^2 + \lambda_-^2 
\ee
For the purpose of generating polynomial invariants, it is enough to stop with $I_1$ and $I_2$.  In fact, from the Cayley-Hamilton Theorem, we have that there is a polynomial in second order in $\rho$, that vanishes identically.  In other words, constants $a,b,c$ exist such that 
\be 
a \rho^2 + b\rho + c \I = 0 
\ee 
so higher powers of $\rho$ can be expressed in terms of $I_1$ and $I_2$ through this relation.  These invariants are indeed algebraically independent and complete, meaning any other polynomial invariant can be expressed in $\{\R, +, \cdot, I_1, I_2\}$.  For instance, 
\be
\text{Det}(\rho) = \frac{1}{2}\left(\Tr(\rho)^2 - \Tr(\rho^2) \right)= \frac{1}{2}\left(I_1 - I_2\right)= 1 - (c^2+a^2+b^2)= \lambda_0\lambda_1
\ee
\end{example}

\subsection*{Topological equivalence of pure and mixed state invariants}


We have considered pure state invariants as well as mixed state invariants, under 
the action of the local unitary group.  We will consider the relation between these two approaches, when cast into our diagrammatic framework.  \marginnote{{\it ``Any effect, constant, theorem or equation named after Professor X was first discovered by Professor Y, for some value of Y not equal to X.''} ---
John C.~Baez}

Let us return to the expression for $J_2$ 
\be 
J_2 = \sum \alpha^{i j}\alpha^{k l}\overline{\alpha}_{i l}\overline{\alpha}_{kj}
\ee 
We will then write 
\begin{equation}\label{eqn:J2-vs-B} 
\alpha^{i j}\overline{\alpha}_{il} =: B^j_{~l}
\end{equation} 
which is given graphically below.  
\begin{center}
 \includegraphics[width=6\xxxscale]{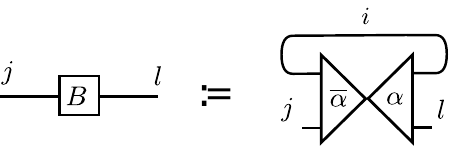}
\end{center} 
We also fix $B^l_{~j} := \alpha^{k l}\overline{\alpha}_{k j}$.  It now follows that 
the expression for $J_2$ becomes 
\be 
J_2 = \sum  B^l_{~j}B^j_{~l} = \Tr(B^2) 
\ee 
In (b) we see that this is found graphically when acting on two reduced states with a SWAP operation, and then tracing over the result.  
\begin{center}
 \includegraphics[width=12\xxxscale]{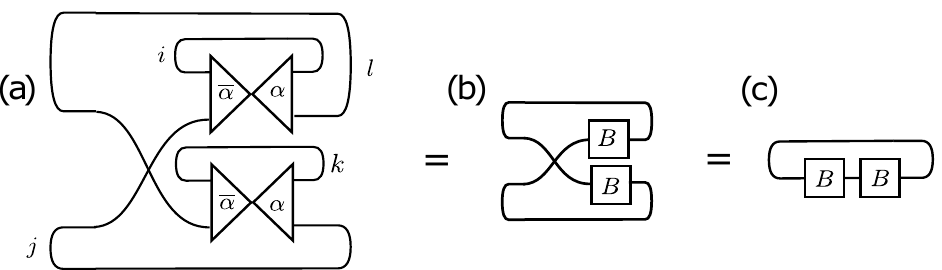}
\end{center}

We will go on to show that this invariant is in fact identical to the pure state invariant.  
From application of the graphical identity (Penrose's graphical representation of a density state is on the right) 
\begin{center}
 \includegraphics[width=4\xxxscale]{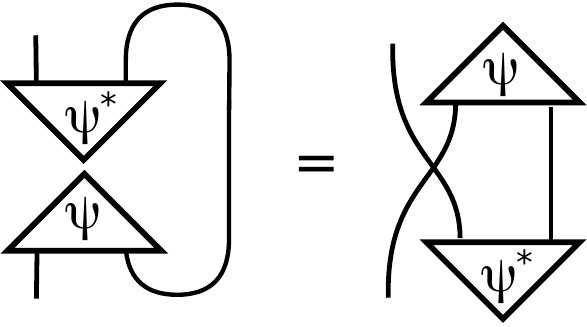}
\end{center} 
it now follows by applying this identity to (a) we arrive at the expression for $J_2$ we considered in Section \ref{sec:qubit-invariants}.  
\begin{center}
 \includegraphics[width=10\xxxscale]{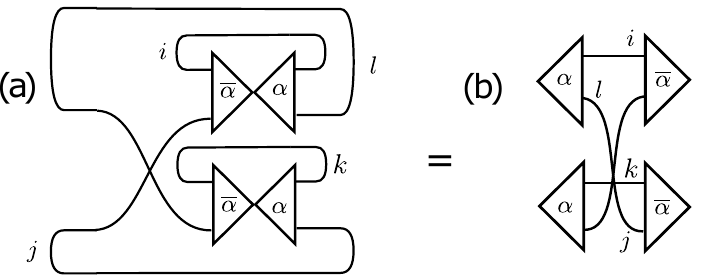}
\end{center}

\section{Some Symmetries of \texorpdfstring{$\rho$}{}} 
We have expressed two Theorems (\ref{theorem:injection} and \ref{theorem:symmetry-class}) which relate density operators to quantum states precisely.  This might be thought of as a type of symmetry.  The other type of symmetry we have considered are invariant polynomials in the coefficients of a density operator.  These symmetries are general, as they are found for arbitrary density operators.    In the present section, we will remind the reader of other types of symmetry that are sometimes considered in quantum theory.  Symmetry plays a key role in the lectures that follow, so here is just an introduction.  In particular, symmetry will be considered in detail in Lecture III.  To begin, let us recall the typical symmetry that is often considered in quantum physics.  

\begin{example}[Group Symmetry of $\rho$]
 We will consider a general density operator $\rho$ and look for $V\in U(d)$ such that 
 \be 
 \rho = V\rho V^\dagger
 \ee 
 this implies that $[\rho, V]=0$ and we arrive at a basis for $G$ by noting the unitary operators that commute with $\rho$.  
 That is, $\{\ket{\lambda_i}\}_i$ such that $\rho = \sum_i p_i\ket{\lambda_i}\bra{\lambda_i}$.  It then follows from $VV^\dagger = \I$ that 
 every $V\in G$ can be written as 
 \be 
 V= \sum_i e^{i\theta_i}\ket{\lambda_i}\bra{\lambda_i}
 \ee 
\end{example}

In quantum computer science another symmetry is often considered.  This is the symmetry found by so called, \emph{stabilizer states}.  
We will discuss this in detail in Lecture III and also cover a bit of it in Lecture IV.  A quantum gate~$U$ which is a tensor product of Pauli operators is said to stabilize the state~$\ket{\psi}$ iff $U \ket{\psi} =
\ket{\psi}$. All the gates stabilizing a given state form a group.  

\begin{example}[Pauli stabilizers]
Here are some standard examples of states stabilized by the Pauli-group
\begin{itemize}\addtolength{\itemsep}{-0.25\baselineskip}
 \item[(i)] \textbf{Pauli-X}: $\sigma^x$ stabilizes $\ket{+}=\ket{0}+\ket{1}$
and $-\sigma^x$ stabilizes $\ket{-}=\ket{0}-\ket{1}$
 \item[(ii)] \textbf{Pauli-Y}: $\sigma^y$ stabilizes
$\ket{y_+}=\ket{0}+i\ket{1}$ and $-\sigma^y$ stabilizes
$\ket{y_-}=\ket{0}-i\ket{1}$
 \item[(iii)] \textbf{Pauli-Z}: $\sigma^z$ stabilizes $\ket{0}$ and $-\sigma^z$
stabilizes $\ket{1}$
\end{itemize}
\end{example}

\begin{remark}[Gottesman-Knill Theorem] A
graphical rewrite proof (by bounding the number of rewrites) of the
Gottesman-Knill theorem follows by considering the action of the black and
plus dots on $\sigma^z$ and $\sigma^x$.  We will set proving this as an exercise in a later Lecture.  
\end{remark}

In addition to these two, one also finds the class of so called, \emph{symmetric states} (or operators for that matter). These are characterized as follows.  

\begin{example}[Invariance under $S_k$]
 The other form of symmetry that is considered, is invariance under the symmetric group.  Diagrammatically this amounts to braiding wires (where here the order of the braids is not relevant).  
\end{example}

We can relate symmetry in density operators and symmetry in states as follows.  
\subsubsection*{Further symmetries of $\rho$.}  Now using these ideas, we can consider other symmetries of $\rho$.  We want to find solutions of 
\be 
\rho = V\rho V^\dagger 
\ee 
We will then consider 
\be 
\rho_\varepsilon' = a\I+ b\varepsilon 
\ee 
We then note that $b\in \7 R$ and that $\varepsilon = i Y$.  So we arrive at the density operator 
\be 
\rho_\varepsilon = a\I+ bY
\ee 
and will search for solutions to the equation 
\be 
\rho_\varepsilon = V\rho_\varepsilon V^\dagger = a V V^\dagger - b i V\varepsilon V^\dagger
\ee 
and note that under wire duality, 
\be
\varepsilon = S\varepsilon S^\top 
\ee 
(see Lemma \ref{lemma:epsilon-invariance}) and so $\rho_\varepsilon$ is invariant for 
\be 
V\in SU(n, \R)
\ee 
as those matrices satisfy $SS^\top=\I$.  We can go further however.  In fact, the following is easy to verify. 

\begin{lemma}[Isomorphism between projectors and self-adjoint unitary maps]
It is readily verified that 
\be 
 U = \I - 2P
 \ee 
 is unitary with $P^2=P$.  The equation is then solved for $P$ in terms of $U$.  
\end{lemma}
We then see that each projector gives rise to a symmetry of $\rho_\varepsilon$, as does the real orthogonal subgroup of the special unitary group.

\section{Polynomial Invariants of Symmetric States}

A symmetric multipartite state or fully symmetric tensor carries the trivial representation of the symmetric group and hence is invariant under any permutation of the parties.  Examples are the three qubit
\GHZ-state $\ket{\text{\sf GHZ}} = \ket{000} +
\ket{111}$ and the three qubit W-state $\ket{\text{\sf W}} =
\ket{001} + \ket{010} + \ket{100}$. 

\begin{remark}[Symmetric basis]
A general $n$ qubit quantum system has the $2^n$ 
orthonormal basis vectors $\{ | 00 \dots 00 \rangle , | 00 \dots 01 \rangle , \dots , | 11 \dots 11 \rangle \}$. 
For the subspace of \emph{symmetric} $n$ qubit states, an orthonormal basis is given by 
the $n+1$ \emph{symmetric} basis states $\{ | S_{0} \rangle , | S_{1} \rangle , \dots , | S_{n} \rangle \}$. 
They are defined as
\begin{equation}
 \ket{S_{k}} = {\binom{n}{k}}^{- \frac{1}{2}} \sum_{\text{perm}}
\; \underbrace{ | 0 \rangle | 0 \rangle \cdots | 0 \rangle }_{n-k}
\underbrace{ | 1 \rangle | 1 \rangle \cdots | 1 \rangle }_{k}
\end{equation}
and called Dicke states.  We can therefore write $\ket{\text{\sf W}} = \ket{{S}_{1}}$ and $\ket{\text{\sf GHZ}} =
\frac{1}{\sqrt{2}} ( \ket{{S}_{0}} + \ket{{S}_{3}})$.
\end{remark}

\begin{myexercise}[Boolean symmetric states]
 Provide the number of symmetric Boolean quantum states.  
\end{myexercise}

\paragraph{A polynomial basis.}
Symmetric qubit states and symmetric polynomials in two variables are isomorphic as vector spaces.  The known mapping is accomplished as follows.  
\begin{equation}
 \ket{0} \leftrightarrow x, ~~~~~~~~~~ \ket{1}\leftrightarrow  y
\end{equation}

\begin{remark}[Creation and annihilation operators]
 If we consider a symmetric polynomial basis in $x^n$, we can define two operators
 \be 
 a^+ : = x
 \ee 
 \be 
 a^- := \frac{\partial}{\partial x} := \partial_x
 \ee 
 From this we can define an inner product on our space.  For instance, we 
 can calculate the norm of $x^n$ as 
 \be 
 \braket{x^n}{x^n}= \braket{x^n}{a^+ x^{n-1}} = \braket{a^- x^n}{x^{n-1}} = n! 
 \ee 
 This concept is readily extended to the case of binary forms in $x$ and $y$ when considering creation and annihilation operators for each variable $x$ and $y$ separately.  
\end{remark}

\begin{example}[Polynomials for the \GHZ-class]
The $n$-partite GHZ state $\ket{0}^{\otimes n}+\ket{1}^{\otimes n}$ 
has homogeneous polynomial $f(x,y)=x^n+y^n$.
\end{example}

\begin{example}[Polynomials corresponding to common quantum states]
The three-qubit W state
$\ket{W_3}=\ket{001}+\ket{010}+\ket{100}$ is isomorphic to the
monomial $3x^2 y$. Furthermore, two appropriately braided copies of $\ket{W_3}$ read 
$\ket{W_{3,3}}=(\ket{003}+\ket{030}+\ket{300})+(\ket{012}+\ket{021}+\ket{102}+\ket{120}+\ket{201}+\ket{210})$ which is given diagrammatically as 
\begin{center}
\includegraphics[width=6\xxxscale]{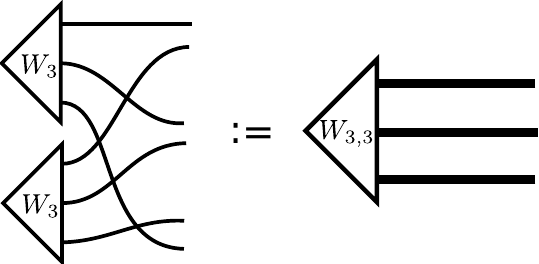}
\end{center}
and becomes a sum of two Dicke states having corresponding homogeneous
polynomials $x_0^2 x_3$ and $x_0x_1x_2$. It is isomorphic to the
homogeneous polynomial $3x_0^2 x_3+6x_0x_1x_2$. 
\end{example}

It turns out that the invariant theory of binary forms is well studied in the literature.  An invariant of a binary form is a polynomial in the coefficients of a binary form in two variables $x$ and $y$ that remains invariant under invertible transformations of the variables $x$ and $y$.

\section*{Invariants of forms} 

\begin{example}
 Consider 
 \be 
 ax^2 + 2bx + c = 0
 \ee 
 where the discriminant is 
 \be 
 \Delta = b^2 - ac
 \ee 
 If $\Delta = 0$ we have a double root, and if $\Delta < 0$ the roots are complex conjugate.  Now if we consider the affine change of variables, where $\alpha \neq 0$ 
 \be 
 x' = \alpha x + \beta 
 \ee 
Upon calculating the discriminant for the polynomial in the new variables we arrive at 
\be 
\Delta' = \frac{1}{\alpha^2} \Delta
\ee 
which is the same as $\Delta$ in the original polynomial, up to a multiplicative factor that depends only on the transformation.  The properties of the roots remain unchanged under such transformations.  
\end{example}

\section{From Binary Forms to Qubit States} 

We will now consider our subject of interest.  Forms in two variables, called quadratic forms, or binary forms.  
For example, the general degree two binary form is given as  
\be 
Q(x,y) = a_2 x^2 + 2 a_1 xy + a_0 y^2
\ee 
which gives rise to the quantum state 
\be 
\psi_Q = a_2\ket{00} + a_1 (\ket{01}+\ket{10}) + a_0 \ket{11}
\ee 

\begin{definition}[General linear invertible transformation]
 A general linear invertible transformation of two variables takes the form 
 \be 
 x' = \alpha x + \beta y
 \ee 
 \be 
 y' = \gamma x + \delta y
 \ee 
 where 
 \be 
 \alpha \delta - \beta \gamma \neq 0
 \ee 
 When a variable is transformed, we get an induced transformation of the coefficients.  It is important to note that this transformation does not change the degree of a given polynomial.  
\end{definition}

\begin{definition}[Invariant of a binary form \cite{oliver}]
 An invariant of a binary form $Q(x,y)$ is a function $I(a) = I(a_0, ..., a_1)$ depending on the coefficients 
 $a= (a_0, ..., a_n)$ of the form, which, up to a determinantal factor, does not change under the general linear 
 transformation 
 \be 
 I(a) = (\alpha \delta - \beta \gamma)^k I(a')
 \ee 
 where $a' = (a_0', ..., a_n')$ are the transformed coefficients.  
\end{definition}

An integer $k$ is called the weight of the invariant and we are concerned with the case that $I(a)$ is a polynomial.  
It is then called, a polynomial invariant.  

\begin{remark}
 Classical invariant theory has provided us elegant solutions to the following questions \cite{hilbert, oliver}.  
 \begin{itemize}
  \item[(i)] How many independent polynomial invariants are there of a given degree?
  \item[(ii)] What do they tell us about the form? 
  \item[(iii)] Can we find a basis for the polynomial invariants? 
 \end{itemize}
\end{remark}

We are missing an important piece to our puzzle.  These are the covariants, which are polynomials.  

\begin{definition}[Covariant \cite{oliver}]
 A covariant of weight $k$ of a binary from $Q$ of degree $n$ is a function 
 \be 
 J(a, x, y)
 \ee 
 depending both on the coefficients $a_i$ and on the independent variables $(x,y)$ which up to a determinantal factor, in unchanged under general linear transformations 
 \be 
 J(a, x, y) =  (\alpha \delta - \beta \gamma)^k J'(a', x', y')
 \ee 
\end{definition}


\section*{Three qubit example} 
The cubic binary form is given as 
\be 
Q(x,y) = a_3 x^3 + 3 a_2 x^2y + 3 a_1 xy^2 + a_0 y^3
\ee 
It is known that there is just one fundamental invariant, the discriminant of the cubic.
\be 
\Delta = a_0^2 a_3^2 - 6 a_0 a_1 a_2 a_3 + 4 a_0 a_2^3 - 3a_1^2 a_2^2 + 4 a_1^3 a_3
\ee 

\begin{remark}
 $\Delta = 0$ iff $Q$ has a double or triple root.  
\end{remark}

\begin{myexercise}
 Consider the discriminant of the cubic.
\be 
\Delta = a_0^2 a_3^2 - 6 a_0 a_1 a_2 a_3 + 4 a_0 a_2^3 - 3a_1^2 a_2^2 + 4 a_1^3 a_3
\ee 
and find maximum and minimum values when $\forall i, a_i\in \{0,1\}$. 
\end{myexercise}

The obvious covariant is the form itself.  Another important covariant is the Hessian 
\be 
H = Q_{xx}Q_{yy} - Q_{xy}^2
\ee 
where we use subscripts to denote partial derivatives.  
For the binary cubic, the Hessian becomes 
\be 
\frac{1}{36} H = (a_1a_3 - a_2)x^2 + (a_0 a_3 - a_1 a_2)xy + (a_0a_2 - a_1^2)y^2 
\ee 

\begin{theorem}[vanishing Hessian \cite{oliver}]
 A binary form $Q(x,y)$ has vanishing Hessian, $H=0$ iff $Q(x,y)= (cx+dy)^n$, that is iff $Q(x,y)$ is the nth power of a linear form.  
\end{theorem}

\begin{remark}[Applications to quantum entanglement]
 It is clear that a linear form 
 \be 
 Q(x,y)= (cx+dy)^n
 \ee 
 gives rise to a factorisable quantum state as 
 \be 
 \psi_Q = (c\ket{0}+ d\ket{1})^n
 \ee 
 So clearly, the Hessian corresponding to a quantum state, determines if the state can be written as a local product of 
 operators.  Other covariants give additional information. 
 
%

\end{remark}

\begin{remark}[Sylvester's tables]
In the mid 1800s Sylvester produced tables\marginnote{{\it ``The object of pure physics is the unfolding of the laws of the intelligible world; the object of pure mathematics that of unfolding the laws of human intelligence.''} --- James Joseph Sylvester.} predicting numbers of invariants and covariants of a given degree.  The table should only be trusted for covariants up to degree 6 and invariants up to degree 8.  
 \begin{center}
\begin{tabular}{c||cccccccccc}
\hline
degree     & 2 & 3 & 4 & 5 & 6 & 7 & 8 & 9 & 10 & 12 \\\hline \hline 
invariants & 1 & 1 & 2 & 4 & 5 & 26 (30) & 9 & 89 & 104 & 109 \\
covariants & 2 & 4 & 5 & 23 & 26 & 124 (130) & 69 & 415 & 475 & 947 \\\hline 
\end{tabular}
\end{center}
\end{remark}

\section{Problems} 

\subsection*{Anti-linear operators} 

\begin{remark}
Tensor contraction is linear.  Tensor networks encapsulate multilinear algebra.  What about antilinear operators or non-linear operators (such as those found in tensor flow software)?  As a prelude, let us first consider some properties of antilinear operators.  Readers should first show that complex conjugate is antilinear. 
\end{remark}

\begin{definition}[Qubit density operator]
Consider a density operator of a single spin (a.k.a.~qubit) as
\begin{equation}
\rho = \frac{1}{2}(\I + a.\sigma)
\end{equation}
where $a=(a^1, a^2, a^3)\in [-1,1]^{\otimes 3}\subset \mathbb{R}^3$, the Pauli vector $\sigma=(\sigma_1, \sigma_2, \sigma_3) = (X, Y, Z)$ and `.' signifies a sum over the entry wise product, viz.~$a.\sigma = \sum_i a^i \sigma_i$.  Here $\Tr \rho = 1$, $\rho \geq 0$, $\rho=\rho^\dagger$, $\sum_{i\in\{1,2,3\}} |a^i|\leq 1$ (the vector $a$ also satisfies $||a||_2\leq 1$). 
\end{definition}
  
\begin{definition}[T-action]\label{def:taction}
For the polarization vector $\sigma$, we define $T$ through the action $T\sigma T^\dagger = - \sigma$. 
\end{definition}

\begin{myexercise}
Prove that $T$ as defined in \ref{def:taction} is not unitary.  (Hint consider preservation of the definition of the Pauli-algebra under unitary group homomorphism). 
\end{myexercise}

\begin{myexercise}
Show that $\prod_{i=1, 2, 3}\Tr \sigma_i \rho = a_1a_2a_3$.  
\end{myexercise}

We will consider the form of $T=UK$ where $U$ is unitary and $K$ is complex conjugation.

\begin{myexercise}
For a single qubit, let $T = e^{-\imath \frac{\pi}{2} Y}K$ and  show  that this definition satisfies the definition of the T-reversal operator from \ref{def:taction}. 
\end{myexercise}

\begin{myexercise} Show that 
$[\rho, T]=0$ for $T = e^{-\imath \frac{\pi}{2} Y}K$ iff $\rho = \frac{1}{2}\I$. 
\end{myexercise}

\begin{myexercise}
Show that $[\rho, K]=0$ iff $a_2=0$, in which case $T\sim Y$.\footnote{$\sim$ denotes equality up to an equivalency class defined by  unit norm complex numbers.}
\end{myexercise}

\begin{remark}[Dagger of qubit maps]
Let us write the general local qubit map, using the following
parameterization which can express any single-qubit map. 
\be 
U(2,\7C) = U(1, \7C)\times SU(2, \7C) 
\ee 
so up to a negligible global phase, every local unitary map on qubits is
effectively given by a special unitary map. Consider then, 
\be 
S = \alpha\ketbra{0}{0} - \overline{\beta}\ketbra{0}{1} + \beta
\ketbra{1}{0} + \overline{\alpha}\ketbra{1}{1}. 
\ee 
It is a straight forward calculation to verify that 
\be 
[\epsilon^\top S \epsilon]^\top = S^\dagger. 
\ee
Which has the diagrammatic form as: 
 \begin{center}
 \includegraphics[width=0.25\textwidth]{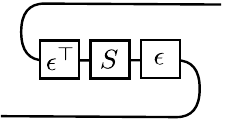} 
\end{center}
\end{remark}

 \part{Tensor Networks for Open Quantum System}

\label{sec:cpmaps}
In this section we recall several common mathematical descriptions for completely-positive trace-preserving maps, and show how several key properties may be captured graphically using the diagrammatic notation. The representations we will consider are the Kraus (or operator-sum) representation, the system-environment (or Stinespring) representation, the Liouville superoperator description based on vectorization of matrices, and the Choi-matrix or dynamical matrix description based on the Choi-Jamio{\l}kowski isomorphism. We will also describe the often used process matrix (or $\chi$-matrix) representation and show how this can be considered as a change of basis of the Choi-matrix. Following this, in \S~\ref{sec:trans} we will show how our framework enables one to freely transform between these representations.  A standard reference for open quantum systems applied to quantum information processing is \cite{BP}. The present chapter on open systems presented in the language of tensor networks follows primarily \cite{2011arXiv1111.6950W}.  Interested readers might consider a formulation of completely positive (for open quantum systems) maps using category theory \cite{SELINGER2007139}.  For historical context, we also mention that the early formulation of {\it atemporal circuits} in \cite{Atemp06} indeed partially considered diagrammatic properties of completely positive maps.  Matrix product state methods applied to open systems and density matrices may also be of interest \cite{Schollw, Chan2016, open1D}.

\begin{center}
\includegraphics[width=13.5\xxxscale]{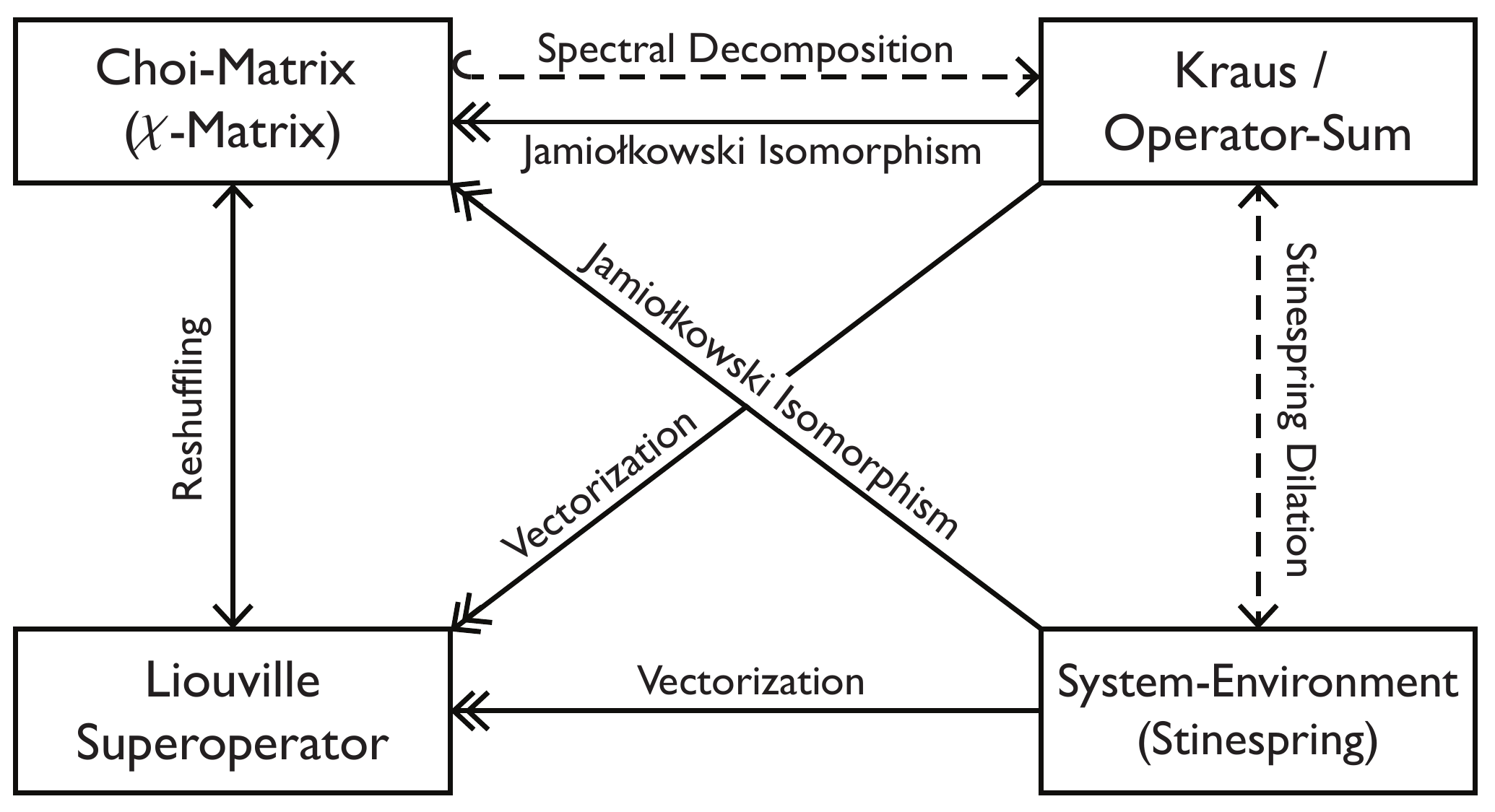}
\label{fig:cpreps}
\end{center}

\begin{remark}
In this chapter we will use the notation that $\2 X, \2 Y, \2 Z$ are finite-dimensional complex Hilbert spaces,  
 
\begin{itemize}
\item[1.] 
${\mathcal L}(\2X,\2Y)$ is the space of bounded linear operators $A: \2 X\rightarrow \2 Y$ (with ${\mathcal L}(\2X)\equiv {\mathcal L}(\2X,\2X)$),
\item[2.]  
$T(\2X,\2Y)$ is the space of operator maps $\2 E:{\mathcal L}(\2X)\rightarrow {\mathcal L}(\2Y)$  (with $T(\2X)\equiv T(\2X,\2X)$),
 \item 
 and $C(\2X,\2Y)$ is the space of operator maps $\2 E$ which are CP.
\end{itemize}

\end{remark} 

\section{Kraus / Operator-Sum Representation}
\label{sec:kraus}

The first representation of CPTP-maps we cast into our framework is the \emph{Kraus}~\cite{Kraus1983} or \emph{operator-sum}~\cite{NC} representation. This representation is particularly useful in phenomenological models of noise in quantum systems.
\begin{theorem}
Kraus's theorem states that a linear map $\2 E\in T(\2X,\2Y)$ is CPTP if and only if it may be written in the form
\begin{equation}
\2 E(\rho)= \sum_{\alpha=1}^D K_\alpha \rho K_\alpha^\dagger
\label{eqn:kraus}
\end{equation}
where the Kraus operators $\{K_\alpha: \alpha=1,...,D\}$, $K_\alpha\in \Lx{X,Y}$, satisfy the completeness relation 
\begin{equation}
\sum_{\alpha=0}^D K_\alpha^\dagger K_\alpha = \I_{\2 X}.
\label{eqn:completeness}
\end{equation} 
\end{theorem}

The Kraus representation of $\2E$ in \eqref{eqn:kraus} has the graphical representation 
\begin{center}
\includegraphics[width=0.6\textwidth]{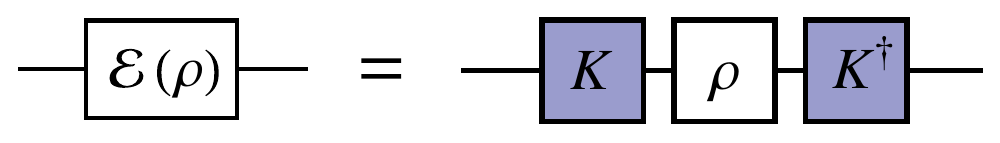}
\label{fig:kraus-evo}
\end{center}

The maximum number of Kraus operators needed for a Kraus description of $\2 E$ is equal to the dimension of ${\mathcal L}(\2X,\2Y)$. For the case where $\2 X\cong\2 Y\cong\C^d$ the maximum number of Kraus operators is $d^2$, and the minimum number case corresponds to unitary evolution where there is only a single Kraus operator.

It is important to note that the Kraus representation of $\2 E$ is not unique as there is unitary freedom in choosing the Kraus operators. We can give preference to a particular representation called the \emph{Canonical Kraus Representation}~\cite{Bengtsson2006} which is the unique set of Kraus operators satisfying the orthogonality relation $\Tr[K^\dagger_\alpha K_\beta]=\lambda_\alpha\delta_{\alpha\beta}$. The canonical Kraus representation will be important when transforming between representations in \S~\ref{sec:trans}.

\section{System-Environment / Stinespring Representation}
\label{sec:se}

The second representation of CPTP-maps we consider is the system-environment model~\cite{NC}, which is typically considered the most physically intuitive description of open system evolution. This representation is closely related to (and sometimes referred to as) the Stinespring representation as it can be thought of as an application of the Stinespring dilation theorem~\cite{Stinespring1955}, which we also describe in this section.
\begin{theorem}
In this model, we consider a system of interest $\2 X$, called the \emph{principle system}, coupled to an additional system $\2 Z$ called the \emph{environment}. The composite system of the principle system and environment is then treated as a closed quantum system which evolves unitarily. We recover the reduced dynamics on the principle system by performing a partial trace over the environment. Suppose the initial state of our composite system is given by $\rho\otimes\tau \in {\mathcal L}(\2X\otimes\2Z)$, where $\tau\in {\mathcal L}(\2Z)$ is the initial state of the environment. The joint evolution is described by a unitary operator $U\in {\mathcal L}(\2X\otimes\2Z)$ and the reduced evolution of the principle system's state $\rho$ is given by
\begin{equation}
\2 E(\rho) = \Tr_{\2 Z}[U(\rho\otimes\tau)U^\dagger]
\label{eqn:sys-env}
\end{equation}
\end{theorem}

For convenience we can assume that the environment starts in a pure state $\tau=\ketbra{v_0}{v_0}$, and in practice one only need consider the case where the Hilbert space describing the environment has at most dimension $d^2$ for $\2 X\cong\C^d$~\cite{NC}. The system-environment representation of the CP-map $\2 E$ may then be represented graphically as 
\begin{center}
\includegraphics[width=0.52\textwidth]{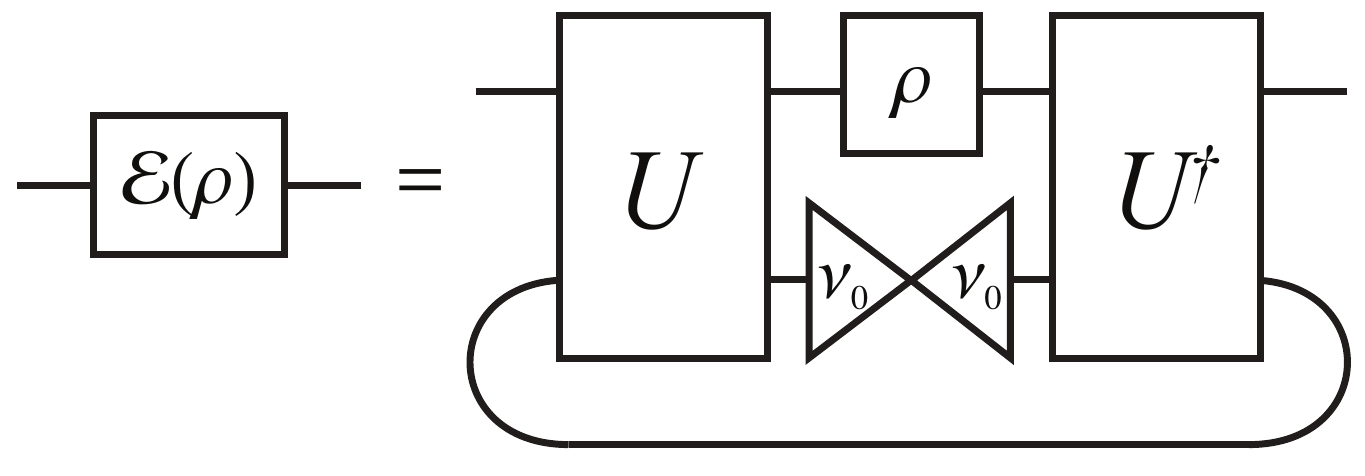}
\label{fig:sys-env}   
\end{center}

The system-environment model is advantageous when modelling the environment as a physical system. However, care must be taken when ascribing physical reality to any particular model as the system-environment description is not unique. This is not surprising as many different physical interactions could give rise to the same reduced dynamics on the principle system. This freedom manifests in an ability to choose the initial state of the environment in the representation and then adjust the unitary operator accordingly. In practice, the system-environment model can be cumbersome for performing many calculations where the explicit dynamics of the environment system are irrelevant. The remaining descriptions, which we cast into diagrammatic form, may be more convenient in these contexts.

Note that the system-environment evolution for the most general case will be an isometry and this is captured in Stinespring's representation~\cite{Stinespring1955}.

\begin{theorem}
Stinespring's dilation theorem states that a CP-map $\2 E\in C(\2X,\2Y)$ can be written in the form
\begin{equation}
\2 E(\rho) = \Tr_Z\left[A\rho A^\dagger\right]
\label{eqn:stinespring}
\end{equation}
where $A\in {\mathcal L}(\2X,\2Y\otimes\2Z)$ and the Hilbert space $\2 Z$ has dimension at most equal to ${\mathcal L}(\2X,\2Y)$. Further, the map $\2 E$ is trace preserving if and only if $A^\dagger A=\I_{\2 X}$ ~\cite{Stinespring1955}.
\end{theorem}

In the case where $\2 Y\cong\2 X$, the Hilbert space $\XZ$ mapped into by the Stinespring operator $A$ is equivalent to the joint system-environment space in the system-environment representation. Hence one may move from the system-environment description to the Stinespring representation as follows:
\begin{center}
\includegraphics[width=0.55\textwidth]{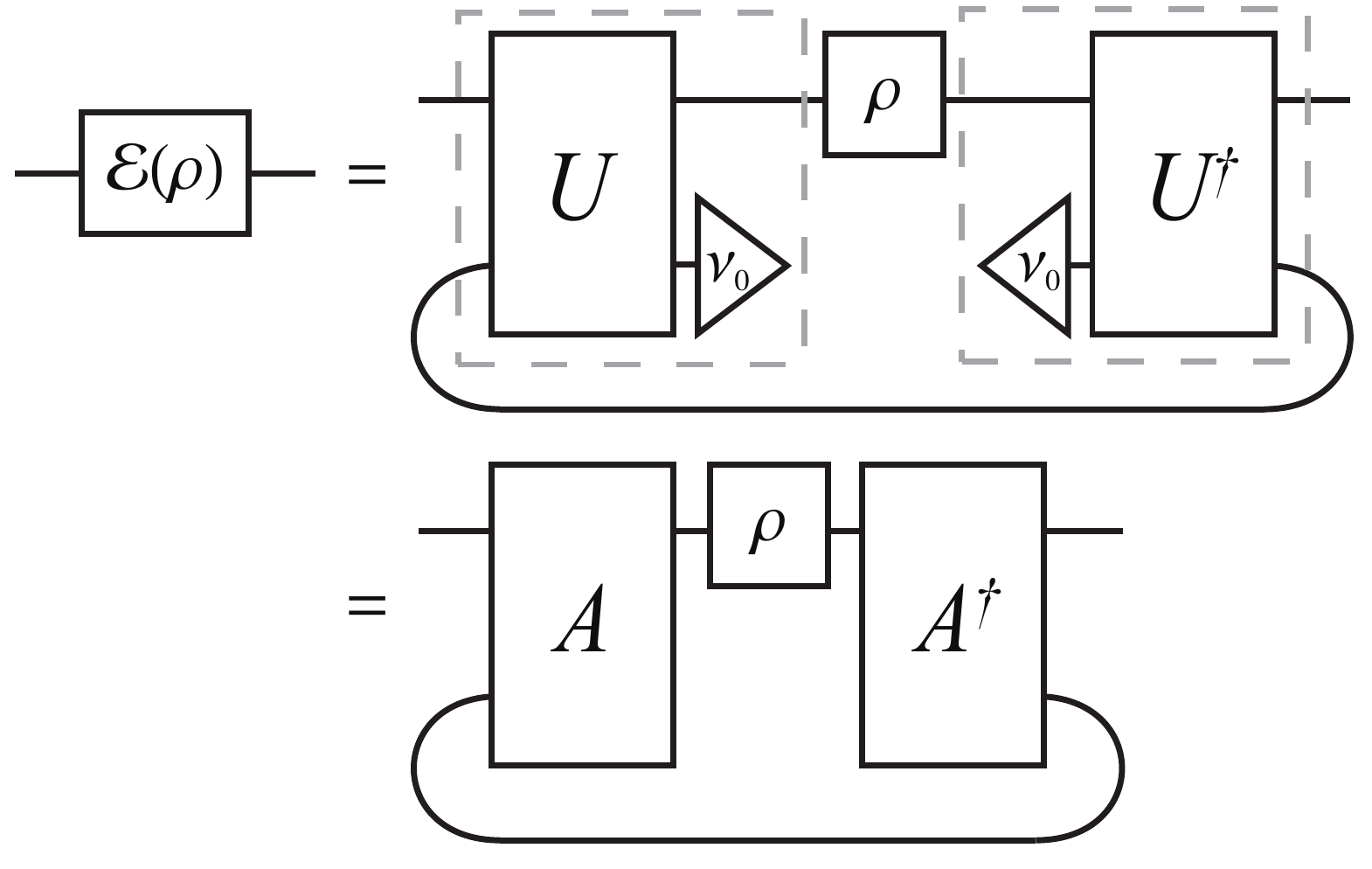}
\label{fig:stinespring}
\end{center}
where $\ket{v_0}\in\2 Z$ is the initial state of the environment, and we have defined the Stinespring operator
\begin{equation}
A = U\cdot(\I_{\2 X}\otimes\ket{v_0}),\label{eqn:stinespring-se}.
\end{equation}

This close relationship is why these two representations are often referred to by the same name, and as we will show in \S~\ref{sec:to-se}, it is straight forward to construct a Stinespring representation from the Kraus representation. However, generating a full description of the joint system-environment unitary operator $U$ from a Stinespring operator $A$ is cumbersome. It involves an algorithmic completion of the matrix elements in the unitary $U$ not contained within the subspace of the initial state of the environment~\cite{Bengtsson2006}. Since it usually suffices to define the action of $U$ when restricted to the initial state of the environment, which by \eqref{eqn:stinespring-se} is the Stinepsring representation, this is often the only transformation one need consider.

\begin{remark}
A further important point is that the evolution of the principle system $\2 E(\rho)$ is guaranteed to be CP if and only if the initial state of the system and environment is separable; $\rho_{\2 X\2 Z}=\rho_{\2X}\otimes\rho_{\3 Z}$. In the case where the physical system is initially correlated with the environment, it is possible to have reduced dynamics which are non-completely positive~\cite{Weinstein2004,Carteret2008}, however such situations are beyond the scope of this chapter.  
\end{remark}

\section{Louiville-Superoperator Representation}
\label{sec:sop}

We now move to the \emph{linear superoperator} or \emph{Liouville} representation of a CP-map $\2 E\in C(\2X,\2Y)$. \begin{definition}
The superoperator representation is based on the vectorization of the density matrix $\rho \mapsto \dket{\rho}_\sigma$ with respect to some orthonormal operator basis $\{\sigma_\alpha : \alpha=0,...,d^2-1\}$ as introduced in \S~\ref{sec:vec}. Once we have chosen a vectorization basis (col-vec in our case) we define the superoperator for a map $\2 E\in T(\2X,\2Y)$ to be the linear map 
\begin{equation}
\2 S:\2X\otimes\2X\rightarrow\2Y\otimes\2Y: \dket{\rho}\mapsto \dket{\2 E(\rho)}	\label{eqn:sop}
\end{equation}

This is depicted graphically as
 \begin{center}
\includegraphics[width=0.4\textwidth]{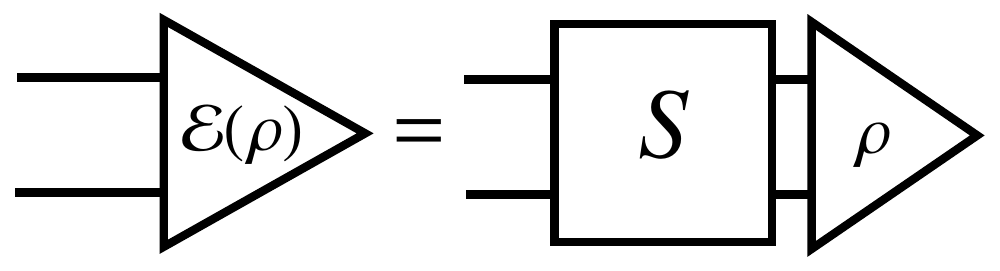}
\label{fig:superop}     
\end{center}

 In the col-vec basis we can express the evolution of a state $\rho$ in terms of tensor components of $\2 S$ as
\begin{eqnarray}
\2 E(\rho)_{mn} &=& \sum_{\mu\nu} \2S_{nm,\nu\mu} \rho_{\mu\nu}.
\end{eqnarray}
\end{definition}
For the case where $\2 E \in T(\2X)$, it is sometimes  useful to change the basis of our superoperators from the col-vec basis to an orthonormal operator basis $\{\sigma_\alpha\}$ for ${\mathcal L}(\2X)$. This is done using the basis transformation operator $T_{c\rightarrow\sigma}$ introduced in \S~\ref{sec:vec}. We have
\begin{eqnarray}
\2 S_\sigma 
	&=& T_{c\rightarrow\sigma}\cdot \2 S \cdot T_{c\rightarrow\sigma}^\dagger \label{eqn:sop-basis-change}\\
	&=& \sum_{\alpha\beta} \2 S_{\alpha\beta} \,\dket{\sigma_\alpha}\dbra{\sigma_\beta}.
\end{eqnarray}
where the subscript $\sigma$ indicates that $\2 S_\sigma$ is the superoperator in the $\sigma$-vec convention. The tensor networks for this transformation is given by
\begin{center}
\includegraphics[width=0.5\textwidth]{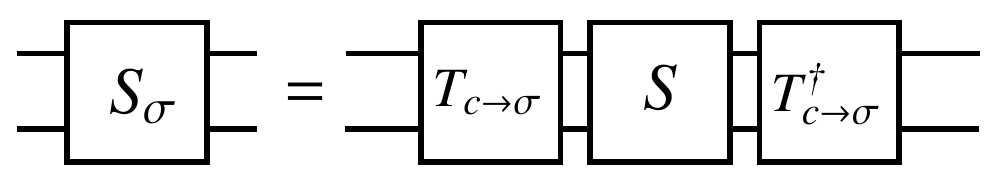}  \label{fig:sop-change-basis}
\end{center}
\begin{remark}
For a general map $\2 E\in T(\2X,\2Y)$ we could do a similar construction but would need different bases for the initial and final Hilbert spaces ${\mathcal L}(\2X)$ and ${\mathcal L}(\2Y)$.
\end{remark}

The structural properties the superoperator $\2 S$ must satisfy for the linear map $\2 E$ to be hermitian-preserving (HP), trace-preserving (TP), and completely positive (CP) are~\cite{Bengtsson2006}:
\begin{eqnarray}
\2 E \mbox{ is HP } 
	&\Longleftrightarrow&  \overline{\2 S}=\2 S^S \label{eqn:sop-hpres}\hspace{10em}\\
	&\Longleftrightarrow&
	\parbox[c]{1em}{\includegraphics[width=0.28\textwidth, trim= 0cm 1cm 0cm 1cm,clip]{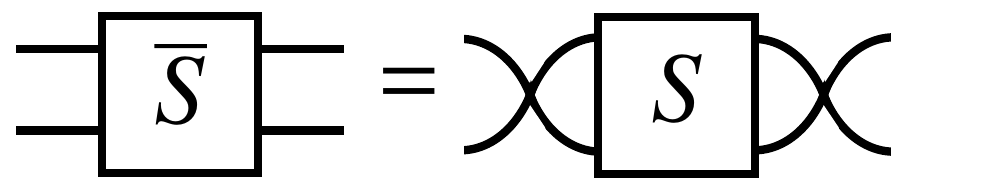}}\\
\2 E \mbox{ is TP } 
	&\Longleftrightarrow&\2 S_{mm,n\nu}=\delta_{n\nu}\\
	&\Longleftrightarrow&
	\parbox[c]{1em}{\includegraphics[width=0.2\textwidth]{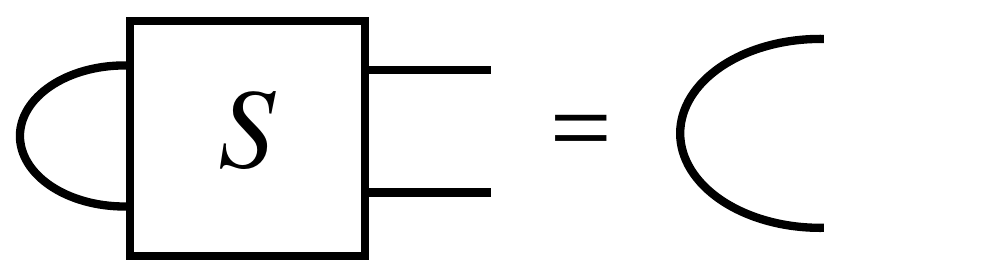}}\\
\2 E \mbox{ is CP } 
	&\Longleftrightarrow&\2 S_{\2 I\otimes \2 E}\dket{\rho_{AB}}\ge 0 ~\forall \rho_{AB}\ge0
\end{eqnarray}
\begin{remark}
There is not a convenient structural criteria on the superoperator $\2 S$ which specifies if $\2 E$ is a CP-map. To test for positivity or complete positivity one generally uses the closely related \emph{Choi-matrix} representation.
\end{remark}

Superoperators are convenient to use for many practical calculations. Unlike the system-environment model the superoperator $S$ is unique with respect to the choice of vectorization basis. Choosing an appropriate basis to express the superoperator in can often expose certain information about a quantum system. For example, if we want to model correlated noise for a mutli-partite system we can vectorize with respect to the mutli-qubit Pauli basis. Correlated noise would then manifest as non-zero entries in the superoperator corresponding to terms such as $\sigma_x\otimes\sigma_x$. 
We discus in more detail how this may be done in \S~\ref{sec:comp-sop}.

\section{Choi-Matrix Representation}
\label{sec:choi}

The final representation shown in Fig.~\ref{fig:cpreps} is the \emph{Choi matrix}~\cite{Choi1975}, or dynamical matrix~\cite{Bengtsson2006}. This is an application of the Choi-Jamio{\l}kowski isomorphism which gives a bijection between linear maps and linear operators~\cite{Jamiolkowski1972}.
\begin{definition}
Similarly to how vectorization mapped linear operators in ${\mathcal L}(\2X,\2Y)$ to vectors in $\XY$ or $\YX$, the Choi-Jamio{\l}kowski isomorphism maps linear operators in $T(\2X,\2Y)$ to linear operators in ${\mathcal L}(\2X\otimes\2Y)$ or ${\mathcal L}(\2Y\otimes\2X)$. 
The two conventions are
\begin{eqnarray}
\mbox{col-}\Lambda&:& 
	T(\2X,\2Y)\rightarrow {\mathcal L}(\2X\otimes\2Y):\,\, \2 E\mapsto \Lambda_c\\
\mbox{row-}\Lambda&:& 
	T(\2X,\2Y)\rightarrow {\mathcal L}(\2Y\otimes\2X):\,\, \2 E\mapsto \Lambda_r.
\end{eqnarray}
For $\2 X\cong\C^d$, the explicit construction of the Choi-matrix is given by
\begin{eqnarray} 
	\label{eqn:col-choi}
	\Lambda_c &=& \sum_{i,j=0}^{d-1}
	 \ketbra{i}{j}\otimes\2 E(\ketbra{i}{j})\\
	\Lambda_r &=& \sum_{i,j=0}^{d-1}
	 \2 E(\ketbra{i}{j})\otimes\ketbra{i}{j}
\end{eqnarray}
where $\{\ket{i}:i=0,\hdots,d-1\}$ is an orthonormal basis for $\2 X$.

\end{definition}

We call the two conventions col-$\Lambda$ and row-$\Lambda$ due to their relationship with the vectorization conventions introduced in \S~\ref{sec:vec}. 
\begin{definition}
The Choi-Jamio{\l}kowski isomorphism can also be thought of as having a map $\2 E\in T(\2X,\2Y)$ act on one half of an unnormalized Bell-state $\ket{\Phi^+}=\sum_i \ket{i}\otimes\ket{i} \in \2X\otimes\2X$, and hence these conventions corresponding to which half of the Bell state it acts on:
\begin{eqnarray} 
	\Lambda_c &=& (\2 I \otimes \2 E)\ketbra{\Phi^+}{\Phi^+}\label{eqn:jamiolkowski}\\
	\Lambda_r &=& (\2 E \otimes \2 I)\ketbra{\Phi^+}{\Phi^+}
\end{eqnarray}
where $\2 I\in T(\2X)$ is the identity map. 
\end{definition}

In what follows we will use the col-$\Lambda$ convention and drop the subscript from $\Lambda_c$. We note that the alternative row-$\Lambda$ Choi-matrix is naturally obtained by applying the bipartite-SWAP operation to $\Lambda_c$. 

As will be considered in \S~\ref{sec:to-choi}, if the evolution of the CP map $\2 E$ is described by a Kraus representation $\{K_i\}$, then the Choi-Jamio{\l}kowski isomorphism states that we construct the Choi-matrix by acting on one half of a bell state with the Kraus map as shown:

\begin{center}
\includegraphics[width=0.7\textwidth]{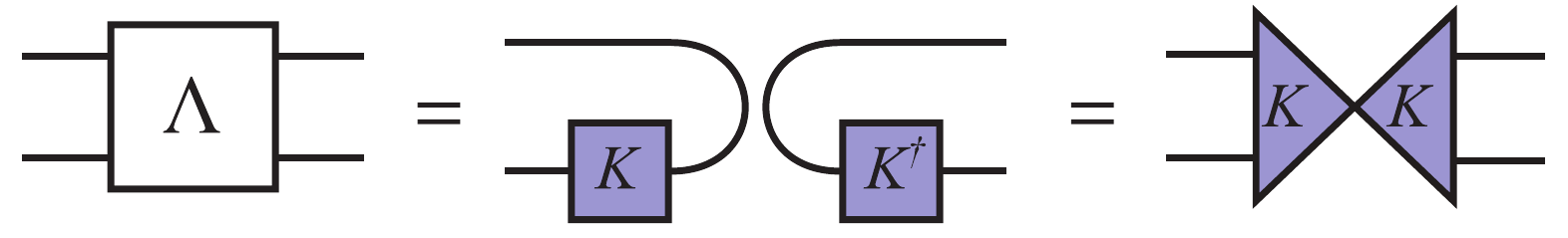}
\label{fig:choi-jamiolkowski}
\end{center}

\begin{remark}
Note that in general any tensor network describing a linear map $\2 E$, not just the Kraus description, may be contracted with one-half of the maximally entangled state $\ketbra{\Phi^+}{\Phi^+}$ to construct the Choi-matrix.
\end{remark}

With the Choi-Jamio{\l}kowski isomorphism defined, the evolution of a quantum state in terms of the Choi-matrix is then given by
\begin{eqnarray}
\2 E(\rho) &=& \Tr_{\2 X} \left[ (\rho^T\otimes \I_{\2 Y} )\Lambda \right]
\label{eqn:choi-evo}
\end{eqnarray}
or in terms of tensor components
\begin{eqnarray}
\2 E(\rho)_{mn}
		&=&\sum_{n,m} \Lambda_{\mu m,\nu n}\rho_{\mu\nu}.
\end{eqnarray}
The tensor network for \eqref{eqn:choi-evo} is given by
\begin{center}
\includegraphics[width=0.6\textwidth]{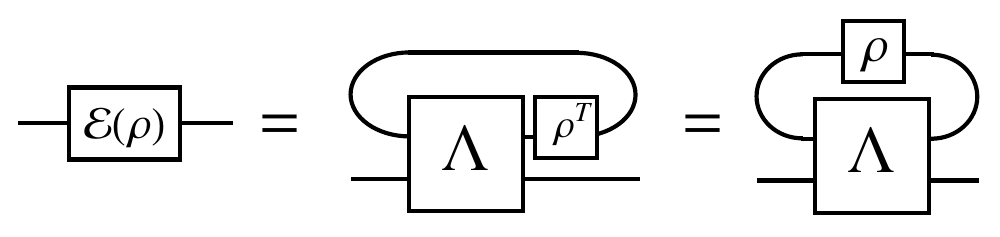}
\label{fig:choi-evo}    
\end{center}
The graphical proof of \eqref{fig:choi-evo} for the case where $\2 E$ is described by a Kraus representation is as follows:
 \begin{center}
\includegraphics[width=0.6\textwidth]{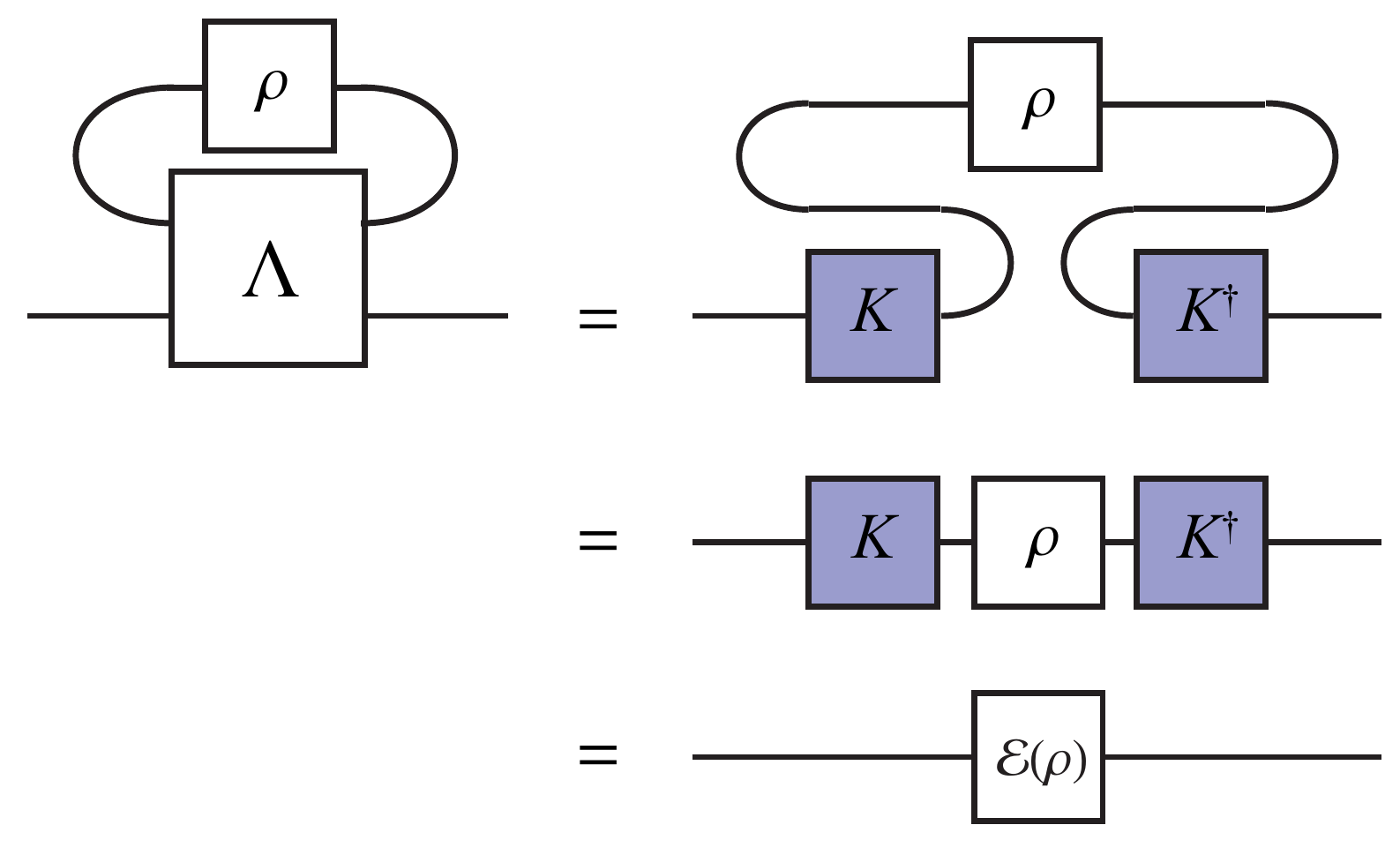}\label{fig:choi-evo-proof}    
 \end{center}
\begin{remark}
The structural properties the Choi-matrix $\Lambda$ must satisfy for the linear map $\2 E$ to be hermitian-preserving (HP), trace-preserving (TP),  and completely positive (CP) are~\cite{Bengtsson2006}:
\begin{eqnarray}
\2 E \mbox{ is HP } 
	&\Longleftrightarrow&  \Lambda^\dagger=\Lambda \label{eqn:choi-HP}\hspace{10em}\\
	&\Longleftrightarrow&
	\parbox[c]{1em}{\includegraphics[width=0.3\textwidth]{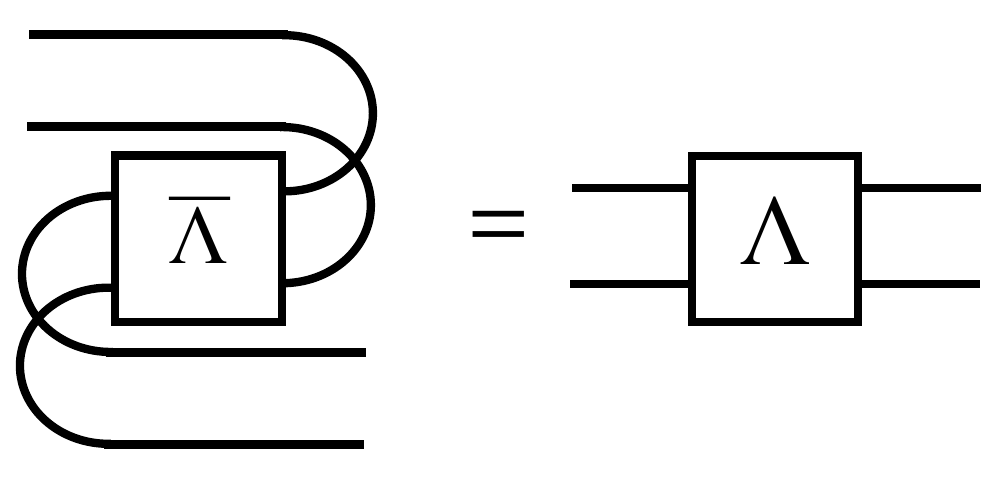}}\\
\2 E \mbox{ is TP } 
	&\Longleftrightarrow& \Tr_{\2 Y}[\Lambda]=\I_{\2 X}\label{eqn:choi-TP}\\
	&\Longleftrightarrow&
	\parbox[c]{1em}{\includegraphics[width=0.3\textwidth]{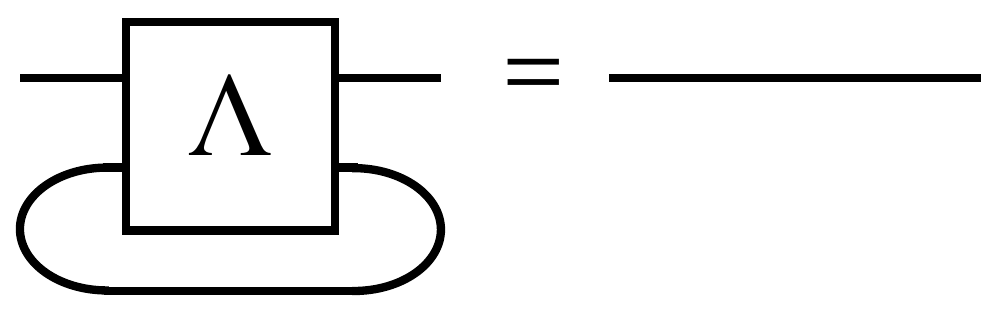}}\label{eqn:choi-tp}\\
\2 E \mbox{ is CP } 
	&\Longleftrightarrow&  \Lambda\ge 0.\label{eqn:choithm}
\end{eqnarray}
\end{remark}

The Choi-matrix for a given map $\2 E$ is unique with respect to the isomorphism convention chosen. We will provide tensor networks to illustrate a close relationship to the superoperator formed with the corresponding vectorization convention in \S~\ref{sec:choi-sop}. The Choi-matrix finds practical utility as one can check the complete-positivity of the map $\2 E$ by computing the eigenvalues of $\Lambda$. It is also necessary to construct the Choi-matrix for a given superoperator to transform to the other representations. 

Due to the similarity of vectorization and the Choi-Jamio{\l}kowski isomorphism, one could then ask what happens if we vectorize in a different basis. This change of basis of the Choi-matrix is more commonly known as the $\chi$-matrix which we will discuss next. However, such a change of basis does not change the eigen-spectrum of a matrix, so the positivity criteria in \eqref{eqn:choithm} holds for any basis.

Another desirable property of Choi matrices is that they can be directly determined for a given system experimentally by \emph{ancilla assisted process tomography (AAPT)}~\cite{DAriano2003,White2003}. This is an experimental realization of the Choi-Jamio{\l}kowski isomorphism which we discuss in detail in \S~\ref{sec:aapt}.

\section{Process Matrix Representation}
\label{sec:chi}

As previously mentioned, one could consider a change of basis of the Choi-matrix analogous to that for the superoperator. The resulting operator is more commonly known as the \emph{$\chi$-matrix} or \emph{process matrix}~\cite{NC}.

\begin{definition}
Consider Hilbert spaces $\2 X\cong\C^{d_x}$, $\2 Y\cong\C^{d_y}$ and let $D=d_x d_y$, and $\2Z \cong \C^D$. If one chooses an orthonormal operator basis  $\{\sigma_\alpha: \alpha=0,...,D-1\}$ for ${\mathcal L}(\2X,\2Y)$, then a CPTP map $\2 E\in C(\2X,\2Y)$ may be expressed in terms of a matrix $\chi\in {\mathcal L}(\2Z)$ as
\begin{eqnarray}
\2 E(\rho)
&=&\sum_{\alpha,\beta=0}^{D-1} \chi_{\alpha\beta} \sigma_\alpha \rho \sigma_\beta^\dagger
\label{eqn:chi-evo}
\end{eqnarray}
where the process matrix $\chi$ is unique with respect to the choice of basis $\{\sigma_\alpha\}$.
\end{definition}

The process matrix with respect to an orthonormal operator basis $\{\sigma_\alpha\}$ is related to the Choi matrix by the change of basis
\begin{eqnarray}
\chi	&=& T_{c\rightarrow\sigma} \cdot \Lambda \cdot T^\dagger_{c\rightarrow\sigma}
\label{eqn:choi-to-chi}\\
\Rightarrow \Lambda&=& \sum_{\alpha,\beta} \chi_{\alpha\beta} \dketdbra{\sigma_\alpha}{\sigma_\beta}
\end{eqnarray}
where $T_{c\rightarrow\sigma}$ is the vectorization change of basis operator introduced in \S~\ref{sec:vec}. Thus evolution in terms of the $\chi$-matrix is analogous to our Choi evolution as shown below:
\begin{center}
\includegraphics[width=0.5\textwidth]{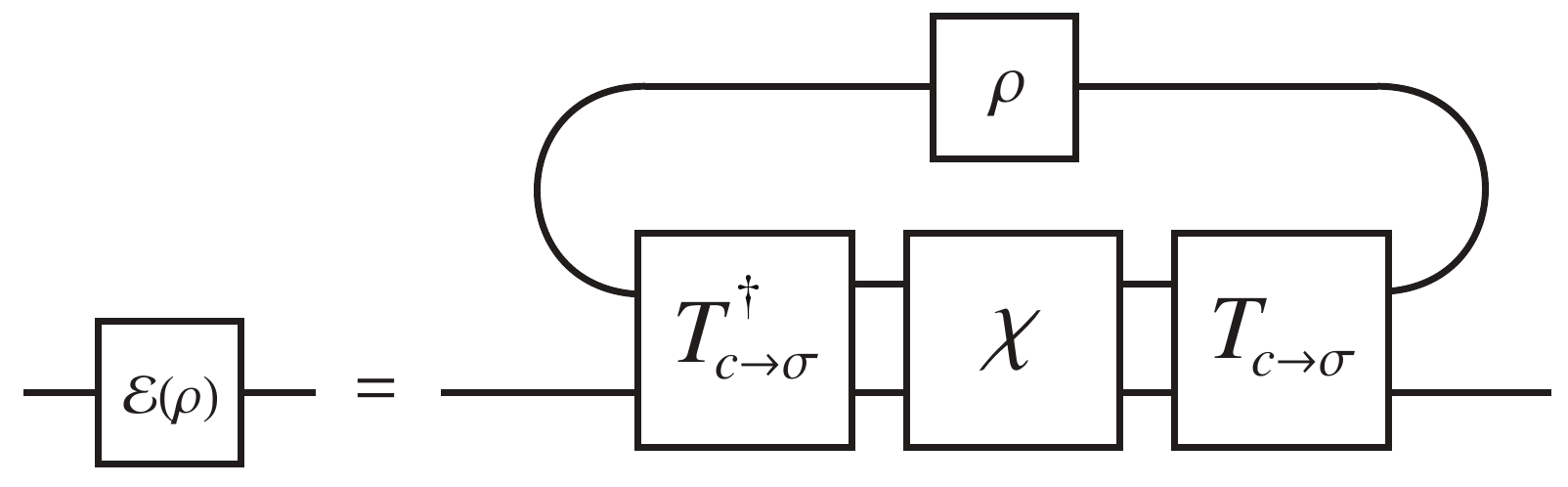}
\label{fig:chi-evo}    
\end{center}

Starting with the expression for process matrix evolution in \eqref{eqn:chi-evo}, the  graphical proof asserting the validity of \eqref{eqn:choi-to-chi} is as follows
\begin{center}
\includegraphics[width=0.55\textwidth]{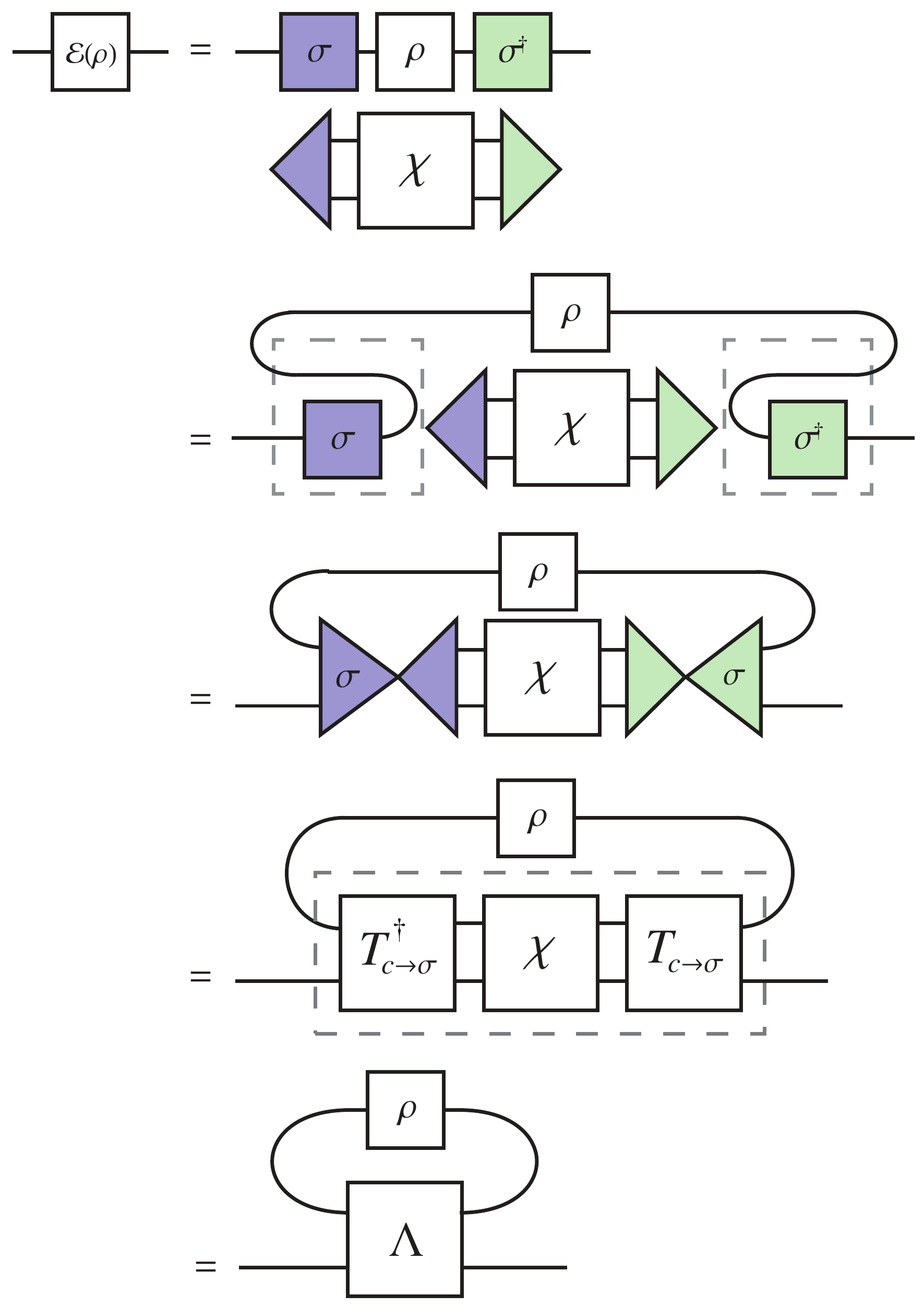} \label{fig:choi-to-chi}   
\end{center}

We also see that if one forms the process matrix with respect to the col-vec basis $\sigma_\alpha = E_{j,i}$ where $\alpha = i+dj$ and $d$ is the dimension of $\2 H$, then we have $\chi=\Lambda$.

\begin{remark}
Since the process matrix is a unitary transformation of the Choi-matrix, it shares the same structural conditions for hermitian preservation and complete-positivity as for the Choi-matrix given in~\eqref{eqn:choi-HP} and \eqref{eqn:choithm} respectively.
The condition for it to be trace preserving may be written in terms of the matrix elements and basis however. These conditions are
\begin{eqnarray}
\2 E \mbox{ is TP} &\Longleftrightarrow&
\Tr_{\2 Y}\left[T^\dagger_{c\rightarrow\sigma}\chi T_{c\rightarrow\sigma}\right]=\I_{\2 X}\\
&\Longleftrightarrow& \sum_{\alpha,\beta} \chi_{\alpha,\beta} \sigma_{\alpha}^T \overline{\sigma}_\beta = \I_{\2 X}\\
\2 E \mbox{ is HP} &\Longleftrightarrow& \chi^\dagger=\chi\\
\2 E \mbox{ is CP} &\Longleftrightarrow& \chi\ge0.
\end{eqnarray}
\end{remark}

To convert a process-matrix $\chi$ in a basis $\{\sigma_\alpha\}$ to another orthonormal operator basis $\{\omega_\alpha\}$, we may use the same change of basis transformation as used for the superoperator change of basis in \S~\ref{sec:sop}. That is
\begin{eqnarray}
\chi^\omega	&=& T_{\sigma\rightarrow\omega} \cdot \chi^\sigma \cdot T^\dagger_{\sigma\rightarrow\omega}\\
		&=& \sum_{\alpha\beta} \chi^\sigma_{\alpha\beta}\,\, \dket{\sigma_\alpha}_\omega\dbra{\sigma_\beta}_\omega
\end{eqnarray}
where the superscripts $\sigma,\omega$ denote the basis of the $\chi$-matries. This is illustrated as
\begin{center}
\includegraphics[width=0.45\textwidth]{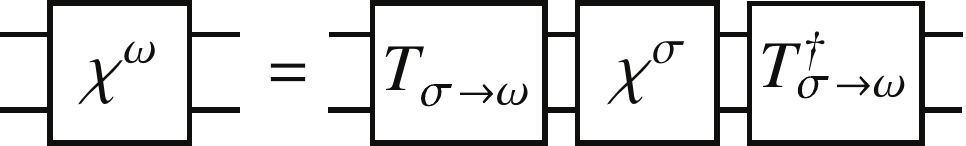}  \label{fig:chi-change-basis}    
\end{center}

\section{Transforming between representations}
\label{sec:trans}

\subsection{Transformations between the Choi-matrix and superoperator representations}
\label{sec:choi-sop}

The Choi-matrix and superoperator are naturally equivalent under the reshuffling wire bending duality introduced in \S~\ref{sec:bipartite}.  In the col (row) convention we may transform between the two by applying the bipartite col (row)-reshuffling operation $R$ introduced in \S~\ref{sec:bipartite}. Let $\Lambda\in \LXY$ be the Choi-matrix, and $\2 S\in \Lx{\XX,\YY}$ be the superoperator, for a map $\2 E\in \Tx{X,Y}$. Then we have
\begin{eqnarray}
\Lambda &=& \2S^{R} \quad\quad\2 S = \Lambda^{R}
\end{eqnarray}
The tensor networks for these transformations using the col convention are

\begin{center}
\begin{tabular}{c|c}
\includegraphics[width=.4\textwidth]{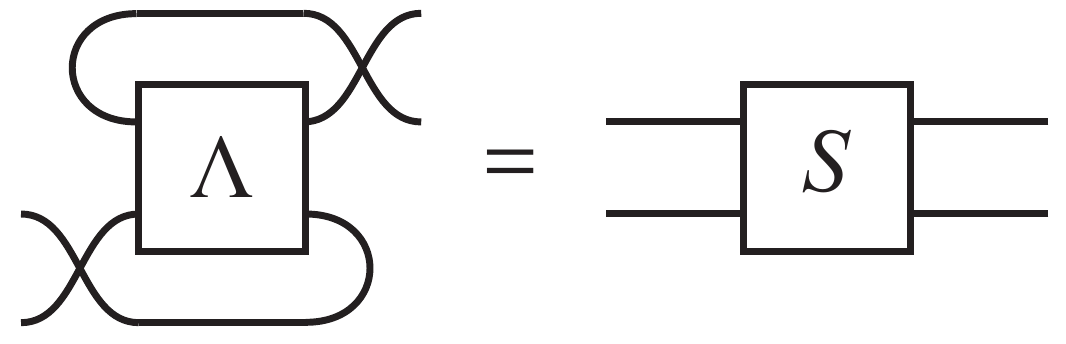}
\hspace{2em} & \hspace{2em}
\includegraphics[width=.4\textwidth]{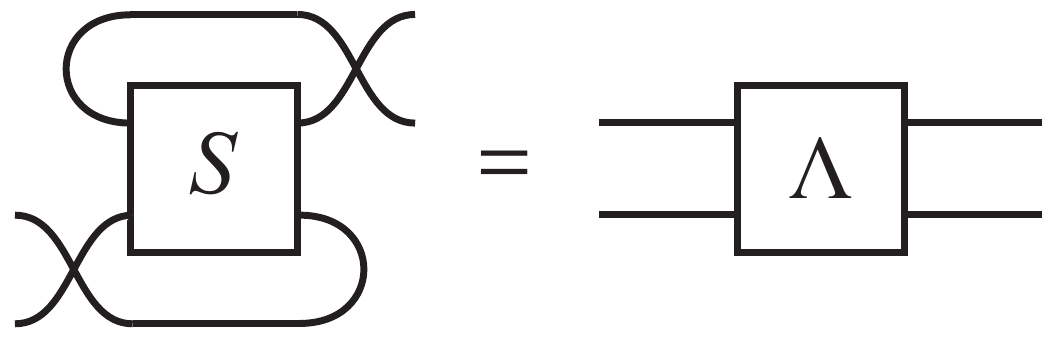}
\end{tabular}
\label{fig:choi-sop}
\end{center}
In terms of tensor components we have
\begin{eqnarray}
\Lambda_{mn,\mu\nu} 	&=& \2 S_{\nu n,\mu m}
\end{eqnarray}
where $m,n$ and $\mu,\nu$ index the standard bases for $\2 X$ and $\2 Y$ respectively. Graphical proofs of the relations $\Lambda^{R_c}=\2 S$ and $\2 S^{R_c}=\Lambda$ are given below

\begin{center}
\begin{tabular}{c|c}
\includegraphics[width=0.4\textwidth]{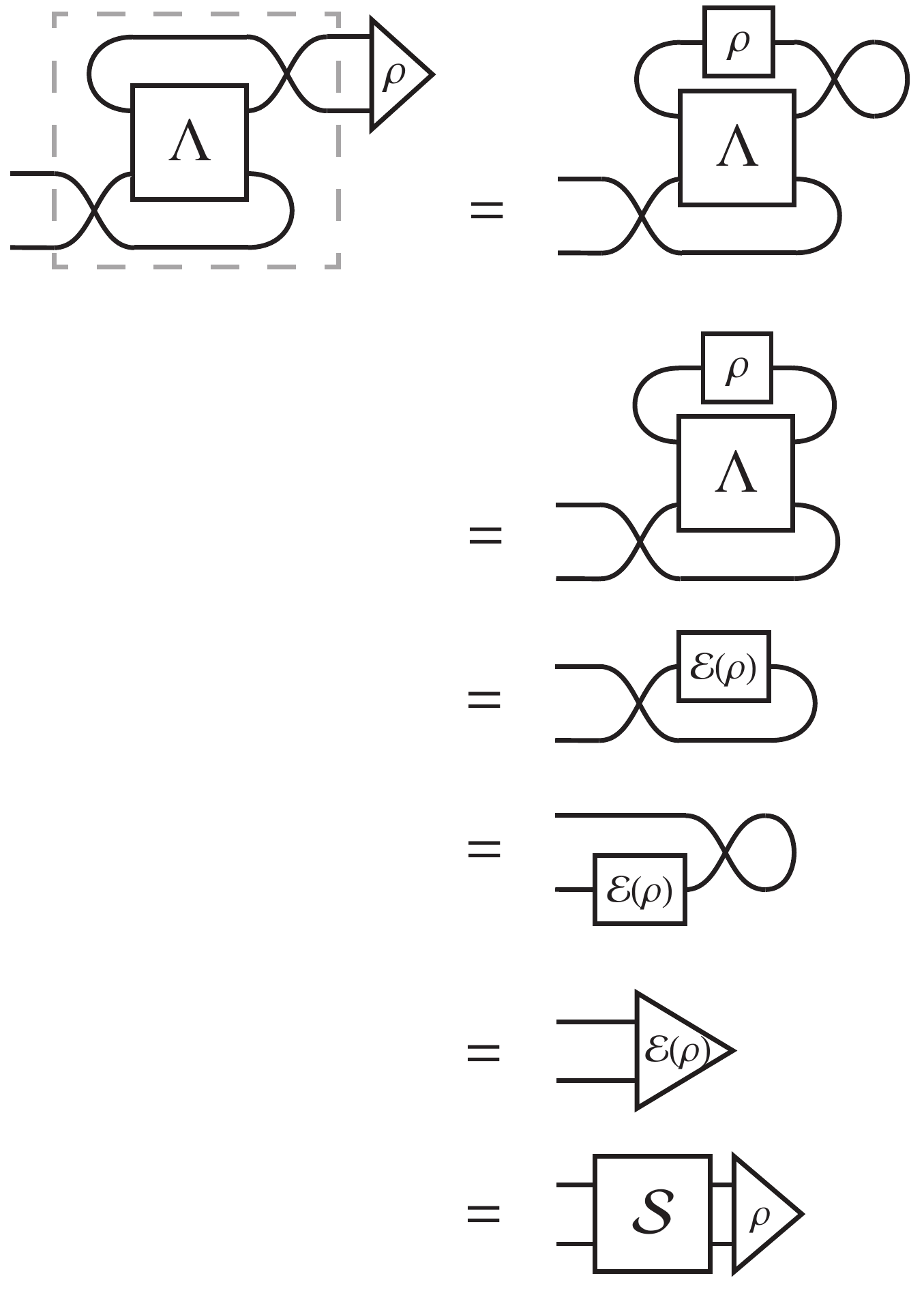} 
\quad\quad&\quad
\includegraphics[width=0.4\textwidth]{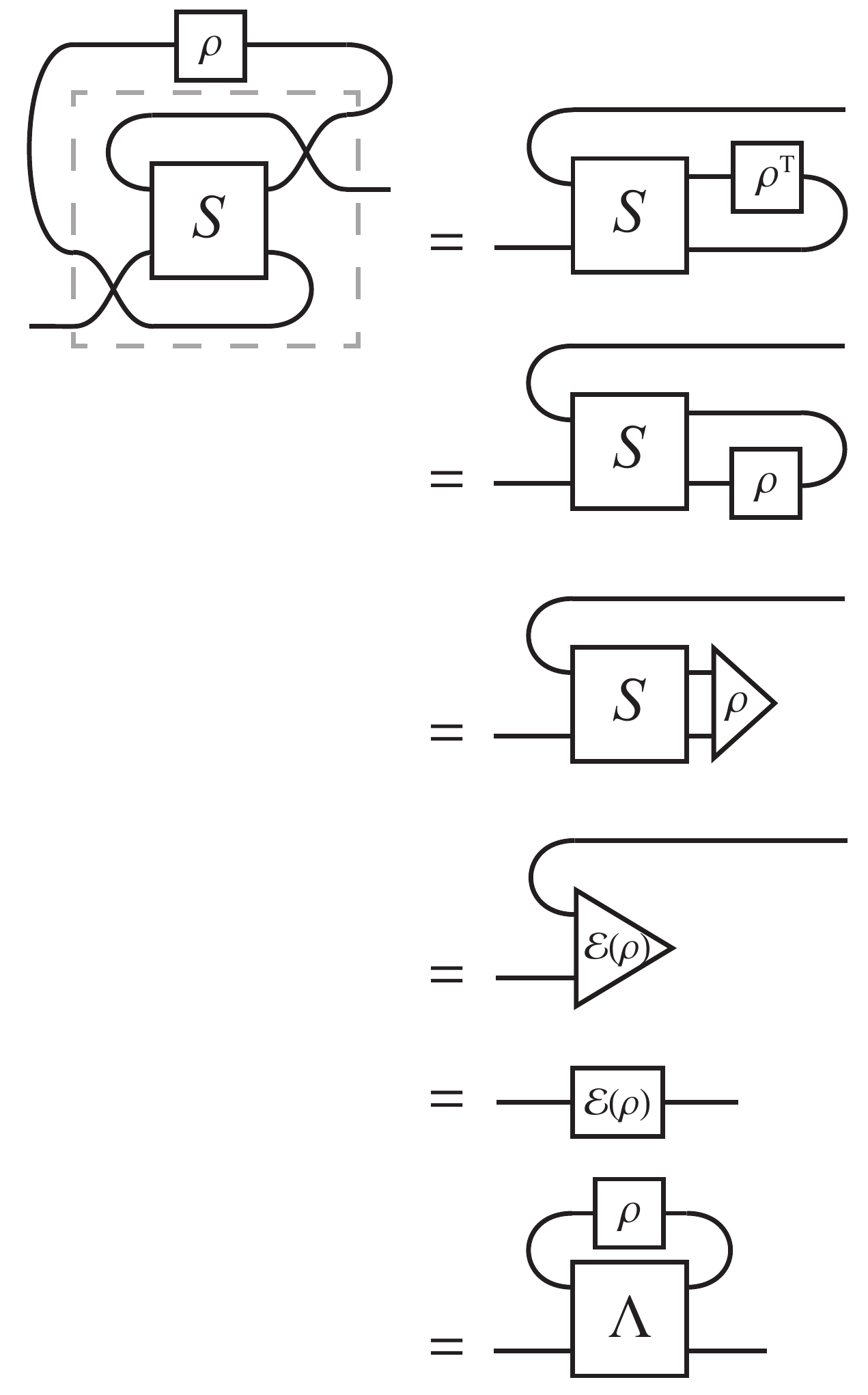}
\end{tabular}
 \label{fig:choi-sop-proof}
\end{center}

\begin{remark}
To transfer between a $\chi$-matrix with respect to an arbitrary operator basis, and a superoperator with respect to an arbitrary vectorization basis, we must first convert both to col-vec (or row-vec) convention and then proceed by reshuffling.
\end{remark}
Note that reshuffling is its own inverse, i.e $(\Lambda^R)^R=\Lambda$, hence the solid bi-directional arrow connecting the Choi-matrix and superoperator representations in Fig.~\ref{fig:cpreps}. This is the only transformation between the representations we consider which is linear, bijective, and self-inverse.


\subsection{Transformations to the superoperator representation}
\label{sec:to-sop}

Transformations to the superoperator from the Kraus and system-environment representations of a CP-map are also accomplished by a wire-bending duality, in this case vectorization. However, unlike the bijective equivalence of the Choi-matrix and superoperator under the reshuffling duality, the vectorization duality is only surjective.

If we start with a Kraus representation for a CPTP map $\2 E\in\Cx{X,Y}$ given by $\{K_\alpha: \alpha=0,...,D-1 \}$, with $K_\alpha\in \Lx{X,Y}$, we can construct the superoperator $\2 S \in \Lx{\XX,\YY}$ by 
\begin{eqnarray}
\2 S &=&\sum_{\alpha=0}^{D-1} \overline{K}_\alpha\otimes K_\alpha.
\label{eqn:kraus-to-choi}
\end{eqnarray}
The corresponding tensor network is  
\begin{center}
\includegraphics[width=0.35\textwidth]{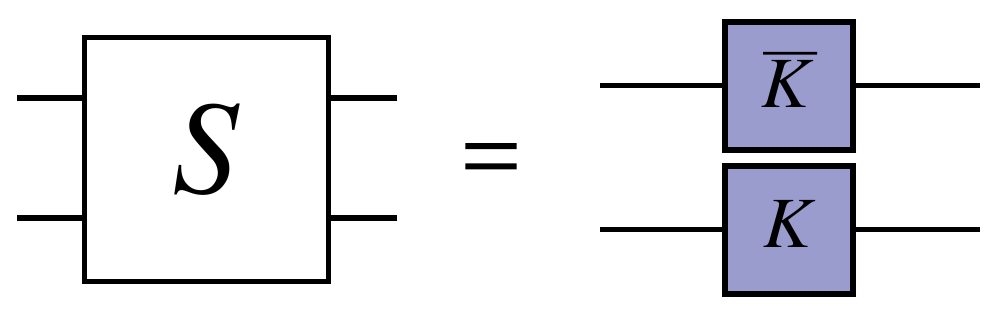}
\label{fig:kraus-sop}
\end{center}
and the graphical proof of this relationship follows directly from Roth's lemma:
\begin{center}
\includegraphics[width=0.9\textwidth]{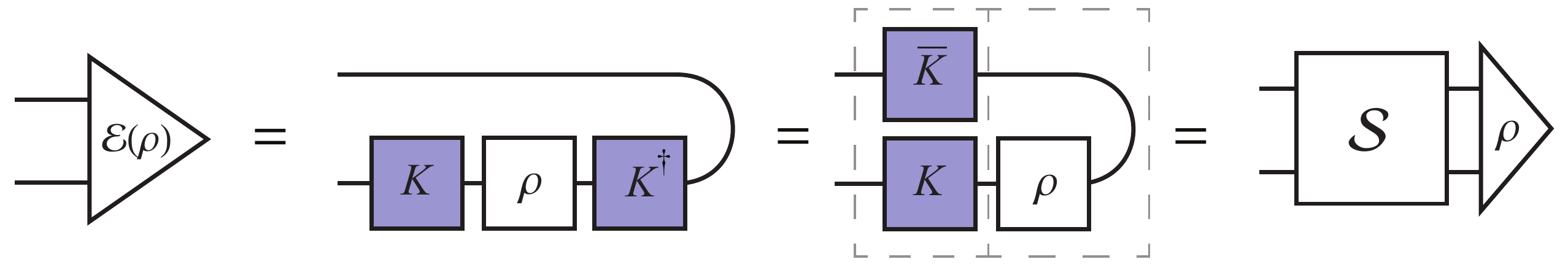}
\label{fig:kraus-sop-proof}
\end{center}

Starting with a system-environment (or Stinespring) representation of a map $\2 E\in\Cx{X,Y}$ with input and output system Hilbert spaces $\2 X\cong \C^{d_x}$ and $\2Y\cong \C^{d_y}$ respectively, and environment Hilbert space $\2 Z\cong \C^D$ with $1\le D \le d_x d_y$, we may construct the superoperator for this map from the joint system-environment unitary $U$ and initial environment state $\ket{v_0}$ by
\begin{eqnarray}
\2 S &=& \sum_{\alpha} \bra{\alpha}\overline{U}\ket{v_0}\otimes \bra{\alpha}U\ket{v_0}, \label{eqn:se-to-sop}
\end{eqnarray}
where $\{\ket{\alpha}:\alpha=0,...,D-1\}$ is a real, orthonormal basis for $\2 Z$. The corresponding tensor network is 
\begin{center}
\includegraphics[width=0.35\textwidth]{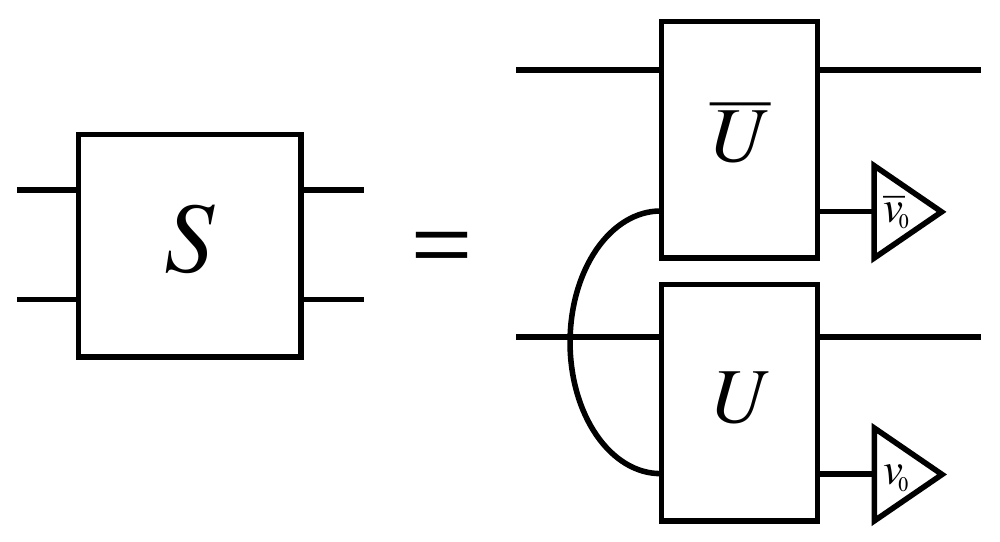}
 \label{fig:se-sop}
\end{center}
 As with the Kraus to superoperator transformation, the proof of \eqref{eqn:se-to-sop} follows from Roth's lemma.

Note that while the vectorization wire bending duality is invertible, these transformations to the superoperator from the Kraus and system-environment representations are single directional. In both cases injectivity fails as the superoperator is unique, while both the Kraus and system-environment representations are not. Hence we have solid single directional arrows in Fig.~\ref{fig:cpreps} connecting both the Kraus and system-environment representations to the superoperator. The inverse transformation from a superoperator to the Kraus or system-environment representation requires a canonical decomposition of the operator $\2 S$ (via first reshuffling to the Choi-matrix), which is detailed in Sections~\ref{sec:to-kraus} and \ref{sec:to-se}.


\subsection{Transformations to the Choi-matrix representation}
\label{sec:to-choi}

Transforming to the Choi-matrix from the Kraus and system-environment representations is accomplished via a wire-bending duality which captures the Choi-Jamio{\l}kowski isomorphism. As with the case of transforming to the superoperator, this duality transformation is surjective but not injective.

Given a set of Kraus matrices $\{ K_\alpha: \alpha=0,...,D-1\}$ where $K_\alpha\in\Lx{X,Y}$ for a CPTP-map $\2 E\in\Cx{X,Y}$, one may form the Choi-Matrix $\Lambda$ as was previously illustrated in \eqref{fig:choi-jamiolkowski} in \S~\ref{sec:choi}. In terms of both Dirac notation and tensor components we have:

\begin{eqnarray}
\Lambda	&=& \sum_{i,j}\left( \ketbra{i}{j}\otimes \sum_\alpha K_\alpha\ketbra{i}{j}K^\dagger_\alpha\right)\\
		&=& \sum_{\alpha} \dketdbra{K_\alpha}{K_\alpha}\\
\Lambda_{mn,\mu\nu}&=&\sum_\alpha (K_\alpha)_{\mu m}(\overline{K}_\alpha)_{\nu n}.
\end{eqnarray}
where $\{\ket{i}\}$ is an orthonormal basis for $\2 X$, $m,n$ index the standard basis for $\2 X$, and $\mu,\nu$ index the standard basis for $\2 Y$.

Given a system-environment representation with joint unitary $U\in \Lx{\XZ}$ and initial environment state $\ket{v_0}\in \2 Z$ we have
\begin{eqnarray}
\Lambda	&=&	\sum_{i,j}\left(\ketbra{i}{j}\otimes 
			\Tr_{\2 Z}\left[U\ketbra{i}{j}\otimes\ketbra{v_0}{v_0}U^\dagger\right]\right)
\end{eqnarray}
Graphically this is given by
\begin{center}
\includegraphics[width=0.4\textwidth]{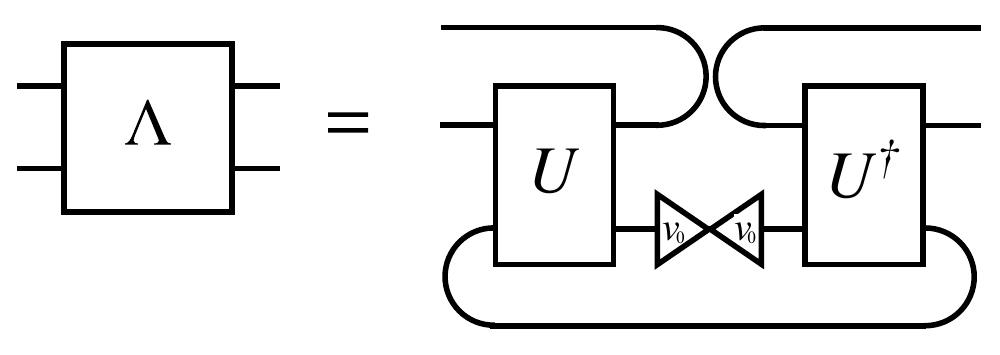}
\label{fig:choi-se}
 \end{center}
The proof of these transformations follow directly from the definition of the Choi-matrix in \eqref{eqn:col-choi}, and the tensor networks for the evolution via the Kraus or system-environment representations given in \eqref{fig:kraus-evo} and \eqref{fig:sys-env} respectively. As with the vectorization transformation to the superoperator discussed in \S~\ref{sec:to-sop}, even though the Choi-Jamio{\l}kowski isomorphism is linear these transformations are single directional as injectivity fails due to the non-uniqueness of both the Kraus and system-environment representations. Hence we have the solid single-directional arrows connecting both the Kraus and system-environment representations to the Choi-matrix in Fig.~\ref{fig:cpreps}. 
 
This completes our description of the linear transformations between the representations of CP-maps in Fig.~\ref{fig:cpreps}. We will now detail the non-linear transformations to the Kraus and system environment representations. 


\subsection{Transformations to the Kraus Representation}
\label{sec:to-kraus}

We may construct a Kraus representations from the Choi matrix or system environment representation by the non-linear operations of spectral-decomposition and partial trace decomposition respectively. To construct a Kraus representation from the Superoperator however, we must first reshuffle to the Choi matrix. 

To construct Kraus matrices from a Choi matrix we first recall the graphical Spectral decomposition we introduced as an example of our color summation convention. If  $\2 E$ is CP, by \eqref{eqn:choithm} we have $\Lambda\ge0$ and hence the spectral decomposition of the Cho -matrix is given by
\begin{equation}
\Lambda=\sum_\alpha \mu_\alpha \ketbra{\phi_\alpha}{\phi_\alpha}, 
\end{equation}
where  $\mu_\alpha \ge 0$ are the eigenvalues, and $\ket{\phi_\alpha}$ the eigenvectors of $\Lambda$. Hence we can define Kraus operators $K_\alpha= \lambda_\alpha A_\alpha$ where $\lambda_\alpha = \sqrt{\mu_\alpha}$ and $A_\alpha$ is the unique operator satisfying $\dket{A_\alpha}=\ket{\phi_\alpha}$ as illustrated: 
\begin{center}
\includegraphics[width=0.35\textwidth]{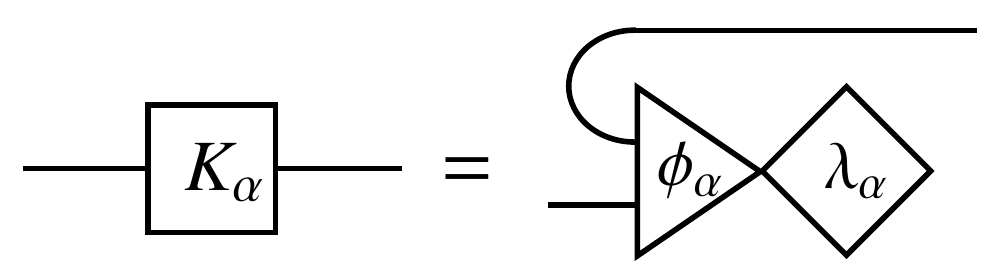}
\label{fig:choi2kraus}
\end{center}
The number of Kraus operators will be equal to the rank $r$ of the Choi matrix, where $1\le r\le \mbox{dim}(\Lx{X,Y})$. The graphical proof of this transformation is as follows:
 \begin{center}
 \includegraphics[width=0.42\textwidth]{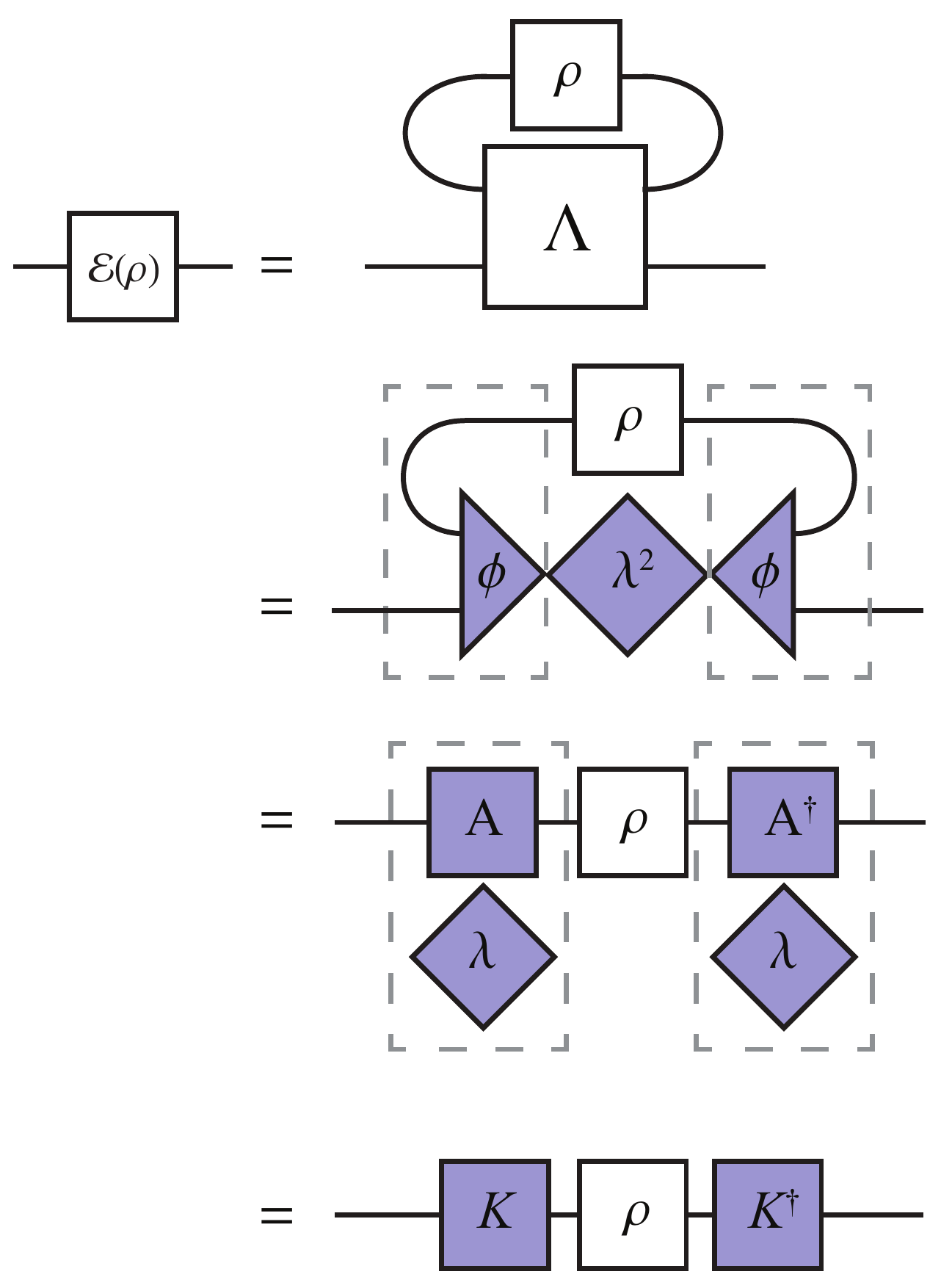} \label{fig:choi-kraus-proof}
 \end{center}
The proof that Kraus operators satisfy the completeness relation follows from the trace preserving property of $\Lambda$ in \eqref{eqn:choi-TP}: 
\begin{center}
\includegraphics[width=0.43\textwidth]{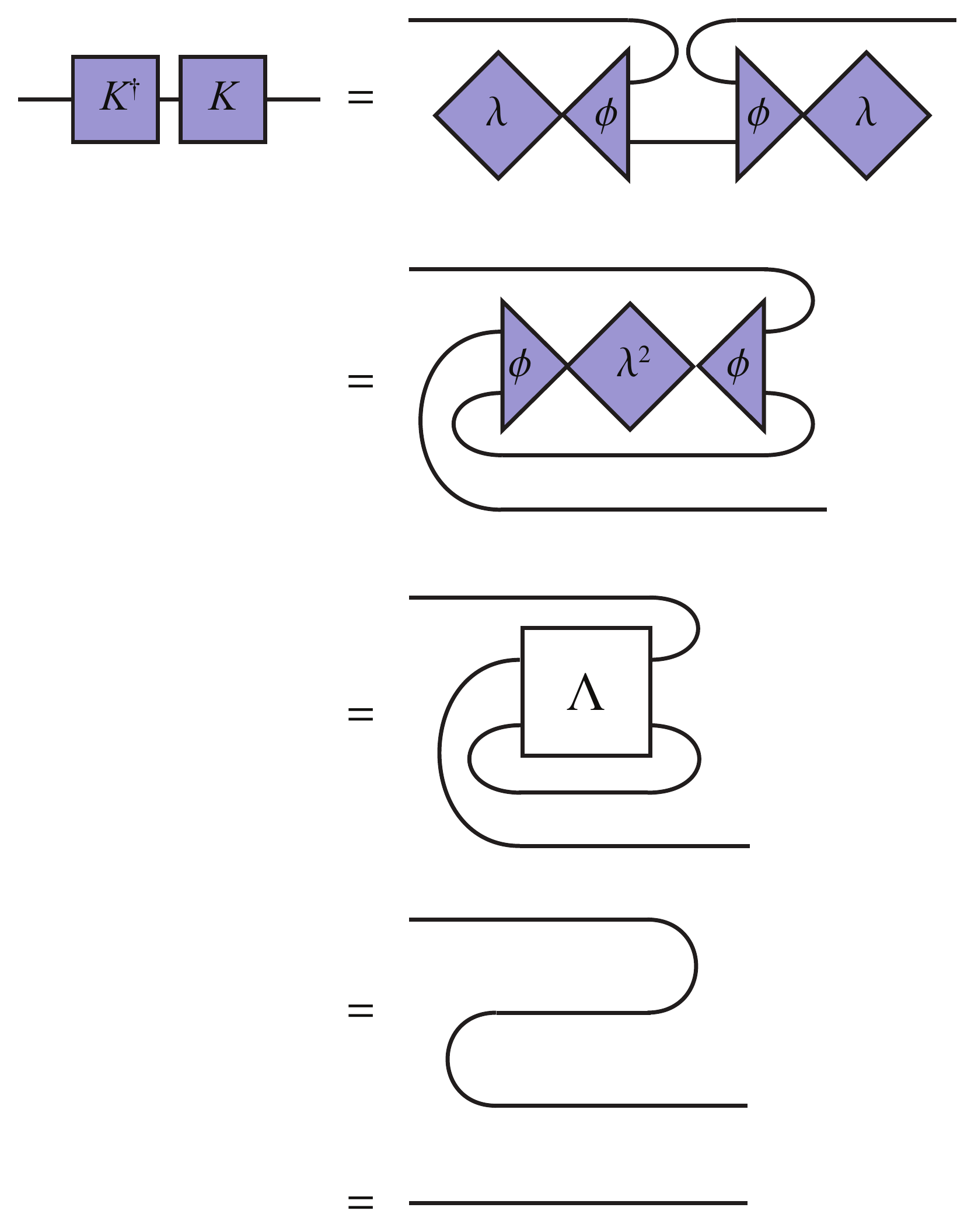}
\label{fig:completeness-proof}
\end{center}
Note that since $\Lambda$, and the $\chi$-matrix are related by a unitary change of basis, the Kraus representations constructed from their respective spectral decompositions will also be related by the same transformation. Each will give a unitarily equivalent Canonical Kraus representation of $\2 E$ since the eigen-vectors are orthogonal. Thus we have described the arrow in Fig.~\ref{fig:cpreps} connecting the Choi matrix to the Kraus representation. It is represented as a dashed arrow as it involves a non-linear decomposition, and is single directional as this representation transformation is injective, but not surjective. Surjectivity fails as we can only construct the canonical Kraus representations for $\2 E$. The reverse transformation is given by the Jamio{\l}kowski isomorphism described in \S~\ref{sec:to-choi}.

Starting with a system-environment representation with joint unitary $U\in \Lx{\XZ}$ and initial environment state $\ket{v_0}\in\2 Z$, we first choose an orthonormal basis $\{\ket{\alpha}:\alpha=0,...,D-1\}$ for $\2 Z$. We then construct the Kraus representation by decomposing the partial trace in this basis as follows
\begin{eqnarray}
\2 E(\rho) 
	&=&\Tr_E\left[U\left(\rho\otimes\ketbra{v}{v}\right)U^\dagger\right] \\
	&=& \sum_{\alpha=0}^{D-1} \bra{\alpha}U\ket{v_0} \rho \bra{v_0}U^\dagger\ket{\alpha}\\
	&=& \sum_{\alpha=0}^{D-1} K_\alpha \rho K_\alpha^\dagger.
\end{eqnarray}
Hence we may define Kraus matrices 
\begin{equation}
K_\alpha = \bra{\alpha}U\ket{v_0}
\label{eqn:kraus-se-proof}
\end{equation}
 leading to the tensor network 
 \begin{center}
\includegraphics[width=0.35\textwidth]{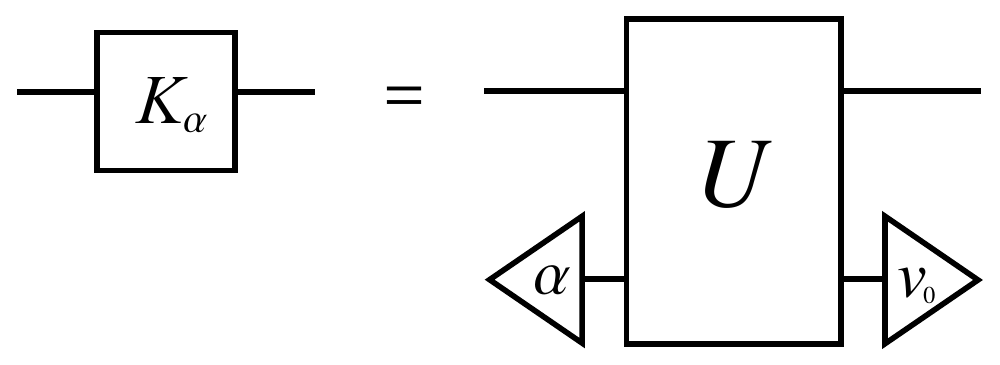} \label{fig:kraus-sys-env}
 \end{center}
The graphical proof of \eqref{eqn:kraus-se-proof} and \eqref{fig:kraus-sys-env} is as follows
\begin{center}
\includegraphics[width=0.45\textwidth]{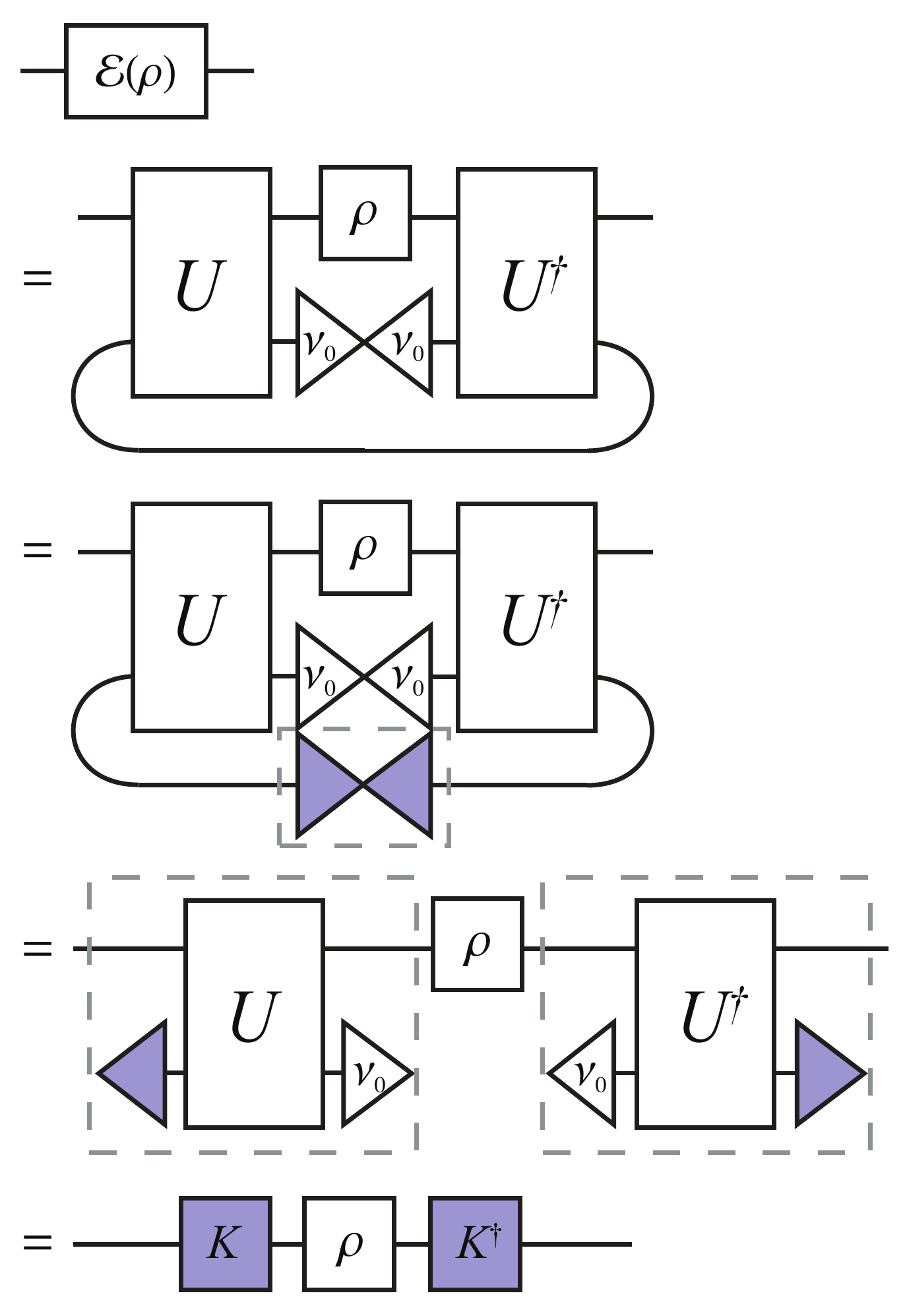} 
\label{fig:kraus-sys-env-proof}
\end{center}
 
Though the Kraus and system-environment representations are both non-unique, for a fixed environment basis this partial trace decomposition is an injective transformation between the Kraus and Stinespring representations (or equivalently between the Kraus and system-environment representations when the joint unitary is restricted to a fixed initial state of the environment). 
To see this  let $\{K_\alpha\}$ and $\{J_\alpha\}$ be two Kraus representations for a CPTP-map $\2 E\in\Cx{X,Y}$, constructed from Stinespring representations $A$ and $B$ respectively. We have that
\begin{eqnarray}
K_\alpha &=& J_\alpha\\
&\Leftrightarrow& (K_\alpha)_{ij}= (J_\alpha)_{ij}\\
&\Leftrightarrow& A_{i\alpha,j} = B_{i\alpha,j}\\
&\Leftrightarrow& A=B.
\end{eqnarray}
Since the Stinespring operators satisfy $A=U\ket{v_0}$ and $B=V\ket{v_0}$ for some joint unitaries $U$ and $V$, we must have that $U_0=V_0$ where $U_0$ and $V_0$ are the joint unitaries restricted to the subspace of the environment spanned by $\ket{v_0}$.

This transformation can be thought of as the reverse application of the Stinespring dilation theorem, and hence for a fixed choice of basis (and initial state of the environment) it is invertible. The inverse transformation is the Stinespring dilation, and as we will show in \S~\ref{sec:to-se}, since the inverse transformation is also injective this transformation is a bijection. However, since the partial trace decomposition involves a choice of basis for the environment it is non-linear --- hence we use a dashed bi-directional arrow to represent the transformation from the system-environment representation to the Kraus representation in Fig.~\ref{fig:cpreps}. 


\subsection{Transformations to the system-environment representation}
\label{sec:to-se}

We now describe the final remaining transformation given in Fig.~\ref{fig:cpreps}, the bijective non-linear transformation from the Kraus representation to the system-environment, or Stinespring, representation. The system-environment representation is the most cumbersome to transform to as it involves the unitary competition of a Stinespring dilation of a Kraus representation. Thus starting from a superoperator one must first reshuffle to the Choi-matrix, then from the Choi-matrix description one must then spectral decompose to the canonical Kraus representation, before finally constructing the system-environment as follows.

Let $\{ K_\alpha:\alpha=0,..., D-1\}$, where $1\le D\le \mbox{dim}(\Lx{X,Y})$, be a Kraus representation for the CP-map $\2 E\in \Tx{X,Y}$. Consider an ancilla Hilbert space $\2 Z\cong \C^D$, this will model the environment. If we choose an orthonormal basis for the environment, $\{\ket{\alpha}: \alpha=0,...,D-1\}$, then by Stinesprings dilation theorem we may construct the Stinespring matrix for the CP map $\2 E$ by 
\begin{equation}
A = \sum_{\alpha=0}^{D-1} K_\alpha \otimes\ket{\alpha}.
\label{eqn:kraus2stinespring}
\end{equation}

Recall from \S~\ref{sec:se} that the Stinespring representation is essentially the system-environment representation when the joint unitary operator is restricted to the subspace spanned by the initial state of the environment. Hence if we let $\ket{v_0}\in \2 Z$ be the initial state of the environment system, then this restricted unitary is given by
\begin{equation}
U_0 = \sum_\alpha K_\alpha\otimes\ketbra{\alpha}{v_0},\label{eqn:kraus-to-U}.
\end{equation} 
The tensor networks for \eqref{eqn:kraus2stinespring} and \eqref{eqn:kraus-to-U} are:
\begin{center}
\begin{tabular}{c|c}
\includegraphics[width=0.42\textwidth]{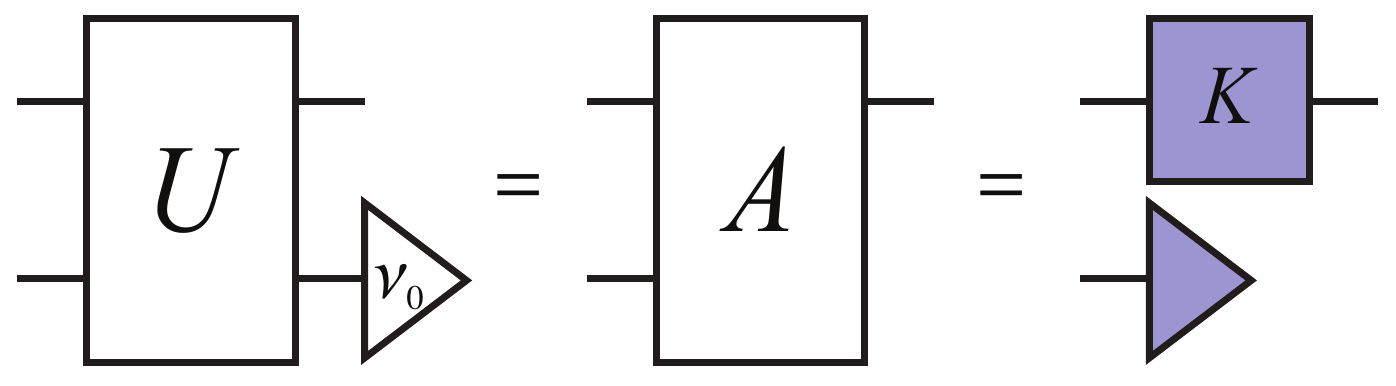}	
\quad\quad&\quad\quad
\includegraphics[width=0.25\textwidth]{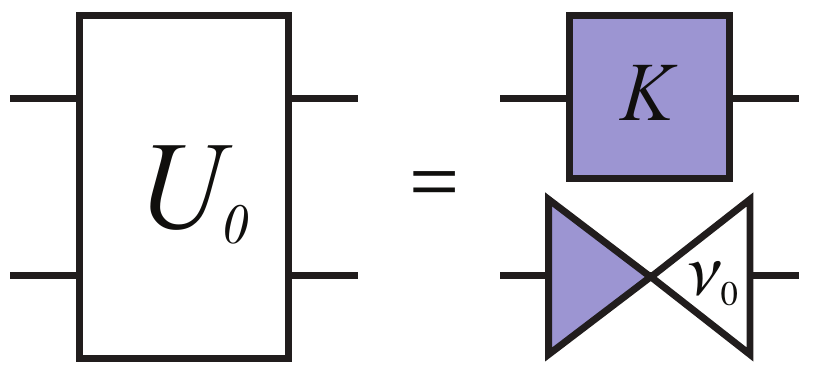}	\\
\footnotesize{(a) Stinespring operator}	&
\footnotesize{(a) Restricted unitary}
\end{tabular}
 \label{fig:sys-env-kraus}
\end{center}
The graphical proof that this construction gives the required evolution of a state $\rho$ is as follows
\begin{center}
\includegraphics[width=0.33\textwidth]{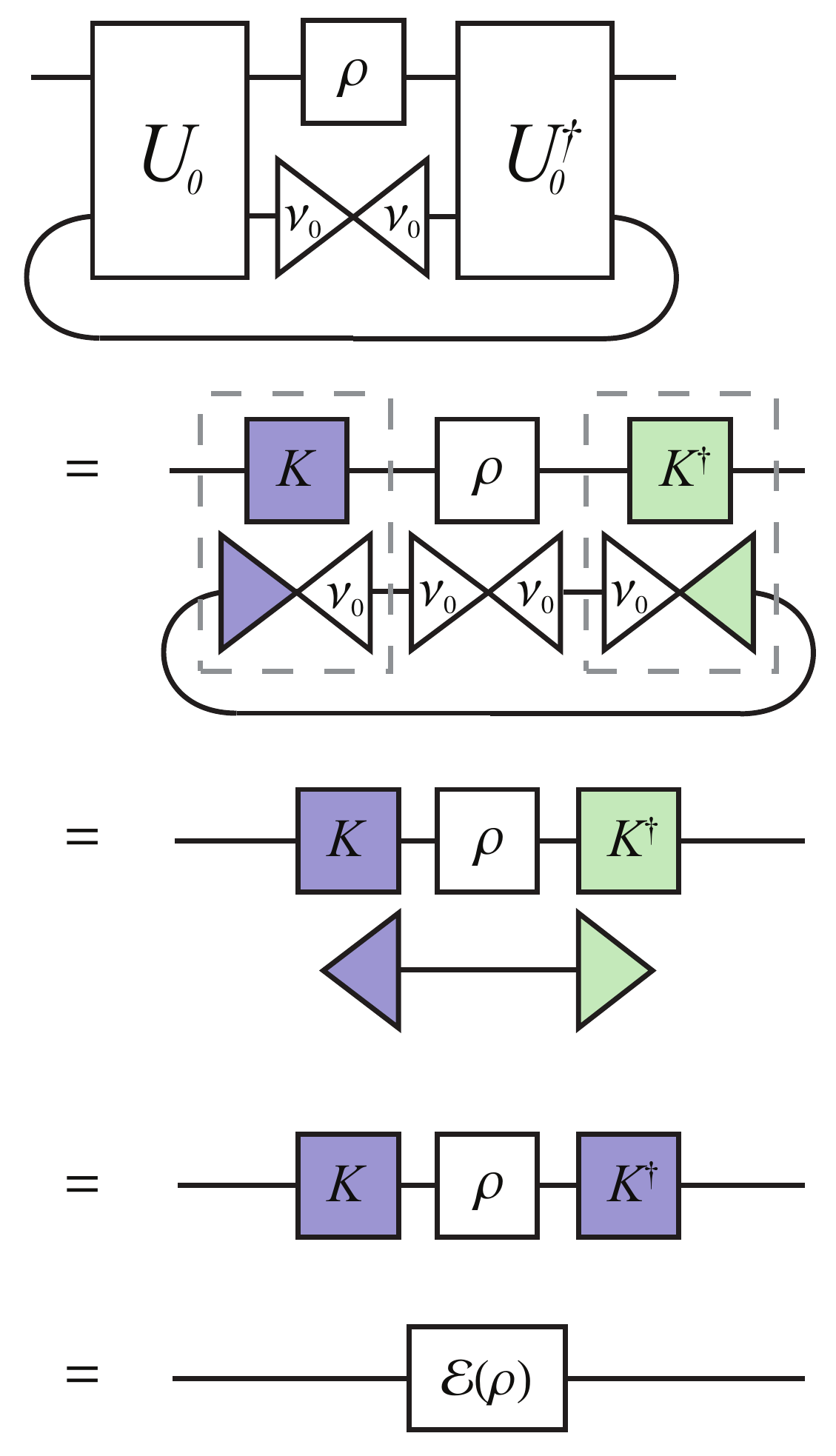}
\label{fig:kraus-se-proof}
\end{center}
In principle, one may complete the remaining entries of this matrix to construct the full matrix description for the unitary $U$, however such a process is cumbersome and is unnecessary to describe the evolution of the CP-map $\2 E$~\cite{Bengtsson2006}.

We have now finished characterizing the final transformations depicted in Fig.~\ref{fig:cpreps} connecting the Kraus representation to the system-environment representation by Stinespring dilation. As previously mentioned in \S~\ref{sec:to-kraus}, for a fixed choice of basis and initial state for the environment, the transformation between Kraus and Stinespring representations is bijective (and hence so is the transformation between Kraus and system-environment representations when restricted to the subspace spanned by the initial state of the environment). Though both these representations are non-unique, by fixing a basis and initial state for the environment we ensure that this transformation is injective. To see this let $U_0$ and $V_0$ be unitaries restricted to the state $\ket{v_0}$ constructed from Kraus representations, $\{K_\alpha\}$ and $\{J_\alpha\}$ respectively, for $\2 E\in \Cx{X,Y}$. Then 
\begin{eqnarray}
U_0=V_0
&\Leftrightarrow&	\sum_\alpha K_\alpha\otimes\ketbra{\alpha}{v_0} = \sum_\alpha J_\alpha \otimes\ketbra{\alpha}{v_0}\\
&\Leftrightarrow&	\sum_\alpha K_\alpha\braket{\beta}{\alpha} = \sum_\alpha J_\alpha \braket{\beta}{\alpha}\\ 
&\Leftrightarrow& K_\beta = J_\beta
\end{eqnarray}
Bijectivity then follows from the injectivity of the inverse transformation --- the previously given construction of a Kraus representation by the partial trace decomposition of a joint unitary operator in \eqref{fig:kraus-sys-env}. 


\section{Applications}
\label{sec:ex}

We have now introduced all the basic elements of our graphical calculus for open quantum systems, and shown how it may be used to graphically depict the various representations of CP-maps, and transformations between representations. In this section we move onto more advanced applications of the graphical calculus. We will demonstrate how to apply vectorization to composite quantum systems, and in particular how to compose multiple superoperators together, and construct effective reduced superoperators from tracing out a subsystem. We also demonstrate the superoperator representation of various linear transformations of matrices. These constructions will be necessary for the remaining examples where we derive a succinct condition for a bipartite state to be used for ancilla assisted process topography, and where we present arguably simpler derivations of the closed form expression for the average gate fidelity and entanglement fidelity of a quantum channel in terms of properties of each of the representations of CP-maps given in \S~\ref{sec:cpmaps}.


\subsection{Vectorization of composite systems}
\label{sec:vec-comp}

We now describe how to deal with vectorization of the general case of composite system of $N$ finite dimensional Hilbert spaces. Let $\2 X_k\cong \C^{d_k}$ be a $d_k$-dimensional complex Hilbert space, and let $\{\ket{i_k}:i_k=0,...,d_k-1\}$ be the standard basis for $\2 X_k$. We are interested in the composite system of $N$ such Hilbert spaces, 
\begin{equation}
\2 X = \2 X_1\otimes...\otimes \2 X_{N}=\bigotimes_{k=1}^{N} \2 X_k
\end{equation}
 which has dimensions $D=\prod_{k=1}^N d_k$. Let $\{\ket{\alpha}: \alpha = 0,..., D-1\}$ be the computational basis for $\2 X$. We can consider vectors in $\2 X$ and the dual space $\2 X^\dagger$ as either 1st-order tensors where their single wire represents an index running over $\alpha$, or as a $N$th-order tensor where each of the $N$ wire corresponds to an individual Hilbert space $\2 X_k$. The correspondence between these two descriptions is made by the concatenation of the composite indices according to the lexicographical order 
 \begin{equation}
 \alpha =\sum_{k=1}^{N} c(k)\, i_k \quad\mbox{where}\quad c(k) := \frac{D}{\prod_{j=1}^k d_j}.
 \end{equation} 
Note that one could also consider the object as any order tensor between 1st and $N$th by the appropriate concatenation of some subset of the the wires.

We define the unnormalized Bell-state on the composite system $\XX$ to be the state formed by the column (or row) vectorization of the identity operator $\I_{\2X} \in {\mathcal L}(\2X)$
\begin{eqnarray}
\dket{\I_{\2 X}} 
	&=& \sum_{\alpha=0}^D \ket{\alpha}\otimes\ket{\alpha}\nonumber\\
	&=& \sum_{i_1=0}^{d_1-1}....\sum_{i_N=0}^{d_N-1} \ket{i_1,..., i_N}\otimes\ket{i_1,...,i_N}.
\label{eqn:comp-bell}
\end{eqnarray}
where $\ket{i_1,...,i_N}:=\ket{i_1}\otimes...\otimes\ket{i_N}$. The tensor network for this state is
\begin{center}
\includegraphics[width=0.25\textwidth]{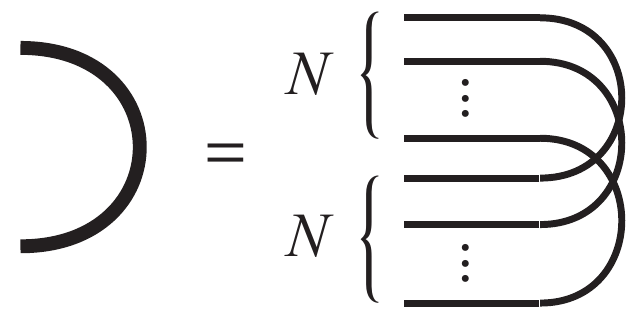}
\label{fig:composite-bell}    
\end{center}

 As with the single system case the column vectorization of a composite linear operator $A \in {\mathcal L}(\2X,\2Y)$, where $\2 Y=\bigotimes_{k=1}^{N} \2 Y_k$,  is given by bending all the system wires upwards, or equivalently by the identity 
 \begin{equation}
 \dket{A}\equiv (\I\otimes A)\dket{\I}.
 \end{equation}
Graphically this is given by
\begin{center}
\includegraphics[width=0.45\textwidth]{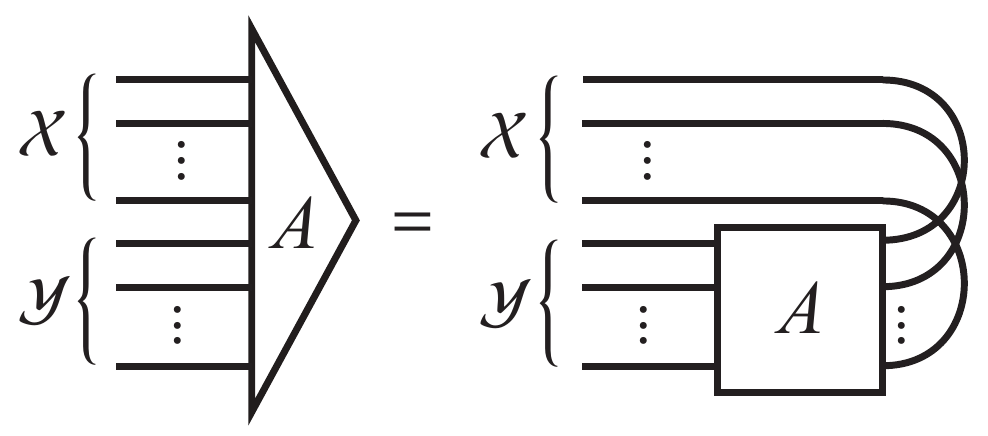}
\label{fig:composite-vec}
\end{center}
 Note that the order of the subsystems for the bent wires is preserved by the vectorization operation.
 
In some situations it may be preferable to consider vectorization of the composite system in terms of vectorization of the individual component systems. Transferring between this component vectorization and the joint-system vectorization can be achieved by an appropriate index permutation of vectorized operators which has a succinct graphical expression when cast in the tensor network framework. 

Suppose the operator $A\in\Lx{X,Y}$, where $\2 X=\bigotimes_{k=1}^N \2 X_k$, $\2 Y=\bigotimes_{k=1}^N \2 Y_k$, is composed of subsystem operators such that 
\begin{equation}
A= A_1\otimes...\otimes A_N
\end{equation} 
where $A_k\in {\mathcal L}(\2 X_k,\2 Y_k)$ for $k=1,...,N$. As previously stated the vectorized composite operator $\dket{A}$ is a vector in the Hilbert space $\XY$. 

We define an operation $\2V_{N}$ called the \emph{unravelling} operation, the action of which unravels a vectorized matrix $\dket{A}= \dket{A_1\otimes\hdots\otimes A_N}$ into the tensor product of vectorized matrices on each individual subsystem $\2 X_k\otimes\2 Y_k$
\begin{equation}
\2V_{N} \dket{A_1\otimes\hdots\otimes A_N}
= \dket{A_1}\otimes\hdots\otimes\dket{A_N} \label{eqn:unravelling}. 
\end{equation}
The inverse operation then undoes the unravelling
\begin{equation}
\2V_{N}^{-1} \big(\dket{A_1}\otimes\hdots\otimes\dket{A_N}\big)
	=\dket{A_1\otimes\hdots\otimes A_n}. 
\end{equation}

More generally the unravelling operation $\2V_{N}$ is given by the map
\begin{eqnarray}
\2V_{N}: \ket{x_{\2 X}}\otimes\ket{y_{\2 Y}}
	&\longmapsto& 
	\bigotimes_{k=1}^{N}\left( \ket{x_k}\otimes\ket{y_{k}}\right)
\end{eqnarray}
where 
$\ket{x_{\2 X}} \equiv\ket{x_1}\otimes\hdots\otimes\ket{x_N},
\ket{y_{\2 Y}}\equiv\ket{y_1}\otimes\hdots\otimes\ket{y_N}$.
Hence we can write $\2 V_N$ in matrix form as
\begin{equation}
\2V_N = \sum_{i_1,\hdots,i_{N}}\sum_{j_1,\hdots,j_N} \ket{i_1,j_{1},\hdots,i_N,j_{N}}\bra{i_{\2 X},j_{\2 Y}}.
\label{eqn:unravelling2}
\end{equation}
where $\ket{i_{\2 X}} \equiv\ket{i_1}\otimes\hdots\otimes\ket{i_N},
\ket{j_{\2 Y}}\equiv\ket{j_1}\otimes\hdots\otimes\ket{j_N}$, and $\ket{i_k},\ket{j_l}$ are the standard bases for $\2 X_k$ and $\2 Y_l$ respectively.

We can also express $\2V_N$ as the composition of SWAP operations between two systems. For the previously considered composite operator $A\in {\mathcal L}(\2X,\2Y)$ we have that $\dket{A}$ has $2N$ subsystems. If we label the SWAP operation between two subsystem Hilbert spaces indexed by $k$ and $l$ by $\mbox{SWAP}_{k:l}$, where $1\le k,l\le2N$, then the unravelling operation can be composed as
\begin{equation}
\2V_{N} = W_{N-1}...W_{1}   
\label{eqn:unravelling3}      
\end{equation}
where
\begin{equation}
W_{k} =\prod_{j=0}^{k-1}\,\mbox{SWAP}_{N-k+2j+1:N-k+2j+2}.
\end{equation}
For example
\begin{eqnarray}
W_1 &=& \mbox{SWAP}_{N:N+1}\\
W_2 &=& \mbox{SWAP}_{N-1:N}\mbox{SWAP}_{N+1:N+2} \nonumber\\
W_{N-1} &=& \mbox{SWAP}_{2:3}\mbox{SWAP}_{4:5}\hdots \mbox{SWAP}_{2N-2:2N-1}. \nonumber
\end{eqnarray}
While this equation looks complicated, it has a more intuitive construction when depicted graphically. The tensor networks for the unravelling operation in the $N=2,3$ and $4$ cases are shown below
\begin{center}
 \begin{tabular}{ccc}
 \includegraphics[width=0.07\textwidth]{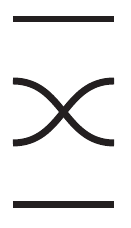}
 \quad\quad&\quad\quad
 \includegraphics[width=0.1\textwidth]{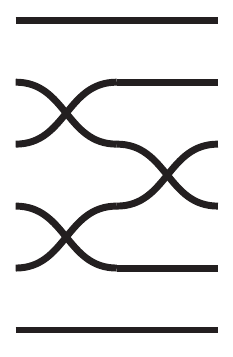}
 \quad\quad&\quad\quad
 \includegraphics[width=0.13\textwidth]{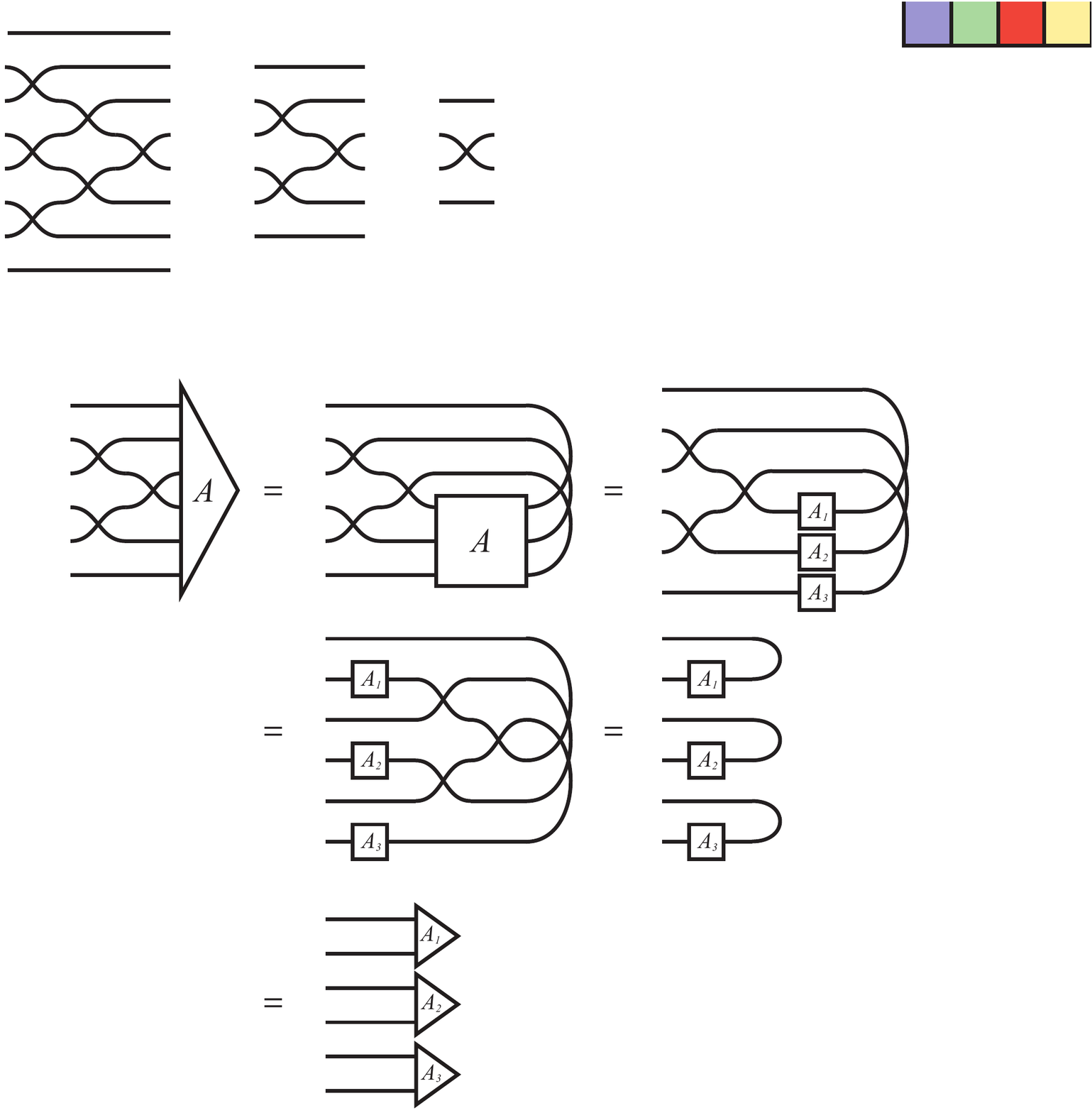}
 \\
\footnotesize{(a) $\2V_2$} 
 \quad\quad&\quad\quad
\footnotesize{(a) $\2V_3$} 
  \quad\quad&\quad\quad
  \footnotesize{(a) $\2V_4$} 
 \end{tabular}
  \label{fig:unravelling}
\end{center}
 We also present a graphical proof of this for the $N=3$ case: 
\begin{center}
\includegraphics[width=0.5\textwidth]{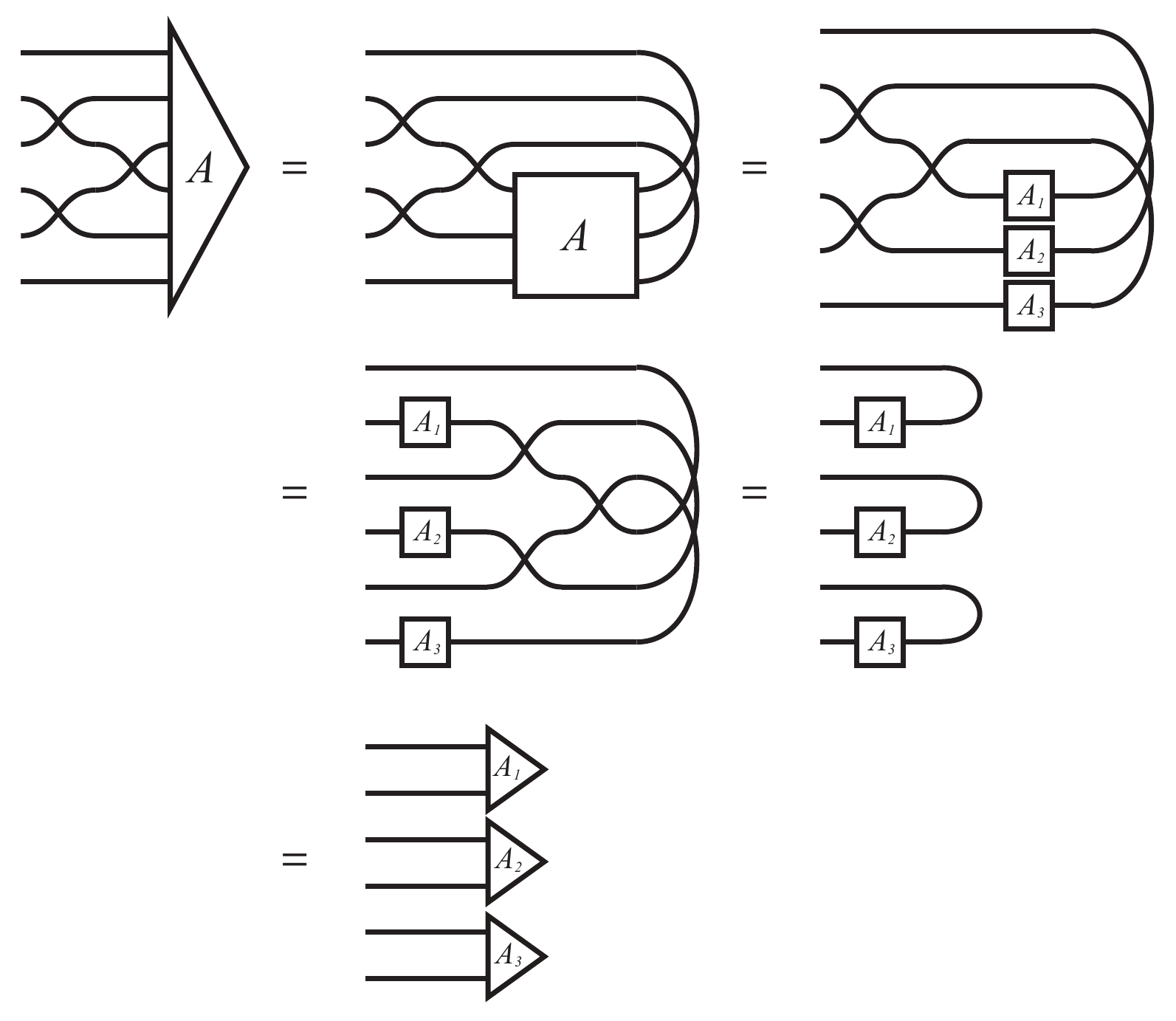}
\label{fig:unravelling-proof}
\end{center}


\subsection{Composing superoperators}
\label{sec:comp-sop}

We now discuss how to compose superoperators on individual subsystems to form the correct superoperator on the composite system, and vice-versa.  Given two superoperators $\2S_1$, and $\2S_2$, if we construct a joint system superoperator via tensor product $(\2S_1\otimes \2S_2)$, this composite operator acts on the tensor product of vectorized inputs $\dket{\rho_1}\otimes\dket{\rho_2}$, rather than the vectorization of the composite input $\dket{\rho_1\otimes\rho_2}$. To construct the correct composite superoperator for input $\dket{\rho_1\otimes\rho_2}$ we may use the unravelling operation $\2V_N$ from \eqref{eqn:unravelling} and its inverse.

If we have a set of superoperators $\{\2S_{k} : k=1,...,N\}$ where $ \2S_k\in {\mathcal L}(\2 X_k\otimes\2X_k,\2Y_k\otimes \2 Y_{k})$, then the joint superoperator $\2 S\in {\mathcal L}(\2X\otimes\2X,\2Y\otimes\2Y)$, where $\2 X=\bigotimes_{k=1}^N \2 X_k$, $\2 Y=\bigotimes_{k=1}^N \2 Y_k$,  is given by
\begin{equation}
\2 S = \2 V_N^{\dagger}\left( \2 S_{1}\otimes\hdots\otimes\2 S_{N}\right)\2 V_N.
\label{eqn:comp-sop}
\end{equation} 
The tensor networks for this transformation in the $N=2$ and $N=3$ cases are shown below
\begin{center}
\begin{tabular}{c|c}
\includegraphics[width=0.3\textwidth]{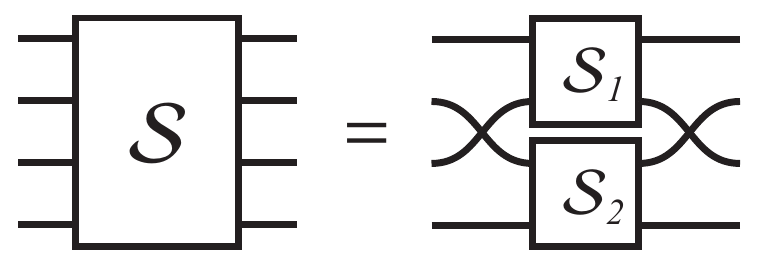}
\quad\quad&\quad\quad
\includegraphics[width=0.38\textwidth]{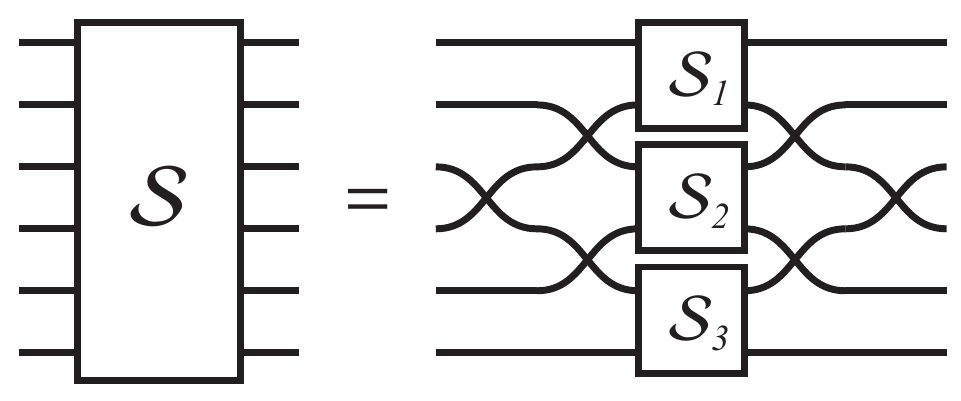}
\\
\footnotesize{$N=2$}
\quad\quad&\quad\quad
\footnotesize{$N=3$}
\end{tabular}
\label{fig:comp-sop}
\end{center}

Composing superoperators from individual subsystem superoperators is useful when performing the same computations for multiple identical systems. For an example we consider vectorization in the Pauli-basis for an $N$-qubit system. While it is generally computationally more efficient to perform vectorization calculations in the col-vec (or row-vec) basis, as these may be implemented using structural operations on arrays, it is often convenient to express the superoperator in the Pauli basis, or the Choi-matrix in the $\chi$-matrix representation, when we are interested in determining the form of correlated errors. However, transforming from the col-vec to the Pauli-basis for multiple (and possibly arbitrary) number of qubits is inconvenient. Using our unravelling operation we can instead compute the single qubit change of basis superoperator $T_{c\rightarrow\sigma}$, where $\sigma=\{\I,X,Y,Z\}/\sqrt{2}$ is the Pauli-basis for a single qubit, and use this to generate the transformation operator for multiple qubits. In the case of $N$-qubits we can construct the basis transformation matrix as
\be
T^{(N)}_{c\rightarrow \sigma} = \2 V_N^{\dagger}\cdot T_{c\rightarrow\sigma}^{\otimes N} \cdot \2 V_N.
\ee
The joint-system superoperator in the Pauli-basis is then given by 
\be
\2 S^\sigma = T^{(N)}_{c\rightarrow \sigma}\cdot\2 S \cdot T^{(N)\dagger}_{c\rightarrow \sigma}
\ee
The same transformation can be used for converting a state $\rho=\rho_1\otimes\hdots\otimes\rho_N$ to the Pauli basis: $\dket{\rho}_\sigma = T^{(N)}_{c\rightarrow \sigma}\dket{\rho}_c$. These unravelling techniques are also useful for applying operations to a limited number of subsystems in a tensor network as used in many tensor network algorithms.


\subsection{Matrix operations as superoperators}
\label{sec:ops-as-sops}

We now show how several common matrix manipulations can be written as superoperators. These expressions are obtained by simply vectorizing the transformed operators. We begin with the trace superoperator $S_{\mbox{\scriptsize{Tr}}}$ which implements the trace of a matrix $S_{\mbox{\scriptsize{Tr}}}\dket{A} := \Tr[A]$ for a square matrix $A\in {\mathcal L}(\2X)$. This operation is simply given by the adjoint of the unnormalized Bell-state:
\be
S_{\mbox{\scriptsize{Tr}}} := \dbra{\I}
\ee
where $\I\in {\mathcal L}(\2X)$ is the identity operator. If $\2X$ is itself a composite system, we simply use the definition of the Bell-state for composite systems from \Eqref{eqn:comp-bell}. This is illustrated in our graphical calculus as
\be
S_{\mbox{\scriptsize{Tr}}}\dket{A} = \parbox[c]{1em}{\includegraphics[width=0.1\textwidth]{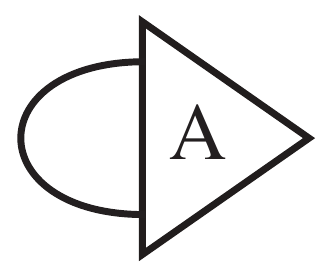}}
\label{fig:vec-tr}
\ee

For a rectangular matrix $B\in {\mathcal L}(\2X,\2Y)$, the transpose superoperator $S_{T}$ which implements the transpose $S_T\dket{B}=\dket{B^T}$ is simply a swap superoperator between $\2X$ and $\2 Y$.
\begin{eqnarray}
S_{T} &=& \mbox{SWAP}	\\
\mbox{SWAP}&:& \XY \mapsto \YX  
\end{eqnarray}
The tensor network for the swap superoperator is 
\be
S_{\mbox{\scriptsize{T}}}\dket{B} = \parbox[c]{0.1\textwidth}{\includegraphics[width=0.1\textwidth]{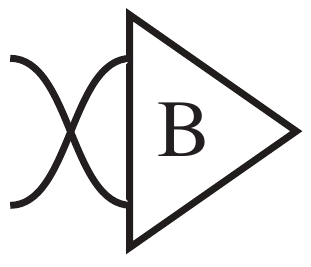}}
\label{fig:vec-trans}
\ee
If $\2X$ and $\2Y$ are composite vector spaces we may split the crossed wires into their respective subsystem wires. 

Next we give the superoperator representations of the bipartite matrix operations in \eqref{fig:bipartite} acting on vectorized square bipartite matrices $M\in {\mathcal L}(\XY)$. These are the partial trace over $\2X$ ($S_{\mbox{\scriptsize{Tr}}_{\2 X}}$) (and $S_{\mbox{\scriptsize{Tr}}_{\2Y}}$ over $\2 Y$), transposition $S_T$, 
and col-reshuffling $(S_{R_c})$. 
\begin{eqnarray}
S_{\mbox{\scriptsize{Tr}}_{\2 X}}
	&:& \XY\otimes\XY \mapsto \2Y\otimes\2Y\\
S_{\mbox{\scriptsize{Tr}}_{\2 Y}}
	&:& \XY\otimes\XY \mapsto \2X\otimes\2X\\
S_{\mbox{\scriptsize{T}}}
	&:&	 \XY\otimes\XY \mapsto \XY\otimes\XY\\
S_{\scriptsize{R_c}}
	&:&	 \XY\otimes\XY \mapsto \XX\otimes\YY
\end{eqnarray}

\begin{fullpage}
The graphical representation of the superoperators for these operations are:
\begin{center}
\begin{tabular}{c|c|c|c}
$S_{\mbox{\scriptsize{Tr}}_{\2 X}}\dket{M}
=$ \parbox[c]{0.11\textwidth}{\includegraphics[width=0.09\textwidth]{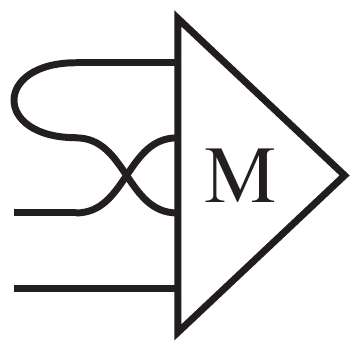}}
&\quad
$S_{\mbox{\scriptsize{Tr}}_{\2 Y}}\dket{M}
=$ \parbox[c]{0.11\textwidth}{\includegraphics[width=0.09\textwidth]{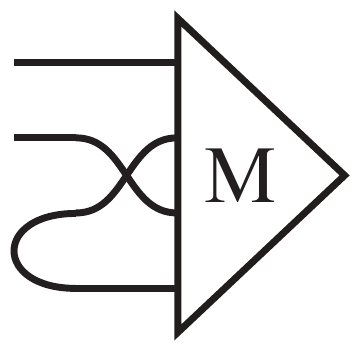}}
&\quad
$S_{\mbox{\scriptsize{T}}}\dket{M}
=$ \parbox[c]{0.11\textwidth}{\includegraphics[width=0.09\textwidth]{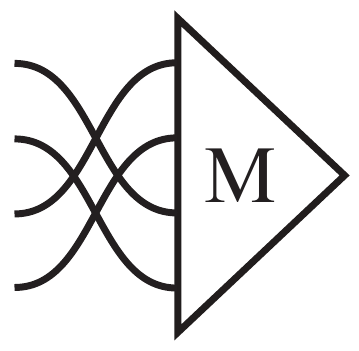}}
&\quad
$S_{\scriptsize{R_c}}\dket{M}
=$ \parbox[c]{0.11\textwidth}{\includegraphics[width=0.09\textwidth]{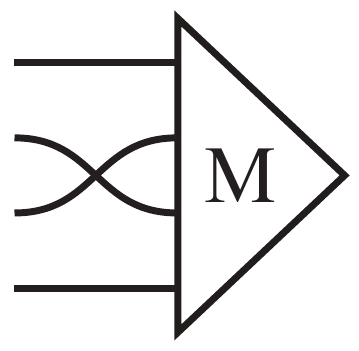}}
\end{tabular}
\label{fig:bipartite-sop}
\end{center}
\end{fullpage} 

Algebraically they are given by
\begin{eqnarray}
S_{\mbox{\scriptsize{Tr}}_{\2 X}}
	&=&	\big[\dbra{\I_{\2 X}}\otimes\I_{\2 Y}\otimes \I_{\2 Y}\big]\2 V_2\\
S_{\mbox{\scriptsize{Tr}}_{\2 Y}}
	&=&	\big[\I_{\2 X}\otimes \I_{\2 X}\otimes\dbra{\I_{\2 Y}}\big]\2 V_2\\
S_{\mbox{\scriptsize{T}}}
	&=&	\mbox{SWAP}_{1:3}\mbox{SWAP}_{2:4}\\
S_{\scriptsize{R_c}}
	&=&	 \2V_2
\end{eqnarray}
where $\2V_2$ is the unravelling operation in \Eqref{eqn:unravelling}. 

In the general multipartite case for a composite matrix $A\in {\mathcal L}(\2X)$ where $\2 X=\bigotimes_{k=1}^N \2 X_k$, we can trace out or transpose a subsystem $j$ by using the unravelling operation in \Eqref{eqn:unravelling} to insert the appropriate superoperator for that subsystem with identity superoperators on the remaining subsystems:
\begin{eqnarray*}
\2S_{O_j} 
&=& 
	\2 V_{N-1}^{-1}\left[
		\left(\bigotimes_{k=1}^{j-1}\2S_{\2 I_{k}}\right)
		\otimes\2S_{O}\otimes
		\left(\bigotimes_{k=j+1}^{N}\2S_{\2 I_{k}}\right)
	\right]\2 V_N	\\
\end{eqnarray*}
where $\2S_O\in T(\2 X_j)$ is the superoperator acting on system $j$ and $\2S_{\2 I_k}\in T(\2 X_k)$ is the identity superoperator for subsystem $ {\mathcal L}(\2 X_k)$. Similarly by inserting the appropriate operators at multiple subsystem locations we can perform the partial trace or partial transpose of any number of subsystems.


\subsection{Reduced superoperators}
\label{sec:reduced-sop}

We now present a simple but useful example of the presented bipartite operations in the superoperator representation to show how to construct an effective reduced superoperator for a a subsystem out of a larger superoperator on a composite system.

Consider states $\rho_{XY}\in {\mathcal L}(\XY)$ which undergo some channel $\2 F\in C(\XY)$ with superoperator representation $\2S$. Suppose system $\2Y$ is an ancilla which we initialize in some state $\tau_{0}\in {\mathcal L}(\2Y)$, and we post-select on the output state of system $\2Y$ being in a state $\tau_{1}$. We may construct the effective reduced map $\2 F^\prime\in T(\2X)$ for this process for arbitrary input and output states of system $\2X$, given by a superoperator $\2S^\prime$, as shown:
\begin{center}
\includegraphics[width=0.35\textwidth]{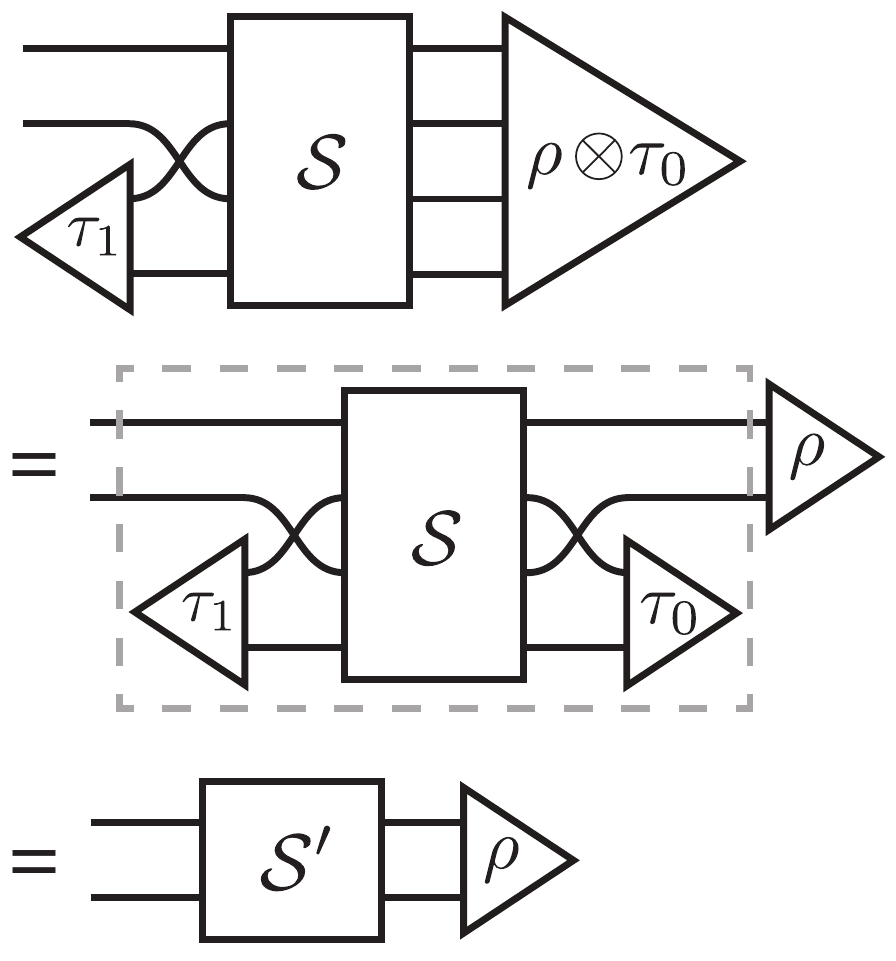}
\label{fig:reduced-sop}
\end{center}

Formally, we are defining the superoperator representation $\2S^\prime$ of the effective channel $\2 F^\prime$ as the map
\begin{eqnarray}
\2 F^\prime(\rho) &=& \Tr_{\2 Y}\left[(\I_{\2X}\otimes\tau_1)\2 F(\rho\otimes\tau_0)\right]	\\
\2 S^\prime &=& \dbra{\tau_1}\2V_2 \2S\2V_2^\dagger\dket{\tau_0}
\end{eqnarray} 
where $\2S$ is the superoperator representation of $\2F$ and $\dket{\tau_j}$ is implicitly assumed to have the identity operation on the vectorization of subsystem $\2X\otimes\2X$ $(\dket{\tau_j}:= \I_{\2X}\otimes\I_{\2X}\otimes \dket{\tau_j})$.


\subsection{Ancilla Assisted Process Tomography}
\label{sec:aapt}

Quantum state tomography is the method of reconstructing an unknown quantum state from the measurement statistics obtained by performing a topographically complete set of measurements on many identical copies of the unknown state~\cite{NC}. Quantum process tomography is an extension of quantum state tomography which reconstructs an unknown quantum channel $\2E\in C(\2X)$ from appropriately generated measurement statistics. One such procedure, known as \emph{standard quantum process tomography}, involves preparing many copies of each of a topographically complete set of input states, subjecting each to the unknown quantum channel, and performing state tomography on the output~\cite{DAriano2001}. 

An alternative approach is to directly measure the Choi-matrix for the channel via a method known as \emph{ancilla assisted process tomography} (AAPT)~\cite{White2003}. The simplest case of AAPT is \emph{entanglement assisted process tomography}(EAPT) which is an experimental realization of the Choi-Jamio{\l}kowski isomorphism. Here an experimenter prepares a a maximally entangled state 
\be
\rho_{\Phi}= \frac{1}{d}\dketdbra{\I}{\I}
\ee 
across the system of interest $\2X$ and an ancilla $\2Z \cong \2X$, and subjects the system to the unknown channel $\2E$, and the ancilla to an identity channel $\2I$. The output of this joint system-ancilla channel is the rescaled Choi-matrix:
\begin{eqnarray}
\rho_{\phi}^\prime = (\2I\otimes\2E)\left(\rho_{\Phi}\right) 
	&=& \frac{\Lambda}{d}.
\end{eqnarray}
which can be measured directly by quantum state tomography. The tensor network for EAPT is
\begin{center}
\includegraphics[width=0.4\textwidth]{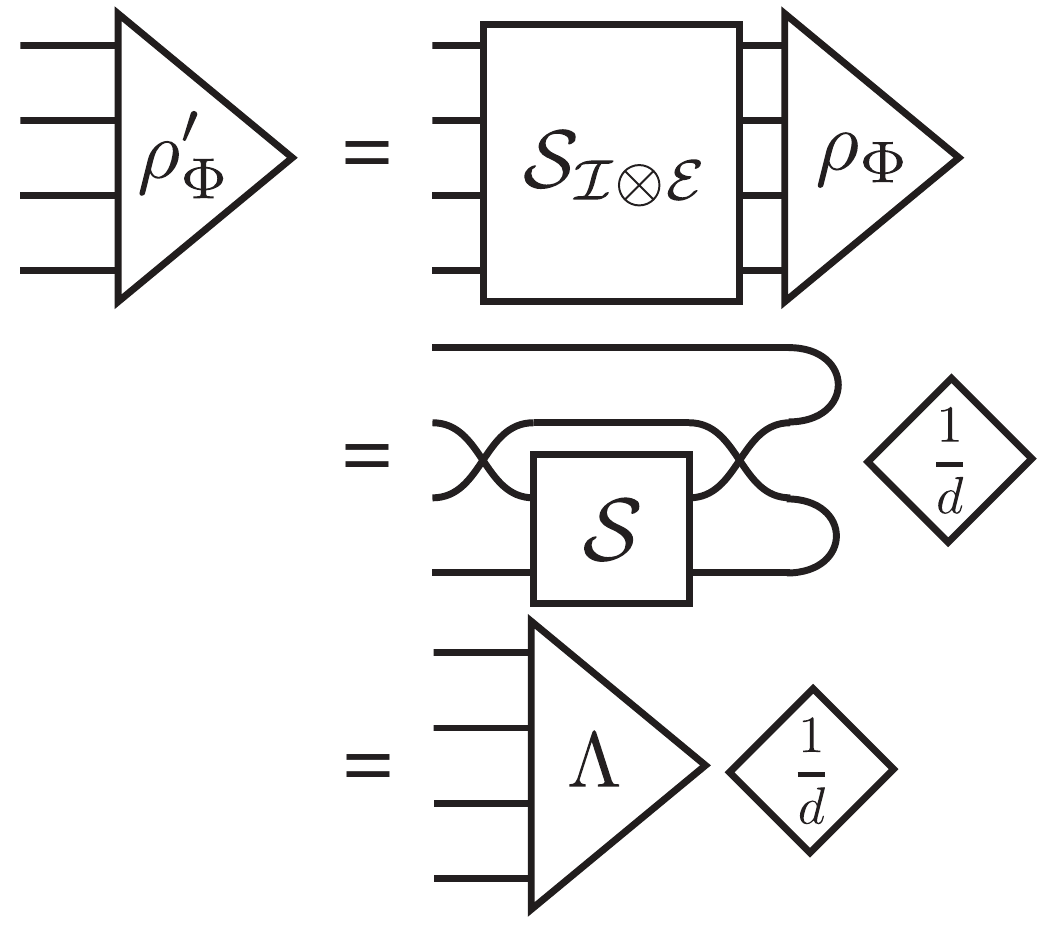}
\label{fig:aapt-choi}
\end{center}

In general AAPT does not require $\rho_{AS}$ to be maximally entangled. It has been demonstrated experimentally that AAPT may be done with a state which does not have any entanglement at all, at the expense of an increase in the estimation error of the reconstructed channel~\cite{White2003}. A necessary and sufficient condition for a general state $\rho_{AS}$ to allow recovery of the Choi-matrix of an unknown channel $\2E$ via AAPT is that it have a \emph{Schmidt number} equal to $d^2$ where $d$ is the dimension of the state space $\2X$~\cite{White2003}. This conditions has previously been called \emph{faithfulness} of the input state, and one can recover the original Choi-matrix for the unknown channel $\2E$ by applying an appropriate inverse map to the output state in post-processing~\cite{DAriano2003}. We provide an arguably simpler derivation of this condition, and the explicit construction of the inverse recovery operator. The essence of this proof is that we can consider the bipartite state $\rho_{AS}$ to be Choi-matrix for an effective channel via the Choi-Jamio{\l}kowski isomorphism (but with trace normalization of 1 instead of $d$) . We can then apply channel transformations to this initial state to convert it into an effective channel acting on the true Choi-matrix, and if this effective channel is invertible we can recover the Choi-matrix for the channel $\2E$ by applying the appropriate inverse channel.

\begin{proposition}
\label{prop:aapt}
$(a)$ A state $\rho_{AS}\in {\mathcal L}(\2X\otimes\2X)$ may be used for AAPT of an unknown channel $\2E\in  C(\2X)$ if and only if the reshuffled density matrix $\2S_{AS} = \rho_{AS}^{R_c}$ is invertible. 

$(b)$ The channel can be reconstructed from the measured output state by $\Lambda_{\2E} = (\2R\otimes \2I)(\rho_{AS}^\prime)$
where $\rho_{AS}^\prime = (\2 I\otimes \2E)(\rho_{AS})$ is the output state reconstructed by quantum state tomography, and $\2R$ is the recovery channel given by superoperator $\2S_{\2R} = (\2S_{AS}^T)^{-1}$.
\end{proposition}

The graphical proof of Prop.~\ref{prop:aapt}  is illustrated in Fig.~\ref{fig:aapt-proof}. This proof demonstrates several useful features of the presented graphical calculus. In particular it applies the vectorized reshuffling transformation to a bipartite density matrix input state to obtain an effective superoperator representation of a state, and uses the unravelling operation for composition of superoperators.
From this construction we find that if the initial state $\rho_{AS}$ is maximally entangled, then it can be expressed as $\rho_{AS}= \dketdbra{V}{V}$ for some unitary $V$. In this case the reshuffled superoperator of the state corresponds to a unitary channel $\2S_{AS} = \overline{V}\otimes V$, and hence is invertible with $\2S_{AS}^{-1} = \2S_{AS}^\dagger$. If the input state is not maximally entangled, then the closer it is to a singular matrix, the larger the condition number and hence the larger the amplification in error when inverting the matrix.

\begin{figure}[ht]
\centering
\includegraphics[width=0.4\textwidth]{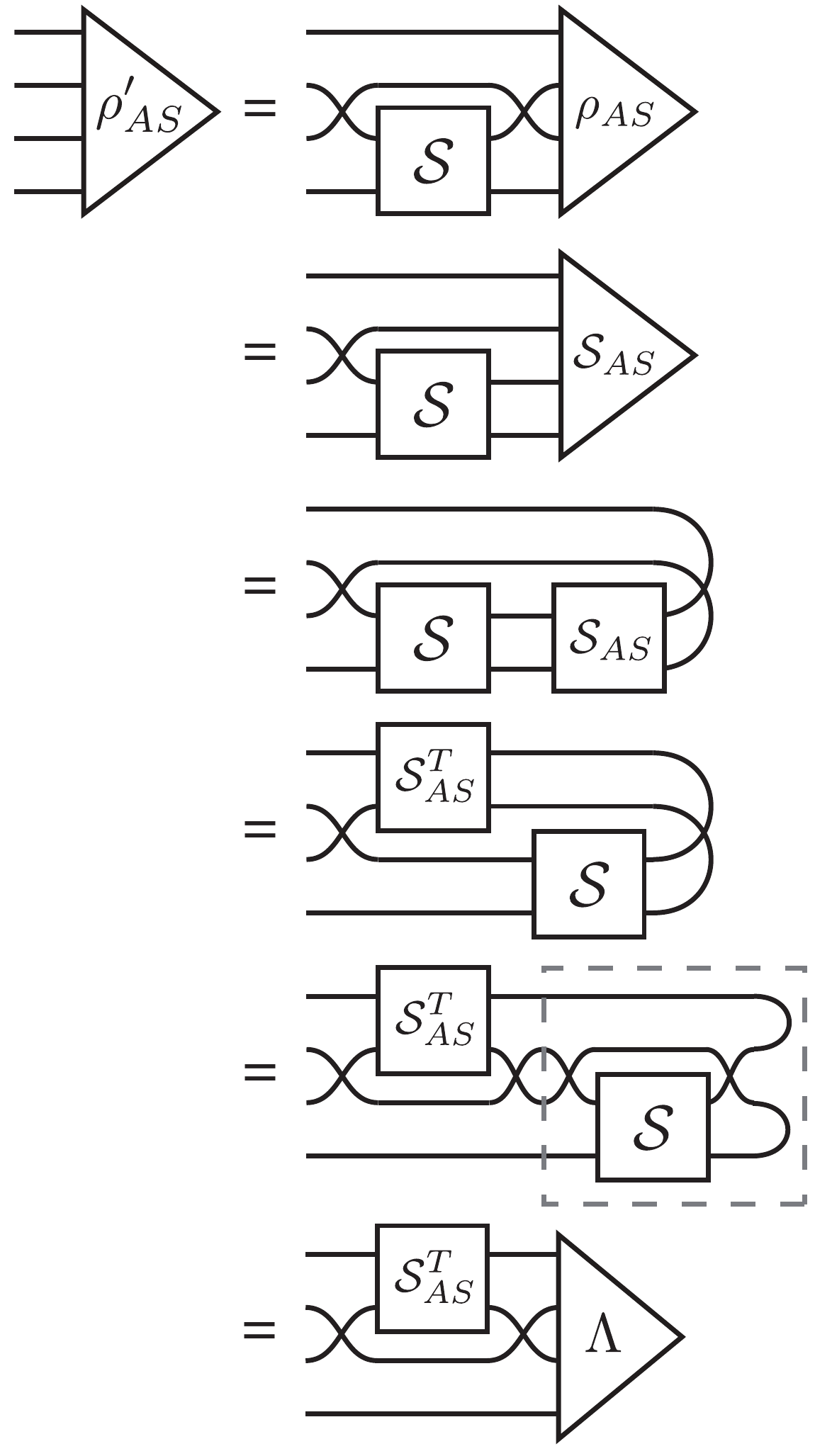}
\caption{Graphical proof of the equivalence of an initial state $\rho_{AS}$ used for performing AAPT of an unknown CPTP map $\2E$ with superoperator representation $\2S$, to a channel $(\2R\otimes\2I)$ acting on the Choi-matrix $\Lambda$ for a channel $\2E$. The Choi-matrix can be recovered if and only if the the superoperator $\2S_{\2R}=\2S_{AS}^T=(\rho_{AS}^{R_c})^T$ is invertible.}
\label{fig:aapt-proof}
\end{figure}


\subsection{Average Gate Fidelity}
\label{sec:gatefid}

When characterizing the performance of a noisy quantum channels a widely used measure of the closeness of a CPTP map $\2E\in C(\2X)$ to a desired quantum channel $\2F\in C(\2X)$ is the \emph{Gate Fidelity}. This is defined to be
\be
F_{\2E,\2F}(\rho) = F(\2F(\rho)\2E(\rho))
\ee
where
\be
F(\rho,\sigma)=\left(\Tr\left[ \sqrt{\sqrt{\rho}\sigma\sqrt{\rho}}\right]\right)^2
\ee
is the fidelity function for quantum states~\cite{NC}.

In general we are interested in comparing a channel $\2E$ to a unitary map $\2U\in C(\2X)$ where $\2U(\rho)=U\rho U^\dagger$. In this case we have
\begin{eqnarray}
F_{\2E,\2U}(\rho) 
	&=& \left[\Tr\sqrt{\sqrt{U\rho U^\dagger}\2E(\rho)\sqrt{U\rho U^\dagger}}\right]^2\\
	&=&  \left[\Tr\sqrt{\sqrt{\rho}U^\dagger\2E(\rho)U\sqrt{\rho}}\right]^2\\
	&=&  \left[\Tr\sqrt{\sqrt{\rho}\,\2U^\dagger(\2E(\rho))\sqrt{\rho}}\right]^2\\
	&=& F_{\2U^\dagger\2E,\2 I}(\rho)
\end{eqnarray}
where $\2I$ is the identity channel and $\2U^\dagger(\rho) = U^\dagger \rho U$, is the adjoint channel of the unitary channel $\2U$. Thus without loss of generality we may consider the gate fidelity $F_{\2E}(\rho)\equiv F_{\2U^\dagger\2F,\2I}(\rho)$ comparing $\2E$ to the identity channel, where we simply define $\2E \equiv \2U^\dagger\2F$ if we wish to compare $\2F$ to a target unitary channel $\2U$.

The most often used quantity derived from the gate fidelity is the \emph{average gate fidelity} taken by averaging $F_{\2E}(\rho)$ over the the Fubini-Study measure. Explicitly the average gate fidelity is defined by
\be
\overline{F}_{\2E} = \int d\,\psi\, \bra{\psi} \2E(\ketbra\psi\psi)\ket{\psi}.
\label{eq:avegf}
\ee
where due to the concavity of quantum states we need only integrate over pure states $F_{\2E}(\ketbra\psi\psi) = \bra\psi \2E(\ketbra\psi\psi)\ket\psi$. 

Average gate fidelity is a widely used figure of merit in part because it is simple to compute. The expression in \eqref{eq:avegf} reduces to explicit expression for $\overline{F}_{\2E}$ in terms of a single parameter of the channel $\2E$ itself. This has previously been given in terms of the Kraus representation~\cite{Horodecki1999,Nielsen2002}, superoperator~\cite{Emerson2005} and Choi-matrix in~\cite{Johnston2011}. We now present an equivalent graphical derivation of the average gate fidelity in terms of the Choi-matrix which we believe is simpler than previous derivations. We start with the tensor network diagram corresponding to \eqref{eq:avegf} and perform graphical manipulations as follows
\begin{center}
\includegraphics[width=0.35\textwidth]{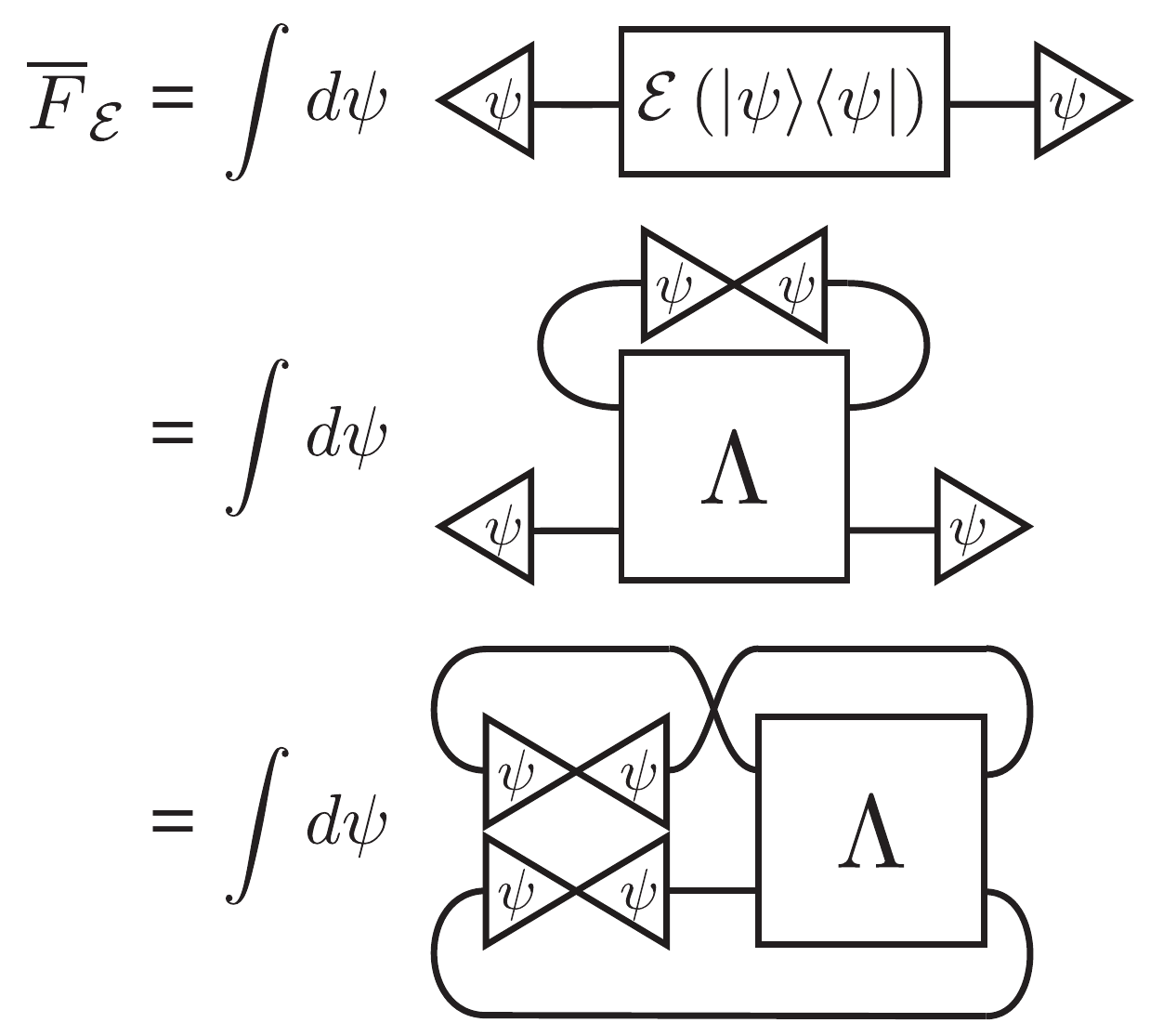}
\label{fig:ave-fid-proof-1}
\end{center}
For the next step of the proof we use the result that the average over $\psi$ of a tensor product of states $\ketbra{\psi}{\psi}^n$ is given by
\be
\int d\psi\, \ketbra\psi\psi^{\otimes n} = \frac{\Pi_{\scriptsize\mbox{sym}}(n,d)}{\Tr[\Pi_{\scriptsize\mbox{sym}}(n,d)]}
\label{eq:symsub}
\ee
where $ \Pi_{\scriptsize\mbox{sym}}(n,d)$ is the projector onto the symmetric subspace of $\2X^{\otimes n}$. This project may be written as~\cite{Magesan2011}
\be
\Pi_{\scriptsize\mbox{sym}}(n,d) = \frac{1}{n!}\sum_\sigma P_{\sigma}
\label{eq:perm-sum}
\ee
where $P_{\sigma}$ are operators for the permutation $\sigma$ of $n$-indices. These permutations may be represented as a swap type operator with $n$ tensor wires. For the case of $n=2$ we have the tensor diagram:
\begin{center}
\includegraphics[width=0.45\textwidth]{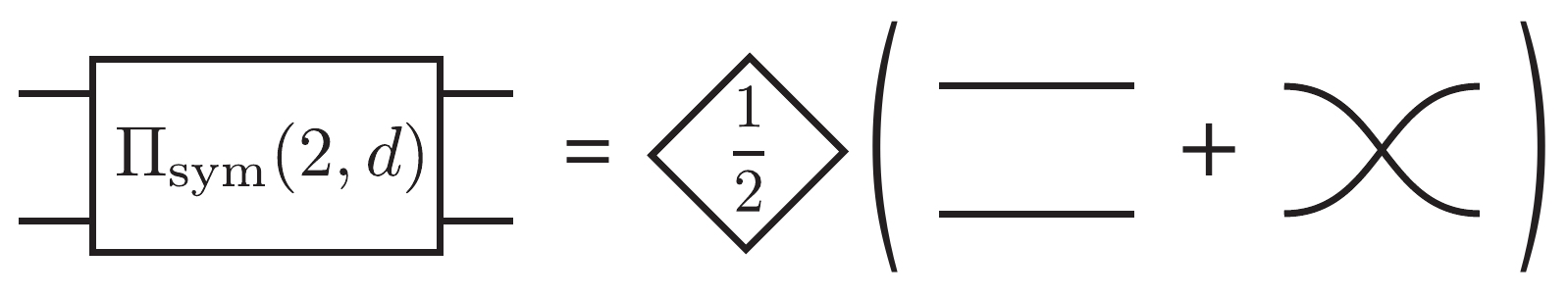}
\label{fig:symsub2}
\end{center}
Here we can see that $\Tr[\Pi_{\scriptsize\mbox{sum}}(2,d)] = (d^2 + d)/2$, and hence we have that
\begin{eqnarray}
\Pi_{\scriptsize\mbox{sym}}(2,d) &=& \frac{1}{2}\left(\I\otimes\I + \mbox{SWAP}\right)	\\
\Tr[\Pi_{\scriptsize\mbox{sym}}(2,d)] &=& \frac{d^2+d}{2}	\\
\Rightarrow
\int d\psi\, \ketbra\psi\psi^2 &=& \frac{\I\otimes\I + \mbox{SWAP}}{d(d+1)}
\end{eqnarray}
where $\2X\cong \C^d$, $\I\in {\mathcal L}(\2X)$ is the identity operator, and SWAP is the SWAP operation on $\XX$. Subsituting \eqref{fig:symsub2} into \eqref{fig:ave-fid-proof-1} completes the proof:
\begin{center}
\includegraphics[width=0.48\textwidth]{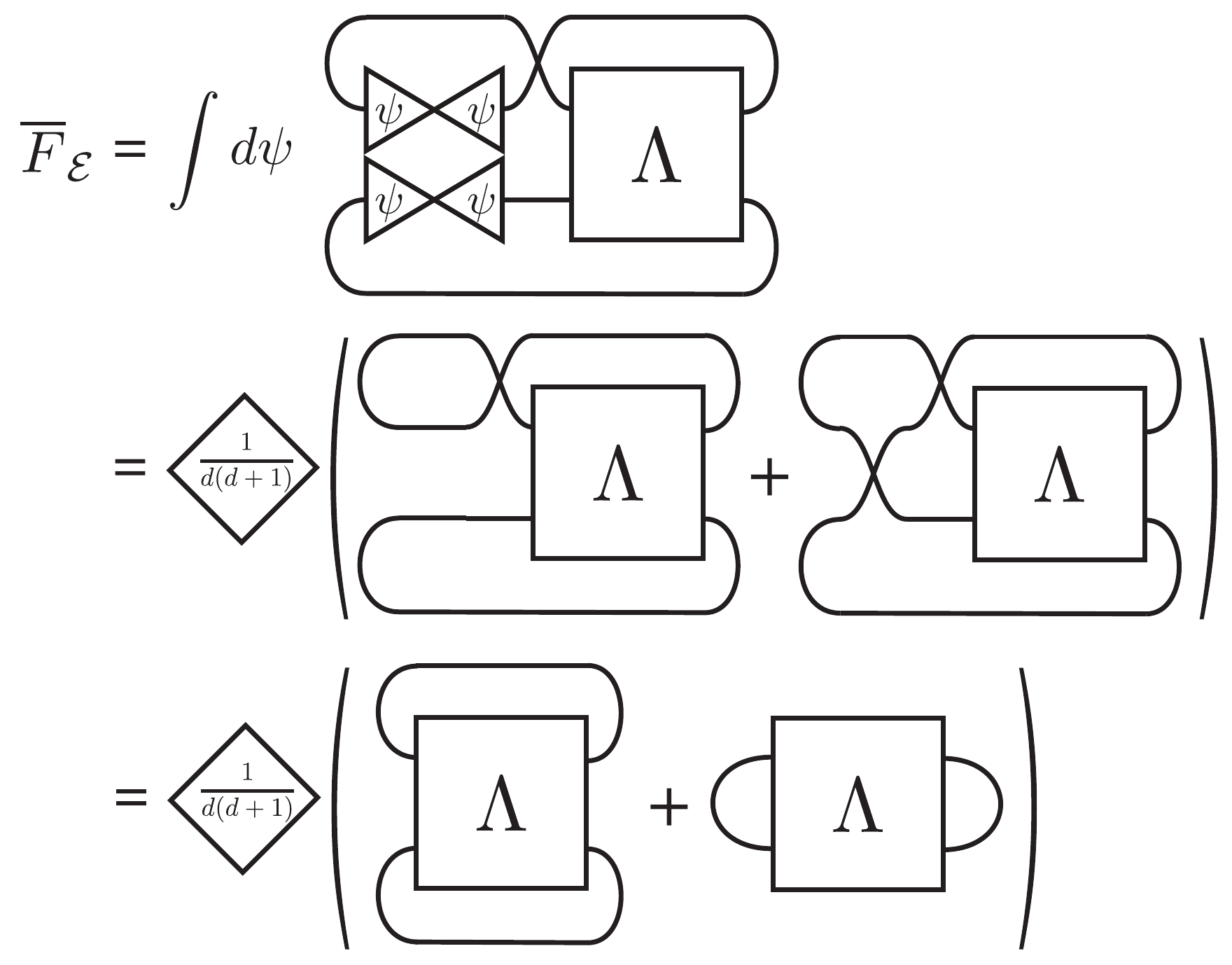}
\label{fig:ave-fid-proof-2}
\end{center}
Hence we have that the average gate fidelity in terms of the Choi-matrix is given by
\be
\overline{F}_{\2E} = \frac{d + \dbra{\I}\Lambda\dket{\I}}{d(d+1)}
\ee
where we have used the fact that the Choi-matrix is normalized such that$\Tr[\Lambda] = d$.
From this proof one may derive expressions for the other representations using the channel transformations in \S~\ref{sec:trans}. The resulting expressions are
\begin{eqnarray}
\overline{F}_{\2E} 
	&=& \frac{d + \Tr[\2S]}{d(d+1)}	\\
	&=& \frac{d + \dbra{\I}\Lambda\dket{\I}}{d(d+1)}	\\
	&=& \frac{d + \sum_j |\Tr[K_j]|^2}{d(d+1)}		\\
	&=& \frac{d + d\chi_{00}}{d(d+1)}		\\
	&=& \frac{d + \Tr_{\2X}[A^\dagger]\cdot\Tr_{\2X}[A]}{d(d+1)}
\end{eqnarray}
where $\2S$, $\Lambda$, $\{K_j\}$, $\chi$, $A$ are the superoperator, Choi-matrix, Kraus, $\chi$-matrix and Strinespring representations for $\2E$ respectively. In the case of the $\chi$-matrix representation, $\chi$ is defined with respect to a basis $\{\sigma_j\}$ satisfying $\Tr[\sigma_j] = \sqrt{d}\delta_{j,0}$.

Similar techniques can be applied for tensor networks that may be graphically manipulated into containing a term $\int d\psi\, \ketbra\psi\psi^{\otimes n} $ for $n>2$. This could prove useful for computing higher order moments of fidelity functions and other quantities defined in terms of averages over quantum states $\ket\psi$. In this case there are $n!$ permutations of the tensor wires for the permutation operator $P_\sigma$ in \eqref{eq:perm-sum}, and these can be decomposed as a series of SWAP gates. For example, in the case of $n=3$ we have
\begin{eqnarray}
\Pi_{\scriptsize\mbox{sym}}(3,d) 
&=& \frac{1}{6}\big(
\I^{\otimes 3} + \mbox{SWAP}_{1:2} 
+ \mbox{SWAP}_{1:3}
+ \mbox{SWAP}_{2:3}\\\nonumber 
&+& \mbox{SWAP}_{1:2}\mbox{SWAP}_{2:3}
+ \mbox{SWAP}_{2:3}\mbox{SWAP}_{1:2}\big) 
\end{eqnarray}

\begin{equation}
\Tr[\Pi_{\scriptsize\mbox{sym}}(3,d)] = \frac{d^3+3d^2+2d}{6}.
\end{equation}


\subsection{Entanglement Fidelity}
\label{sec:entfid}

Another useful fidelity quantity is the \emph{entanglement fidelity} which quantifies how well a channel preserves entanglement with an ancilla~\cite{Schumarcher1996,NC}. For a CPTP map $\2E\in C(\2X)$ and density matrix $\rho\in {\mathcal L}(\2X)$ the entanglement fidelity is given by 
\begin{eqnarray}
F_{\scriptsize\mbox{e}}(\2E,\rho) 
&=& \inf
	\big\{
	F\left(\ketbra{\psi}{\psi}, (\2I_{\2Z}\otimes\2E)(\ketbra{\psi}{\psi})\right):
\nonumber\\&&
	\Tr_{\2Z}[\ketbra{\psi}{\psi}]=\rho \big\}
\label{eq:ent-fid}
\end{eqnarray}
where $\ket\psi\in \2X\otimes\2Z$ is a purification of $\rho$ over an ancilla $\2Z$. Entanglement fidelity turns out to be independent of the choice of purification $\ket\psi$, and a closed form expression has been given in terms of the Kraus representation~\cite{NC} and Choi-matrix~\cite{Fletcher2007}. Here we present a simple equivalent derivation in terms of the Choi-matrix representation of the channel $\2E$ using graphical techniques. Then by applying the channel transformations of \S~\ref{sec:trans} we obtain expressions in terms of the other representations. The resulting expressions for entanglement fidelity are:
\begin{eqnarray}
F_{\scriptsize\mbox{e}}(\2E,\rho) 
	&=& \dbra{\rho}\Lambda\dket{\rho}	\\
	&=& \Tr\left[(\rho^T\otimes\rho)\2S\right]	\\
	&=& \sum_j |\Tr[\rho K_j]|^2 \\
	&=& \sum_{i,j}\chi_{ij} \Tr[\rho\, \sigma_i]\Tr[\rho \sigma^\dagger_j]\\
	&=&\Tr_{\2X}[\rho A^\dagger]\cdot\Tr_{\2X}[A\rho]
\end{eqnarray}
where $\2S$, $\Lambda$, $\{K_j\}$, $\chi$, $A$ are the superoperator, Choi-matrix, Kraus, $\chi$-matrix and Strinespring representations for $\2E$ respectively. In the case of the $\chi$-matrix representation, $\chi$ is defined with respect to a basis $\{\sigma_j\}$ satisfying $\Tr[\sigma_j] = \sqrt{d}\delta_{j,0}$. 

For the graphical proof in terms of the Choi-representation we start with \eqref{eq:ent-fid} and perform the following tensor manipulations
\begin{center}
\includegraphics[width=0.45\textwidth]{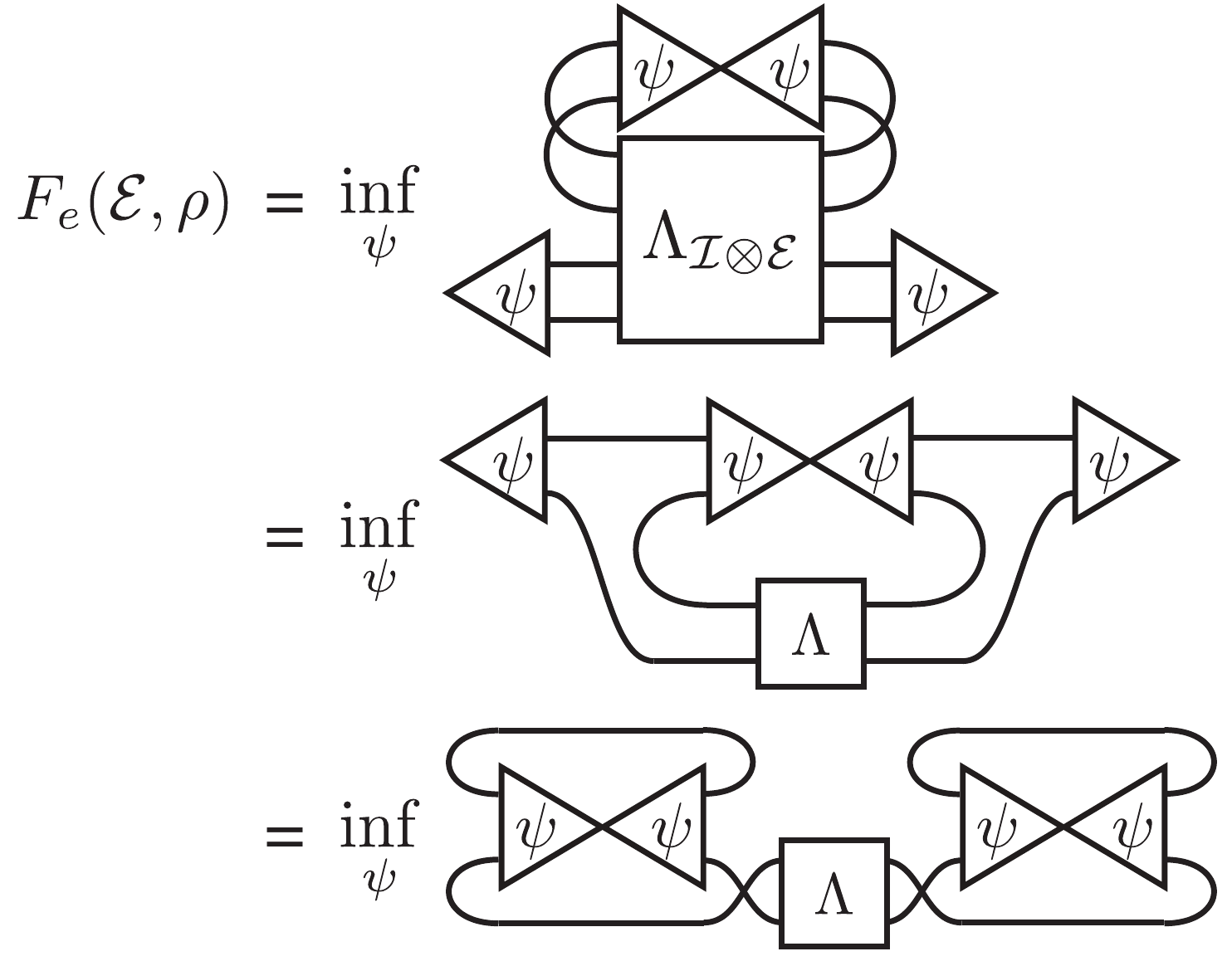}
\label{fig:ent-fid-proof-1}
\end{center}
Now since the infimum is over all $\ket{\psi}\in \2Z\otimes\2X$ satisfying $\Tr_{\2Z}[\ketbra{\psi}{\psi}]=\rho$ the result is independent of the specific purification $\psi$ and we have:
\begin{center}
\includegraphics[width=0.45\textwidth]{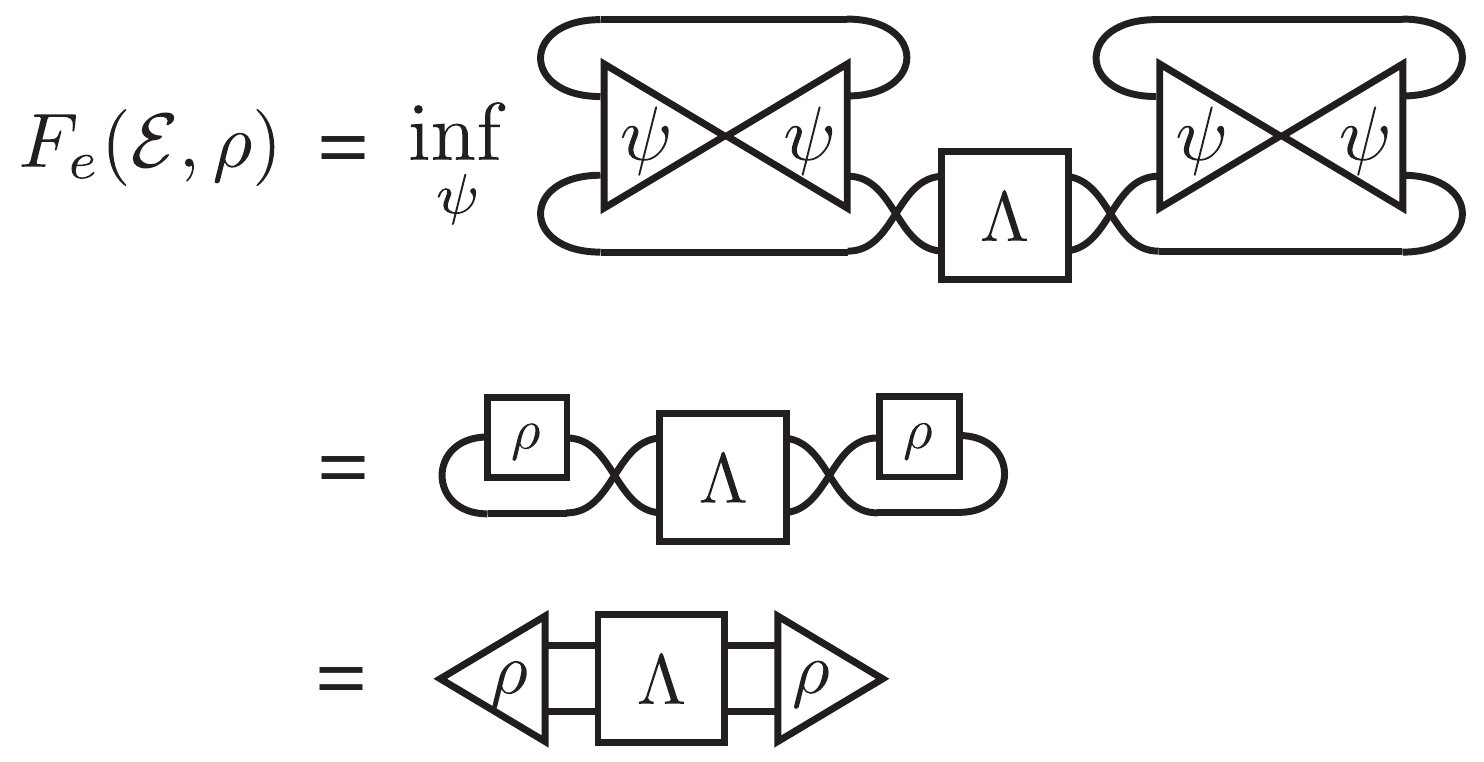}
\label{fig:ent-fid-proof-2}
\end{center}

Entanglement fidelity is equivalent to gate fidelity for pure states and hence average entanglement fidelity is equivalent to average gate fidelity. This can be shown graphically as follows
\begin{center}
\includegraphics[width=0.45\textwidth]{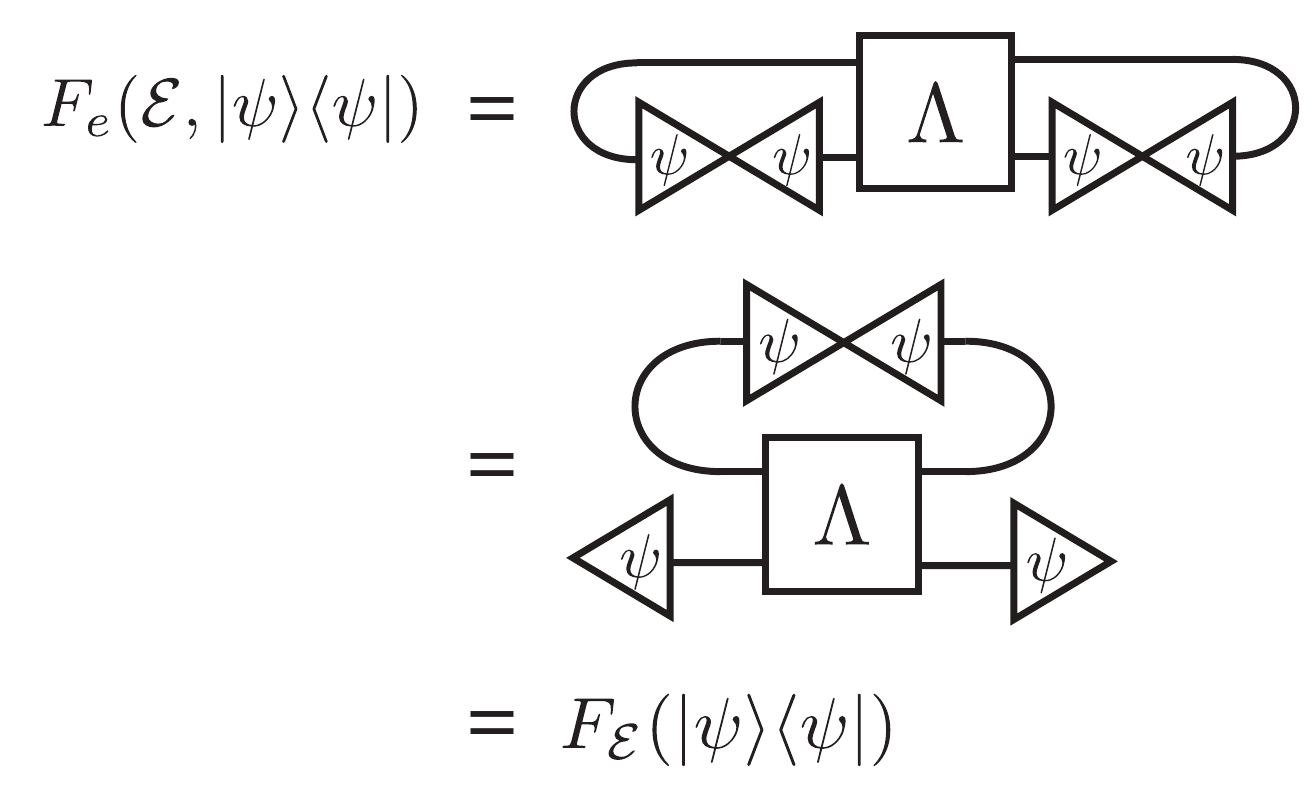}
\label{fig:ent-gate-proof}
\end{center}
Alternatively we can also define the average gate fidelity in terms of the entanglement fidelity with the identity operator
\be
\overline{F}_{\2E} = \frac{d + F_e(\2E, \I)}{d(d+1)}.
\ee

\section{Further Studies}

Open quantum systems represents an active and poorly understood area of research with results appearing frequently.  

Further directions include and are not limited to. 
\begin{itemize}
    \item[1.] 
    Present a tensor network explanation for entanglement breaking channels---see Figure 1 in \cite{Filippov_2013}.  (suggested by Sergey Filippov) 
    \item[2.] 
     Extending open dynamics to  non-Markovianity by considering the process tensor which corresponds to a recent debate in the community \cite{Milz_2017} see also Figure 1 \cite{luchnikov2019machine}. (suggested by Kavan Modi and then Sergey Filippov) 
    
    \item[3.]
     Extend ideas related to \eqref{eqn:choithm}. $\2 E$ is unital iff $\Tr_{\2 X}[\Lambda]=\I_{\2 Y}$.  According to Konstantin Antipin, an interesting tensor network corresponds to that too---similar to \eqref{eqn:choi-TP}. (suggested by Konstantin Antipin and then Sergey Filippov) 
\end{itemize} 

 \part{Counting Solutions by Tensor Contraction}\label{part:count}

This chapter presents methods to count via tensor contractions. Starting first with Boolean tensor contractions, the chapter ends with tensor contractions for edge coloring's of 3-regular planner graphs.  The chapter follows partially \cite{2015JSP...160.1389B} and less so \cite{Penrose} as presented in \cite{biamonte2017tensor}. 

\section{Returning to Boolean Quantum States}

Boolean states were considered in detail in \S~\ref{sec:btn}, while properties of Boolean algebra are reviewed in Appendix \ref{sec:Boolean}.  Here we again recall certain key notations, 
with a succinct presentation tailored towards the use of tensor networks for counting \cite{2015JSP...160.1389B, VTN, Morton_2012, 2012PhRvL.109c0503C,Kourtis2019}.  

\begin{remark}
A quantum state is called {\it Boolean} if and only if it can be written in a local basis with amplitude coefficients taking only binary values $0$ or $1$. We relate such states with Boolean functions, allowing for a host of tools from algebra to be applied to their analysis. The present note derives several relations of these states, related to the contraction of the corresponding tensor networks.   
\end{remark}

\begin{remark}
 Note that {\it quantum Boolean functions} have alternatively been studied \cite{cj10-01} as unitary projectors (for unitary projector $f$,  spec($f$)$\in\{0,1\}$). 
\end{remark}

\begin{remark}[Notation]
We use $\7B$ to denote a Boolean bit, given by an element of the set $\{0,1\}$.  
A number in $\7B^n$ then denotes an $n$-long Boolean bit string.  If $x$ is a bit string, 
then we use $\ket{x}$ as an index for a basis state.  
If $f:\7B^n\rightarrow \7B$ then $\ket{f(x)}$ also indexes a basis state.  
\end{remark}

\begin{definition}[The class of Boolean quantum states \cite{2015JSP...160.1389B}]
 Let 
 \be 
 f:\7 B^n \rightarrow \7B
 \ee 
 be any switching function.  Then 
 \be 
 \psi_{\7 B} = \sum_{\x} \ket{\x}\ket{f(\x)}
 \ee 
 is an arbitrary representative in the class of Boolean states.  In this fashion, every Boolean function gives rise to a quantum state.  Conversely, every quantum state written in a local basis with amplitude coefficients taking binary values in $\{0,1\}$ gives rise to a Boolean function.  This defines the so called, class of Boolean quantum states \cite{CTNS}.
\end{definition}

\begin{theorem}[Boolean tensor network states \cite{CTNS}]\label{theorem:btns}
 A tensor network representing a Boolean quantum state is determined from the classical network description of the corresponding function.  
\end{theorem}

Theorem \ref{theorem:btns} was developed in \S~\ref{sec:btn}, where the quantum tensor networks are found by letting each classical gate act on a linear space and from changing the composition of functions, to the contraction of tensors.

\subsection*{Contracting networks to solve SAT instances} 

\begin{theorem}[Counting 3-SAT solutions]\label{theorem:3-SAT}
Let $f$ be given to represent a 3-SAT instance. Then the standard two-norm length squared can be made to give the number of satisfying assignments of the instance \cite{2015JSP...160.1389B}.

\begin{proof}
 The quantum state takes the form 
 \be
 \psi_f = \sum_{\x} \ket{\x}\braket{f(\x)}{1} = \sum_{\x} f(\x) \ket{\x}
 \ee 
 We calculate the inner product of this state with itself viz 
 \be 
 ||\psi||^2=\sum_{\x \y} f(\x) f(\y) \langle \x, \y\rangle = \sum_\x f(\x) 
 \ee 
 which gives exactly the number of satisfying inputs.  This follows since $f(\x)f(\y)=\delta_{\x\y}$.   
 We note that for Boolean states, the square of the two-norm in fact equals the one-norm.  
\end{proof}
\end{theorem}

\begin{remark}[Counting 3-SAT solutions]
  We note that solving the counting problem \eqref{theorem:3-SAT} for general formula is known to be {\sf \#P}-complete.
\end{remark}

\begin{corollary}[Solving 3-SAT instances] 
 The condition 
 \be 
 ||\psi_f|| > 0
 \ee 
 implies that the SAT instance corresponding to $f$ has a satisfying assignment.  Determining if this condition holds for general Boolean states is 
 an {\sf NP}-complete decision problem.\sn{Instead of SAT, sometimes the term UNIQUE-SAT or USAT is used to describe the problem of determining whether a formula known to have either zero or one satisfying assignments has zero or has one.}  Note that determining this for 
 \be 
 \psi_{\overline f} = \sum_{\x} \ket{\x}\braket{f(\x)}{0} = \sum_{\x} (1-f(\x)) \ket{\x}
 \ee
 in general is a tautology problem.   
\end{corollary}

\subsection*{Graphical depiction and physicality} 

\begin{remark}[Graphical depiction]
 The algorithm is depicted below.  (a) gives a network realization of the function and determining if the network in (b) contracts to a value greater than zero solves a SAT instance.   
\begin{center}
\includegraphics[width=10\xxxscale]{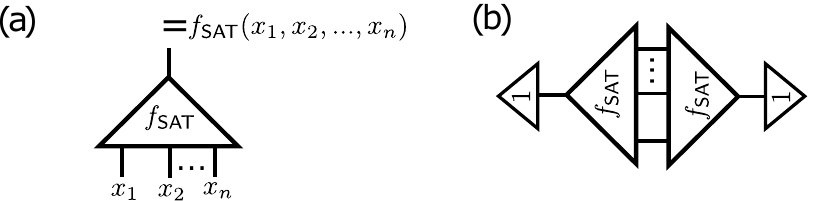}
\end{center}
\end{remark}

As it happens, some time ago Penrose proved a theorem which applies directly to the physicality of Boolean satisfiable states.  
We changed the wording of the theorem only slightly, changing spin network to tensor network.  

\begin{theorem}[Penrose, 1967]
 The norm of a tensor network vanishes iff the physical situation it represents is forbidden by the rules of quantum mechanics \cite{Penrose67}. 
\end{theorem}

The above theorem applies to quantum states.  Consider instead a process that involved the impossibility of measuring a state to be in 
a certain state.  To capture when such a process is impossible, we modify Penrose's theorem as follows.  

\begin{corollary}[Forbidden process]
 The contraction of a tensor network vanishes iff the physical situation it represents is forbidden by the rules of quantum mechanics.  
\end{corollary}

\begin{example}[Examples of Penrose's theorem]
 Consider a Bell state $\Phi^+ = \ket{00}+\ket{11}$.  The amplitude of the first party measuring $\ket{0}$ followed by the second party measuring $\ket{1}$ is zero. This vanishing tensor network contraction is given by $\braket{01}{\Phi^+}$.  A second example is found by considering the norm of a state $\ket{\psi}$ formed by a network of connected tensors, by taking an inner product with a conjugated copy of itself $\braket{\psi}{\psi}$.  If this inner product vanishes, the network necessarily represents a non-physical quantum state, by Penrose's theorem. 
\end{example}

\begin{corollary}
All physical Boolean states are satisfiable.  
\end{corollary}

\begin{remark}[read-once]
A function $f$ is called \textit{read-once} iff it can be represented
as a Boolean expression using the operations conjunction, disjunction and negation, in
which every variable appears exactly once. We call such a factored expression a read-once
expression for $f$.  These correspond exactly to fan-out only circuits.  From this structure we conclude directly that  
\end{remark}

\begin{corollary}
All \textit{read-once} formula are satisfiable \cite{2015JSP...160.1389B}.  
\end{corollary}

\begin{remark}[Quantum read-once]
A quantum quantum state is called read-once if it can be represented by a tensor tree containing only isometries.   
\end{remark}

\begin{corollary}
All quantum \textit{read-once} formula are satisfiable, with the evaluation of k-point functions polynomial in the particle number.  
\end{corollary}

\section{Returning to Stabilizer Tensor Theory} 

In \S~\ref{sec:stabtt}, we talked in detail about the Clifford group and
stabilizer theory.  In fact, we proved the following theorem, which
connects the theory of stabilizer states to the theory tensor networks
which represent pseudo Boolean forms, which we have developed in our work,
and presented in this book.  

\begin{theorem}[Stabilizer states as pseudo Boolean forms]
Let 
\be 
f, g, k:\7B^n\rightarrow \7B
\ee 
then the quantum state 
\be 
\psi_{\7B} = \sum (-1)^{f(\1 x)}(i)^{g(\1 x)}k(\1 x)\ket{\1 x}
\ee 
is sufficient to express any stabilizer state.  
\end{theorem}

We will now take a step in the other direction.  That is, we wish to
understand what Boolean states are stabilizer states.  Here we will
consider the class of linear quantum states.  We consider the general
theory elsewhere.  

\begin{definition}[The class of linear quantum states]
 We define the linear class of quantum states as quantum states of the
form 
 \be 
 \psi_{\oplus L} = \sum c_0\oplus c_1x_1\oplus c_2x_2\oplus ...\oplus
c_nx_n \ket{x_1, x_2, ..., x_n}
 \ee
 where $\forall i, c_i = 0,1$ selects the linear function uniquely.  (As
we have mentioned, $c_0=1$ results technically in the affine class of
classical circuits, but we still define this full class as, the class of
linear quantum states.) 
\end{definition}

Note that the laws of the algebra enforce the strong constraint, $x\oplus
x = 0$ and $0\oplus y = y$.  So we find immediately that we need only
consider two fully entangled states in this class, as every other state
is found from a direct product of states of this form.  The first is 
\be 
 \psi_1 = \sum x_1\oplus x_2\oplus ...\oplus x_n
\ket{x_1, x_2, ..., x_n}
\ee 
and the second is given by 
\be 
 \psi_2 = \sum 1\oplus x_1\oplus x_2\oplus ...\oplus x_n
\ket{x_1, x_2, ..., x_n}
\ee 
The tensor network differs only by contraction with the constant
$\ket{1}$.  $\psi_1$ is shown in (a) and $\psi_2$ is shown in (b).  
\begin{center}
 \includegraphics[width=0.45\textwidth]{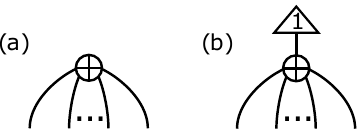} 
\end{center}

We will now consider the stabilizers of each of these cases, (a) and (b).  

\begin{remark}[Stabilizers of case (a)]
 The network in (a) is found from a Hadamard transform on all the legs of
a \COPY-tensor.  The $2^n$ stabilizers of the \COPY-tensor are generated
by the $n$ operators 
\be 
X_1\otimes X_2\otimes ...\otimes X_n
\ee 
\be 
Z_i\otimes Z_j, ~~~ 0 \leq i < j \leq n
\ee 
We have considered in lecture III how stabilizers transform.  Under the
Hadamard transform, the stabilizer generators transform to 
\be 
Z_1\otimes Z_2\otimes ...\otimes Z_n
\ee 
\be 
X_i\otimes X_j, ~~~ 0 \leq i < j \leq n
\ee 
\end{remark}

\section{Elementary Theorems of Tensor Contraction} 

The following can be used to prove graphical identities and represents a conceptual tool to aid 
in the analysis and design of tensor networks as a conceptual framework (as advocated in this lecture series) as well as a numerical tool for the simulation of quantum and classical physics.  

\begin{remark}[Linearity of tensor contraction] 
 Tensor contraction is linear in its arguments.  If $A$ is a tensor in a fully contracted network $\2C\{A\}$, if we let $A\mapsto A'+ B$ and then $A\mapsto kA$ we readily find that the contraction becomes $\2C\{A'\}+\2C\{B\}$ and $k\cdot \2C\{A\}$ respectively.  
\end{remark}

\begin{theorem}[Contraction to sum of products transform]
 Given a tensor $\Gamma^{...ijk}_{~~~~lmn...}$ in a fully contracted network (e.g.\ a network without open legs), the following graphical identity transforms the contraction, to a sum over products.  
\begin{center}
 \includegraphics[width=0.65\textwidth]{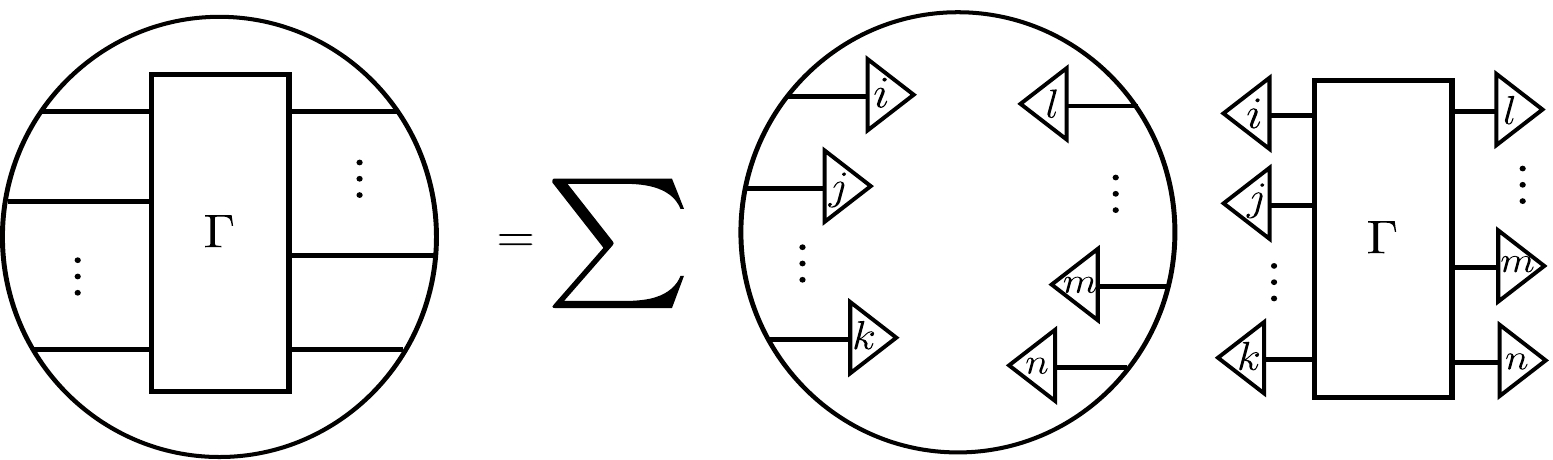} 
\end{center}
The circle is meant as an abstraction depicting a fully contracted but otherwise unknown network.  
\end{theorem}

\begin{theorem}[COPY-tensors as a resolution of identity]
The following sequence of graphical rewrites hold.  
\begin{center}
 \includegraphics[width=0.95\textwidth]{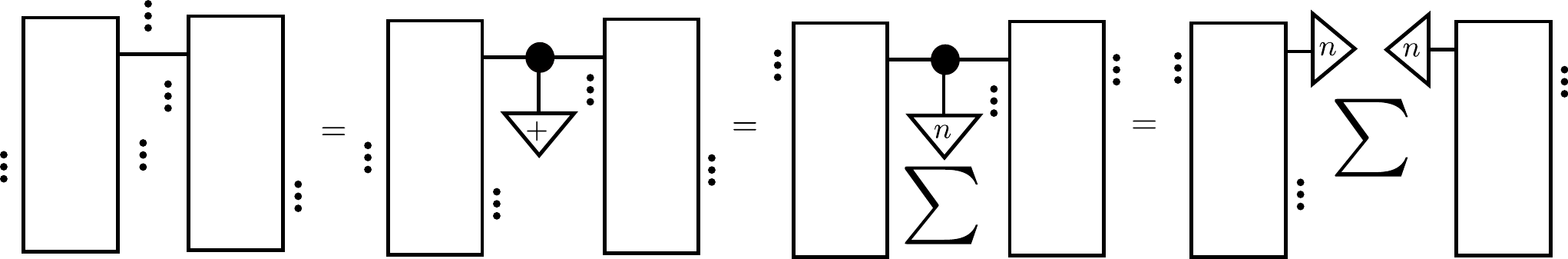} 
\end{center}
\begin{proof}
In the above figure, on the left, we abstractly depict a other wise arbitrary tensor network, by showing only one single wire.  The unit for the COPY-tensor is the plus state $\ket{+}$.  This state is defined as a sum over basis states $\ket{n}$.  The tensor copies these basis states, and splits apart.  
\end{proof}
\end{theorem}

\begin{theorem}[A tensor contraction inequality]
Given a contracted network and a partition into two halves $x$, and $y$.  Writing the contraction as 
$\2C\{x,y\}$ the following inequality holds.  
\be 
\2C\{x,y\}\leq \2C\{x,x\}\cdot \2C\{y,y\}
\ee 
with graphical depiction.  
\begin{center}
 \includegraphics[width=0.67\textwidth]{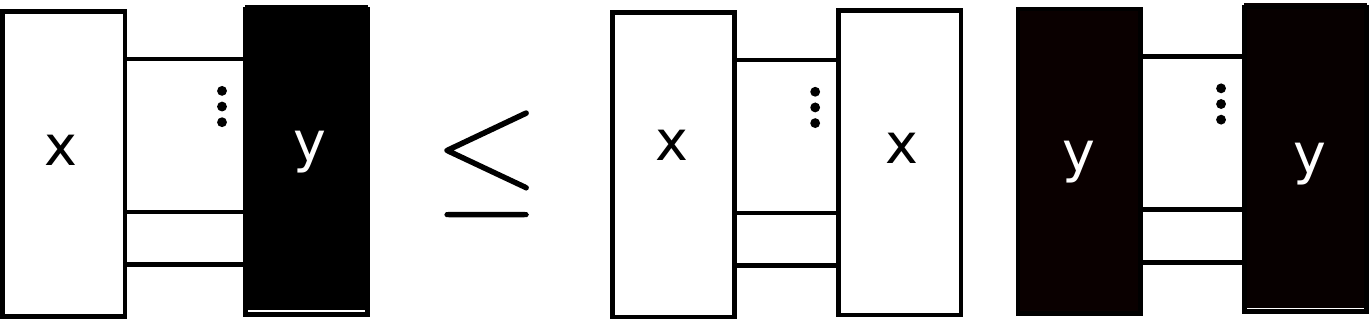} 
\end{center}
\begin{proof}
 By the linearity of tensor contraction, we arrive at an abstract form of the 
 Cauchy-Schwarz inequality, with equality in the contraction iff $x = \alpha \cdot y$.  This leads 
 directly to the concept of an angle between tensors, 
 \be 
 \cos \theta_{xy} = \frac{\2C\{x,y\}}{\2C\{x\}\cdot \2C\{y\}}
 \ee
 where the right side is either real valued, or we take the modulus.  
\end{proof}
\end{theorem}

\section{A 3-fold way}  

We will now unify three concepts.  

\begin{remark}[Pseudo Boolean function]
 A function is called pseudo Boolean when it is total with type 
 \be 
 f(x): \7B^n \rightarrow \7C
 \ee 
\end{remark}

\subsection*{Fold I. Quantum States}

A quantum state is a map 
\be 
\psi: \7 C \rightarrow \2 H
\ee 
since $\psi (c) = c\cdot \psi\in \2 H$ for $c\in \7C$ and $\psi (1) = \psi$ uniquely determines $\psi$ by linearity.  We typically fix $\|\psi \|=1$.  Note that 
\be 
\braket{\psi}{i j\cdots k} = c_{i j\cdots k}
\ee 
where 
\be 
\psi = \sum_x c_x \ket{x}
\ee 
given basis $\ket{x}$.  We arrive at the following 
\be 
\exists ! f:\7B^n \rightarrow \7C ~|~ f(x)=c_x 
\ee 
and $f$ is pseudo Boolean and $\psi$ can be expressed as 
\be 
\psi = \sum_x f(x)\ket{x}
\ee

\subsection*{Fold II. Linear Operators in $\2H\rightarrow \2H$} 

\begin{theorem}[Isomorphism between states and diagonal maps]
 There is an isomorphism sending every state $\psi\in \2 H$ to a diagonal map $\2 L(\psi) \in \2 H\rightarrow \2 H$.  Moreover, the eigenvalues of $\2 L(\psi)$ are the amplitudes of $\psi$ expressed in the spin basis. 
 \begin{proof}
 We write 
 \be
 \psi = \sum \alpha_\x \ket{\x}
 \ee 
 and then by constructing a map such that 
 \be 
 \ket{\x} \rightarrow \ket{\x}\ket{\x}
 \ee 
 one can construct $\2 L(\psi)$ as 
 \be 
 \2 L(\psi) =  \sum \alpha_\x \ket{\x}\bra{\x}
\ee
We then let 
\be 
\ket{+} := \ket{0} + \ket{1}+\cdots + \ket{n} 
\ee 
Assume we are considering $n$ qubits, then 
\be
\psi = 
 \2 L(\psi)\ket{+}^{\otimes n} 
 \ee 
 establishes the bijection.  We also note that 
 \be 
 \2 L(\psi)\ket{\1k} = \alpha_{\1k} \ket{\1k}
 \ee 
 satisfying the eigenvalue condition and hence, proving the result.  
\end{proof}
\end{theorem}

\begin{lemma}[Tensor networks equating states and diagonal maps]
The maps relied on in the above theorem can be given in terms of tensor networks.  
$\2 L(\psi)$ is shown in (a).  This map is invertible as shown in (b).  

\begin{center}
\includegraphics[width=12\xxxscale]{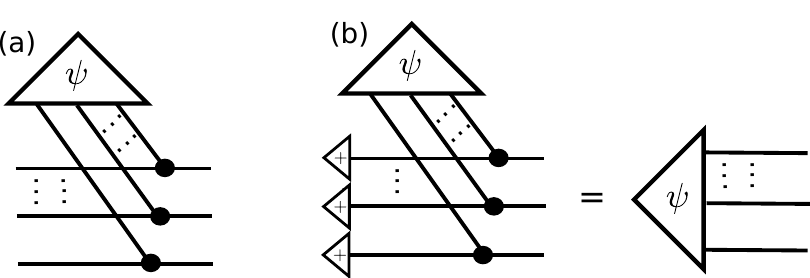}
\end{center}
\end{lemma}

\begin{remark}[Proof strategy]
Here we sketch what we call the argument by linearity.  We contract all wires of an open diagram as follows.  On the left we contract with $\bra{x,y,z}$ and on the right with $\ket{q,p,r}$.  
 \begin{center}
 \includegraphics[width=8\xxxscale]{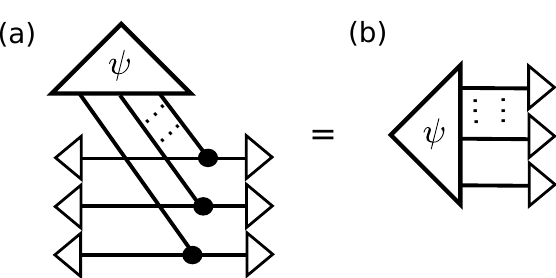}
\end{center}
This is equivalent to evaluating the coefficients of three party delta functions as 
\be 
\delta_{xq}^{~~i}\delta_{yp}^{~~j}\delta_{zr}^{~~k}\psi_{ijk} = c_{ijk} ~ (=c_{xyz}=c_{qpr})
\ee 
by linearity, we recover the map $\2L(\psi)$. 
\end{remark}

\begin{remark}
 Note that from $\psi = \sum f(x) \ket{x}$ we have 
 \be 
 \2L (\psi) = \2 L \left( \sum f(x) \ket{x}\right) = \sum f(x) \2L (\ket{x}) = \sum f(x) \ket{x}\bra{x} 
 \ee 
\end{remark}

\subsection*{Fold III.  Classical Spin Hamiltonians} 
We will consider a generalized Ising spin Hamiltonian.  Let 
\be 
h = \sum h^i Z_i + \sum J^{ij}Z_i Z_j + \cdots + \sum k^{ij...k}Z_i Z_j \cdots Z_k 
\ee 
and let $s_i$ be a spin variable taking values $\pm 1$ and let $x_i$ be a Boolean valued $0,1$.  
Use 
\be 
s_i = 1 - 2 x_i 
\ee 
then 
\be 
Z_i = \I - 2 \ket{1}\bra{1}
\ee 
and so we arrive at 
\begin{equation} 
\begin{split}
h_x & = \sum h^i (1 - 2 x_i) + \sum J^{ij}(1 - 2 x_i)(1 - 2 x_j) + \cdots \\  
 &~~~ + \sum k^{ij...k}(1 - 2 x_i) (1 - 2 x_j) \cdots (1 - 2 x_k) 
\end{split}
\end{equation}
and we arrive at the expression for $\psi$ 
\be 
\psi = \sum h_x(x) \ket{x} 
\ee 

In \cite{JDB08} we developed general methods to reason about spin Hamiltonians.  We explored an embedding reducing k-body interactions (e.g.\ $k^{ij...k}Z_i Z_j \cdots Z_k$) into two-body interactions by adding additional qubits.  

\subsection*{A three fold way}
We have shown that three different concepts, are effectively equivalent by constructing mathematical dualities that relate them precisely.  

\begin{center}
 \includegraphics[width=0.95\textwidth]{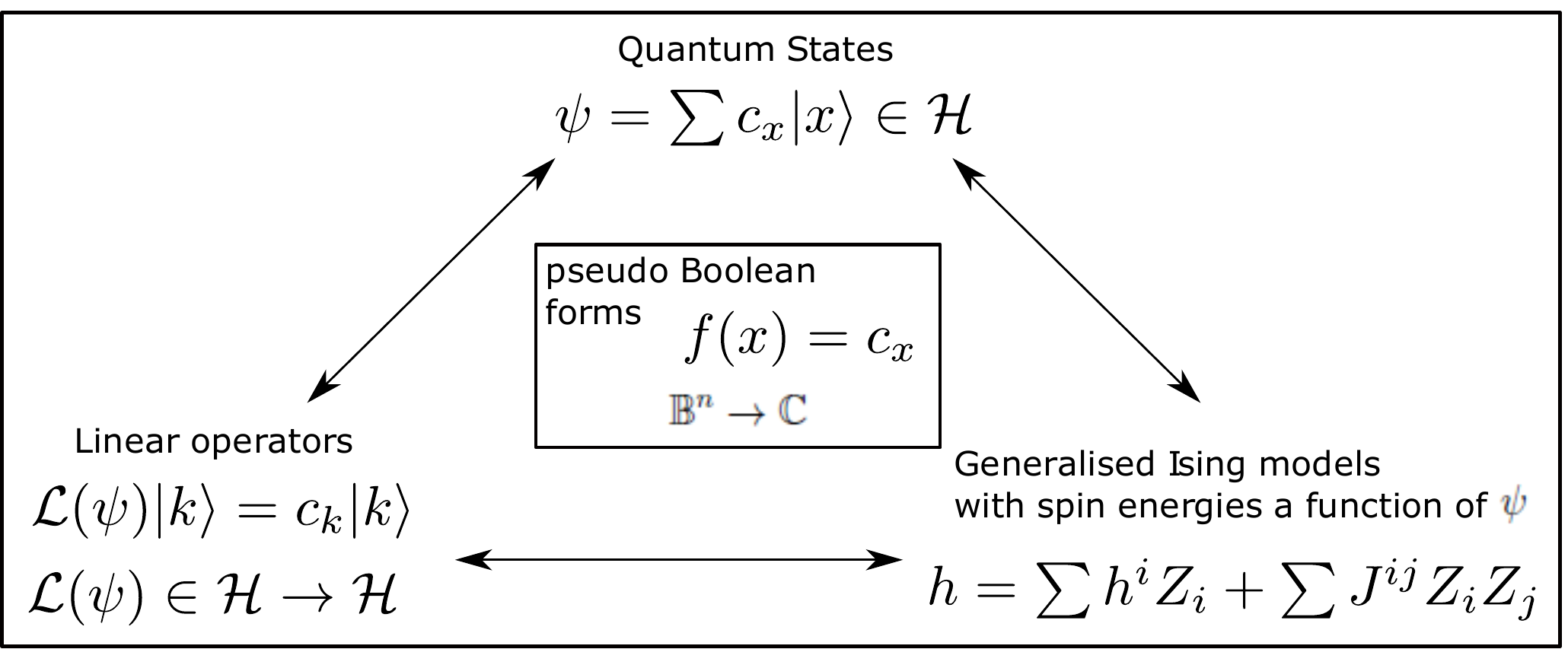} 
\end{center}

\begin{remark}[Factorization of quantum states into Tensor Networks]
 In the lectures, we also presented a universal a factorization of quantum states, into networks comprised of the building blocks (Quantum Legos), not included here. 
 A related factorization appeared in \cite{CTNS}.  
\end{remark}

\subsection*{Counting Graph Colorings}
Given a $3$-regular planar graph\sn{A graph is $k$-regular iff every node has exactly $k$~edges connected to it.},
how many possible edge colorings using three colors exist,
such that all edges connected to each node have distinct colors?
This counting problem can be solved in an interesting (if not computationally efficient) way
using the order-3 $\epsilon$~tensor, which is defined in terms of components as
\begin{align}
\notag
 &\epsilon_{012} = \epsilon_{120} = \epsilon_{201} = 1,\\
 &\epsilon_{021} = \epsilon_{210} = \epsilon_{102} = -1,
\end{align}
otherwise zero. The counting algorithm is stated as
\begin{theorem}[Planar graph $3$-colorings, Penrose 1971 \cite{Penrose}]\label{thm:3-color}
The number~$K$ of proper $3$-edge-colorings of a planar $3$-regular graph
is obtained by replacing each node with an order-3 epsilon tensor,
replacing each edge with a wire,
and then contracting the resulting tensor network.
\end{theorem}

We will first consider the simplest case, a graph with just two nodes.
In this case we obtain
\begin{center} 
\includegraphics{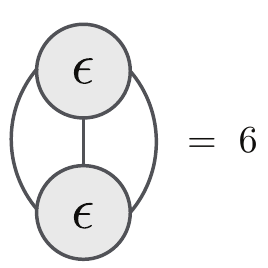}
\end{center} 
There are indeed $6$ distinct edge colorings for this graph, given as
\begin{center}
\includegraphics{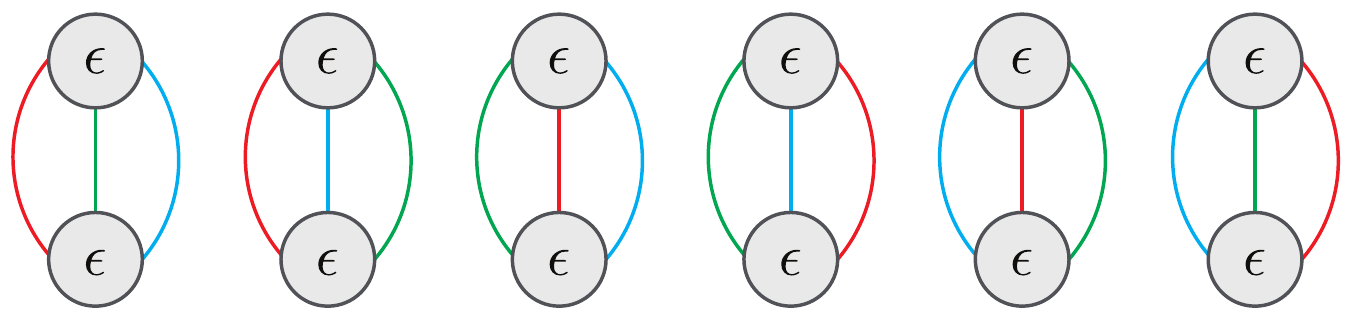} 
\end{center}
To understand Theorem~\ref{thm:3-color},
note first that the contraction~$K$ of the epsilon tensor network is the sum of all possible individual assignments of the index values
to the epsilon tensors comprising the network.
Each of the three possible index values can be understood as a color choice for the corresponding edge.
Whenever the index values for a given epsilon tensor are not all different, the corresponding term in~$K$ is zero.
Hence only allowed color assignments result in nonzero contributions to~$K$,
and for a graph that does not admit a proper $3$-edge-coloring we will have~$K=0$.
For instance, for the non-$3$-colorable Petersen graph we obtain
\begin{center}
\includegraphics{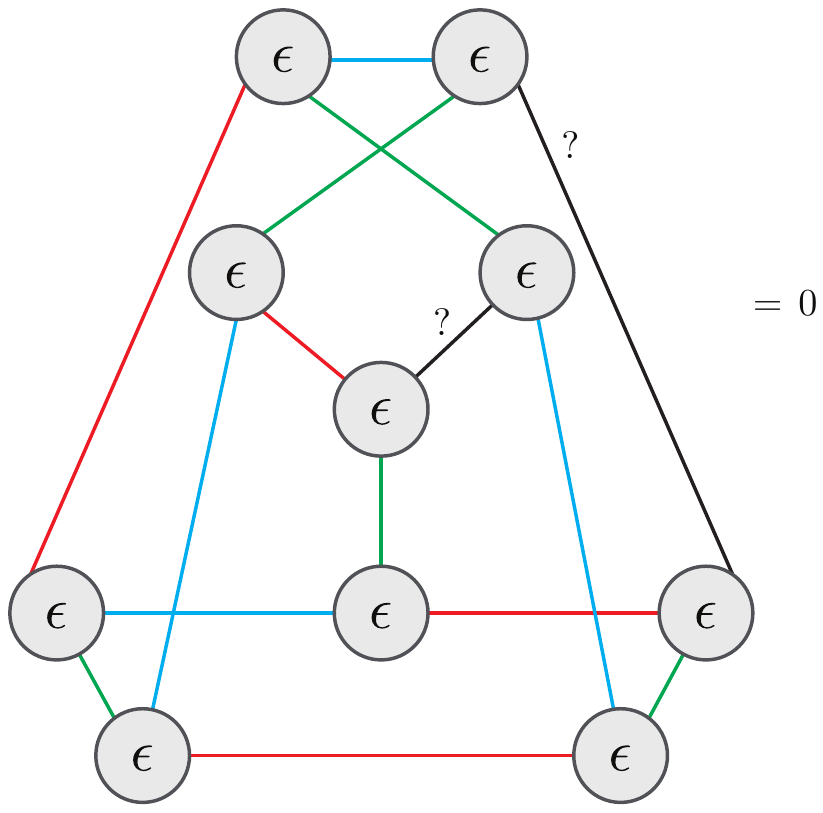}
\end{center}

However, for~$K$ to actually equal the number of allowed colorings, each nonzero term must have the value~$1$ (and not~$-1$).
This is only guaranteed if the graph is planar,
as can be seen by considering the non-planar graph $K_{3,3}$:
\begin{center}
\includegraphics{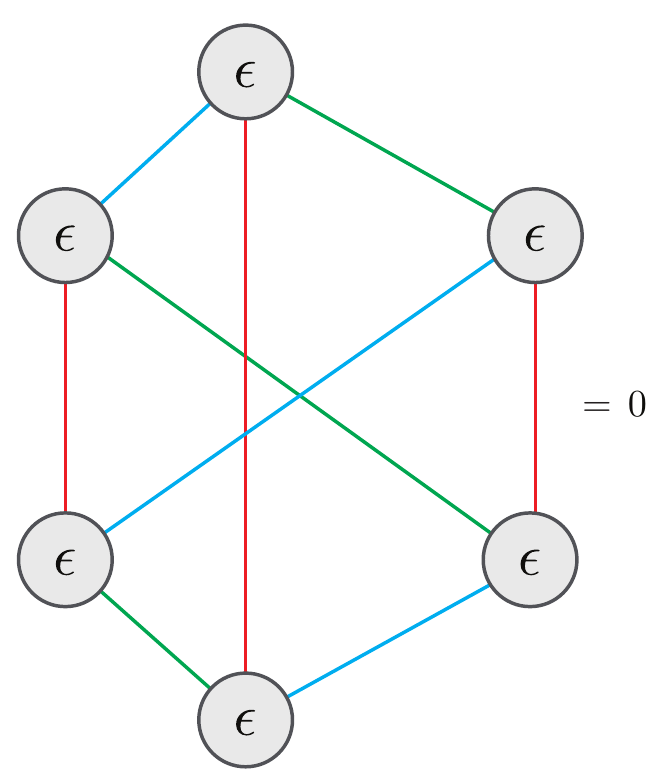}
\end{center}
The edges can be colored with three colors---in 12 different ways---yet the contraction vanishes.

The computational complexity of this problem has been studied in \cite{xia2007computational}. Interesting, by a well known result (Heawood 1897), the 3-colorings as stated above, are one quarter of the ways of coloring the faces of the graph with four colors, so that no 
 two like-colored faces have an edge in common. 
\begin{example}[Physical implementation of $\epsilon_{abc}$ in quantum computing]
In quantum computing, typically one works with qubits (two level quantum systems) but implementations using qutrits exist (three level quantum systems, available in e.g.~nitrogen vacancy centers in diamond---see for instance \cite{2014NatCo...5E3371D}). 
The epsilon tensor $\epsilon_{abc}$ could be realized directly as a locally invariant 3-party state using qutrits, and can also be embedded into a qubit system.  We leave it to the reader to show that by pairing qubits, $\epsilon_{abc}$ can be represented with six qubits, where each leg now represents a qubit pair.  (Note that a basis of 3 states can be isometrically embedded in 4-dimensional space in any number of ways.) Show further that the construction can be done such that the two qubit pairs (together representing one leg) are symmetric under exchange.
\end{example}
 
\begin{myexercise}[Representation of the Epsilon Tensor on Qubit States]  
The epsilon tensor is typically considered 
 in $\7C^3\otimes \7C^3\otimes \7C^3$.  Show that by pairing qubits, this can be represented in the space $\7C^{2\otimes 6}$ of six qubits, where each leg now represents a qubit pair.  Show further that the construction can be done such that the two qubit pairs are symmetric under exchange. 
\end{myexercise}
 

 \appendix
 
\part{Appendix}






\section{Algebra on Quantum States}\label{sec:newalgebra}
We are concerned with a network theory of quantum states.  This on the one hand can
be used as a tool to solve problems about states and operators in quantum theory, but
does have a physical interpretation on the other.  This is not foundational
\textit{per se} but instead largely based on what one 
might call an operational interpretation of quantum states and processes.  A related idea has been used 
to study non-locality in quantum physics \cite{Edwards10}.  This appendix stems from those ideas \cite{Edwards10} which Bill Edwards introduced me to in Oxford around circa 2010.

We call
an algebra a pairing on a vector space, taking two vectors and producing a third
(you might instead call it a monoid if there is a unit, and then a group if the
set of considered vectors is closed under the product).  Let's now examine how every
tripartite quantum state forms an algebra.

Consider a tripartite quantum state (subsystems labeled 1,2 and 3), and then ask the
question: ``how would the state of the third system change after measurement of 
systems one and two?''   Enter Algebras: as stated, an algebra on a vector space, or
on a Hilbert space is formed by a \textit{product} taking two elements from the
vector space to produce a third element in the vector space.  Algebra on states can
then be studied by considering duality of the state, that is considering the
adjunction between the maps of type
\begin{equation}
\I\rightarrow \2 H\otimes \2H \otimes \2H~~~~~~~~\text{and}~~~~~~~~\overline{\2
H}\otimes \overline{\2H} \rightarrow \2H
\end{equation}
This duality is made evident by using the $\dagger$-compact structure of the category
(e.g.\ the cups and caps).  It is given vivid physical meaning by considering the
effect measuring (that is two events) two components of a state has on the third
component.

\begin{remark}[Overbar notation on Spaces] Given a Hilbert space $\2 H$, we can
consider the Hilbert space $\overline{\2H}$ which can be thought of simply as the
Hilbert space $\2 H$ with all basis vectors complex conjugates (overbar).  That is,
$\overline{\2H}$ is a vector space whose elements are in one-to-one
correspondence with the elements of $\2 H$:
\begin{equation}
    \overline{\2H} = \{\overline{v} \mid v \in \2 H\},
\end{equation}
with the following rules for addition and scalar multiplication:
\begin{equation}
    \overline v + \overline{w} =
\overline{\,v+w\,}\quad\text{and}\quad\alpha\,\overline v = \overline{\,\overline
\alpha \,v\,}.
\end{equation}
\end{remark}

\begin{remark}[Definition of Algebra]
We consider an algebra as a vector space $\2A$ endowed with a product, taking a pair
of elements (e.g.\ from $\2A\otimes \2A$) and producing an element in $\2A$.  So the
product is a map $\2A\otimes \2A\rightarrow \2A$, which may not be associative or
have a unit (that is, a multiplicative identity --- see Example \ref{ex:weakunits}
for an
example of an algebra on a quantum state without a unit).
\end{remark}

\begin{observation}[Every tripartite Quantum State Forms an Algebra]
Let $\ket{\psi}\in \2 H\otimes \2H \otimes \2H$ be a quantum state and let
$M_i$, $M_j$ be complete sets of measurement operators.  Then $(\ket{\psi}, M_i,M_j)$
forms
an algebra.
\begin{center}
\includegraphics[width=14cm]{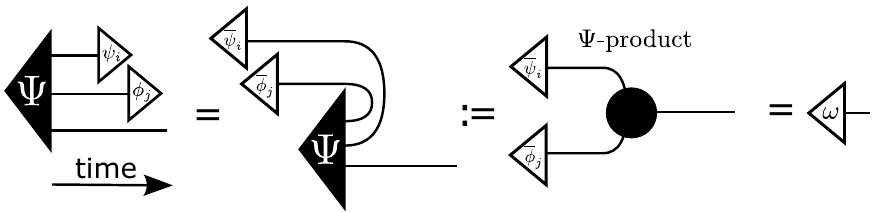}
\end{center}
\end{observation}

The quantum state $\ket{\Psi} = \sum_{ijk}\psi^{ijk}\ket{ijk}$ is drawn as a
triangle, with
the identity operator on each
subsystem acting as time goes to the right on the page (represented as a wire). 
Projective measurements
with respect to $M_i$ and $M_j$ are made.  We define these complete
measurement operators as
\begin{gather}
M_1=\sum_{i=1}^{N}i\ket{\psi_i}\bra{\psi_i}
\end{gather}
\begin{gather}
M_2=\sum_{j=1}^{N}j\ket{\phi_j}\bra{\phi_j}
\end{gather}
such that we recover the identity operator on the $N$-level subsystem viz
\begin{equation}
\sum_{j=1}^{N}\ket{\phi_j}\bra{\phi_j}=\sum_{i=1}^{N}\ket{\psi_i}\bra{\psi_i}=\I_N
\end{equation}
The measurements result in eigenvalues $i,j$ leaving the state of the unmeasured
system in
\begin{equation}
\ket{\omega} = \sum_{xyz}\psi^{xyz} \braket{\overline{\psi}^x}{x}
\braket{\overline{\phi}^y}{y} \ket{z}
\end{equation}
where $\bra{\overline{Q}} \bydef \ket{Q}^\top$ that is, the transpose is factored
into: (i) taking the dagger (diagrammatically this mirrors states across the page)
and (ii) taking the complex conjugate.  Hence,
\begin{equation}
\ket{\overline{Q}}^\dagger = \ket{Q}^\top = \bra{\overline{Q}} =
\overline{\ket{Q}^\dagger}
\end{equation}
and if we pick a real valued basis for $\ket{x},\ket{y},\ket{z}=\ket{0},\ket{1}$ we
recover
\begin{equation}
\ket{\omega} = \sum_{xyz}\psi^{xyz}
\braket{x}{\psi_x}\braket{y}{\phi_y}\ket{z}
\end{equation}

As stated, this physical
interpretation is not our main interest. Even in its absence, we're able to write
down and represent a quantum
state purely in terms of a connected network, where each component is fully defined
in terms of algebraic laws.

\section{\XOR-algebra}\label{sec:Boolean}
Here we review the concept of an algebraic normal form (\ANF) for Boolean
polynomials, commonly known as \PPRM s, (Positive Polarity Reed Muller Forms).  See the reference
book~\cite{Davio78} and the historical references~\cite{Cohn62,
xor70} for further details. 

\begin{definition}
The \XOR{}-algebra forms a commutative ring with presentation
$M=\{\7B,\wedge,\oplus\}$ where the following product is called \XOR
\begin{equation}
\text{---}\oplus\text{---}:\7B\times \7B\mapsto \7B: (a,b)\rightarrow
a+b-ab~\text{mod}~2
\end{equation}
and conjunction is given as
\begin{equation}
\text{---}\wedge\text{---}:\7B\times \7B\mapsto \7B: (a,b)\rightarrow a\cdot b, 
\end{equation}
where $a\cdot b$ is regular multiplication over the reals.  One defines left negation
$\neg(\text{---})$ in terms of $\oplus$ as
$\neg(\text{---})\equiv$
\begin{equation}
  \text{1}\oplus(\text{---}):\7B \mapsto \7B: a\rightarrow 1-a.
\end{equation}
In the \XOR-algebra, 1-5 hold.  (i) $a\oplus 0 = a$, (ii) $a\oplus 1 = \neg a$, (iii)
$a\oplus a = 0$, (iv) $a\oplus \neg a = 1$ and (v) $a\vee b = a\oplus b\oplus
(a\wedge b)$.  Hence, $0$ is the unit of \XOR{} and $1$ is the unit of \AND. The 5th
rule
reduces to $a\vee b = a\oplus b$ whenever $a\wedge b=0$, which is the case for
disjoint ($\text{mod}~2$) sums.  The 
truth table for $\AND$ follows
\begin{center}
\begin{tabular}{c|c|c}
$~x_1~$ & $~x_2~$ & $f(x_1,x_2)=x_1\wedge x_2$ \\ \hline
0 & 0 & 0 \\
0 & 1 & 0 \\
1 & 0 & 0 \\
1 & 1 & 1
\end{tabular}
\end{center}
\end{definition}

\begin{definition}\label{def:FPRM}
Any Boolean equation may be uniquely expanded to the fixed polarity
Reed-Muller form as:\setlength{\arraycolsep}{0.140em}
\begin{eqnarray}\label{eqn:rm_exp}
&&f(x_1,x_2,...,x_k) = c_0\oplus c_1 x_1^{\sigma_1}\oplus c_2 x_2^{\sigma_2}\oplus
\cdots \oplus c_n x_n^{\sigma_n}\oplus\nonumber\\
&&~~~~~~~~c_{n+1}x_1^{\sigma_1} x_n^{\sigma_{n}}\oplus \cdots \oplus
c_{2k-1}x_1^{\sigma_1} x_2^{\sigma_2},...,x_k^{\sigma_k},
\end{eqnarray}
where selection variable $\sigma_i\in \{0,1\}$, literal
$x_i^{\sigma_i}$ represents a variable or its negation and any $c$
term labeled $c_0$ through $c_j$ is a binary constant $0$ or $1$. In
Equation~\eqref{eqn:rm_exp} only fixed polarity variables appear such that
each is in either un-complemented or complemented form.
\end{definition}

Let us now consider derivation of the form from Definition~\ref{def:FPRM}.  Because
of the structure of
the algebra, without loss of generality, one avoids keeping track of
indices in the $N$ node case, by considering the case where $N\equiv 2^n=8$.

\begin{example}

The vector
\begin{equation}
    \underline{c}=(c_0,c_1,c_2,c_3,c_4,c_5,c_6,c_7,)^\intercal
\end{equation}
represents
all possible outputs of any function $f(x_1,x_2,x_3)$ over the
algebra formed from linear extension of $\7Z_2\times \7Z_2\times \7Z_2$.  We wish to
construct a normal form in terms of the vector $\underline{c}$, where each $c_i\in
\{0,1\}$, and therefore $\underline{c}$ is a selection vector
that simply represents the output of the function
\begin{equation}
    f:\7B\times\7B\times\7B\rightarrow
\7B:(x_1,x_2,x_3)\mapsto f(x_1,x_2,x_3). 
\end{equation}
One may expand $f$ as:
\begin{eqnarray}\label{eqn:generic}
f(x_1,x_2,x_3) &=& (c_0\cdot \neg x_1\cdot \neg x_2\cdot \neg
x_3)\vee(c_1\cdot \neg x_1\cdot \neg x_2\cdot x_3)\vee(c_2\cdot
\neg x_1\cdot x_2\cdot \neg x_3)\nonumber\\
&&\vee(c_3\cdot \neg x_1\cdot x_2\cdot x_3)\vee(c_4\cdot x_1\cdot
\neg x_2\cdot \neg x_3)\vee(c_5\cdot x_1\cdot
\neg x_2\cdot x_3)\nonumber\\
&&\vee(c_6\cdot x_1\cdot x_2\cdot \neg x_3)\vee(c_7\cdot x_1\cdot
x_2\cdot x_3)
\end{eqnarray}

Since each disjunctive term is disjoint the logical \OR{} operation may
be replaced with the logical \XOR{} operation.  By making the substitution $ \neg
a=a\oplus 1$ for all variables and rearranging terms one arrives at
the following normal form:\marginnote{For instance, $\neg x_1\cdot
\neg x_2\cdot\neg x_3=(1\oplus x_1)\cdot(1\oplus x_2)\cdot(1\oplus
x_3)=(1\oplus x_1\oplus x_2\oplus x_2\cdot x_3)\cdot(1\oplus x_3)=
1\oplus x_1\oplus x_2\oplus x_3\oplus x_1\cdot x_3\oplus x_2\cdot
x_3\oplus x_1\cdot x_2\cdot x_3$.}

\begin{eqnarray}\label{eqn:generic2}
f(x_1,x_2,x_3) &=&c_0\oplus(c_0\oplus c_4)\cdot x_1\oplus(c_0\oplus
c_2)\cdot x_2\oplus(c_0\oplus c_1)\cdot x_3\\
&&\oplus(c_0\oplus
c_2\oplus c_4\oplus c_6)\cdot x_1\cdot
x_2\nonumber\\
&& \oplus(c_0\oplus c_1\oplus c_4 \oplus c_5)\cdot x_1\cdot
x_3\oplus(c_0 \oplus c_1 \oplus c_2 \oplus c_3)\cdot x_2\cdot x_3\nonumber\\
&&\oplus (c_0\oplus c_1\oplus c_2\oplus c_3\oplus c_4 \oplus c_5
\oplus c_6\oplus c_7)\cdot x_1\cdot x_2\cdot x_3
\end{eqnarray}

The set of linearly independent vectors, $\{x_1,x_2,x_3,x_1\cdot
x_2,x_1\cdot x_3,x_2\cdot x_3,x_1\cdot x_2\cdot x_3\}$ combined with
a set of scalars from Equation~\ref{eqn:generic2} spans the eight
dimensional space of the Hypercube representing the Algebra.
A similar form holds for arbitrary $N$.

\begin{eqnarray}\label{eqn:generic3}
f(x_1,x_2,x_3) &=& (a_1)\cdot x_1\oplus(a_2)\cdot
x_2\oplus(x_3)\cdot x_3\oplus(a_1\oplus a_2\oplus a_1\oplus
c_2)\cdot x_1\cdot
x_2\nonumber\\
&& \oplus(a_1\oplus a_3\oplus a_1\oplus c_3)\cdot x_1\cdot
x_3\oplus(a_2\oplus a_3\oplus a_2\oplus
c_3)\cdot x_2\cdot x_3\nonumber\\
&&\oplus (a_1\oplus a_2\oplus a_3\oplus a_1\oplus a_2\oplus
a_3)\cdot x_1\cdot x_2\cdot x_3
\end{eqnarray}
\end{example}

\section{The Minimization Method of Karnaugh}\label{appendix:kmap}

The \emph{Karnaugh map} is a tool to facilitate the algebraic reduction of Boolean functions.  Many excellent texts and online tutorials cover the use of Karnaugh maps and should be consulted for more detail.\footnote{This includes the wikipedia entry (\href{http://en.wikipedia.org/wiki/Karnaugh_map}{\tt http://en.wikipedia.org}) and the articles linked to therein as well as the straight forward reference~\cite{Ros99}.}  This Appendix briefly introduces these maps to make the lecture notes self contained.

\begin{table}
\begin{tabular}{c c c}
\textbf{(a)}\begin{karnaugh-map}[4][2][1][\Large$x_1x_2$][\Large$z_*$] 
\terms{0}{0}
\terms{1}{1}
\terms{2}{2}
\terms{3}{3}
\terms{4}{4}
\terms{5}{5}
\terms{6}{6}
\terms{7}{7}
\implicant{3}{7}
\implicant{5}{7}
\implicant{7}{6}
\end{karnaugh-map}
&\quad
\textbf{(b)}\begin{karnaugh-map}[4][2][1][\Large$x_1x_2$][\Large$z_*$]
\terms{0}{0}
\terms{1}{1}
\terms{2}{2}
\terms{3}{3}
\terms{4}{4}
\terms{5}{5}
\terms{6}{6}
\terms{7}{7}
\implicant{4}{6}
\implicant{1}{7}
\implicant{3}{6}
\end{karnaugh-map}\\
\textbf{(c)}\begin{karnaugh-map}[4][2][1][\Large$x_1x_2$][\Large$z_*$] 
\terms{0}{0}
\terms{1}{0}
\terms{2}{0}
\terms{3}{.}
\terms{4}{.}
\terms{5}{.}
\terms{6}{.}
\terms{7}{0}
\implicant{3}{7}
\implicant{5}{7}
\implicant{7}{6}
\implicant{4}{6}
\end{karnaugh-map}
\end{tabular}
\caption{Karnaugh maps: (a) 2-local (positive polarity) variable couplings. (b) Linear (positive polarity) terms.  (c) A Karnaugh map illustrating (with ovals) the linear and quadratic terms needed to an example function.}\label{fig:quad}
\end{table}


Karnaugh maps (see Table~\ref{fig:quad} for three examples), or more compactly K-maps, are organized so that the truth table of a given equation, such as a Boolean equation ($f:\mathbb{B}^n\rightarrow \mathbb{B}$) or multi-linear form ($f:\mathbb{B}^n\rightarrow \mathbb{R}$), is arranged in a grid form and between any two adjacent boxes only one domain variable can change value.

This ordering results as the rows and columns are ordered according to Gray code --- a binary numeral system where two successive values differ in only one digit.  For example, the 4-bit Gray code is given as:
\begin{eqnarray*}
&&\{0000, 0001, 0011, 0010, 0110, 0111, 0101, 0100, 1100,\\
&&1101, 1111, 1110, 1010, 1011, 1001, 1000\}.
\end{eqnarray*}
By arranging the truth table of a given function in this way, a K-map can be used to derive a minimized function.

To use a K-map to minimize a Boolean function one \emph{covers} the 1s on the map by rectangular \emph{coverings} containing a number of boxes equal to a power of 2.  For example, one could circle a map of size $2^n$ for any constant function $f=1$.  Table~\ref{fig:quad} (a) and (b) contain three circles each --- all of 2 and 4 boxes respectively. After the 1s are covered, a term in a \emph{sum of products expression}~\cite{Weg87} is produced by finding the variables that do not change throughout the entire covering, and taking a 1 to mean that variable ($x_i$) and a 0 as its negation ($\overline{x_i}$). Doing this for every covering yields a function which \emph{matches} the truth table.

For instance consider Table~\ref{fig:quad} (a) and (b).  Here the boxes contain simply labels representing the decimal value of the corresponding Gray code ordering.  The circling in Table~\ref{fig:quad} (a) would correspond to the truth vector (ordered $z_\star, x_1$ then $x_2$)
\begin{equation}\label{eqn:tt1}
\left(0,0,0,1,0,1,1,1\right)^T.
\end{equation}
The cubes 3 and 7 circled in Table~\ref{fig:quad} correspond to the sum of products term $x_1x_2$.  Likewise (5,7) corresponds to $z_\star x_2$ and finally (7,6) corresponds to $z_\star x_1$. The sum of products representation of (\ref{eqn:tt1}) is simply
\begin{equation*}
f(z_\star,x_1,x_2)=x_1x_2\vee z_\star x_2\vee z_\star x_1.
\end{equation*}
Let us repeat the same procedure for Table~\ref{fig:quad} b.) by again assuming the circled cubes correspond to 1s in the functions truth table.  In this case one finds $z_\star$ for the circling of cubes ladled (4,5,7,6), $x_2$ for (1,3,5,7) and $x_1$ for (3,2,7,6) resulting in the function
\begin{equation*}
f(z_\star,x_1,x_2)=x_1\vee z_\star \vee x_2.
\end{equation*}

\begin{definition}
 (Davio Expansion) The Davio expansion is a decomposition of a boolean function.  
 For a boolean function $f(x_1,...,x_n)$ we set with respect to $x_i$:
    \begin{align} f_{x_i}(x) & = f(x_1,...,x_{i-1},1,x_{i+1},...,x_n) \\[3pt] f_{\overline{x_i}}(x)& = f(x_1,...,x_{i-1},0,x_{i+1},...,x_n) \\[3pt] \frac{\partial f}{\partial x_i} & = f_{x_i}(x) \oplus f_{\overline{x_i}}(x)\, \end{align}
as the positive and negative cofactors of $f$, and the boolean derivation of $f$.  Then we have for the Reed-Muller or positive Davio expansion:
\be
    f = f_{\overline{x_i}} \oplus x_i \frac{\partial f}{\partial x_i}
\ee
\end{definition}

\section{Tensors and Tensor Products}
\label{app:tensor}

The definition of a tensor starts with the \emph{tensor product}~$\otimes$.
There are many equivalent ways to define it, but perhaps the simplest one is
through basis vectors.
Let $V$ and~$W$ be finite-dimensional vector spaces over the same field of scalars~$\K$.
In physics-related applications~$\K$ is typically either the real numbers~$\R$
or the complex numbers~$\C$.
Now $V \otimes W$ is also a vector space over~$\K$.
If $V$ and~$W$ have the bases
$\{e_j\}_j$ and $\{f_k\}_k$, respectively,
the symbols $\{e_j \otimes f_k\}_{jk}$ form a basis for~$V \otimes W$.
Thus, for finite-dimensional spaces
$\dim(V \otimes W) = \dim V \, \dim W$.

The tensor product of two individual vectors $v \in V$ and $w \in W$
is denoted as $v \otimes w$.
For vectors the tensor product is a bilinear map
$V \times W \to V \otimes W$,
i.e.~one that is linear in both input variables.
For finite-dimensional spaces
one can obtain the standard basis coordinates of the tensor product of two vectors
as the \emph{Kronecker product} of the standard basis coordinates of the individual vectors:
\be
(v \otimes w)^{jk} = v^j w^k.
\ee
It is important to notice that due to the bilinearity $\otimes$ maps
many different pairs of vectors~$(v,w)$ to the same product vector:
$v \otimes (s w) = (s v) \otimes w = s (v \otimes w)$, where $s \in \K$.
For inner product spaces (such as the Hilbert spaces encountered in
quantum mechanics) the tensor product space inherits the inner product
from its constituent spaces:
\be
\inprod{v_1 \otimes w_1}{v_2 \otimes w_2}_{V \otimes W} = \inprod{v_1}{v_2}_V \inprod{w_1}{w_2}_W.
\ee

A \emph{tensor}~$T$ is an element of the tensor product of a finite
number of vector spaces over a common field of scalars~$\K$.
The dual space~$V^*$ of a vector space~$V$ is defined as the space of
linear maps from~$V$ to~$\K$. It is not hard to show that $V^*$~is a
vector space over~$\K$ on its own. This leads us to define
the concept of an order-$(p,q)$ tensor, an element of the tensor product
of $p$~primal spaces and $q$~dual spaces:
\be
T \in W_1 \otimes W_2 \otimes \ldots \otimes W_p \otimes V_1^* \otimes V_2^* \otimes \ldots \otimes V_q^*.
\ee
Given a basis $\{\bv\spidx{i}{}{k}\}_k$ for each vector
space~$W_i$ and a dual basis $\{\eta\spidx{i}{k}{}\}_k$ for each dual space~$V_i^*$,
we may expand~T in the tensor products of these basis vectors:
\be
\label{eq:tensorexpansion}
T = T\indices{^{i_1 \ldots i_p}_{j_1 \ldots j_q}}
\bv\spidx{1}{}{i_1} \otimes \ldots \otimes \bv\spidx{p}{}{i_p}
\otimes
\eta\spidx{1}{j_1}{} \otimes \ldots \otimes \eta\spidx{q}{j_q}{}.
\ee
$T\indices{^{i_1 \ldots i_p}_{j_1 \ldots j_q}}$ is simply an array of
scalars containing the basis expansion coefficients.
Here we have introduced the \emph{Einstein summation convention}, in which
any index that is repeated exactly twice in a term, once up, once
down, is summed over. This allows us to save a considerable number of
sum signs, without compromising on the readability of the formulas.
Traditionally basis vectors carry a lower (covariant) index and dual basis vectors
an upper (contravariant) index.

A tensor is said to be \emph{simple} if it can be written as the tensor product
of some elements of the underlying vector spaces:
$T = v\spidx{1}{}{} \otimes \ldots \otimes v\spidx{q}{}{} \otimes \varphi\spidx{1}{}{} \otimes \ldots \otimes \varphi\spidx{p}{}{}$.
This is not true for most tensors; indeed, in addition to the
bilinearity, this is one of the properties that separates tensors from mere Cartesian products of vectors.
However, any tensor can be written as a linear combination of simple tensors,
e.g.~as in Eq.~\eqref{eq:tensorexpansion}.

For every vector space~$W$ there is a unique bilinear map
$W \otimes W^* \to \K$,
$w \otimes \phi \mapsto \phi(w)$
called a natural pairing,
where the dual vector maps the primal vector to a scalar.
One can apply this map to any pair of matching primal and dual spaces in a tensor.
It is called a \emph{contraction} of the corresponding upper and lower indices.
For example, if we happen to have $W_1 = V_1$ we may contract the corresponding indices on~$T$:
\begin{align}
\notag
C_{1,1}(T)
&=
T\indices{^{i_1 \ldots i_p}_{j_1 \ldots j_q}}
\eta\spidx{1}{j_1}{}(\bv\spidx{1}{}{i_1})
\:\:
\bv\spidx{2}{}{i_2} \otimes \ldots \otimes \bv\spidx{p}{}{i_p}
\otimes
\eta\spidx{2}{j_2}{} \otimes \ldots \otimes \eta\spidx{q}{j_q}{}\\
&=
T\indices{^{k \, i_2 \ldots i_p}_{k \, i_2 \ldots j_q}}
\bv\spidx{2}{}{i_2} \otimes \ldots \otimes \bv\spidx{p}{}{i_p}
\otimes
\eta\spidx{2}{j_2}{} \otimes \ldots \otimes \eta\spidx{q}{j_q}{},
\end{align}
since the defining property of a dual basis is
$\eta\spidx{1}{j_1}{}(\bv\spidx{1}{}{i_1}) = \delta\indices{^{j_1}_{i_1}}$.
Hence the contraction eliminates the affected indices ($k$~is summed over),
lowering the tensor order by~$(1,1)$.

We can see that an order-$(1,0)$ tensor is simply a vector,
an order-$(0,1)$ tensor is a dual vector,
and can define an order-$(0,0)$ tensor to correspond to a plain scalar.
But what about general, order-$(p,q)$ tensors?
How should they be understood?
Using contraction, they can be immediately
reinterpreted as multilinear maps from vectors to vectors:
\begin{align}
\label{eq:Tmultilin}
\notag
T':\quad& V_1 \otimes \ldots \otimes V_q \to W_1 \otimes \ldots \otimes W_p,\\
T'(v\spidx{1}{}{} \otimes \ldots \otimes v\spidx{q}{}{}) &=
T\indices{^{i_1 \ldots i_p}_{j_1 \ldots j_q}}
\bv\spidx{1}{}{i_1} \otimes \ldots \otimes \bv\spidx{p}{}{i_p}
\times
\eta\spidx{1}{j_1}{}(v\spidx{1}{}{}) \times \ldots \times \eta\spidx{q}{j_q}{}(v\spidx{q}{}{}),
\end{align}
where we tensor-multiply $T$ and the vectors to be mapped together, and then contract the corresponding indices.
However, this is not the only possible interpretation.
We could just as easily see them as mapping dual vectors to dual vectors:
\begin{align}
\label{eq:Tmultilin2}
\notag
T'':\quad& W^*_1 \otimes \ldots \otimes W^*_p \to V^*_1 \otimes \ldots \otimes V^*_q,\\
T''(\varphi\spidx{1}{}{} \otimes \ldots \otimes \varphi\spidx{p}{}{}) &=
T\indices{^{i_1 \ldots i_p}_{j_1 \ldots j_q}}
\varphi\spidx{1}{}{}(\bv\spidx{1}{}{i_1}) \times \ldots \times \varphi\spidx{p}{}{}(\bv\spidx{p}{}{i_p})
\times
\eta\spidx{1}{j_1}{} \otimes \ldots \otimes \eta\spidx{q}{j_q}{}.
\end{align}
Essentially we may move any of the vector spaces to the
other side of the arrow by taking their dual:
\begin{align}
\label{eq:tensor-interp}
&W \otimes V^*
\quad \isom \quad
\K \to W \otimes V^*
\quad \isom \quad
V \to W
\quad \isom \quad
V \otimes W^* \to \K
\quad \isom \quad
W^* \to V^*,
\end{align}
where all the arrows denote linear maps.
Any and all input vectors are mapped to scalars by the corresponding
dual basis vectors in expansion~\eqref{eq:tensorexpansion}, whereas all input dual vectors map the
corresponding primal basis vectors to scalars.

If we expand the input vectors~$v\spidx{k}{}{}$ in
Eq.~\eqref{eq:Tmultilin} using the same bases as when expanding the tensor~T,
we obtain the following equation
for the expansion coefficients:
\begin{align}
T'(v\spidx{1}{}{} \otimes \ldots \otimes v\spidx{q}{}{})\indices{^{i_1 \ldots i_p}}
&=
T\indices{^{i_1 \ldots i_p}_{j_1 \ldots j_q}}
v\spidx{1}{j_1}{}
\cdots
v\spidx{q}{j_q}{}.
\end{align}
This is much less cumbersome than Eq.~\eqref{eq:Tmultilin}, and
contains the same information.
This leads us to adopt the \emph{abstract index notation} for tensors,
in which the indices no longer denote the components of the tensor in a
particular basis, but instead signify the tensor's order.
Tensor products are denoted by simply placing the tensor symbols next to each other.
Within each term, any repeated index symbol must appear once up and once down, and denotes
contraction over those indices.
Hence, $x^a$~denotes a vector (with one contravariant index),
$\omega_a$ a dual vector (with one covariant index),
and $T\indices{^{ab}_c}$ an order-$(2,1)$ tensor with two
contravariant and one covariant indices.
$S\indices{^{ab}_{cde}} x^c y^d P\indices{^e_a}$ denotes the contraction of
an order-$(2,3)$ tensor~$S$,
an order-$(1,1)$ tensor~$P$,
and two vectors, $x$ and $y$, resulting in an order-$(1,0)$ tensor
with one uncontracted index,~$b$.

In many applications, for example in differential geometry, the vector spaces
associated with a tensor are often copies of the same vector space~$V$
or its dual~$V^*$,
which means that any pair of upper and lower indices can be contracted,
and leads to the tensor components transforming
in a very specific way under basis changes.
This specific type of a tensor is called an
order-$(p,q)$ tensor \textit{on the vector space}~$V$.
However, here we adopt a more general definition, allowing $\{V_k\}_k$ and
$\{W_k\}_k$ to be all different vector spaces.

{    
    \bibliographystyle{unsrt}   
   \bibliography{qc.bib}
}



\end{document}